\documentclass[12pt]{article}
\usepackage[T1]{fontenc}
\usepackage[utf8]{inputenc}
\usepackage{amssymb}
\usepackage{amsmath}
\usepackage{amsfonts}
\usepackage{mathtools}
\usepackage{amsthm}
\usepackage{upgreek}
\usepackage{enumitem}

\usepackage{lmodern}
\usepackage{avant}

\usepackage{geometry}
\usepackage{hyperref}
\hypersetup{colorlinks=true,linkcolor=blue, citecolor=blue, urlcolor = blue}

\usepackage[authoryear]{natbib}
\usepackage{comment}
\usepackage{setspace}
\usepackage{mathptmx}

\usepackage{chngcntr}

\counterwithin*{equation}{section}

\theoremstyle{definition}
\newtheorem{theorem}{Theorem}
\newtheorem{definition}{Definition}
\newtheorem{assumption}{Assumption}
\newtheorem{lemma}{Lemma}
\newtheorem{example}{Example}

\newtheorem{proposition}{Proposition}
\newtheorem{remark}{Remark}
\newtheorem{corollary}{Corollary}
\newtheorem{exercise}{Exercise}

\newcommand\norm[1]{\left\lVert#1\right\rVert}

\DeclarePairedDelimiter{\floor}{\lfloor}{\rfloor}

\newcommand{\nc}{\newcommand}
\nc{\mbb}{\mathbb}\nc{\bb}{\mathbb}
\nc{\mbf}{\mathbf}\nc{\mb}{\mathbf}
\nc{\mc}{\mathcal}
\nc{\msf}{\mathsf}\nc{\ms}{\mathsf}
\nc{\acc}{\ms{acc}}
\nc{\ack}{\ms{ack}}
\nc{\alp}{\alpha}\nc{\al}{\alpha}\nc{\gka}{\alpha}
\nc{\ap}{\ms{ap}}
\nc{\apd}{\ms{apd}}
\nc{\base}{\ms{base}}\nc{\ba}{\ms{base}}
\nc{\bet}{\beta}\nc{\gkb}{\beta}
\nc{\boucle}{\ms{loop}}\nc{\Loop}{\ms{loop}}\nc{\lo}{\ms{loop}}
\nc{\bu}{\bullet}
\nc*{\cc}{\raisebox{-3pt}{\scalebox{2}{$\cdot$}}}
\nc{\centre}{\ms{center}}\nc{\Center}{\ms{center}}\nc{\cen}{\ms{center}}\nc{\ce}{\ms{center}}
\nc{\ci}{\circ}
\nc{\code}{\ms{code}}\nc{\cod}{\ms{code}}\nc{\decode}{\ms{decode}}\nc{\encode}{\ms{encode}}
\nc{\de}{:\equiv}
\nc{\dr}{\right}\nc{\ga}{\left}
\nc{\ds}{\displaystyle}
\nc{\ep}{\varepsilon}
\nc{\eq}{\equiv}
\nc{\ev}{\ms{ev}}
\nc{\fib}{\ms{fib}}
\nc{\funext}{\ms{funext}}\nc{\fu}{\ms{funext}}
\nc{\gam}{\gamma}
\nc{\glue}{\ms{glue}}\nc{\gl}{\ms{glue}}
\nc{\happly}{\ms{happly}}\nc{\ha}{\ms{happly}}
\nc{\id}{\ms{id}}
\nc{\ima}{\ms{im}}
\nc{\inc}{\subseteq}
\nc{\ind}{\ms{ind}}
\nc{\inl}{\ms{inl}}
\nc{\inr}{\ms{inr}}
\nc{\isContr}{\ms{isContr}}\nc{\co}{\ms{isContr}}\nc{\iC}{\ms{isContr}}\nc{\ic}{\ms{isContr}}
\nc{\isequiv}{\ms{isequiv}}\nc{\iseq}{\ms{isequiv}}\nc{\ieq}{\ms{isequiv}}
\nc{\ishae}{\ms{ishae}}\nc{\ish}{\ms{ishae}}\nc{\ih}{\ms{ishae}}
\nc{\isProp}{\ms{isProp}}\nc{\prop}{\ms{isProp}}\nc{\iP}{\ms{isProp}}\nc{\ip}{\ms{isProp}}
\nc{\isSet}{\ms{isSet}}\nc{\isS}{\ms{isSet}}\nc{\iss}{\ms{isSet}}\nc{\iS}{\ms{isSet}}\nc{\is}{\ms{isSet}}
\nc{\lam}{\lambda}
\nc{\LEM}{\ms{LEM}}\nc{\lem}{\ms{LEM}}\nc{\LE}{\ms{LEM}}
\nc{\lv}{\lvert}\nc{\rv}{\rvert}\nc{\lV}{\lVert}\nc{\rV}{\rVert}
\nc{\Map}{\ms{Map}}
\nc{\merid}{\ms{merid}}\nc{\meri}{\ms{merid}}\nc{\mer}{\ms{merid}}\nc{\me}{\ms{merid}}
\nc{\N}{\bb N}
\nc{\na}{\ms{nat}}
\nc{\nn}{\noindent}
\nc{\one}{\mb1}
\nc{\oo}{\operatorname}
\nc{\pd}{\prod}
\nc{\ps}{\mc P}
\nc{\pa}{\ms{pair}^=}
\nc{\ph}{\varphi}
\nc{\ppmap}{\ms{ppmap}}
\nc{\pr}{\ms{pr}}
\nc{\Prop}{\ms{Prop}}
\nc{\qinv}{\ms{qinv}}\nc{\qin}{\ms{qinv}}\nc{\qi}{\ms{qinv}}
\nc{\rec}{\ms{rec}}
\nc{\refl}{\ms{refl}}
\nc{\seg}{\ms{seg}}
\nc{\Set}{\ms{Set}}
\nc{\sm}{\scriptstyle}
\nc{\sms}{\ms s}
\nc{\sq}{\square}
\nc{\suc}{\ms{succ}}\nc{\su}{\ms{succ}}
\nc{\tb}{\textbf}
\nc{\then}{\Rightarrow}
\nc{\tms}{\ms t}
\nc{\tx}{\text}
\nc{\transport}{\ms{transport}}\nc{\tr}{\ms{transport}}
\nc{\two}{\mb2}
\nc{\Type}{\text-\ms{Type}}\nc{\type}{\text-\ms{Type}}\nc{\ty}{\text-\ms{Type}}
\nc{\U}{\mc U}
\nc{\ua}{\ms{ua}}
\nc{\uniq}{\ms{uniq}}
\nc{\univalence}{\ms{univalence}}
\nc{\vide}{\varnothing}
\nc{\ws}{\ms{sup}}
\nc{\zero}{\mb0}

\title{\textbf{Limit Theory under Network Dependence and Nonstationarity}\footnote{These lecture notes were prepared during my Ph.D. studies at the Department of Economics of the School of Economic, Social and Political Sciences, University of Southampton, Southampton SO17 1BJ, United Kingdom. This draft was updated during the academic year 2022-2023, Department of Economics, University of Exeter.\\

Dr. Christis Katsouris is a Lecturer in Economics, University of Exeter Business School, Exeter EX4 4PU, United Kingdom. Email: \textcolor{blue}{c.katsouris@exeter.ac.uk} }}

\author{\textbf{Christis Katsouris}\\ Department of Economics, University of Southampton}

\date{\today}

\begin{document}

\maketitle

\begin{abstract}
\vspace*{-0.8 em}
These lecture notes represent supplementary material for a short course on time series econometrics and network econometrics. We give emphasis on limit theory for time series regression models as well as the use of the local-to-unity parametrization when modeling time series nonstationarity. Moreover, we present various non-asymptotic theory results for moderate deviation principles when considering the eigenvalues of covariance matrices as well as asymptotics for unit root moderate deviations in nonstationary autoregressive processes. Although not all applications from the literature are covered we also discuss some open problems in the time series and network econometrics literature.    
\end{abstract}

\begin{figure}[h!]

\begin{center}
\includegraphics[width=5.5cm]{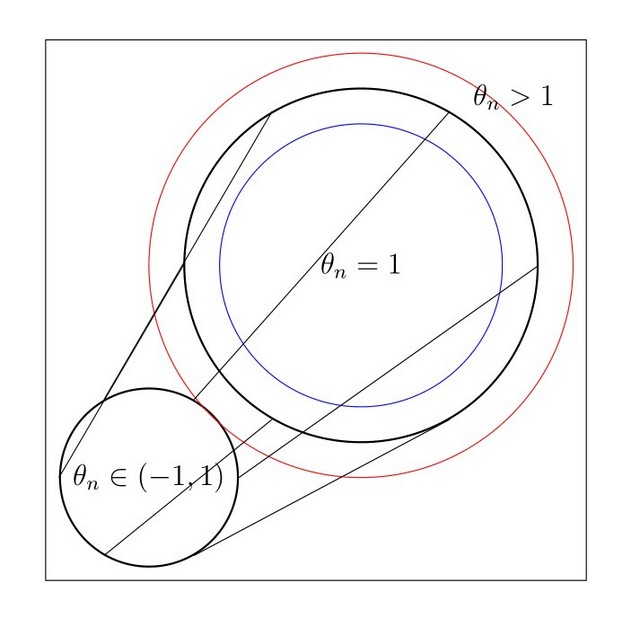}
\end{center}

\end{figure}


\maketitle

\newpage 
   
\begin{small}
\begin{spacing}{0.9}
\tableofcontents
\end{spacing}
\end{small}

\newpage

\section{Introduction}

This set of lecture notes is about systems of predictive regressions under network dependence. The first question that we will address is: \textit{why should we model network dependence?} Our answer to this issue is threefold. First, under both network dependence and time series nonstationarity (using a local-to-unity parametrization) conventional asymptotic theory results are not directly applicable unless we understand the required regularity conditions and modifications that are needed.

\subsection{Time Series as a Stochastic process}

Consider $\left\{ X_{\theta} \right\}_{ \theta \in \Theta }$, a stochastic process with index set $\Theta$. This is a family of random elements defined on a probability space $( \Omega, \mathcal{F}, \mathbb{P} )$. More formally, we have a mapping $X: \Theta \times \Omega \mapsto S$, and thus we can write $X( \theta, \omega )$ so that the element is a function of $\theta$ (possibly time or space) and $\omega$ is randomness.  As we shall see in the next section, we can write down a coordinate representation of the probability space so that $\Omega = R_{ \infty } \left(  \times_{0}^{\infty} R \right)$  such that $\omega = \left( x_0, x_1,... \right) \in R_{ \infty }$. Here $\omega$ is a typical trajectory or history of the process. We have the following natural maps
\begin{align}
\text{given} \ \omega : \ t \mapsto X(t, \omega) = x_t
\end{align} 
The above map picks off the $t'$th element and records a sample path as we change $t$. 
\begin{align}
\text{given} \ t : \ \omega \mapsto X(t, \omega) = X_t
\end{align} 
The above map produces a random variable that lives on $R$ (the state space) which in other words is a real valued random variable. This perspective is useful because it makes explicit what is the most distinctive feature of time series analysis. The fact that in making inferences about the probability law $\mathbb{P}$ that governs the process the evidence we have available takes the form of a single history $\omega$. In practise, of course, the situation is even worse: we only have a finite stretch of $\omega$ such as $\left( x_0, ..., x_n \right)$ on which to base inference.  

\begin{enumerate}
\item[(i)] \textbf{Weak dependence}: As we move along a history we accumulate new information about the process, because events separated by long enough stretches of time are nearly independent. This is what we mean by weak independence, asymptotic independence.  

\item[(ii)] \textbf{Stationarity}: The probability law remains unchanged as we move along in time. Thus, there is enough stability about the law to ensure that information within a given trajectory can be usefully employed to estimate some aspect of the data generating mechanism. The assumption of strict stationarity implies that the probability law is fixed and all marginal distributions are the same.   
\end{enumerate}

\newpage

\begin{theorem}
(Necessary and sufficient condition for ergodicity)
Let $\left( \Omega, \mathcal{F}, \mathbb{P} \right)$ is a probability space and $S : \Omega \times \Omega$ be a map on $\Omega$ and $S$ is ergodic iff
\begin{align}
\frac{1}{n} \sum_{ k = 0 }^{ n - 1} \mathbb{P} \left( F \subset S^{-k} G  \right)  \mapsto \mathbb{P} (F) \mathbb{P} (G), \ \ \ \forall \ \ F, G \in \mathcal{F}
\end{align}
\end{theorem}

\textit{Proof.}
Suppose S is ergodic and take the random variable $1_{G}$. Then, by the ergodic theorem, 
\begin{align}
\frac{1}{n} \sum_{k=0}^{n-1} 1_g \left( S^k \omega \right) \to E( 1_G ) = \int_{ \omega } dP = P(G), 
\end{align}

\begin{definition}
Let $\left\{ X_t \in \mathbb{N} \right\}$, for $\mathbb{N} = 0,1,2,...$ be a stochastic process with mean function
\begin{align}
\mathbb{E} \left[ X_t \right] = \mu, \ \ \ \text{for all} \ \ t \in \mathbb{N}    
\end{align}
Then, the sequence 
\begin{align}
\bar{X}_T = \frac{1}{T} \sum_{t=1}^T X_t, \ \ \ T = 1,2,...    
\end{align}
is said to be ergodic if and only if 
\begin{align}
\underset{ T \to \infty }{ \mathsf{lim} } \left( \bar{X}_T - \mu \right) = \underset{ T \to \infty }{ \mathsf{lim} } \mathbb{E} \left[ \left| \bar{X}_T - \mu \right|^2 \right] = 0.   
\end{align}
\end{definition}

\begin{definition}
The autocovariance of a time series at some lag or interval, $k$, is defined to be the covariance between $X_t$ and $X_{t+k}$ such that 
\begin{align}
\gamma_k = \mathsf{Cov} \left( X_t, X_{t+k} \right)  = \mathbb{E} \big[ \left( X_t - \mu \right) \left( X_{t+k} - \mu \right) \big]   
\end{align}    
\end{definition}

\begin{definition}
The autocorrelation of a time series is the standardization of the autocovariance of a time series relative to the variance of a time series, and the autocorrelation at lag $k$, $\rho_k$, is bounded between $+1$ and $-1$ such that
\begin{align}
\rho_k = \frac{ \mathbb{E} \big[ \left( X_t - \mu \right) \left( X_{t+k} - \mu \right) \big]  }{ \sqrt{ \mathbb{E} \left[ ( X_t - \mu )^2 \right] \mathbb{E} \left[ ( X_{t+k} - \mu )^2 \right]  } }    
\end{align}    
\end{definition}

\begin{remark}
Both the sample autocovariance (acf) and sample autocorrelation functions are considered as measured of dependency for stationary time-series data. A simple structural time series model, is the additive decomposition model. Usually, in financial time series there are often changes which are long term. Thus, the long term level of at time t is a stochastic process and is assumed to be a function of present and past values of $\left\{ X_t, X_{t-1},... \right\}$. At time $t$ there will also be a fluctuation component $f_t$ which represents new behaviour at time additional to the existing previous level at time $(t-1)$, which is assumed to be stationary, that is, none of its statistical properties change with time.
\end{remark}

\newpage

The above aspects are considered to be the fundamental assumptions in time series models. However, often these assumptions could vary depending on both the econometric model under consideration as well as on the empirical data application. For example, when modelling stock returns the presence of serial correlation  has an important interpretation in the financial economics literature. Specifically, incorporating serial correlation in models of stock prices and expected returns can provide additional evidence on the validity of the efficient market hypothesis and explain other market anomalies.

Further stylized facts of financial returns include the presence of heavy-tailedness and volatility clustering which implies that related econometric assumptions require the existence of fourth moments for the underline distribution function of the model innovations. Further assumptions can include the presence of heteroscedastic errors by imposing a parametric GARCH specification as well as other type of structures such as autoregressive errors. Although the main assumption regarding the innovation sequence of the the stationary AR$(1)$ time series model is that $u_t$ is an i.i.d sequence with mean zero and known variance, this assumption can be further relaxed to incorporate weakly or strongly dependent error sequences. By imposing a linear process representation then we can introduce the notation of weakly dependent errors.

Following the work of \cite{granger1978invertibility} for a linear stationary process it holds that 
\begin{align*}
\mathsf{Corr} \left( X^2_t, X^2_{t-k} \right) = [ \mathsf{Corr} \left( X_t, X_{t-k} \right) ]^2.
\end{align*}
for some time series $X_t$ across all $k$. 

Therefore, departures from the above expression indicate the presence of nonlinearity in time series. In particular, using the squared residuals from a linear model we can apply a standard Box-Ljung Portmanteau test for serial correlation (see, \cite{lee1993testing}). In other words, the BLP test is sensitive to departures from linearity in mean.

\begin{definition}
A process $\left\{ X_t \in \mathbb{N} \right\}$ is said to be a moving average process of order $m$ if it can be expressed
\begin{align}
X_t = \sum_{j=0}^m c_j \varepsilon_{t-j}    
\end{align}
Thus, an autoregressive process of order $p$, provided that certain conditions are satisfied concerning the roots of the associated polynomial, can be  represented by a moving average process of infinite order. 
\end{definition}

\begin{remark}
Notice that many time series exhibit "long memory", which implies that the autocorrelation function decays slowly with respect to the lag. In general such a time series characteristic traditionally has been modelled in the literature with unit roots (using nonstationary asymptotics), nonlinear dynamics (or regime switching) as well as structural breaks. 
\end{remark}

\subsection{Martingale Theorems}

\begin{theorem}[CLT for linear processes]
Consider the stationary linear process $X_t$ with $\left( c_j \right)$ satisfying  
\begin{align}
\frac{1}{\sqrt{n}} \sum_{t=1}^n \varepsilon_t \to_d \mathcal{N} \left( 0, \sigma^2 \right) \ \ \ \text{as} \ \ n \to \infty, \ \ \text{with} \ \ \ \sum_{i=0}^{\infty} c_i \neq 0.    
\end{align}
Then, it holds that 
\begin{align}
\frac{1}{\sqrt{n}} \sum_{t=1}^n X_t \to_d \mathcal{N} \left( 0, \omega^2 \right) \ \ \ \text{as} \ n \to \infty.   
\end{align}
\end{theorem}

\begin{theorem}[Martingale CLT] Let $\left\{ \xi_{n,j} : 1 \leq j \leq k_n, n \in \mathbb{N} \right\}$ be an $\mathcal{F}_{n,j}$ \textit{martingale difference array} that satisfies the Lindeberg condition. Then, it holds that 
\begin{align}
\sum_{j=1}^{k_n} \mathbb{E} \left[ \xi^2_{n,j} | \mathcal{F}_{n,j-1} \right] \overset{p}{\to} \eta^2 > 0, \ \ \ \text{as} \ n \to \infty,   
\end{align}
which implies that $\sum_{j=1}^{k_n} \xi_{n,j} \to_d Z \sim \mathcal{N} ( 0, \eta^2 )$. 
\end{theorem}

\begin{remark}
Notice that the sample mean of the linear process $X_t$ can be written as 
\begin{align*}
\frac{1}{\sqrt{n}} \sum_{t=1}^n X_t 
= \tilde{c}_{-1} \frac{1}{\sqrt{n}} \sum_{t=1}^n \varepsilon_t -  \frac{1}{\sqrt{n}} \sum_{t=1}^n \Delta \tilde{\varepsilon}_t     
&= 
\tilde{c}_{-1} \frac{1}{\sqrt{n}} \sum_{t=1}^n \varepsilon_t - \frac{1}{\sqrt{n}} \left( \tilde{\varepsilon}_n +  \tilde{\varepsilon}_0 \right)
\\
&= \tilde{c}_{-1} \frac{1}{\sqrt{n}} \sum_{t=1}^n \varepsilon_t + \mathcal{O}_p \left( \frac{1}{\sqrt{n}} \right).
\end{align*}
In other words, as long as the sequence $\varepsilon_t$ satisfies a CL(e.g., when $\varepsilon_t$ are \textit{i.i.d} $(0, \sigma^2 )$, it follows
\begin{align}
\frac{1}{\sqrt{n}} \sum_{t=1}^n \varepsilon_t \to_d \mathcal{N} \left( 0, \sigma^2 \right) \Rightarrow \frac{1}{\sqrt{n}} \sum_{t=1}^n X_t \to_d \mathcal{N} \left( 0, \omega^2 \right).     
\end{align}
where we have that 
\begin{align}
\omega^2 = \tilde{c}_{-1} \sigma^2 = \left( \sum_{j=0}^{\infty} c_j \right)^2 \sigma^2.    
\end{align}
is called the \textbf{long run variance} of $X_t$. The term emphasises the fact that, under short memory correlation, the variance of $X_t$ is
\begin{align}
\mathbb{E} \left( X_t^2 \right) = \underset{ n \to \infty }{ \mathsf{plim} } \frac{1}{n} \sum_{t=1}^n X_t^2 \neq \omega^2 = \underset{ n \to \infty }{ \mathsf{lim} } \ \mathbb{E} \left( \frac{1}{\sqrt{n}} \sum_{t=1}^n X_t \right)^2, 
\end{align}

\newpage

Thus, the probability limit of the WLLN for the sample variance is not equal to the asymptotic variance of the CLT as in the case of $\textit{i.i.d}$ random variables. Moreover, it can be shown that $\omega^2 = \sum_{ h = - \infty }^{\infty} \gamma (h)$, 
\begin{align}
\gamma (h) = \mathbb{E} \left( X_t X_{t-h} \right) = \sigma^2 \sum_{j=0}^{\infty} c_j c_{j + |h| }, \ \ \ h \in \mathbb{Z}.     
\end{align}
\end{remark}

\begin{example}(CLT for sums of linear processes under SM) Consider the linear process $X_t = \sum_{j=0}^{\infty} c_j \varepsilon_{t-j}$, 
\begin{align}
\sum_{j=0}^{\infty} | c_j | < \infty \ \ \ \text{and} \ \ \ \mathbb{E} \left( \varepsilon_t^2 \right) = \sigma^2 \in ( 0, \sigma ) \ \ \ \text{for all} \ t. 
\end{align}
Then, it holds that 
\begin{align}
\norm{ \frac{1}{ \sqrt{n} } \sum_{=1}^n - \left( \sum_{j=0}^{n-1} \frac{1}{\sqrt{n}} \varepsilon_t \right) }_{L_1} \to 0 \ \ \text{as} \ n \to \infty.    
\end{align}
To see this, we write 
\begin{align}
\frac{1}{\sqrt{n}} \sum_{t=1}^n X_t = \frac{1}{\sqrt{n}} \sum_{t=1}^n \sum_{j=0}^{t-1} c_j \varepsilon_{t-j} + \frac{1}{\sqrt{n}} \sum_{t=1}^n \sum_{j=t}^{\infty} c_j \varepsilon_{t-j} =: S_{n1} + S_{n1}.     
\end{align}
\end{example}

\begin{example}
\begin{align}
X_t = \rho X_{t-1} \epsilon_t, \ \ \ \epsilon_t \sim ( 0, \sigma^2 ).
\end{align}
For a stationary process we assume that the first two moments are constant, such that $\mathbb{E} ( X_t ) =  \mu$ and $\mathsf{Var} (X_t ) = \sigma^2$. Moreover, the covariance depends on $h$. 
Consider the martingale transform 
\begin{align}
M_n = \sum_{t=1}^n X_{t-1} \varepsilon_t
\end{align}
Moreover, it holds that $\mathbb{E} \left( M_n^2 \right) = O(n)$. Thus, consider
\begin{align}
\frac{1}{\sqrt{n}} \sum_{t=1}^n X_{t-1} \varepsilon_t = \sum_{t=1}^n \xi_{n,t} \ \ \ \text{with} \ \ \ \xi_{n,t} = \frac{1}{\sqrt{n}} X_{t-1} \varepsilon_t 
\end{align}
First we compute the probability limit of the conditional variance for the martingale CLT 
\begin{align}
\sum_{t=1}^n \mathbb{E} \left( \xi^2_{n,t} | \mathcal{F}_{t-1} \right) = \frac{1}{n} \sum_{t=1}^n  X_{t-1}^2 \mathbb{E} \left( \varepsilon^2_{t} | \mathcal{F}_{t-1} \right) = \frac{\sigma^2}{n} \sum_{t=1}^n X_{t-1}^2  \overset{p}{\to} \sigma^2 \gamma_X(0) = \frac{\sigma^4}{1- \rho^2}  \ \ \ \text{as} \ \ n \to \infty.
\end{align}
by $\mathcal{F}_{t-1}$ measurability of $X_{t-1}$ and WLLN for the sample ACF of linear processes. The conditional variance probability limit is given by $\tilde{\sigma} = \frac{  \displaystyle \sigma^4}{ \displaystyle 1 - \rho^2}$, which implies a distribution limit by the MCLT:

\newpage 

\begin{align}
\sum_{t=1}^n \xi_{n,t} = \frac{1}{\sqrt{n}} \sum_{t=1}^n X_{t-1} \varepsilon_t  \overset{ d }{\to }  \mathcal{N} \left( 0, \frac{ \sigma^4 }{ 1 - \rho^2 } \right)
\end{align}
To prove that indeed the above convergence in distribution holds, it remains to prove the Lindeberg condition for the CLT such that
\begin{align}
\mathcal{L}_n \left( \delta \right) = \frac{1}{n} \sum_{t=1}^n \mathbb{E} \big[ X_{t-1}^2 \varepsilon^2_t \mathbf{1} \left\{ | X_{t-1} | | \varepsilon_t | > \sqrt{n \delta} \right\} \big] \to 0.    
\end{align}
Assuming that the above result can be proved then, we can conclude that using the Slutsky theorem, the asymptotic distribution of the OLS estimator is given by
\begin{align}
\sqrt{n} \left( \hat{\rho}_n - \rho \right) = \frac{ \displaystyle  \frac{1}{\sqrt{n}} \sum_{t=1}^n X_{t-1} \varepsilon_t  }{  \displaystyle \frac{1}{n} \sum_{t=1}^n X_{t-1}^2 } \overset{d}{ \to} \frac{ \mathcal{N} \left(  0, \frac{\sigma^4}{1-\rho^2} \right) }{ \frac{\sigma^2}{1 - \rho^2} } =_d \mathcal{N} \left( 0, 1 - \rho^2 \right). 
\end{align}
\end{example}

\subsection{Covariance Matrix Estimation}

\subsubsection{Long-Run variance}

Denote with $\bar{w} = n^{-1} \sum_{t=1}^n w_t$. When the series is strict white noise (that is, zero mean \textit{iid}), then the central limit theorem implies that $n^{1/2} \bar{w} \overset{d}{\to} \mathcal{N} (0, V_w)$, where $V_w = \mathsf{Var} (w)$. Furthermore, when the series $w_t$ exhibit  \textit{serial dependence} but is a martingale difference with respect to its own past then this convergence in distribution still holds. However, if the series $w_t$ exhibits \textit{serial correlation}, the asymptotic normality still holds, but the asymptotic variance becomes
\begin{align}
V_w = \underset{ n \to \infty }{ \mathsf{lim} } \ \mathsf{Var}  \left( \sqrt{n} \bar{w} \right),     
\end{align}
which differs from the variance var$(w)$. Since we have that, 
\begin{align*}
V_w 
=
\underset{ n \to \infty }{ \mathsf{lim} } \ \mathsf{Var}  \left( \sqrt{n} \bar{w} \right)  
= 
\underset{ n \to \infty }{ \mathsf{lim} } \ \mathbb{E} \left[ \left( \sqrt{n} \bar{w} \right) \left( \sqrt{n} \bar{w} \right)^{\top} \right]
=
\underset{ n \to \infty }{ \mathsf{lim} } \ \mathbb{E} \left[ \frac{1}{n} \sum_{t=1}^n \sum_{s=1}^n  \bar{w}_t \bar{w}_s^{\top} \right]
=
\underset{ n \to \infty }{ \mathsf{lim} } \ \frac{1}{n} \sum_{t=1}^n \sum_{s=1}^n  \mathbb{E} \big[ \bar{w}_t \bar{w}_s^{\top} \big].
\end{align*}
Moreover, we denote the autocovariance of $w_t$ at delay $j$ by $\rho (j) = \mathbb{E} \big[ w_t w_{t-j}^{\top} \big]$. Using that $\rho (-j) = \rho(j)^{\top}$ and the Toeplitz lemma, we obtain 
\begin{align*}
V_w 
= 
\underset{ n \to \infty }{ \mathsf{lim} } \ \frac{1}{n} \left[ n \rho (0) + \sum_{j=1}^{n-1} (n-j) \big( \rho(j) + \rho(j)^{\top}  \big) \right] 
= 
\underset{ n \to \infty }{ \mathsf{lim} } \ \left[ n \rho (0) + \sum_{j=1}^{n-1} \left( 1 - \frac{j}{n} \right) \big( \rho(j) + \rho(j)^{\top}  \big) \right] 
= 
\sum_{j= - \infty}^{ \infty } \rho(j). 
\end{align*}

\newpage

Notice that the above expression for $V_w$ is termed the long-run variance. If instead of zero $\mathbb{E} [ w_t ]$ we consider possibly non-zero $\mathbb{E} \big[ m(w_t) \big]$ for some function $m(.)$, the long-run variance
\begin{align}
V_m = \underset{ n \to \infty }{ \mathsf{lim} } \mathsf{Var} \left( \sqrt{n} \sum_{t=1}^n m(w_t) \right)  
\end{align}
has the same representation as before, but the autocovariance of $m (w_t)$ at delay $j$ takes the form $\rho (j) := \mathbb{E} \big[ \big( m (w_t) - \mathbb{E} \left[ m (w_t) \right] \big)  \big( m (w_{t-j}) - \mathbb{E} \left[ m (w_t) \right] \big)^{\prime} \big]$. Then, the long-run covariance between, say, $m( w_t )$ and $\mu ( w_t )$ is defined similarly such that
\begin{align*}
C_{m, \mu} 
&:= 
\underset{ n \to \infty }{ \mathsf{lim} } \ \mathsf{Cov} \left( \frac{1}{ \sqrt{n} } \sum_{t=1}^n  m( w_t ),  \frac{1}{ \sqrt{n} } \sum_{t=1}^n  \mu( w_t )  \right)  
= 
\sum_{j= - \infty}^{ \infty } \gamma(j). 
\\
\gamma(j) 
&:=  \mathbb{E} \big[ \big( m (w_t) - \mathbb{E} \left[ m (w_t) \right] \big)  \big( \mu ( w_{t-j} ) - \mathbb{E} \left[ \mu ( w_{t-j} ) \right] \big)^{\prime} \big]   
\end{align*}
is the cross-covariance between $m ( w_t )$ and $\mu ( w_t )$ at delay $j$. Consistent estimation of the long-run variance is not straightforward. Specifically, each autocovariance $\rho(j)$ can be consistently estimated by 
\begin{align}
\widehat{\rho(j)} = \frac{1}{n} \sum_{t = j+1}^n \left( m (w_t) - \frac{1}{n} \sum_{s=1}^n m (w_s) \right) \left( m (w_{t-j}) - \frac{1}{n} \sum_{s=1}^n m (w_s) \right)    
\end{align}
for non-negative $j$, and by $\widehat{\rho(j)} = \left( \widehat{\rho (-j)} \right)^{\prime}$ for negative $j$. A natural estimator of $V_m$ is \begin{align}
\widehat{V}^0_m = \sum_{j = - (n-1)}^{n-1} \widehat{\rho(j)}.    
\end{align}
However, this estimator is not consistent because it is constructed using covariance estimates whose number increases proportionately with the sample size. Therefore, a possible solution is to truncate the summation at $-p$ and $p$, where $p < n - 1$ such that 
\begin{align}
\widehat{V}_m^{HH} = \sum_{j = - p}^p \widehat{\rho(j)}.     
\end{align}
Due to an ill-posed estimation problem, since the positive-definiteness property of the sample covariance matrix is not guaranteed, NW proposed the following modified version of the covariance matrix:
\begin{align}
\widehat{V}_m^{NW} = \sum_{j= -p}^p \left( 1 - \frac{ |j| }{p+1} \right)  \widehat{\rho(j)}.     
\end{align}
where the weights ensure that this NW estimator is positive semi-definite. The use of the long-run variance estimators is widespread in estimation and inference for models with dependent data. In particular, $V_n ( \bar{\theta} )$ is a consistent estimator of the long-run variance of the moment function $m ( w_t, \theta_0 )$

\newpage 

Define with  
\begin{align}
V_m = \underset{ n \to \infty }{ \mathsf{lim} } \ \mathsf{var} \big( \sqrt{n} m_n ( \theta_0 ) \big) = \sum_{ j = - \infty}^{ \infty }  \mathbb{E} \big[ m( w_t, \theta_0 ) m( w_{t-j}, \theta_0 )^{\top} \big].    
\end{align}
Then, the feasible optimal weighting matrix is then an estimate of $V_m$ which also uses a preliminary estimator $\bar{\theta}$ in place of $\theta_0$. Similarly, for the OLS estimator in a linear time series regression $y_t = x_t^{\prime} \theta_0 + e_t, \mathbb{E} \big[ e_t | x_t \big] = 0$, the asymptotic variance of $\hat{\theta}$ can be estimated using the long-run variance estimator of $x_t \hat{e}_t$, where $\hat{e}_t$ are the OLS residuals of the model. Moreover because, the long-run variance estimators automatically take into account the serial correlation and conditional heteroscedasticity in $e_t$, these estimators are often referred to as HAC estimators (see, \cite{newey1986simple}, \cite{andrews1991heteroskedasticity} and \cite{andrews1992improved}).

\subsubsection{Consistent Covariance Matrix Estimation for Linear Processes}

Consistency of kernel estimators of the long-run covariance matrix of a linear process is established under weak moment and memory conditions. 
\begin{align}
V_t = \sum_{ j = 0 }^{ \infty } C_j e_{t - j}. 
\end{align}
We consider the estimation of $\Omega = \displaystyle \underset{ T \to \infty }{ \text{lim} } \frac{1}{T} \sum_{t=1}^T \sum_{s=1}^T \mathbb{E} \left( V_t V_t^{\prime} \right)$, the long-run covariance matrix of $V_t$. 

Consistency of kernel estimators of $\Omega$ is established under weak conditions on $\left\{ C_j : j \geq 1 \right\}$ and $\left\{ e_t : t \in \mathbb{Z} \right\}$. We consider a kernel estimator of the form
\begin{align}
\hat{\Gamma}_T = T^{-1} \sum_{ t=s}^T \sum_{s=1}^{t-1} k \left( \frac{ | t - s | }{ b_T } \right) V_t V_s^{\prime}, 
\end{align}
where $k(.)$ is a measurable kernel function and $\left\{ b_T: T \geq 1 \right\}$ is a sequence of bandwidth parameters. The corresponding estimator of $\Omega$ is given by 
\begin{align}
\hat{\Omega}_T = T^{-1} \sum_{t=1}^T \sum_{s=1}^T  \left( \frac{ | t - s | }{ b_T } \right) V_t V_s^{\prime},
\end{align}
which can be written as $\hat{\Gamma}_T + \hat{\Gamma}_T^{\prime} + \hat{\Sigma}_T$, where $\hat{\Sigma}_T = \sum_{ t= 1}^T V_t V_t^{\prime}$. 
Notice that $\sqrt{n} \left( \hat{ \theta }_n - \theta \right) = \mathcal{O}_p(1)$ implies that $\hat{ \theta }_n \overset{p}{\to} \theta $ with a $\sqrt{n}-$rate of convergence. The order of convergence indicates the rate at which the distance of the sample estimator to the population parameter get smaller as the sample size increases, $n \to \infty$, i.e., $\hat{ \theta }_n = \mathcal{O}_p( \frac{1}{ \sqrt{n} })$. According to \cite{kiefer2000simple} and \cite{kiefer2002heteroskedasticity}, 
\begin{align*}
\sqrt{T} \left( \hat{\beta} - \beta_0 \right) 
\Rightarrow 
Q^{-1} \Lambda W_k(1) \sim \mathcal{N} \big( 0, Q^{-1} \Lambda \Lambda^{\prime} Q^{-1} \big) 
\equiv 
\mathcal{N} \big( 0, Q^{-1} \Omega Q^{-1} \big) = \mathcal{N} (0, V ). 
\end{align*}

\newpage

In other words, the asymptotic distribution is a $k-$variate normal distribution with mean 0 and variance-covariance matrix $V = Q^{-1} \Omega Q^{-1}$. Then, the asymptotic distribution of $\hat{\beta}$ can be used to test hypothesis about $\beta$. However, to do this an estimate of $V$ is required. In particular, a natural estimate of $Q$ is 
\begin{align}
\widehat{Q} = \frac{1}{T} \boldsymbol{F}^{\prime} ( \hat{\beta} ) W \boldsymbol{F}^{\prime} ( \hat{\beta} ) 
\end{align} 
Therefore, $\Omega$ can be estimated by a HAC estimator, $\widehat{\Omega}$. Letting $\hat{u}_t$ be the residuals of the transformed model, the HAC estimate would use $\hat{v}_t = \tilde{F}_t ( \hat{\beta} ) \hat{u}_t$ to estimate nonparametrically the spectral density of $v_t$ at frequency 0, and hence $\Omega$. A typical estimator takes the form 
\begin{align}
\widehat{\Omega} &= \sum_{ j = - (T-1) }^{ T - 1 } k \left( \frac{j}{s(T) } \right) \widehat{\Gamma}_j.  
\\
\widehat{\Gamma}_j &= \frac{1}{T} \sum_{t=j+1}^T \hat{v}_t \hat{v}_{t-j}^{\prime}, \ \ \text{for} \ \ j \geq 0.  
\\
\widehat{\Gamma}_j &= \frac{1}{T} \sum_{ t = -j+1 }^T \hat{v}_t \hat{v}_{t-j}^{\prime}, \ \ \text{for} \ \ j \geq 0.  
\end{align}
where $k(x)$ is a kernel function satisfying $k(x) = k(-x)$, $k(0) = 1$, $| k(x) | \leq 1$ continuous at $x = 0$ and $\int_{- \infty}^{ \infty } k^2 (x) dx < \infty$. Moreover, the tuning parameter, $s(T)$, is often called the truncation lag or bandwidth. A typical condition for consistency of $\widehat{\Omega}$ is that $s(T) \to \infty$ as $T \to \infty$ but $s(T) / T \to 0$. Therefore, to test hypotheses about $\beta$ in the standard approach, $\widehat{V} = \widehat{Q}^{-1} \widehat{\Omega} \widehat{Q}^{-1}$ is used to transform $\sqrt{T} \left( \hat{\beta} - \beta_0 \right)$ to obtain 
\begin{align*}
\widehat{V}^{ - \frac{1}{2} } \sqrt{T} \left( \hat{\beta} - \beta_0 \right) 
\Rightarrow
\Omega^{ - \frac{1}{2} } Q Q^{-1} \Lambda W_k(1) 
= 
W_k(1) \sim \mathcal{N} \big( 0, \boldsymbol{I}_k \big). 
\end{align*}
In a new approach \cite{hong2023kolmogorov} (see also \cite{sun2022adjusted}), use a similar method, but transform $\sqrt{T} \left( \hat{\beta} - \beta_0 \right)$ in such a manner that the asymptotic distribution no longer depends on unknown parameters. The essential difference between the two approaches is that their approach does not require an explicit estimate of $\Omega$ and takes the additional sampling variation associated with not knowing the covariance matrix into account in the asymptotic approximation. On the other hand, HAC estimates treat the variance-covariance matrix as known asymptotically and ignore the impact of finite-sample variability from $\widehat{\Omega}$ on the distribution of test statistics. Thus, a practical limitation of the HAC approach is that to obtain HAC estimates, a truncation lag for a spectral density estimator must be chosen. Although asymptotic theory dictates the rate at which the truncation lag much increase as the sample size grows, no concrete guidance is provided. Consider the following scaled partial sum empirical process:
\begin{align}
T^{- 1 / 2} \widehat{S}_{ \floor{Tr} } 
=
T^{- 1 / 2} \sum_{t=1}^{ \floor{Tr} } \hat{v}_t 
= 
T^{- 1 / 2} \sum_{t=1}^{ \floor{Tr} } \tilde{F}_t \left( \hat{\beta} \right) \hat{u}_t.
\end{align}

\newpage 

Their approach provides an elegant solution to this practical problem by avoiding a consistently estimate of the variance-covariance matrix and thus removing the need to choose a truncation lag. In other words, a data-dependent transformation is applied to the NLS estimates of the parameters of interest. This transformation is chosen such that it ensures that the asymptotic distribution of the transformed estimator does not depend on nuisance parameters. The transformed estimator can then be used to construct a test for general hypotheses on the parameters of interest.

\subsection{Continuous Asymptotics}

Before considering the continuous time asymptotics we consider the following example which can be found on page 328 in the book of \cite{csorgo1997limit}. 

\begin{example}
\begin{align}
X_t = \rho X_{t-1} + \epsilon_t, 
\end{align}
Let $| \rho | \neq 1$, then the partial sums of the residuals $\displaystyle \sum_{ t = 1 }^n \hat{\epsilon}_t$ have asymptotic behaviour similar to the partial sums of the \textit{i.i.d} errors. On the other hand when the autocorrelation coefficient is within the unit root boundary, i.e., $| \rho | = 1$, then the limit of the suitably normalized residual sum $\displaystyle \sum_{ t = 1 }^n \hat{\epsilon}_t$ can be proved that is not a Wiener process. The particular result proved in various studies appeared in the literature is of particular importance when considering persistence processes with predictive regression models.    

Consider $X_t = S(t)$ the corresponding partial sum process such that  $S(t) = \sum_{ j = 1 }^t \epsilon_j$. Furthermore, consider that $\hat{\rho}$ is the maximum likelihood estimator for $\rho = 1$, then we have that 
\begin{align}
\hat{\rho} - 1 = \frac{ \displaystyle \sum_{ t = 1 }^T \epsilon_j S( t -1 )  }{ \displaystyle \sum_{ t = 1 }^T S^2( t -1 ) }.      
\end{align} 
Hence, we have that the partial sum process is given by
\begin{align*}
\sum_{ j = 1 }^t \hat{ \epsilon }_j 
= 
\big( 1 - \hat{\rho}    \big) \sum_{ j = 1 }^k S(j-1) + \sum_{ j = 1 }^k \epsilon_j
= S(k) - \frac{ \displaystyle \sum_{ j = 1 }^n \epsilon_j S(j-1) }{ \displaystyle \sum_{ j = 1 }^n  S^2(j-1) } \sum_{ j = 1 }^k S(j-1). 
\end{align*}
Since, we have that 
\begin{align}
\frac{1}{ \sqrt{n} \uptau } S( nt ) \to W(t), \ \  \text{where} \ \ \big\{ W(t), 0 \leq t \leq 1 \big\} \ \text{is a Wiener process}.
\end{align}
\end{example}

\newpage

Consider a triangular array of random variables $\left\{ \left\{ y_{nt} \right\}_{t=1}^{T_n} \right\}_{n = 1}^{ \infty }$. The triangular array then provides a formal framework within which $h$ may vary and by means of which we may investigate limiting behaviour as $h_n \to 0$. Let $S_{ ni } = \sum_{ j = 1 }^i u_{nj}$ such that $1 \leq i \leq T_n$ with $S_{n0} = 0$. 

Furthermore, we form the random function as below
\begin{align}
Y_n &= \sigma^{-1} S_{ni - 1}, \ \ (i-1) / T_n \leq t \leq i / T_n, \ \ i = 1,...,T_n
\\
Y_n(1) &= \sigma^{-1} S_{n T_n } 
\end{align}
As $n \to \infty, Y_n(r)$ converges weakly to a constant multiple of a standard Wiener process.

\medskip

\begin{lemma}
If the following conditions hold: 
\begin{itemize}
\item[(a)] $\left\{ \left\{ u_{nt} \right\} \right\}$ is a triangular array of i.i.d $(0, \sigma^2 h_n )$,

\item[(b)]  $T_n \in Z^{+}$, $T_n \to \infty$, and $h_n \to 0$ as $n \to \infty$ in such a way that the product $T_n h_n = N > 0$ remains constant, 
\end{itemize}
then $Y_n ( r ) \Rightarrow N^{1 / 2} W(r)$ as $n \to \infty$ where $W(r)$ as $n \to \infty$, where $W(r)$ is a standard Wiener process.   
\end{lemma}
We define with 
\begin{align}
\hat{ \alpha }_n  = \frac{ \displaystyle \sum_{t = 1}^{ T_n } y_{nt} y_{nt-1} }{ \displaystyle \sum_{ t = 1 }^{ T_n } y_{nt-1}^2 }
\end{align}
Then, we have the t-statistic given by 
\begin{align}
t_{ \alpha_n } = \left( \sum_{t=1}^{ T_n } y_{nt-1}^2 \right) \left( \hat{\alpha}_n - 1 \right) / s_n
\end{align}
where an unbiased estimator for the variance is given by 
\begin{align}
s_n = \left\{ T_n^{-1} \sum_{ t = 1 }^{ T_n } \left( y_{nt} - \hat{\alpha}_n y_{nt-1} \right)^2  \right\}^{1 / 2}
\end{align}
Furthermore, we have that
\begin{align}
\int_0^1 X_T(r)^2 dr \equiv \int_0^1 W(r)^2 dr + O_p( T^{-1} ) 
\end{align}
and this expansion may be verified directly by developing an expansion for the characteristic function of $\int_0^1 X_T (r)^2 dr$.

\newpage

\begin{theorem}
If $y_t$ is generated from the random walk with $\alpha = 1$ and initial value $y_0 = 0$ and if the $u_i$ are i.i.d $\mathcal{N} (0,1)$, then we have that 
\begin{align}
T \left( \hat{\alpha} - 1 \right) \equiv \frac{ \displaystyle (1 / 2) \left( W(1)^2 - 1 \right) - \left( 1 / \sqrt{2 T } \right) \xi }{ \displaystyle \int_0^1 W(r)^2 dr }
\end{align}
where $W(r)$ is a standard Wiener process and $\xi$ is $\mathcal{N} (0,1)$ and independent of $W(r)$.
\end{theorem}

\subsection{The Ornstein–Uhlenbeck process}

Hazard rates processes are functionals of piecewise-continuous Gaussian martingales, such as the OU processes and the pure jump process.

\paragraph{Modelling the item state by a stochastic process}

\begin{align}
X(n+1) - X(n) = \sigma \left( X(n) \right) \epsilon_n + \mu \left( X(n) \right) h, 
\end{align}

The resulting stochastic differential equation is given by 
\begin{align}
d X(t) = \sigma X(t) d \gamma (t) + \mu \left[ X(t) \right] dt, 
\end{align}

\paragraph{Modelling item failure rate by a stochastic process}

Dynamic environments exert random stresses on the item. An approach for modelling the lifelength of items in a dynamic environment is to describe the item's failure rate by a stochastic process. Moreover, there is an interesting relation between a hazard rate process and a doubly stochastic Poisson process (known as the Cox process). 

Let T denote the lifelength of the item and let $\left\{ \lambda(s), s \leq 0 \right\}$ be a non-negative, real-valued, right-continuous process.Then, $\left\{ \lambda(s), s \leq 0 \right\}$ is said to be the hazard rate process of T if
\begin{align}
\mathbb{P} \big( T \leq t | \lambda(s), 0 \leq s \leq t \big) = \text{exp} \left\{ - \int_0^t \lambda (s) ds \right\}
\end{align}  
Consequently, 
\begin{align}
\mathbb{P} \big( T \leq t \big) = \mathbb{E} \left[ \text{exp} \left\{ - \int_0^t \lambda (s) ds \right\} \right], t \geq 0
\end{align} 

\newpage 

\subsection{Higher-Order Autoregressive Processes}

\subsubsection{AR(p) Time Series Regression}

An autoregressive process of order AR$(p)$ is defined by 
\begin{align}
X_t = \theta_1 X_{t-1} + ... + \theta_p X_{t-p} + \varepsilon_t
\end{align}
for any $t \geq 1$ or equivalently, in a compact form we have that $X_t = \theta^{\top} \Phi_{t-1} + \varepsilon_t$, 
where $\theta = \big( \theta_1,..., \theta_p \big)^{\top}$ is a vector parameter, $\Phi_0$ is an arbitrary initial random vector and $\Phi_t = \big( X_t,..., X_{t-p+1}  \big)^{\top}$ and $\varepsilon_t$ is a strong white noise having a finite variance $\sigma^2$. Then, the corresponding characteristic polynomial is defined by 
\begin{align}
\Theta (z) = 1 - \theta_1 z - ... - \theta_p z^p \ \ \ \text{and} \ \ \ \Phi_t = C_{\theta} \Phi_{t-1} + E_t,
\end{align}
where $E_t = \big( \varepsilon_t, 0, ..., 0 \big)^{\top}$ is a $p-$dimensional noise. Notice that it is well-known that the stability of this $p-$dimensional process is closely related to the eigenvalues of the companion matrix that we will denote and arrange in increasing order as below
\begin{align}
\rho ( C_{\theta} ) = | \lambda_1 | \geq | \lambda_2 | \geq ... \geq | \lambda_p |.
\end{align}

\subsubsection{The Autoregressive and Moving Average Processes (ARMA)}

A stochastic process, or time series, can be repeated as the output resulting from a white noise input, $\varepsilon_t$, \begin{align*}
X_t = \varepsilon_t + \phi_1 \varepsilon_{t-1} + \phi_2 \varepsilon_{t-2} + ...  \Rightarrow X_t = \varepsilon_t + \sum_{j=1}^{\infty} \phi_j \varepsilon_{t-j}
\end{align*}
where $X_t$ is considered to be the output series after demeaning e.g., $X_t = \tilde{X}_t - \mu$. In other words, the general linear process representation allows to represent the output of a time series, as a function of the current and previous value of the white noise process, $\varepsilon_t$, which may be represented as a series of shocks. The autocorrelation function of a linear process is given by 
\begin{align}
\gamma_k = \sigma_{\varepsilon}^2 \sum_{j=0}^{\infty} \phi_j \phi_{j+k} 
\end{align}
Thus, the backward shift operator, $B$, is defined as $B X_t = X_{t-1}$ and $B^j X_t = X_{t-j}$. The autocorrelation generating function may be written as $
\gamma (B) = \sum_{j= - \infty}^{+ \infty } \gamma_k B^k$. Therefore, the strict stationarity assumption holds if and only if the $\phi$ weights of a linear process must satisfy that $\phi(B)$ converges on or lies within the unit circle. Then, an autoregressive, AR, model, the current value of the time series and a random shock $\varepsilon_t$, such that 
\begin{align}
X_t = \phi_1 X_{t-1} + \phi_2 X_{t-2} + ... + \phi_p X_{t-p} + \varepsilon_t    
\end{align}

\newpage

Therefore, the autoregressive operator of order $P$ is given by 
\begin{align}
\phi(B) = 1 - \phi_1 B^2 - \phi_2 B^2 - ... - \phi_p B^p.    
\end{align}

\begin{example}

Consider the following $m-$dimensional vector ARMA model (see, the study of \cite{gronneberg2019partial}) 
\begin{align}
\boldsymbol{\Phi} (B) \boldsymbol{X}_t = \boldsymbol{\Theta} (B) \boldsymbol{\epsilon}_t,
\end{align}
where $B$ is the back-shift operator, and $\boldsymbol{\Phi} (B)$ and $\boldsymbol{\Theta} (B)$ are defined as below 
\begin{align}
\boldsymbol{\Phi} (B) = \boldsymbol{I} - \sum_{i=1}^p \boldsymbol{\Phi}_i B^i, \ \ \ \ \ \boldsymbol{\Theta} (B) = \boldsymbol{I} + \sum_{j=1}^q \boldsymbol{\Theta}_j B^j
\end{align}
where $\boldsymbol{\Phi}_i$ and $\boldsymbol{\Theta}_j$ are $( m \times m )$ matrices with $1 \leq i \leq p$ and $1 \leq j \leq q$. Therefore, we denote the unknown parameters such that $\boldsymbol{\beta} = \mathsf{vec} \big( \boldsymbol{\Phi}_1,..., \boldsymbol{\Phi}_p, \boldsymbol{\Theta}_1,..., \boldsymbol{\Theta}_q \big)$ and its true parameters by 
\begin{align}
\boldsymbol{\beta}_0 = \mathsf{vec} \big( \boldsymbol{\Phi}_{01},..., \boldsymbol{\Phi}_{0p}, \boldsymbol{\Theta}_{01},..., \boldsymbol{\Theta}_{0q} \big)
\end{align}

Furthermore, assume that $\boldsymbol{\beta} \in \mathcal{C}$, a compact subspace in $\mathbb{R}^{ ( p + q ) m^2 }$, and $\boldsymbol{ \beta}_0$ is an interior point in $\mathcal{C}$.   

\medskip
   
\begin{assumption}
All roots of $\mathsf{det} \big[ \boldsymbol{\Phi} (z) \big] = 0$ and all the roots of $\mathsf{det} \big[ \boldsymbol{\Theta} (z) \big] = 0$ lie outside the unit circle. 
\end{assumption}   

In particular, the above assumption is a sufficient condition for stationarity and invertability, and it implies that the time series $\left\{ \boldsymbol{X}_t \right\}$ has the following representation:
\begin{align}
\boldsymbol{X}_t = \boldsymbol{\Psi} (B) \boldsymbol{\epsilon}_t,  
\end{align}  
where $\boldsymbol{\Psi}(z) := \boldsymbol{\Phi}^{-1} (z) \boldsymbol{\Theta} (z) = \boldsymbol{I} + \sum_{j=1}^{\infty} \boldsymbol{\Psi}_j z^j$ with $\boldsymbol{\Psi}_i = \mathcal{O}_p \left( \rho^i \right)$ for some $\rho \in (0,1)$. Then, the periodgram function of $\left\{ \boldsymbol{X}_1,..., \boldsymbol{X}_n \right\}$ is defined as below 
\begin{align}
\boldsymbol{I}_x 
= 
\boldsymbol{J}_x ( \lambda ) \boldsymbol{J}^{\prime}_x ( - \lambda )
= 
\sum_{ | k | < n } \boldsymbol{\Gamma}_X ( k ) e^{ - i \lambda k }, 
\end{align} 
where $\lambda \in [ - \pi, \pi ]$ and $\boldsymbol{J}_x ( \lambda ) = \sum_{t=1}^n \boldsymbol{X}_t e^{- i \lambda t} / \sqrt{n}$, and for $k \geq 0$,  
\begin{align}
\Gamma_X (k) 
= 
\sum_{ t = 1 }^{ n - k } \boldsymbol{X}_{t+k} \boldsymbol{X}_t^{\prime} / n. 
\end{align}

\end{example}

\newpage 

Furthermore, the device proposed by\cite{phillips1992asymptotics}, is suitable to derive invariance principles and functional central limit theorems in different econometric environments. In particular, \cite{gronneberg2019partial} focus on establishing limit results for partial sum processes based on ARMAX residuals. The following lemma from \cite{gronneberg2019partial} adapts the Phillips-Solo device although unlike \cite{phillips1992asymptotics}, the authors do not impose a time dependence structure, and allow two dimensional arrays. 
\begin{lemma}[\cite{gronneberg2019partial}]
Let $C(L):= \sum_{k=0}^{\infty} c_k L^k$ be a lag series with uniformly geometrically decreasing coefficients, that is, there exists $M \in \mathbb{R}$ and $\rho \in [0,1]$ such that for all $k \in \mathbb{N}$, $| c_k | M \rho^k$. If a constant $\epsilon > 0$, then 
\begin{align}
\underset{ ( n, t ) \in [ 1, \infty ] \times \mathbb{Z} : t \leq n }{ \mathsf{sup} } \ \mathbb{E} | Z_{n,t} |^{ 1 + \epsilon } < \infty
\end{align}
\end{lemma}
Therefore, by the Beveridge-Nelson decomposition as stated in Lemma 1 of \cite{phillips1992asymptotics}, 
\begin{align}
\tilde{Z}_{n,t} := \sum_{j=0}^{\infty} \left( \sum_{k=j+1}^{\infty} c_k \right) Z_{n,t-j}, \ \ \text{as} \ \ n \to \infty,
\end{align}
\begin{align*}
C(B) Z_{n,t} 
&= 
C(1) Z_{n,t} - (1 - B) \tilde{Z}_{n,t} = C(1) Z_{n,t} - \left( \tilde{Z}_{n,t} -  \tilde{Z}_{n,t-1} \right)
\\
&\Rightarrow
\frac{1}{n} \sum_{t=1}^{ \floor{ \lambda n} } C(B) Z_{n,t} = C(1) \left[ \frac{1}{n} \sum_{t=1}^{ \floor{ \lambda n} } Z_{n,t} \right] + \frac{ \tilde{Z}_{n,0} }{n} - \frac{ \tilde{Z}_{n, \floor{ \lambda n }} }{n}  
\\
&\Rightarrow
\underset{ \lambda \in [0,1] }{ \mathsf{sup} } \left| \frac{1}{n} \sum_{t=1}^{ \floor{ \lambda n} } C(B) Z_{n,t} - C(1) \frac{1}{n} \sum_{t=1}^{ \floor{ \lambda n} } Z_{n,t} \right|
= \underset{ \lambda \in [0,1] }{ \mathsf{sup} } \left| \frac{ \tilde{Z}_{n,0} }{n} - \frac{ \tilde{Z}_{n, \floor{ \lambda n }} }{n} \right|
\\
&\leq 
\frac{2}{n} \underset{ \lambda \in [0,1] }{ \mathsf{sup} } \left| \tilde{Z}_{n,t} \right| = o_p(1). 
\end{align*}
The second result is stated as below
\begin{align}
\underset{ \lambda \in [0,1] }{ \mathsf{sup} } \left| \frac{1}{n} \sum_{t=1}^{ \floor{ \lambda n } } \sum_{j=0}^{t} C(B) Z_{n,j} - C(1) \frac{1}{n} \sum_{t=1}^{ \floor{ \lambda n } } \sum_{j=0}^{t} Z_{n,j} \right| = o_p(1), \ \ \ \text{as} \ \ n \to \infty.
\end{align}
Denote with $\tilde{Z}_{n,t} := \sum_{j=0}^{\infty} \left( \sum_{k=j+1}^{\infty} c_k \right) Z_{n,t-j}$. Then, we have that 
\begin{align*}
C(B) Z_{n,t} 
&= 
C(1) Z_{n,t} - (1 - B) \tilde{Z}_{n,t} = C(1) Z_{n,t} - \left( \tilde{Z}_{n,t} -  \tilde{Z}_{n,t-1} \right)
\\
&\Rightarrow
\frac{1}{n} \sum_{t=1}^{ \floor{ \lambda n} } \sum_{j=0}^{t} C(B) Z_{n,t} = C(1) \left[ \frac{1}{n} \sum_{t=1}^{ \floor{ \lambda n} } \sum_{j=0}^{t} Z_{n,t} \right] + \frac{ \tilde{Z}_{n,0} }{n} - \frac{ \tilde{Z}_{n, \floor{ \lambda n }} }{n}  
\\
&\Rightarrow
\underset{ \lambda \in [0,1] }{ \mathsf{sup} } \left| \frac{1}{n} \sum_{t=1}^{ \floor{ \lambda n} } \sum_{j=0}^{t} C(B) Z_{n,t} - C(1) \frac{1}{n} \sum_{t=1}^{ \floor{ \lambda n} } \sum_{j=0}^{t} Z_{n,t} \right|
= \underset{ \lambda \in [0,1] }{ \mathsf{sup} } \left| \frac{ \tilde{Z}_{n,0} }{n} - \frac{ \tilde{Z}_{n, \floor{ \lambda n }} }{n} \right|
\leq 
\frac{2}{n} \underset{ \lambda \in [0,1] }{ \mathsf{sup} } \left| \tilde{Z}_{n,t} \right| = o_p(1). 
\end{align*}

\newpage

The above result holds due to the following property
\begin{align}
\frac{1}{n} \sum_{t=1}^{ \floor{ \lambda n} } \sum_{j=0}^{t} \left( \tilde{Z}_{n,t} -  \tilde{Z}_{n,t-1} \right) = \frac{ \tilde{Z}_{n,t} }{ n } - \frac{ \tilde{Z}_{n, \floor{ \lambda n  } } }{ n }
\end{align}
Let $\eta > 0$ be given. Then, by definition of supremuem, as $n \to \infty$, we have that 
\begin{align*}
\mathbb{P} \left( \frac{1}{n} \underset{ t \in [0,n] }{ \mathsf{sup} } \ \left| \tilde{Z}_{n,t} \right| > \eta \right)
&= 
\mathbb{P} \left( \bigcup_{t=0}^n \left\{ \frac{1}{n} \left| \tilde{Z}_{n,t} \right| > \eta \right\} \right)
\\
&
\leq
\sum_{t=0}^n \mathbb{P} \bigg( \frac{1}{n} \left| \tilde{Z}_{n,t} \right| > \eta \bigg) 
= 
\sum_{t=0}^n \mathbb{P} \bigg(  \left( \frac{1}{n} \left| \tilde{Z}_{n,t} \right|  \right)^{ 1 + \epsilon } > \eta^{ 1 + \epsilon } \bigg)
\\
&=
\sum_{t=0}^n \mathbb{P} \bigg( \frac{1}{ n^{1 + \epsilon} } \left| \tilde{Z}_{n,t} \right|^{ 1 + \epsilon } > \eta^{ 1 + \epsilon} \bigg)
\\
&\leq
\left( \underset{ (n,t) \in [ 1, \infty ] \times \mathbb{N}: t \leq n }{ \mathsf{sup} } \ \mathbb{E} \left| \tilde{Z}_{n,t} \right|^{ 1 + \epsilon } \right) \frac{1}{\eta^{ 1 + \epsilon}} \sum_{t=0}^n \frac{1}{ n^{1 + \epsilon} }
\\
&=
\left( \underset{ (n,t) \in [ 1, \infty ] \times \mathbb{N}: t \leq n }{ \mathsf{sup} } \ \mathbb{E} \left| \tilde{Z}_{n,t} \right|^{ 1 + \epsilon } \right) \frac{1}{\eta^{ 1 + \epsilon}} \frac{2}{ n^{\epsilon} } = o(1).
\end{align*}

\medskip

\begin{lemma}
(Stability of uniform partial-sum LLN under linear filtering). 
Let $\left( \tilde{S}_{n,t} \right)$ be a stochastic array
\begin{itemize}

\item[(a)] for a constant $\epsilon > 0$, $\underset{ ( n, t ) \in [ 1, \infty ] \times \mathbb{Z} : t \leq n }{ \mathsf{sup} } \ \mathbb{E} | Z_{n,t} |^{ 1 + \epsilon } < \infty$, and that  

\item[(b)] $\underset{ \lambda \in [0,1] }{ \mathsf{sup} } \left| \frac{1}{n} \sum_{t=1}^{ \floor{ \lambda n}  } \tilde{S}_{n,t} \right| = o_{p}(1)$.

\end{itemize}
Then, for all sequences $\left( \omega_j \right)_{ j \in \mathbb{N} }$ such that there exist $M > 0$ and $\rho \in [0,1]$ for all $j \in \mathbb{N}$ such that $| \omega_j | \leq M \rho^j$ and
\begin{align}
\underset{ \lambda \in [0,1] }{ \mathsf{sup} } \ \left| \frac{1}{n} \sum_{t=1}^{ \floor{ \lambda n} } \sum_{j=0}^{\infty} \omega_j \tilde{S}_{n,t-j}  \right| = o_p(1) \ \ \text{as} \ \ n \to \infty.
\end{align}
\end{lemma}
Let $\Gamma(B) = \sum_{j=0}^{\infty} \omega_j B^j$ such that 
\begin{align}
\frac{1}{n} \sum_{t= 1}^{ \floor{ \lambda n} } \sum_{j=0}^{\infty} \omega_j \tilde{S}_{n,t-j} 
=
\frac{1}{n} \sum_{t= 1}^{ \floor{ \lambda n} } \Gamma(B) \tilde{S}_{n,t}     
\end{align}
where $\Gamma(B)\tilde{S}_{n,t}$ is well-defined. Then, addition and subtraction of $\Gamma(1) \sum_{t= 1}^{ \floor{ \lambda n} }\tilde{S}_{n,t}$ with the triangle inequality yields the following result

\newpage 

\subsection{Moderate Deviations for Multivariate Martingales}

According to \cite{grama2006asymptotic}, limit theorems for probabilities of moderate deviations for sums of independent random variables have been extensively examined before. For martingales, however the number of results on moderate deviations in the literature is still rather limited. Moderate deviation results are related to bounds on the rate of convergence in the central limit theorem. In the multivariate case, quadratic characteristics are matrix valued processes. Therefore, to establish uniform estimates of the rate of convergence in the multivariate martingale CLT, moderate deviations require special handling. 

\paragraph{Univariate Case}

Let $\left( \xi_{nk}, \mathcal{F}_{nk} \right)_{ 0 \leq k \leq n}$ be a square integrable martingale difference sequence with $\xi_{n0} = 0$. We set
\begin{align}
X_k^n = \sum_{ i = 1}^k \xi_{ni}, \ \ 1 \leq k \leq n.
\end{align}

Denote by $\langle X^n \rangle$ the quadratic characteristic of the martingale $X^n$, such that 
\begin{align}
\langle X^n \rangle_k = \sum_{i = 1}^k a_{ni}, \ \ a_{nk} = \mathbb{E} \left( \xi_{nk}^2 | \mathcal{F}_{n,k-1} \right), \ \ 1 \leq k \leq n.  
\end{align}

\paragraph{Multivariate Case}

Denote by $\langle X^n \rangle$ the quadratic characteristic of the martingale $X^n$, i.e.,  
\begin{align}
\langle X^n \rangle_k = \sum_{i = 1}^k a_{ni}, \ \ a_{nk} = \mathbb{E} \left( \xi_{nk} \xi_{nk}^{\prime} | \mathcal{F}_{n,k-1} \right), \ \ 1 \leq k \leq n.  
\end{align}

\begin{example}
Consider the autoregressive model given by 
\begin{align}
X_t = \rho_1 X_{t-1} + ... + \rho_d X_{t-d} + \epsilon_t,
\end{align}
Then, the OLS estimator is given by 
\begin{align}
\hat{\rho}_n = \left( \sum_{t=1}^n \mathbf{X}_{t-1} \mathbf{X}_{t-1}^{\prime} \right)^{-1} \sum_{t=1}^n \mathbf{X}_{t-1}  X_t
\end{align}
Therefore, we have that 
\begin{align}
\frac{1}{ \sqrt{n} } \sum_{t=1}^n \mathbf{X}_{t-1} \mathbf{X}_{t-1}^{\prime} \left( \hat{\rho}_n - \rho \right) = \frac{1}{ \sqrt{n} } \sum_{t=1}^n \mathbf{X}_{t-1} \epsilon_t,
\end{align}
Thus, the $d-$dimensional random vectors $\xi_{nt} = \mathbf{X}_{t-1} \epsilon_t / \sqrt{n}$ form a square integrable martingale difference array with respect to the $\sigma-$fields $\mathcal{F}_{nt}$, which implies that asymptotic results for the OLS estimator $\hat{\rho}_n$ may be derived from martingale limit theory. 
\end{example}

Moreover, it is well known that the distributional properties of the time series $( X_t )_{ t \geq - d + 1 }$ are strongly dependent on the value of $\rho \in \mathbb{R}^d$ are are entirely different for different $\rho-$regions in $\mathbb{R}^d$. One of the basic results states that there exists a stationary distribution of $\mathbf{X}_0 = \left( X_0, ..., X_{-d+1} \right)^{\prime}$, that is, a distribution which results in a strictly stationary process $( X_t )_{ t \geq - d + 1 }$ , if and only if all (possible complex) roots of the equation 
\begin{align}
\lambda^d - \rho_1 \lambda^{d-1} - \rho_2 \lambda^{d-2} - ... - \rho_{d-1} \lambda - \rho_d = 0
\end{align}

have modules strictly less than one. Then, the quadratic characteristic of the martingale difference array $\left( \xi_{nt}, \mathcal{F}_{nt}  \right)$ satisfies
\begin{align}
\sum_{t=1}^n \mathbb{E} \left( \xi_{nt}  \xi_{nt}^{\prime} | \mathcal{F }_{n, t-1} \right) = \frac{\sigma^2}{n} \sum_{t=1}^n \mathbf{X}_{t-1} \mathbf{X}_{t-1}^{\prime} \to \sigma^4 \Sigma 
\end{align}
in probability as $n \to \infty$, since $\mathbf{X}_{t-1}$ is measurable w.r.t $\mathcal{F }_{n, t-1}$ and $\epsilon_t$ and $\mathcal{F }_{n, t-1}$ are independent. Moreover, the Lindeberg condition is satisfied, and the multivariate CLT implies that 
\begin{align}
\sum_{t=1}^n \xi_{nt} \to \mathcal{N} ( 0, \sigma^4 \Sigma ), \ \ \text{as} \ n \to \infty, 
\end{align}
Therefore, we obtain the CLT for the least squares estimator $\hat{\rho}_n$ as below
\begin{align}
\sqrt{n} \left( \hat{\rho}_n - \rho \right) \to \mathcal{N} \left( 0, \Sigma^{-1} \right), \ \ \text{as} \ \ n \to \infty, 
\end{align}
Consider also the $\norm{ \ . \ }_{ \Sigma }$ to be the norm on $\mathbb{R}^d$ pertaining to the inner product induced by $\Sigma$, that is, $\norm{ x }_{ \Sigma } = \norm{ \Sigma^{1/ 2} x }$ for all $x \in \mathbb{R}^d$, where $\norm{ . }$ is the Euclidean norm in $\mathbb{R}^d$. 

Then, for any sequence $x_n \to \infty$ with $x_n = \mathcal{O}( n^{\delta - \epsilon} )$ for some $\epsilon > 0$ as $n \to \infty$, we have the moderate deviations limit theorem
\begin{align*}
&\mathbb{P} \left( \norm{ \sqrt{n} \left( \widehat{\rho}_n - \rho \right) }_{ \Sigma } \geq \sqrt{2 \text{log} x_n} \right)
\\
&= \frac{1}{ \Gamma( d/ 2) } x_n^{-1} \left( \text{log} x_n \right)^{d/2 -1} \left( 1  + \mathcal{O} \left( \frac{1}{\text{log} x_n }  \right) \right).
\end{align*}
 
Consider the martingale difference array such that $\tilde{\xi}_{nt} = \sigma^{-2} \Sigma^{-1 / 2} \xi_{nt}$. Using this martingale difference array we have that 
\begin{align}
\tilde{L}_{\delta} = \sum_{ t = 1}^k \mathbb{E} \norm{ \tilde{\xi}_{nt}  }^{2 + 2 \delta} = \mathcal{O}( n^{- \delta} ).
\end{align}

\newpage

\subsection{Estimators for Gaussian Autoregressive Process}

Consider the OU diffusion solution of the stochastic differential equation
\begin{align}
dX_t = \theta X_t dt + dW_t, \ \ X_0 = x, 
\end{align}
where $W$ is a standard Brownian motion and $\theta$ is an unknown parameter in $\mathbb{R}$. For a continuous observation of $X$ over $[0,T]$, it is usual to consider an estimating function, the score function, that is, the derivative of the log-likelihood, which is given by 
\begin{align}
Y_T( \theta ) = \frac{1}{2} \left( X_T^2 - x^2 \right) - \int_0^T \left( \theta X_t^2 + \frac{1}{2} \right) dt. 
\end{align}
We denote by $\hat{ \theta }_T$ the maximum likelihood estimator (MLE), solution of $Y_T ( \theta ) = 0$, explicitly as 
\begin{align}
\hat{\theta}_T = \frac{ \displaystyle \int_0^T X_t dX_t }{ \displaystyle \int_0^T X_t^2 dt}. 
\end{align}

\begin{example}
Consider $(X_n)$ to be the stable autoregressive model of order $p$, and dimensions $d$
\begin{align}
X_n = A_1 X_{n-1} + A_2 X_{n-2} + ... + A_p X_{n-p} + \epsilon_n 
\end{align}
with an initial state $X_0^{(p)} = \left( X_0, X_{-1}, ..., X_{-p +1}   \right)$ independent from the noise. By stable model we mean that all roots of the polynomial $z \mapsto \text{det} \left( I - A_1 z - ... - A_p z^p \right)$ have their modules $> 1$.    
\end{example}

\begin{example}
Consider the linear autoregressive model in $\mathbb{R}^d$
\begin{align}
X_n = \theta X_{n-1} + \xi_n
\end{align}
Notice that a stationary solution of the autoregressive equation is given by 
\begin{align}
X_n = \sum_{ p = 0}^{ \infty } \theta^p \xi_{ n - p}, \ \ \ n \geq 0
\end{align}    
\end{example}
Thus, as in \cite{yu2009moderate}, for estimating $\theta$, the following two estimators are widely used
\begin{align}
\widehat{\theta}^{ols}_n &= \left( \sum_{ k = 1 }^n X_k X_{ k - 1}^{\top}    \right) \left( \sum_{k = 1}^n X_{k-1} X_{k-1}^{\top} \right). 
\\
\widetilde{\theta}^{YW}_n &= \left( \sum_{ k = 1 }^n X_k X_{ k - 1}^{\top}    \right) \left( \sum_{k = 0}^n X_{k} X_{k}^{\top} \right). 
\end{align}

\newpage 

Consider the least-squares estimator $\theta_n$ of the parameter $\theta$ which is defined as below (see, \cite{worms2001large})
\begin{align}
\hat{\theta}_n = \left( \sum_{j=1}^n X_{j-1}^2 \right) \sum_{j=1}^n X_{j-1} X_j. 
\end{align}

It is well-known that $\left( \hat{\theta}_n \right)$ is strongly consistent, but the behaviour in distribution and the corresponding speeds are different according to the true parameter value of $\theta$, when $X_0 = 0$.

\begin{itemize}
\item in the stable case $\left( | \theta | < 1 \right)$, $\sqrt{n} \tilde{\theta} / \sqrt{1 - \theta^2 }$ converges in distribution to the Gaussian law $\mathcal{N} ( 0, 1)$, 

\item in the explosive case $\left( | \theta | > 1 \right)$, $\sqrt{n} \tilde{\theta} / \sqrt{1 - \theta^2 }$ converges in distribution to a Gauchy distribution,

\item in the unstable case $\left( | \theta | = 1 \right)$, $\tilde{\theta}$ converges in distribution to $B_1 / \int_0^1 B_s^2 ds $, where $\left( B_t \right)$ is some standard Brownian motion. 

\end{itemize}

\subsection{The Exogeneity Assumption in Time Series Models}

The second important issue when considering econometric modelling is the aspect of exogeneity. In other words, the concept of exogeneity is essential to uniquely characterize the implications that certain variables are exogenous according to particular definitions. 
\begin{itemize}
    \item \textbf{Weak exogeneity:}  

    A  good intuition of the concept of weak exogeneity, is that it guarantees that the parameters of the conditional model and those of the marginal model are variation free, it offers a natural framework for analyzing the structural invariance of parameters of conditional models. However, by itslef, weak exogeneity is neither necessary nor sufficient for structural invariance of a conditional model.

    \item \textbf{Strong exogeneity:} 
    
    If in addition to being weakly exogenous, $z_t$, is not caused in the sence of Granger by any of the endogenous variables in the system, then $z_t$ is defined to be strongly exogenous.  
    
\end{itemize}
The concept of structurally invariant conditional models characterizes the conditions which guarantee the appropriateness of "policy implications" or other control exercises, since any change in the distribution of the conditioning variables has no effect on the conditional model and thus on the conditional forecasts of the endogenous variables  (see, \cite{ericsson1998exogeneity}, \cite{white2014granger}). It might be more helpful to understand what is not causality, thus we consider the definition of Granger Noncausality.  
\begin{definition}
$Y_{t-1}$ does not Granger cause $z_t$ with respect to $X_{t-1}$ if and only if 
\begin{align}
D ( z_t | X_{t-1}, \theta ) = D \left( z_t | Z_{t-1}, Y_0, \theta  \right)    
\end{align}
\end{definition}

\medskip

\begin{remark}
Granger noncausality is neither necessary nor sufficient for weak exogeneity. Granger noncausality in combination with weak exogeneity, however, defines strong exogeneity. 
\end{remark}

\newpage 

\begin{example}(time-varying covariates: endogeneity versus exogeneity) 
Consider for simplicity a generalized linear model, in which case we assume that $\left\{ X_i \right\}$ is stochastic which implies that $\left( Y_{i1}, X_{i1} \right)$,...,$\left( Y_{i n_i }, X_{i n_i } \right)$ are independently distributed. Furthermore, this assumption implies
\begin{align}
\mathbb{E} \left[ Y_{ij} | X_{i1}, X_{i2},..., X_{i n_i} \right] = \mathbb{E} \left[ Y_{ij} | X_{ij} \right],
\end{align}
which is called the Full Covariate Conditional Mean (FCCM) assumption. In other words, if current outcomes predict future values of the covariates, then the above equality does not hold. In particular, if $f_X \left( X_{ij+1} | Y_{ij}, X_{ij} \right) \neq f_X \left(  X_{ij+1} | X_{ij} \right)$, then we have that
\begin{align}
f_Y \left( Y_{ij+1} | X_{ij}, X_{ij+1} \right) \neq f_Y \left(  Y_{ij+1} | X_{ij} \right) 
\end{align}
and so the expectation $\mathbb{E}\left[ Y_{ij} | X_{ij} \right]$ may be incorrect. Therefore, it is important to distinguish between the presence of endogeneity versus exogeneity for time-varying covariates. Assume that we have follow-up times which are discrete, e.g., $t = 1,...,T$. Let $
\mathcal{H}_{it}^Y : \left\{ Y_{i1}, ...., Y_{it} \right\}$ and $\mathcal{H}_{it}^X : \left\{ X_{i1}, ...., X_{it} \right\}$, to denote the response history of subject $i$ at time $t$, and the covariate history of subject $i$ at time $t$ (filtration) respectively, where $Z_i$ to denote the baseline covariates of subject $i$. Then, exogeneity implies 
\begin{align}
f_X \left( X_{it} | \mathcal{H}_{it}^Y, \mathcal{H}_{it-1}^X, Z_i      \right) = f_X \left( X_{it} | \mathcal{H}_{it-1}^X, Z_i \right) 
\end{align}
and endogeneity gives that 
\begin{align}
f_X \left( X_{it} | \mathcal{H}_{it}^Y, \mathcal{H}_{it-1}^X, Z_i      \right) \neq f_X \left( X_{it} | \mathcal{H}_{it-1}^X, Z_i \right)
\end{align}
where $f_X$ is the pdf/cdf of $X_{it}$. Exogeneity assumption implies that $
f_{X,Y} \left( Y_{it}, X_{it} | Z_i, \theta \right) = \ell_Y ( \theta ) . \ell ( \theta )$, which implies that the likelihood of $Y$ can be maximized independently of the likelihood of $X$. Also, if $\theta = \left( \theta_1, \theta_2 \right)$ and $f_{X,Y} \left( Y_{it}, X_{it} | Z_i, \theta \right) = \ell_Y ( \theta_1 ) . \ell ( \theta_2 )$, there will be no loss of information/no loss of efficiency if we condition on $X$ instead of modelling the joint density. Exogeneity further implies
\begin{align}
\mathbb{E} \left[ Y_{it} | X_{i1},..., X_{iT}, Z_i \right] = \mathbb{E} \left[ Y_{it} | X_{i1},..., X_{iT}, Z_i   \right]
\end{align}
that is, $Y_{it}$ is conditionally independent of all future outcomes given the covariate history. 
\end{example}
\begin{remark}
\textit{Conditional invariance} is a property of the data generating process, defined without regard to any model, unlike super exogeneity (see, \cite{white1996estimation}). Thus, \textit{conditional invariance} can be viewed as a form of necessary condition for examining the effect of $Y_t$ of policy changes (i.e., changes in the distribution of $W_t$), because if $\left\{ Y_t | W_t \right\}$ is not conditionally invariant, then changing the distribution of $W_t$ may change the relationship between $Y_t$ and $W_t$ in unpredictable ways. In this situation, even a correctly specified model of the old relationship is likely to be of little value in describing the behaviour of $Y_t$ under the new regime. Numerous examples can be given to show that dynamic misspecification may adversely affect consistency, and if not consistency, then attainment of the asymptotic variance bound. Furthermore, neglected heteroscedasticity may also adversely affect attainment of the asymptotic variance bound. 
\end{remark}

\newpage

\subsection{Autoregressive Processes with Exogenous Variables}

Following the framework of \cite{damon2005estimation}, we consider the following autoregressive Hilbertian with exogenous variables of order one model, denoted by ARHX(1) such that
\begin{align}
Y_t = \rho ( Y_{t-1} ) + \beta_1 ( X_{t,1} ) + ... + \beta_p ( X_{t,p} ) + \varepsilon_t, \ \ \ t \in \mathbb{Z},     
\end{align}

where $X_{t,1} + ... +  X_{t,p}$ are $p$ autoregressive of order one exogenous variables associated respectively with operators $\varphi_1,..., \varphi_p$ and strong white noises $u_{t,1},..., u_{t,p}$ such that 
\begin{align}
X_{t,i} =  \varphi_j ( X_{t-1,i} )  + u_{t,i}, \ \ \ \text{for all} \ \ i \in \left\{ 1,..., p\right\}.
\end{align}

\begin{remark}
Notice that regressors of 'mixed integration order' it should represent the case where the autoregressive model includes regressors of different integration order such as stationary versus nonstationary rather regressors which are LUR but fall in a different persistence class due to different values and sign of the localizing coefficients of persistence.
\end{remark}
Furthermore, define with 
\begin{align}
\mathcal{Y}_t =
\begin{bmatrix}
Y_t
\\
X_{t+1,1}
\\
\vdots
\\
X_{t+1,p}
\end{bmatrix},
\ \ 
\mathcal{E}_t =
\begin{bmatrix}
\varepsilon_t
\\
u_{t,1}
\\
\vdots
\\
u_{t,p}
\end{bmatrix}
\ \ \ 
\text{and} 
\ \ \ 
\mathcal{R} = 
\begin{bmatrix}
\rho & \beta_1 & \hdots & \hdots & \beta_p
\\
0   & \varphi_1 & 0 & \hdots & 0
\\
0   &    0      & \varphi_2 &  0 & \vdots
\\
\vdots & \vdots &   & \ddots & 0
\\
0 & 0 & 0 & 0 & \varphi_p
\end{bmatrix}
\end{align}

In particular, when the sequence $\left\{ Y_t \right\}$ is an ARHX (1) model, then $\left\{ \mathcal{Y}_t \right\}$ is an $H^{p + 1}-$valued ARHX(1) process, as seen by the following autoregressive representation 
\begin{align}
\mathcal{Y}_t = \mathcal{R} ( \mathcal{Y}_{t-1} ) +  \mathcal{E}_t, \ t \in \mathbb{Z}. 
\end{align}
We suppose that the noises $\varepsilon_t$ and $u_{t,1},..., u_{t,p}$ are independent which ensures that for all $i \in \left\{ 1,..., p  \right\}$ the whole processes $( \varepsilon_t )$ and $\left(  X_{t,i}  \right)$ are independent since $X_{t,i}$ can be expressed as an infinite moving average of the $u_{t-q,1}$ where $q = 0,..., \infty$.

\begin{proposition}
Equation (1) has a unique stationary solution given by 
\begin{align}
X_k = \sum_{j=0}^{ \infty } \left( P_1 \mathcal{R}^j \right) \mathcal{E}_{k-j}, \ \ \ t \in \mathbb{Z}.    
\end{align}
Denote by $P_i$ the projection operators $\big( X_1,..., X_{q+1} \big) \mapsto x_i.$ Then, the series converge almost surely and in $\mathcal{L}^2_H \left( \Omega, \mathcal{A}, \mathcal{P} \right)$.  
\end{proposition}

\begin{definition}
The covariance operator $C^{ X,Y }$ of two $H-$valued random variables $X$ and $Y$ is
\begin{align}
C^{ X,Y } (x) := \mathbb{E} \big[ \langle X, x \rangle Y \big], \ \ \ x \in H,  
\end{align}
and $C^X$ stands for $C^{X,Y}$. The autocovariance of a stationary process $( X_n )$ is the sequence of operators defined by $ C_h = \big( C^{ X_0, X_h}, h \in \mathbb{Z} \big)$. 
\end{definition}

\medskip

Consider for example the framework given by \cite{seo2019cointegrated}. 
A cointegrated linear process in Bayes Hilbert space is isomorphic to a cointegrated linear process in a Hilbert space of centered square-integrable real functions. We illustrate the use of this isomorphism for modeling nonstationary time series of probability densities. In particular, a recent literature on functional data analysis deals with datasets whose observations take values in an infinite dimensional Banach or Hilbert space, typically a space of functions. 

\medskip

\begin{proposition}
Let $H_1$ and $H_2$ be real separable Hilbert spaces, and $S: H_1 \to H_2$ an isomorphism between them. If $X$ is a cointegrated linear process in $H_2$, then $S^{-1} X = S^{-1} X_t, t \geq 0$, is a cointegrated linear process in $H_1$, and its attractor space is the inverse image of the attractor space of $X$ under $S$.
\end{proposition}

\medskip

\begin{proof}
Each $S^{-1} X_t$ is a random element of $H_1$ (since $S^{-1}$ is Borel measurable) that satisfies the finite moment condition $\mathbb{E} \norm{ S^{-1} X_t }_{ H_1 }^2 < + \infty$ (since $S^{-1}$ is norm-preserving) and thus it holds that $\mathbb{E} S^{-1} X_t = 0$ (since $S^{-1}$ is linear). Therefore, $S^{-1} X_t$ is indeed a sequence in $L_{ H_1 }^2$. 

Furthermore, its differences $\Delta S^{-1} X_t$ satisfy 
\begin{align}
\label{limit}
\Delta S^{-1} X_t  = \sum_{ j = 0 }^{ + \infty } S^{-1} \Psi_k S \big( S^{-1} \left( \epsilon_{t-k} \right) \big), \ \ \ t \geq 1.     
\end{align}
Thus, we will show that expression \eqref{limit} above constitutes a valid representation for $\Delta S^{-1} X_t$ with coefficients $S^{-1} \Psi_k S$ and innovations $S^{-1} \left( \epsilon_t \right)$. In particular, the argument used above to show that $S^{-1} X_t \in L_{ H_1 }^2$ also shows that the innovations $S^{-1} \epsilon_t$ belong to $L_{ H_1 }^2$. Moreover, they are \textit{i.i.d} due to the \textit{i.i.d} property of the $\epsilon_t$'s and the Borel measurability of $S^{-1}$. They have covariance operator $S^{-1} \Sigma S$ (since $S$ is the adjoint to its inverse), which is positive definite since 
\begin{align}
\langle  S^{-1} \Sigma S(f), f \rangle_{ H_1 } =  \langle \Sigma S(f), S(f) \rangle_{ H_2 }   > 0 
\end{align}
for any nonzero $f \in H_1$ due to the positive definiteness of $\Sigma$.

\end{proof}

\newpage

\section{Temporal Dependence}

Econometric models provide parsimonious representations of economic phenomena; thus the development of robust statistical methods is of paramount importance in improving our understanding of financial markets and economic decision making. In this section we study some aspects of estimation and inference in time series regression under temporal dependence. 

\subsection{Setting the time series framework}

Consider the following model
\begin{align}
\label{modelDGP}
X_t = F \left( X_{t-1}, M_{t-1}, \epsilon_t \right), \ \ t \in \mathbb{Z}, 
\end{align}
where $\left( X_{t} \right)_{ t \in \mathbb{Z} }$ is the stochastic time series we aim to model,  $\left( M_{t} \right)_{ t \in \mathbb{Z} }$ is a covariate process and $\left( \epsilon_{t} \right)_{ t \in \mathbb{Z} }$ a noise process. Here, we notice that $X_t = f_t \left(  X_{t-1} \right)$ holds for the random function defined by $f_t (x) = F \left( x, M_{t-1}, \epsilon_t \right)$. Therefore, the stochastic time series model \eqref{modelDGP} shows that the stochastic sequence $\left\{  X_t \right\}$ transforms a sequence of i.i.d random innovations $\epsilon_t$ into a time series sequence of dependent random variables (maps) as described by the random function $\left( f_t \right)_{t \in \mathbb{Z} }$.  Moreover, the stochastic process $\left( M_{t} \right)_{ t \in \mathbb{Z} }$ allows to incorporate in the model exogenous covariates which typically exhibit temporal dependence.  Under this general setting we can examine both linear and nonlinear time series models without restricting the estimation methodology employed.     

\subsection{Example: Nonlinear time series models}

A nonlinear time series methodology which is commonly employed to capture conditional heteroscedasticity and dynamic variances in financial markets includes the ARCH and GARCH models. Specifically, these two nonlinear type of models have been the cornerstone of financial econometrics and statistical risk management fields. The ARCH model has been generalized to to GARCH by \cite{bollerslev1986generalized}. A GARCH$(p,q)$ sequence $\left\{ X_t, - \infty < t < \infty  \right\}$ has the following form 
\begin{align}
X_t = \sigma_t \epsilon_t
\end{align} 
\color{black}
\begin{align}
\sigma_t^2 = \alpha_0 + \sum_{i=1}^p \alpha_i X^2_{t-i} + \sum_{j=1}^q \beta_j \sigma^2_{t-j}
\end{align}
The literature has evolved considerably since the introduction of the GARCH$(p,q)$ model. The aforementioned stochastic linear and nonlinear time series models exhibit different asymptotic behaviour under the presence of exogenous covariates, especially considering these time series models as iterations of dependent random variables capturing linear and nonlinear dynamics. 

\newpage

Furthermore, considering these time series models as stochastically recursive sequences, when the independence assumption is removed under certain regulatory conditions, certain interesting asymptotic cases appear. For instance, we can study stationarity of asymptotic power GARCH processes which encompasses the case of exogenous regressors in nonlinear dynamics. Notice that for strictly exogenous regressors, the processes  $\left( M_{t} \right)_{ t \in \mathbb{Z} }$ and $\left( \epsilon_{t} \right)_{ t \in \mathbb{Z} }$ are independent. 

We are interested for example in the dynamics of predictive regression models described by \eqref{modelDGP} with regressors that are not necessarily strictly exogenous, assuming that at any time $t$, the noise $\epsilon_t$ is independent from the past information $\sigma \left( \left( M_s, \epsilon_s \right) : s \leq t - 1 \right)$. Therefore, this independence assumption is weaker than the independence assumption between the two processes $\left( \epsilon_{t} \right)_{ t \in \mathbb{Z} }$ and $\left( M_{t} \right)_{ t \in \mathbb{Z} }$. The latter independence condition implies strict exogeneity, a terminology intially defined by Sims. Strict exogeneity is useful for deriving the conditional likelihood of the $X_t$'s conditionally on the $M_t$'s. However, strict exogeneity is a strong assumption. Moreover, \cite{chamberlain1982multivariate} shown has shown that this assumption is equivalent to the non-Granger causality, that is, $M_t$ is independent of $\left( X_{s} \right)_{ s \leq t }$ conditionally on $\left( X_{s} \right)_{ s \leq t - 1 }$. This idea means that the covariate process $\left( M_{t} \right)_{ t \in \mathbb{Z} }$ evolves in a totally autonomous way. On the other hand, our exogeneity condition allows general covariates of the form $M_t = H \left( \eta_t, \eta_{t-1}, ...., \right)$ with $H$ a measurable function and a sequence $\left( \eta_t, \epsilon_t \right)_{t \in \mathbb{Z}}$ of i.i.d random vectors, $\epsilon_t$ being possibly correlated with $\eta_t$. Thus, the error $\epsilon_t$ can then still have an influence on future values of the covariates. For linear models, the two technical independence conditions discussed above between the noise and the covariate processes are often used as a distinction between weak and strict exogeneity (see, \cite{debaly2021iterations}). 

The family of ARCH and GARCH models are commonly used parametric models for capturing important stylized features of financial time series such as long memory, volatility clustering, fat tails in the distribution of stock returns etc. In particular, the GARCH(1,1) model or AR(1)-GARCH(1,1) model have repeatedly proved to provide robust parsimonious representations of market volatility (see \cite{berkes2003garch}). A GARCH(1,1) regression process is given by the following equations
\begin{align*}
y_t &= \mu + \eta_t , \ \eta_t = \sigma_t \epsilon_t  , \ \ \ \epsilon_t \sim N(0,1)
\\
\sigma^2_t &= \omega + \alpha \epsilon^2_{t-1} + \beta \sigma^2_{t-1}
\end{align*}
An AR(1)-GARCH(1,1) regression process is given by 
\begin{align*}
y_t &= \mu + \rho y_{t-1} + \eta_t , \ \eta_t = \sigma_t \epsilon_t \ \ epsilon_t \sim N(0,1) 
\\
\sigma^2_t &= \omega + \alpha \epsilon^2_{t-1} + \beta \sigma^2_{t-1}
\end{align*}
The above non-linear models are valid under certain regulatory conditions which include $\omega , \alpha, \beta \geq 0$ and $\alpha + \beta < 1$ in order to avoid the existence of IGARCH effects. The recursively estimated sequences $(\sigma^2_t)$ for $t=0 ,..., + \infty$ are assumed to be non negative with probability 1 and common unconditional local mean given by $\sigma^2=\frac{\omega}{1 - \alpha - \beta}$. Statistical estimation can be done by imposing certain parametric assumptions. A commonly used methodology in the literature is the implementation of the QMLE. 

\newpage

The use of QML estimation can restore any inefficient model estimations (i.e., existence of negative coefficients or out of bounds persistence) and allows for higher estimation precision since the proposed methodology is less restrictive in the moment assumptions of the observed process (no restriction on normality assumption - which induces distribution-free inference). 
\begin{align}
L_N(\theta) \equiv L_N(\theta ; \epsilon_1,...,\epsilon_N) = \prod_{t=1}^N \frac{1}{\sqrt{2 \pi \hat{\sigma}^2_t } }\exp \bigg( - \frac{\epsilon^2_t}{2 \hat{\sigma}^2_t  } \bigg), \hat{\sigma}^2_t = \omega + \alpha \epsilon^2_{t-1} + \beta \hat{\sigma}^2_{t-1}.
\end{align}
and the QMLE for the set of model parameters $\theta=\{\omega, \alpha, \beta \}$ of GARCH(1,1) to be  
\begin{align*}
\hat{\theta}_N = \underset{ \theta \in \Theta}{\text{argmax}} \ L_N(\theta)
\end{align*}
which is equivalent to 
\begin{align*}
 \underset{\theta \in \Theta}{\text{argmin}} \ \hat{l}_N(\theta), \ \text{where} \ \hat{\ell}_N(\theta) = \sum_t \bigg( \log[\hat{\sigma}^2_t] + \frac{\epsilon^2_t}{\hat{\sigma^2_t}}  \bigg).
\end{align*}
Then asymptotic normality holds, which implies that, 
\begin{align}
\displaystyle{ \sqrt{N}(\hat{\theta}_N - \theta_0^{*}) \overset{d}{\to} N(\mathbf{0}, \mathbf{V_0} )}, \ \mathbf{V_0}=\mathbf{B_0^{-1}A_0 B_0^{-1}}.
\end{align}

\subsection{Heavy Tail Time series}

\subsubsection{Max-stable processes}

\begin{definition}
A time series $\left\{ \mathbf{X}_j , j \in \mathbb{Z} \right\}$ is max-stable if all its finite dimensional distributions are max stable.  
\end{definition}

\begin{definition}
We call a Lebesgue measurable function $L$ from $\mathbb{R}$ into $(0, \infty)$ slowly varying iff
\begin{align}
\underset{ \lambda \to \infty }{ \text{lim}  } \frac{ L (\lambda t ) }{ L ( \lambda ) } = 1, \ \ \text{for every} \ \ t > 0. 
\end{align}     
\end{definition}

\subsubsection{Tail behaviour of sample quantiles}

\begin{proposition}
Let $\mathbf{Y}_j = \left( Y_j(t), t \in T  \right)$, $j = 1,2,...$ be a sequence of \textit{i.i.d} measurable stochastic processes. The stochastic processes
\begin{align}
X_j(t) = \int_{ \mathcal{S} } f(t,s) M_j(ds), \ \ t \in T,
\end{align}
$j \geq 1$, are well defined and independent. 
\end{proposition}

\newpage

Moreover, we have that  
\begin{align}
\bigg( X(t), t \in T \bigg) \overset{d}{=} \bigg( \sum_{j=1}^{\infty} X_j(t), t \in T \bigg), 
\end{align} 
with the sum converging a.s for every $t \in T$. We have already proved that each $\left( X_j (t), t \in T \right)$ has a measurable version. Since the pointwise limit of measurable functions is measurable, we have constructed our measurable version of $\left( X (t), t \in T \right)$. 

\subsection{Extreme Value Theory for nonstationary time series}

Consider the limit theory for extreme values of a class of nonstationary time series with the form 
\begin{align}
Y_t = \mu_t + \xi_t , \ \ \ \xi_t = \sum_{j=0}^{\infty} c_j Z_{t-j}
\end{align}
where $\left\{ Z_t = \sigma_t \eta_t , - \infty < t < \infty \right\}$ and 
$\left\{ \eta_t ; - \infty < t < \infty \right\}$ is a sequence of \textit{i.i.d} random variables with regularly varying tail probabilities. Some convergence results for point processes based on one-sided moving averages $\left\{ \xi_t = \sum_{j=0}^{\infty} c_j Z_{t-j} \right\}$ are derived. Extreme properties of the nonstationary sequence $\left\{ Y_t \right\}$ are then obtained from the convergence results (see, \cite{niu1997extreme}). 
\begin{lemma}
Suppose for each $n \geq 1$, $\left\{ W_{n,j}, j \geq 1 \right\}$ is a sequence of independent nonidentically distributed random elements of $( E, \mathcal{E} )$. Define with $N_n = \sum_{j=1}^{\infty} \epsilon_{( j / n, W_{n,j}) }$. Then, $N_n \Rightarrow N$ iff
\begin{align}
\sum_{j=1}^{\infty} \epsilon_{j/n} (.) \mathbb{P} \left\{ W_{n,j} \in .  \right\} \to_{\nu} \nu \times \mu.
\end{align}
\end{lemma}
We assume that 
\begin{align}
\mathbb{P} \left\{ | \eta_1 | > x \right\} \in RV_{- \alpha }, \ \ \ \alpha > 0
\end{align}
and 
\begin{align}
\underset{ x \to \infty }{ \text{lim} } \frac{ \mathbb{P} ( \eta_1 > x ) }{   \mathbb{P} ( | \eta_1 | > x ) } = \pi_0, \ \ \ \underset{ x \to \infty }{ \text{lim} } \frac{ \mathbb{P} ( \eta_1 < - x ) }{   \mathbb{P} ( | \eta_1 | > x ) } = 1 - \pi_0, 
\end{align}
where $0 \leq \pi_0 \leq 1$. Consider one-sided moving averages of the form 
\begin{align}
\xi_t = \sum_{j=0}^{ \infty } c_j Z_{t-j} = \sum_{j=0}^{ \infty } c_j \sigma_{t-j} \eta_{t-j}, \ \ \ - \infty < t < \infty
\end{align}
where $\left\{ c_j \right\}$ is a sequence of real constants with $c_0 = 1$

\newpage

where
\begin{align}
\sum_{j=0}^{\infty} | c_j |^{\gamma} < \infty, \ \ \ \text{for some} \ \ 0 < \gamma < 1 
\end{align}

Furthermore, we assume that 
\begin{align}
\frac{1}{n} \sum_{t=1}^n \sigma_t^{\alpha} \to \sigma^{\alpha}, \ \ \text{as} \ n \to \infty, 
\end{align}

\begin{definition}
Let $\left\{ X_{nk}, u_n \leq k \leq v_n, n \geq 1 \right\}$ be an array of random variables and $\left\{ a_{nk}, u_n \leq k \leq v_n, n \geq 1 \right\}$ an array of constants with $\sum_{ k = v_n }^{ v_n } | a_{nk} | \leq C$ for all $n \in \mathbb{N}$ and some constant $C > 0$. Moreover, let $\left\{  h(n), n \geq 1 \right\}$ be an increasing sequence of positive constants, with $h(n) \to \infty$ as $n \to \infty$. The array $\left\{ X_{nk} \right\}$ is said to be $h-$integrable with respect to the array of constants $\left\{ a_{nk} \right\}$ if the following conditions hold: 
\begin{align}
\underset{ n \geq 1 }{ \text{sup} } \sum_{k = v_n }^{v_n} | a_{nk} | \mathbf{E} | X_{nk} | < \infty, \ \ \ \text{and} \ \ \ \underset{ n \to \infty }{ \text{lim} } \sum_{k = v_n }^{ v_n } | a_{nk} | \mathbf{E} | X_{nk} | \mathbf{1} \left\{ | X_{nk} | > h(n) \right\} = 0.
\end{align}  
\end{definition}

A linear process can be written in the following form 
\begin{align}
X_t = \sum_j \psi_j \epsilon_{t-j}, \ \ \ t \in \mathbb{Z}
\end{align}

is regularly varying with index $\alpha > 0$ if the \textit{i.i.d} sequence $\left( \epsilon_{t} \right)$ is regularly varying with index $\alpha$. Under mildly conditions on $\left( \psi_j  \right)$ we have that 
\begin{align}
\frac{ \mathbb{P} \left( X_0 > x \right) }{ \mathbb{P} \left( | \epsilon_0 | > x \right) } \sim \sum_j | \psi_j |^{\alpha} \left( p \mathbf{I} \left\{ \psi_j > 0 \right\} +  \mathbf{I} \left\{ \psi_j < 0 \right\}  \right) = \norm{ \psi }_{\alpha}^a , \ \ x \to \infty
\end{align}
For an \textit{i.i.d} sequence $\left( A_t, B_t \right)_{t \in \mathbb{Z}}$, $A, B > 0$, the stochastic recurrence equation can be written as below  
\begin{align}
X_t = A_t X_{t-1} + B_t, \ \ t \in \mathbb{Z},
\end{align}
has a unique stationary solution given by the following expression 
\begin{align}
X_t = B_t + \sum_{j = - \infty}^{t-1} A_t .... A_{j+1} B_j, \ \ t \in \mathbb{Z},
\end{align}

\begin{remark}
Overall, the nature of dependence can vary and unless specific assumptions are made about the dependence between random variables, no meaningful statistical model can be assumed. A measure of dependence indicates how closely two random variables $X$ and  $Y$ are. Measures of dependence could be conditions based on order or time between random variables, or could be conditions expressed in terms of a covariance or a correlation coefficient. Distance is also considered as a measure of dependence.
\end{remark}

\newpage 

\section{Limit Theory under Network Dependence}

\subsection{Random Graphs and Moderate Deviations Theory}

\subsubsection{Spectral Edge in Sparse Random Graphs: Tail Large Deviations}

Understanding asymptotic properties of spectral statistics arising from random matrices has been the subject of intense study in recent years. Within this theme, a particularly important research direction is in deriving large deviation principles (rare event probabilities) for spectral functionals, such as the empirical spectral measure and the extreme eigenvalues. In the results above, we consider the large deviations of the edge eigenvalues of $\mathcal{G}_{n,p}$, in the regime $p << \sqrt{\text{log} n / \text{log} \text{log} n}/ n$. Notice for example that the largest eigenvalue $\lambda_1 \left( \mathcal{G}_{n,p} \right)$ is asymptotically $\sqrt{L}_p$, which is governed by the maximum degree of the graph, and corresponding eigenvector is localized; whereas above this threshold the typical value is asymptotically $np$, which is dictated by the total number of edges in the graph, and the corresponding eigenvector is completely delocalized.  

\begin{proposition}
The space $\big( \mathcal{F}, \norm{ . }_n \big)$ is a pseudo-normed space.
\end{proposition}

\medskip

\begin{proof}
Note that $\norm{ f }_n = \left( \frac{1}{n} \sum_{i=1}^n f^2 \big( \boldsymbol{x}_i \big) \right)^{1/2}$.

(i) Based on the definition of $\norm{ . }_n$, it is clear that $\norm{ . }_n \geq 0$, for any $f \in \mathcal{F}$.

(ii) For any $\lambda \in \mathbb{R}$ and $f \in \mathcal{F}$, 
\begin{align}
\norm{ \lambda f }_n = \left( \frac{1}{n} \sum_{i=1}^n \lambda^2 f^2 \big( \boldsymbol{x}_i \big)  \right)^{1/2} = | \lambda | \norm{ f }_n.
\end{align}

(iii) For any $f, g \in \mathcal{F}$,
\begin{align*}
\norm{ f + g }_n 
&= 
\left( \frac{1}{n} \sum_{i=1}^n \big[ f \left( \boldsymbol{x}_i \right) + g \left( \boldsymbol{x}_i \right) \big]^2 \right)^{1/2} 
= 
\left( \sum_{i=1}^n \left[ \frac{1}{ \sqrt{n} }  f \left( \boldsymbol{x}_i \right) + \frac{1}{ \sqrt{n} } g \left( \boldsymbol{x}_i \right) \right]^2 \right)^{1/2} 
\\
&\leq
\left( \sum_{i=1}^n \left[ \frac{1}{ \sqrt{n} }  f \left( \boldsymbol{x}_i \right) \right]^2 \right)^{1/2} + 
\left( \sum_{i=1}^n \left[ \frac{1}{ \sqrt{n} }  g \left( \boldsymbol{x}_i \right) \right]^2 \right)^{1/2}
\\
&= 
\norm{ f }_n + \norm{ g }_n, 
\end{align*}
where we applied the triangle inequality to the classical Euclidean norm. Therefore, we showed that $\big( \mathcal{F}, \norm{ . }_n \big)$ is a pseudo-normed space.

\end{proof}

\newpage

\subsection{Limit theorems for network dependent processes}

\subsubsection{Network dependence condition}

We discuss in details the framework proposed by \cite{kojevnikov2021limit} as this is useful in understanding the relevant limit theory. In particular, we provide a sufficient condition for the shape of the network that ensures that our limit theorems hold (LLN and CLT). A crucial aspect for the network based limit theorems to hold, the number of neighbours at distance $s$ should not grow too fast as $s$ increases. The precise condition for such neighbourhood  shells depends on the dependence coefficients $\theta_{n,s}$, so that if $\theta_{n,s}$ decreases fast as $s$ increases, the requirement for the neighbourhood shells can be weakened. To introduce sufficient conditions, let
\begin{align}
\delta_n^{ \partial } (s;k) = \frac{1}{n} \sum_{ i \in N_n } | N_n^{\partial} (i;s) |^k , 
\end{align}        
where $N_n^{\partial} (i;s)$. When $k = 1$, we simply write that $\delta_n^{ \partial } (s;1) = \delta^{ \partial }_n(s)$. This quantity measures the denseness of a network. Let us introduce a further notation. Define the following expression 
\begin{align}
\Delta_n (s, m;k ) = \frac{1}{n} \sum_{ i \in N_n } \underset{ j \in N_n^{\partial} (i;s) }{\text{max}} |  N_n(i;m) \backslash  N_n(j;s-1) |^k
\end{align} 
where $N_n^{\partial}(i;s)$ is defined above, and we take $N_n(j;s-1) = 0$ if $s = 0$. We also define the following expression
\begin{align}
c_n ( s, m ; k ) = \underset{ a > 1 }{ \text{inf} } \left[ \Delta_n(s,m ; ka) \right]^{ \frac{1}{a} } \left[ \delta_n^{\partial} \left( s; \frac{\alpha}{1 - \alpha} \right) \right]^{1 - \frac{1}{\alpha} }
\end{align} 
The quantity $c_n (s, m; k)$ is easy to compute when a network is given, and it captures the network properties that are relevant for the limit theorems. It consists of two components: $\Delta_n (s,m; k \alpha )$ and $\delta_n^{\partial} \left( s; \alpha / (\alpha - 1) \right)$. These capture the denseness of the network through the average neighbourhood sizes and the average neighbourhood shell size. We summarize a sufficient condition for the network and the weak dependence coefficient as follows
\begin{assumption}[\cite{kojevnikov2021limit}]
(\textbf{Condition ND.}) There exists $p > 4$ and a sequence $m_n \to \infty$ s.t 
\begin{itemize}
\item[\textbf{(a)}] $\theta^{1 - 1 /p}_{n, m_n} = o \left( n^{-3 / 2} \right)$, 
\item[\textbf{(b)}] for each $k \in \left\{ 1, 2 \right\}$, 
\begin{align}
\frac{1}{ n^{k / 2} } \sum_{ s \geq 0 } c_n \left( s, m_n ; k \right) \theta_{n,s}^{1 - \frac{k+2}{p} } = o_{a.s}(1), \ \text{and}
\end{align}
\item[\textbf{(c)}] sup$_{ n \geq 1 } \text{max}_{ i \in N_n } \mathbf{E} \left[ | Y_{n,i} |^p | \mathcal{C}_n \right] < \infty$ almost surely. 
\end{itemize}
\end{assumption}  

\newpage

Below we show that Condition ND (i) is sufficient for the LLN and CLT. Notice that $\Delta_n (s , m_n ; k)$ tends to decrease fast to zero as $s$ goes beyond a certain level, because the set $N_n ( j ; s-1)$ quickly becomes large. 

\begin{lemma}[\cite{kojevnikov2021limit}]
Suppose that network $G_n$ is generated as above, and let $c_n$ be a $\sigma-$field such that the adjacency matrix of network $G_n$ is $\mathcal{C}_n-$measurable. Suppose further that $\left\{ Y_{n,i} \right\}$ is conditionally $\psi-$dependent given $\left\{ \mathcal{C}_n \right\}$, with dependent coefficients $\left\{ \theta_n \right\}$ satisfying  Condition NF. Then Conditions ND(a) and (b) hold. 
\end{lemma}

\subsubsection{Law of Large numbers}

Let $\left\{ Y_{n,i} \right\}$ be conditionally $\psi-$dependent given $\left\{ \mathcal{C}_n \right\}$. Since a LLN can be applied element-by-element in the vector case, without loss of generality we can assume that $Y_{n,i} \in \mathbb{R}$ in this sections, that is, $v = 1$. 

Define with,
\begin{align}
\norm{ Y_{n,i} }_{ \mathcal{C}_{n,p} } = \left( \mathbf{E} \left[ | Y_{n,i} |^p | \mathcal{C}_n \right] \right)^{1 / p}
\end{align}
Assume the following moment condition holds. 
\begin{assumption}[\cite{kojevnikov2021limit}]
For some $\epsilon > 0$, 
\begin{align}
\underset{ n \geq 1 }{ \mathsf{sup} }  \ \underset{ i \in N_n }{ \mathsf{max} } \ \norm{ Y_{n,i} }_{ \mathcal{C}_{n, 1 + \epsilon} } < \infty,
\end{align}
almost surely. 
\end{assumption}

The next assumption puts a restriction on the denseness of the network and the rate of decay of dependence with the network distance. 

\begin{assumption}[\cite{kojevnikov2021limit}]
$n^{-1} \sum_{ s \geq 1 } \delta_n^{\partial} (s) \theta_{n,s} \to 0$ almost surely. 
\end{assumption}

For example, Assumption 4 above can fail, for example, if there is a node connected to almost every other node in the network as in the following example. Consider a network with the star topology, which has a central node or hub connected to every other node. 

\begin{theorem}[\cite{kojevnikov2021limit}]
Suppose that $\left\{ Y_{n,i} \right\}$ is conditionally $\psi-$dependent given $\left\{ \mathcal{C}_n \right\}$ and that the Assumptions above hold. Then, as $n \to \infty$, 
\begin{align}
\norm{ \frac{1}{n} \sum_{i \in N_n } \left( Y_{n,i} - \mathbf{E} \left[ Y_{n,i} | \mathcal{C}_n \right] \right)}_{ \mathcal{C}_{n,1} } \to 0.
\end{align}
\end{theorem}
An unconditional version of the result, which replaces the conditional norm in Theorem 3.1 with the unconditional norm, can be established in a similar manner by replacing the conditional moment with the unconditional moment.

\newpage

Next, we discuss LNNs for nonlinear functions of $\left\{ Y_{n,i} \right\}$. When $f \in \mathcal{L}_{ \nu, 1}$, a LLN for a nonlinear transformation $f \left( Y_{n,i} \right)$ follows immediately from the definition of the $\psi-$ dependence. 

In that case, we have that
\begin{align}
\norm{ \frac{1}{n} \sum_{ i \in N_n }  \big( f \left( Y_{n,i} \right)  - \mathbf{E} \left[ f \left( Y_{n,i} \right) | \mathcal{C}_n \right]   \big) }^2_{ \mathcal{C}_{n,2} } \leq \frac{2}{n} \norm{ f }_{\infty}^2 + \psi_{1,1} (f,f) \frac{1}{n} \sum_{s \geq 1} \delta_n^{ \partial } (s) \theta_{n,s}. 
\end{align}
We have the following result. 
\begin{proposition}[\cite{kojevnikov2021limit}]
Suppose that $\left\{  Y_{n,i} \right\}$ is conditionally $\psi-$dependent given $\left\{ \mathcal{C}_n \right\}$, then if Assumption 3.2., holds and $f \in \mathcal{L}_{v,1}$ as $n \to \infty$ we have that
\begin{align}
\norm{ \frac{1}{n} \sum_{ i \in N_n }  \big( f \left( Y_{n,i} \right)  - \mathbf{E} \left[ f \left( Y_{n,i} \right) | \mathcal{C}_n \right]   \big) }_{ \mathcal{C}_{n,2} } \to 0 \ \textit{almost surely}. 
\end{align}
\end{proposition}
However, in general nonlinear transformations of $\psi-$dependent processes are not necessarily $\psi-$dependent. In such cases, LLNs for nonlinear transformations can be established using the covariance inequalities for transformation functions. For example, suppose that for some nonlinear function $h(.)$ of a $\psi-$dependent process $\left\{ Y_{n,i} \right\}$, and that $\theta_{n,s}$ is bounded by a constant uniformly over $s \geq 1$ and $n \geq 1$. 

In that case for some constants $C > 0$ and $p > 2$, the conditional covariance given $\mathcal{C}_n$ between $\big\{ h(Y_i) - \mathbf{E} \left[ h( Y_{n,i} ) | \mathcal{C}_n \right] \big\}$ and $\big\{ h( Y_{n,j} ) - \mathbf{E} \left[ h( Y_{n,j} ) | \mathcal{C}_n \right] \big\}$ is bounded by 
\begin{align}
C . \ \underset{ n,i }{ \text{sup} } \norm{ h \left( Y_{n,i} \right) }^2_{\mathcal{C}_{n,p} } . \theta_{ n, d_n (i,j)}^{1 - \frac{2}{p} }.
\end{align}

Therefore, as $n \to \infty$, 
\begin{align}
\norm{ \frac{1}{n} \sum_{ i \in N_n }  \big( h \left( Y_{n,i} \right)  - \mathbf{E} \left[ h \left( Y_{n,i} \right) | \mathcal{C}_n \right]   \big) }_{ \mathcal{C}_{n,2} } \to 0 \ \text{almost surely}. 
\end{align}
provided that sup$_{n,i} \norm{ h \left( Y_{n,i} \right)}_{ \mathcal{C}_{n,p}} < \infty$ \textit{almost surely}.

\subsubsection{Central limit theorem}

In this section, we study the CLT for a sum of random variables that are conditionally $\psi-$dependent. Define with
\begin{align}
\sigma_n^2 = \text{Var} \left(  S_n | \mathcal{C}_n \right), \ \ \ S_n = \sum_{ i \in N_n } Y_{n,i}.
\end{align}

\newpage

\begin{assumption}[\cite{kojevnikov2021limit}]
There exists a positive sequence $m_n \to \infty$ such that for $k = 1,2$, 
\begin{align}
\frac{n}{ \sigma^{2+k} } \sum_{ s \geq 0 } c_n (s, m_n ; k) \theta_{n,s}^{1 - \frac{2+k}{p}} \to 0 \ \ \textit{almost surely} 
\end{align}
and 
\begin{align}
\frac{ n^2 \theta_{ n, m_n }^{1 - \frac{2+k}{p}} }{ \sigma_n } \to 0 \ \ \textit{almost surely},
\end{align}
as $n \to \infty$, where $p > 4$. 
\end{assumption}
It is not hard to see that Condition ND is a sufficient condition for this assumption, when $\sigma_n \geq c \sqrt{n}$ with probability one, for some constant $c > 0$ that does not depend on $n$. The latter condition is satisfied if the "long-run variance", Var $\left( S_n | \mathcal{C}_n \right) / n$ is bounded away from $c^2 > 0$ for all $n \geq 1$. The theorem below establishes the CLT for the normalized sum $S_n / \sigma_n$.  

\begin{theorem}[\cite{kojevnikov2021limit}]
Suppose that Assumptions hold, and that $\mathbf{E} \left[ Y_{n,i} | \mathcal{C}_n \right] = 0$ \textit{almost surely}. Then, we have that
\begin{align}
\underset{ t \in \mathbb{R} }{ \text{sup} } \left| \mathbb{P} \left\{ \frac{S_n}{ \sigma_n } \leq t \big| \mathcal{C}_n  \right\} - \Phi(t) \right| \to 0, \ \text{as} \ n \to \infty \ \ \textit{almost surely}. 
\end{align} 
where $\Phi$ denotes the distribution function of $\mathcal{N}(0,1)$. 
\end{theorem}

The proof of the CLT uses Stein's Lemma and the CLT immediately gives a stable convergence of a normalized sum of random variables under appropriate conditions. Suppose that  
\begin{align}
\sigma_n^2 / \left( n v^2 \right) \to 1 \ \textit{almost surely}, 
\end{align}
where $v^2$ is a random variable that is $\mathcal{C}-$measurable and $\mathcal{C}$ is a sub $\sigma-$field of $\mathcal{C}_n$ for all $n \geq 1$. Then, it follows that  $S_n / \sqrt{n}$ converges stably to a mixture normal random variable. 

\subsubsection{Network HAC}

In this section, we develop network HAC estimation of the conditional variance of $S_n / \sqrt{n}$ given $\mathcal{C}_n$, where $S_n = \sum_{ i \in N_n } Y_{n,i}$. We assume that $\mathbf{E} \left[ Y_{n,i} | \mathcal{C}_n \right] = 0$ \textit{almost surely} for all $i \in N_n$. Let    
\begin{align}
\Omega_n (s) = n^{-1} \sum_{ i \in N_n } \sum_{ j \in N_n^{\partial} (i;s) } \mathbf{E} \left[ Y_{n,i} Y^{\prime}_{n,j} \right].
\end{align}

\newpage

Then, the conditional variance of $S_n / \sqrt{n}$ given $\mathcal{C}_n$ is expressed as below
\begin{align}
V_n = \text{Var} \left( S_n / \sqrt{n} | \mathcal{C}_n \right) = \sum_{s \geq 0} \Omega_n(s) \ \textit{almost surely}
\end{align}
Similarly to the time-series case, the asymptotic consistency of an estimator $V_n$ requires a restriction on weights given to the estimated "autocovariance" terms $\Omega_n(.)$. Consider a kernel function $\omega: \mathbb{R} \to [-1,1]$ such that $\omega(0) = 1$, $\omega(z) = 0$ for $|z| > 1$, and $\omega(z) = \omega( -z)$ for all $z \in \mathbb{R}$. Let $b_n$ denote the bandwidth or the lag truncation parameter. Then, the kernel HAC estimator of $V_n$ is given by 
\begin{align}
\widetilde{V}_n = \sum_{ s \geq 0 } \omega_n (s) \widetilde{\Omega}_n(s),
\end{align}
where $\omega_n(s) = \omega( s / b_n )$, and 
\begin{align}
\widetilde{\Omega}_n(s) = n^{-1} \sum_{ i \in N_n } \sum_{ j \in N_n^{\partial} (i;s) } Y_{n,i} Y_{n,j}^{\prime}.
\end{align}

\begin{remark}
Notice that the weight given for each sample covariance term $\widetilde{\Omega}_n(s)$ is a function of distance $s$ implied by the structure of a network. Moreover, notice that if nodes $i$ and $j$ are disconnected then $d_n(i,j) = \infty$ so that $\omega_n \left( d_n(i,j) \right) = 0$. Moreover, unlike the time series case, the number of terms included in the double sum depends on the shape of the network. Hence, if there are many empty neighbourhood shells, a large value of the bandwidth can still produce a HAC estimator that performs well in finite samples.  We assume that $\mathbf{E} \left[ Y_{n,i} | \mathcal{C}_n \right] = \Lambda_n$ \textit{almost surely} for all $i \in N_n$ and the sequence of common conditional expectations $\left\{ \Lambda_n \right\}$ is unknown.     
\end{remark}
We have that $\bar{Y} = S_n / n$ is a consistent estimator of $\Lambda_n$ such that 
\begin{align}
\mathbf{E} \left[ \norm{ \bar{Y}_n - \Lambda_n } \big| \mathcal{C}_n \right] \to 0 \ \ \textit{almost surely}. 
\end{align}
We redefine the kernel HAC estimator as follows
\begin{align}
\widehat{V}_n = \sum_{ s \geq 0 } \omega_n (s) \widehat{\Omega}_n(s), 
\end{align}
where
\begin{align}
\widetilde{\Omega}_n(s) = n^{-1} \sum_{ i \in N_n } \sum_{ j \in N_n^{\partial} (i;s) } \left( Y_{n,i} - \bar{Y}_n \right) \left( Y_{n,j} - \bar{Y}_n \right)^{\prime}.
\end{align}

\newpage

\subsection{Stable Limit Theorems under Network Dependence}

We borrow some of the derivations presented in the framework of \cite{lee2019stable}. A relevant question of concern is: \textit{Does imposing conditions on the degree of connectedness in the network allows to control the stability of the nodes?}

\medskip

\begin{itemize}

\item Consider a new notion of stochastic dependence among a set of random variables. Suppose that we are given a set of random variables $\left\{  Y_i \right\}_{ i \in N_n }$ indexed by a set $N_n$, where the set $N_n$ is endowed with a neigborhood system so that each $i \in N_n$ is associated with a subset $\nu_n (i) \subset N_n \ \left\{ i \right\}$ called the neighborhood of $i$. In this paper, we call the map $\nu_n : N_n \to 2^{N_n}$ a \textit{neighborhood system}.  

\item Given a neighborhood system $\nu_n$ and a set of $\sigma-$fields $\mathcal{M} \equiv \left( \mathcal{M}_i \right)_{ i \in N_n }$, we say that $\left\{  Y_i \right\}_{ i \in N_n }$ is \textit{conditionally neighborhood dependent} (CND) with respect to $\left( \nu_n, \mathcal{M}   \right)$ if for any two non-adjacent subsets $A$ and $B$ of $N_n$, $\left( \left( Y_i \right)_{ i \in A },  \left( \mathcal{M}_i \right)_{ i \in A } \right)$ and $\left( \left( Y_i \right)_{ i \in B },  \left( \mathcal{M}_i \right)_{ i \in B } \right)$ are conditionally independent given $\left( \mathcal{M}_i \right)$, where $\nu_n \left( A \right)$ is the union of the neighbourhoods of $i \in A$ with the set $A$ left removed. 

\item The CND property is a generalization of both dependency graphs and Markov random fields with a global Markov property. A set of random variables have a graph as a dependency graph, if two sets of random variables are allowed to be dependent only when the two sets are adjacent in the graph. This dependence can be viewed as restrictive in many applications, as it requires that any random variables be independent even if their indices are indirectly connected in the graph. In contrast, CND random variables are allowed to be dependent even if they are not adjacent in the graph. 

\item The CND property captures the notation that "any two random variables are independent once we condition on the source of their joint dependence".  In this sense, the CND property is closely related to a Markov property in the literature of random fields. However, in contrast to the Markov property, the CND property does not require that the $\sigma-$fields $\mathcal{M}_i$ be generated by $Y_i$ itself.

\end{itemize}

\subsubsection{Stable Convergence of an Empirical Process}

Suppose that $\left\{  Y_i \right\}_{ i \in N_n }$  is a given triangular array of $\mathbb{R}-$valued random variables with is CND with respect to $\left( \nu_n, \mathcal{M} \right)$. Let $\mathcal{H}$ be a given class of real measurable functions on $\mathbb{R}$, having a measurable envelope H. Then, we consider the following empirical process: 
\begin{align}
\big\{ \mathbb{G}_n (h) : h \in \mathcal{H} \big\}, 
\end{align}
where for each $h \in \mathcal{H}$, we have that 
\begin{align}
\mathbb{G}_n (h) = \frac{1}{ \sqrt{n} } \sum_{ i \in N_n } \bigg( h (Y_i) - \mathbf{E} \left[ h(Y_i) | \mathcal{M}_{\nu_n (i) } \right] \bigg). 
\end{align} 

\medskip

The empirical process $\nu_n$ takes a value in $\ell^{\infty} \left( \mathcal{H} \right)$, the collection of bounded functions on $\mathcal{H}$ which is endowed with the sup norm so that $\big( \ell^{\infty} \left(  \mathcal{H} \right), \norm{ . }_{\infty} \big)$ forms the metric space $\left( \mathbb{D}, d \right)$ with the sup norm $\norm{h}_{\infty} \equiv \text{sup}_{y \in \mathbb{R}} \left|  h(y)  \right|$. \cite{lee2019stable} explore conditions for the class $\mathcal{H}$ and the joint distribution of the triangular array $\left\{  Y_i \right\}_{ i \in N_n }$  which delivers the stable convergence of the empirical process. Stable convergence in completer separable metric spaces can be defined as a weak convergence of Markov kernels. However, this definition does not extend to the case of empirical processes taking values in $\mathbb{D}$ that is endowed with the sup norm, due to the non-measurability. Weak convergence of an empirical process to a Gaussian process is often established in three steps. First, they  show that the class of functions is totally bounded with respect to a certain pseudo-metric $\rho$. Second, they show that each finite dimensional projection of the empirical process converges in distribution to a multivariate normal random vector. Third, they establish the asymptotic $\rho-$equicontinuity of the empirical process. 

\subsubsection{Maximal Inequality}

This subsection presents a maximal inequality in terms of bracketing entropy bounds. The maximal inequality is useful primarily for establishing asymptotic $\rho-$equicontinuity of the empirical process. We begin with a tail bound for a sum of CND random variables. The following exponential tail bound is crucial for our maximal inequality. 

\begin{lemma}[\cite{lee2019stable}]
Suppose that $\left\{ X_i \right\}_{ i \in N_n }$ is a triangular array of random variables that take values in $[-M, M]$ and are CND with respect to $( \nu_n, \mathcal{M} )$, with $\mathbf{E} \left[ X_i | \mathcal{M}_{\nu_n(i)} \right] = 0$, and let $\sigma_i^2 = \text{Var} \left( X_i | \mathcal{M}_{ \nu_n(i) } \right)$ and $V_n = \sum_{i \in N_n } \mathbf{E} \left[ \sigma_i^2 | \mathcal{G} \right]$ with $\mathcal{G}$ defined above. 

Then, for any $\eta > 0$ we have that
\begin{align}
\mathbb{P} \left\{ \left| \sum_{i \in N_n} X_i \right| \geq \eta | \mathcal{G} \right\} \leq 2 \text{exp} \left( - \frac{\eta^2}{ 2 (d_{mx} +1) \left[ 2 (d_{mx} + 1)  V_n + M \frac{\eta}{3} \right] }  \right)
\end{align}
\textit{almost surely} for all $n \geq 1$. 
\end{lemma} 

Furthermore, if Condition A holds and the condition $\mathbf{E} \left[ X_i | \mathcal{M}_{\nu_n(i)} \right] = 0$ is replaced by $\mathbf{E} \left[ X_i | \mathcal{G} \right] = 0$ and the $\sigma-$fields $\mathcal{M}_{\nu)n(i)}$ in $\sigma_i'$s are replace by $\mathcal{G}$, then the following hods for any $\eta > 0$, 
\begin{align}
\mathbb{P} \left\{ \left| \sum_{i \in N_n} X_i \right| \geq \eta | \mathcal{G} \right\} \leq 8 \text{exp} \left( - \frac{\eta^2}{ 25  (d_{mx} + 1) \left( V_n + M \frac{\eta}{3} \right)} \right)
\end{align}

\newpage

\begin{corollary}[\cite{lee2019stable}]
Suppose that $\left\{  Y_i \right\}_{ i \in N_n }$ is a triangular array of random variables that are CND with respect to $\left( \nu_n, \mathcal{M} \right)$. Let for each $h \in \mathcal{H}$,
\begin{align}
V_n(h) &= n^{-1}\sum_{i \in N_n } \mathbf{E} \left[ \sigma_i^2 (h) | \mathcal{G} \right]
\\
\sigma_i^2 &= \text{Var} \left( h(Y_i) | \mathcal{ M}_{ \nu_n (i)}    \right) 
\end{align}
with $\mathcal{G}$ as defined above. Then, there exists an absolute constant $C > 0$ such that 
\begin{align}
\mathbf{E} &\bigg[ \underset{ 1 \leq s \leq m }{ \text{max} }  | \mathbb{G}_n ( h_s ) | \bigg| \mathcal{G} \bigg]   
\leq C \left( d_{mx} + 1 \right) \left( \frac{J}{ \sqrt{n} } \text{log} ( 1 + m ) + \sqrt{\text{log} (1+m) \underset{ 1 \leq s \leq m }{ \text{max} } V_n (h_s) } \right)
\end{align} 
\textit{almost surely} for any $n \leq 1$ and any $m \geq 1$ with a finite subset $\left\{ h_1,...,h_m \right\}$ of $\mathcal{H}$ such that for some constant $J > 0$, it holds that
\begin{align*}
\underset{ 1 \leq s \leq m }{ \mathsf{max} } \ \underset{ x \in \mathbb{R}  }{ \mathsf{sup} } \ \big| h_s(x) \big| \leq J.
\end{align*} 
 
\end{corollary}

\subsubsection{Stable Central Limit Theorem}  

\medskip

\begin{definition}[\cite{lee2019stable}]
Assume that a stochastic process $\left\{ \mathbb{G} (h): h \in \mathcal{H} \right\}$ is a $\mathcal{G}-$\textit{mixture Gaussian} process if for any finite collection $\left\{ h_1, ... , h_m \right\} \subset \mathcal{H}$, the distribution of random vectors
\begin{align*}
\big[\mathbb{G}(h_1) ,...,  \mathbb{G}(h_m) \big]
\end{align*}
conditional on $\mathcal{G}$ is a multivariate normal distribution.   
\end{definition}

\begin{theorem}[\cite{lee2019stable}]
Suppose that $\left\{  Y_i \right\}_{ i \in N_n }$ is a triangular array of random variables that are CND with respect to $\left( \nu_n, \mathcal{M} \right)$, satisfying Assumption 3.1. Suppose further that there exists $C > 0$ such that for each $n \geq 1$, $\text{max}_{ i \in \mathcal{N}_n} \ \mathbf{E} \left[ H(Y_i)^4 \right] < C$, where $H$ is an envelope of $\mathcal{H}$. 

Then $\nu_n$ converges to a $\mathcal{G}-$mixture Gaussian process $\mathbb{G}$ in $\ell^{\infty} \left( \mathcal{H} \right)$, $\mathcal{G}-$stably, such that for any 
\begin{align}
h_1, h_2 \in \mathcal{H}, \mathbf{E} \left[ \mathbb{G}(h_1) \mathbb{G}(h_2)  | \mathcal{G} \right] = K \left( h_1, h_2 | \mathcal{G} \right)
\end{align}
\textit{almost surely}. 
\end{theorem}

\begin{remark}
The two main theorems which are considered to be the building tools when developing asymptotic theory in time series regression models, are the \textit{continuous mapping theorem} as well as the implementation of a suitable \textit{invariance law for partial sum processes}; both are used to establish the weak convergence theory. Furthermore, \textit{empirical processes} have been extensively used the past decades to help in the development of asymptotic statistics of various inference problems. However, one area which has not been studied in the literature in great length is a modification for empirical processes
to accommodate network dependence in data.
\end{remark}

\newpage 

\subsection{Moderate Deviations: Eigenvalues of Covariance Matrices}

Although due to the convenient structure of the covariance matrix, the estimation procedure of the entries of the matrix is many times overlooked therefore the spectral analysis of high dimensional matrices mainly focused on the asymptotic behaviour of the eigenvalues without taking into consideration how the entries are estimated. In our framework we give emphasis on the estimation of the risk matrix with the use of quantile predictive regression models which clearly accounts for more features regarding the stationarity and persistence properties of the time series. In practise, the estimation of the covariance matrix with the conventional methodologies which consider the use of a high dimensional vector does not account for much features of the time series under examination (see, \cite{katsouris2021optimal, katsouris2023statistical}).

\medskip

According to \cite{zhang2020estimation}, in recent years, large sample properties for high-dimensional sample covariance matrices, including their eigenvalues and eigenvectors, have been proved to be useful. In fact, random matrix theory (RMT) provides a plethora of useful methods for estimation and testing in high dimensional environments. Currently, there are two steams of literature about asymptotic theory of the largest eigenvalues of high dimensional random matrices. The first stream of literature is concerned with the Tracey-Widom law. In particular, it is well known that the limiting distribution of the largest eigenvalues of high-dimensional random matrices, such as Wigner matrices, follow the Tracy-Widom law which holds for Gaussian Wigner ensembles. Moreover, the second stream of literature is interested in the asymptotic behaviour of spiked eigenvalues.

Most of the existing studies operate under the assumption that observations of high dimensional data are independent. However, applications in finance and economics with high dimensional settings involve data which can be temporally dependent or even nonstationary. For instance, the first framework  that examines the asymptotic behaviour of the largest eigenvalues of sample covariance matrices generated from high-dimensional nonstationary time series is proposed by \cite{zhang2018clt}. Investigating the asymptotic  behaviour of the largest eigenvalues of random matrix models is important in understanding the stability behaviour of these systems. See Remark 4 (page 34) of \cite{grama2006asymptotic} who mention: \textit{"The key point in the proof when considering the eigenvalues of a covariance or state matrix, is the fact that the eigenvalues of a matrix depend continuously on its elements. In practise, this means that the maximal eigenvalue of a matrix is a predictable discrete time process in $t$ and therefore the random times $\tau_m$ are predictable stopping times"}. 
\color{black}

\newpage

\subsubsection{Limit of the smallest eigenvalue of a Large Sample Covariance Matrix}

The related framework is presented by \cite{bai1999exact}.  

\begin{lemma}
We assume that the entries of $X_n$ have already been truncated at $\delta \sqrt{n}$ for some slowly varying $\delta = \delta_n \to 0$. We define with 
\begin{align}
V_{ij} = X_{ij} \mathbf{1} \big( \left| X_{ij} \right| \leq \delta \sqrt{n} \big) - \mathbb{E} \left[ X_{ij} \mathbf{1} \big( \left| X_{ij} \right| \leq \delta \sqrt{n} \big) \right].
\end{align}
Moreover, it has been proved that 
\begin{align}
\sum_{i=1}^p \lambda_i \eta_i \geq \text{trace} \left( A^{\prime} B \right),
\end{align}
if $A$ and $B$ are $p \times n$ matrices with singular values $\lambda_1 \geq ... \geq \lambda_p$ and $\eta_1 \geq ... \geq \eta_p$, respectively. Therefore, using the von Neumann's inequality, we have that 
\begin{align*}
&\left| \lambda_{ \text{min} }^{1/2} \left( n^{1/2} \hat{X}_n \hat{X}^{\prime}_n \right) - \lambda_{ \text{min} }^{1/2} \left( n^{1/2} \hat{V}_n \hat{V}^{\prime}_n \right)  \right|
\\
&\leq  \sum_{i=1}^p \left\{  \lambda_{ k }^{1/2} \left( n^{1/2} \hat{X}_n \hat{X}^{\prime}_n \right) - \lambda_{ k }^{1/2} \left( n^{1/2} \hat{V}_n \hat{V}^{\prime}_n \right) \right\}^2
\\
&\leq \frac{1}{n} \text{trace} \left( \hat{X}_n - V_n \right) \left( \hat{X}_n - V_n \right)^{\prime} 
\\
&\leq p \mathbb{E}^2 \left| X_{11} \right| \mathbf{1} \left( \left| X_{11} \right| > \delta \sqrt{n} \right) \to 0, 
\end{align*}
where $\hat{X}_n$ and $V$ are $n \times p$ matrices with $( u, v)-$th entries $X_{u,v} \mathbf{1} \left(  \left| X_{11} \right| \leq \delta \sqrt{n} \right)$ and $V_{uv}$, respectively. The above convergence is true provided that $n \delta^3 \to 0$. Therefore, we assume that for each $n$ the entries $X_{uv} = X_{uv} (n)$ of the matrix $X_n$ are \textit{i.i.d} and satisfy
\begin{align}
\mathbb{E} X_{uv} = 0, \ \ \mathbb{E} X^2_{uv} \leq 1 \ \ \text{and} \ \ \mathbb{E} X^2_{uv} \to 1 \ \ \text{as} \ \ n \to \infty,
\end{align}
\end{lemma}

\paragraph{Proof of Lemma 3 of the paper}

The following inequality holds, 
\begin{align*}
\norm{ \frac{1}{n} X^{(1)} X^{(1) \prime} } 
&\leq 
\norm{ T(1) } + \norm{ \text{diag} \left[ \frac{1}{n} \sum_{j=1}^n X_{ij}^2 , \ i = 1,...,p \right] }
\\
&\leq \norm{ T(1) } + \frac{1}{n} \underset{ i \leq p }{ \text{max}  } \sum_{j=1}^n X_{ij}^2. 
\end{align*}

\newpage

\subsubsection{Limit Theory for the largest eigenvalues of sample covariance matrices with heavy-tails}

Theory that consistently estimates the spectrum of a large dimensional covariance matrix using RMT. Thus, statistical considerations will be our motivation for a random matrix model with heavy-tailed and dependent entries.
Notice that results on the global behaviour of the eigenvalues of $XX^{\top}$ mostly concern the spectral distribution, that is, the random probability measure of its eigenvalues $p^{-1} \sum_{i=1}^p \epsilon_{ n^{-1} \lambda_{(i)} }$, where $\epsilon$ denotes the Dirac measure.  The spectral distribution converges, as $n, p \to \infty$ with $p / n \to \gamma \in (0,1]$, to a deterministic measure with density function 
\begin{align}
\frac{1}{2 \pi x \gamma } \sqrt{ ( x_{+} - x ) ( x - x_{-} ) } \mathbf{1} ( x_{-}, x_{+} ) (x),
\end{align}
This is the so-called Marcenko-Pastur law. Therefore, one obtains a different result if $X X^{\top}$ is perturbed via an affine transformation. Although the eigenvalues of $XX^{\top}$ offer various interesting local properties to be examined, we will only focus on the joint asymptotic behaviour of the $k$ largest eigenvalues $\left( \lambda_{(1)}, ..., \lambda_{(k)} \right), k \in \mathbb{N}$. This is motivated from a statistical point of view since the variances of the first $k$ principal components are given by the $k$ largest eigenvalues of the covariance matrix. Therefore, a Marcenko-Pastur type result holds for the asymptotics of the spectrum of heavy-tailed random matrices.

\subsubsection{Limit Spectral Distribution for symmetric random matrices with correlated entries}

Let $\left( X_{k, \ell} \right)_{ ( k, \ell ) \in \mathbb{Z}^2 }$ be an array of real-valued random variables, and consider its associated systemic random matrix $\mathbf{X}_n$ of order $n$ defined by 
\begin{align}
\left( \mathbf{X}_n \right) = X_{i,j}, \ \text{if} \ 1 \leq j \leq i \leq n \ \ \ \text{and} \ \ \ \left( \mathbf{X}_n \right) = X_{j,i}, \ \text{if} \ 1 \leq i \leq j \leq n 
\end{align}

We then define 
\begin{align}
\mathbb{X}_n := n^{- 1 / 2} \mathbf{X}_n.  
\end{align}

The aim of this section is to study the limiting spectral empirical distribution function of the symmetric matrix $\mathbf{X}_n$ when the process $\left( X_{k, \ell} \right)_{ ( k, \ell ) \in \mathbb{Z}^2 }$ has the following dependence structure: for any $(k , \ell) \in \mathbb{Z}^2$, 
\begin{align}
X_{ k, \ell } = g \left( \xi_{k - i, \ell - j} : (i, j) \in \mathbb{Z}^2 \right), 
\end{align} 
where $\left( \xi_{ i, j } \right)_{ (i,j) \in \mathbb{Z}^2}$ is an array of \textit{i.i.d} real-valued random variables given on a common probability space $\left( \Omega, \mathcal{F}, \mathbb{P} \right)$ and $g$ is a measurable function $\mathbb{R}^{ \mathbb{Z}^2 } \mapsto \mathbb{R}$ such that $\mathbb{E} \left( X_{0,0} \right) = 0$ and $\norm{ X_{0,0} }_2 < \infty$.

\newpage

The Theorem below shows a universality scheme for the random matrix $\mathbb{X}_n$ when the entries of the symmetric matrix $\sqrt{n} \mathbb{X}_n$ have the above dependence structure. Notice that the particular result does not require rate of convergence to zero of the correlation between the entries of the risk matrix. 

\bigskip

\begin{theorem}
\label{theoremABC}
Let $\left( X_{k, \ell} \right)_{ ( k, \ell ) \in \mathbb{Z}^2 }$ be a real-valued stationary random field. Define the symmetric matrix $\mathbf{X}_n$. let $\left( G_{k, \ell} \right)_{ ( k, \ell) \in \mathbb{Z}^2 }$ be a real-valued centred Gaussian random field, with covariance function given by 
\begin{align}
\mathbb{E} \left(  G_{k, \ell} G_{ i, j }  \right) = \mathbb{E} \left( X_{k, \ell} X_{i,j} \right) \ \ \text{for any} \ \ (k, \ell) \ \text{and} \ (i,j) \ \text{in} \ \mathbb{Z}^2. 
\end{align} 
Let $\mathbf{G}_n$ be the symmetric random matrix defined by $\left( \mathbf{G}_n \right)_{i,j} = G_{i,j}$ if $1 \leq j \leq i \leq n$ and $\left( \mathbf{G}_n \right)_{i,j} = G_{j,i}$ if $1 \leq i \leq j \leq n$. Denote with $\mathbb{G}_n = \frac{1}{ \sqrt{n} } \mathbf{G}_n$. Then, for any $z \in \mathbb{C}^{+}$, 
\begin{align}
\underset{ n \to \infty }{ \text{lim} } \left| S_{ \mathbb{X}_n } (z) - \mathbb{E} \left( S_{ \mathbb{X}_n } (z) \right) \right| = 0, \ \text{almost surely}. 
\end{align}
\end{theorem}
The above theorem is important since it shows that the study of the limiting spectral distribution function of a symmetric matrix whose entries are functions of $\textit{i.i.d}$ random variables can be reduced to studying the same problem as for a Gaussian matrix with the same covariance structure. 

\paragraph{Proof of Theorem \ref{theoremABC}} 

\begin{proof}
For $m$ a positive integer (fixed for the moment) and for any $(u,v) \in \mathbb{Z}^2$ define
\begin{align}
X_{u,v}^{ (m) } = \mathbb{E} \left( X_{u,v} | \mathcal{F}_{u,v}^{ (m) }  \right)
\end{align}
where $mathcal{F}_{u,v}^{ (m) } := \sigma \left( \xi_{i,j} : u - m \leq i \leq u + m,  v- m \leq j \leq v + m  \right)$. 

Let $\mathbf{X}_n^{(m)}$ be the symmetric random matrix of order $n$ associated with $\left( X_{u,v}^{ (m) } \right)_{ (u,v) \in \mathbb{Z}^2 }$ and defined by $\left(  \mathbf{X}_n^{(m)}  \right)_{i,j} = X_{ i, j }^{ (m) }$ if $1 \leq j \leq i \leq n$ and  $\left(  \mathbf{X}_n^{(m)}  \right)_{i,j} = X_{ j, i }^{ (m) }$ if $1 \leq i \leq j \leq n$. Let $\mathbb{X}_n^{ (m) } = n^{- 1 / 2} \mathbf{X}_n^{ (m) }$.  

We first show that, for any $z \in \mathbb{C}^{+}$, 
\begin{align}
\underset{ m \to \infty }{ \text{lim} } \underset{ n \to \infty }{ \text{lim sup} } \big|  S_{ \mathbb{X}_n } (z) - \mathbb{E} \left( S_{ \mathbb{X}_n } (z) \right) \big| = 0, \ \text{almost surely}, 
\end{align}
Moreover, we have that 
\begin{align}
\big|  S_{ \mathbb{X}_n } (z) - \mathbb{E} \left( S_{ \mathbb{X}_n } (z) \right) \big|^2 \leq \frac{2}{n^2 v^4 } \sum_{ 1 \leq \ell \leq k \leq n } \left( X_{k, \ell} -  X_{k, \ell}^{(m)} \right)^2. 
\end{align}

\newpage

Since the shift is ergodic with respect to the measure generated by a sequence of $\textit{i.i.d}$ random variables and the sets of summations are on regular sets, the ergodic theorem entails that 
\begin{align}
\underset{ n \to \infty }{ \text{lim} } \sum_{ 1 \leq k, \ell \leq n } \left( X_{k, \ell} - X_{k, \ell}^{(m)} \right)^2 = \mathbb{E} \left( \left( X_{0,0} - X_{0,0}^{(m)} \right)^2 \right) \ \ \text{almost surely}. 
\end{align} 
Therefore, 
\begin{align}
\underset{ n \to \infty }{ \text{lim sup} } \left| S_{ \mathbb{X}_n } (z) - \mathbb{E} \left( S_{ \mathbb{X}_n } (z) \right) \right|^2 \leq 2 v^{-4} \norm{ X_{0,0} - X_{0,0}^{(m)} }_2^2 \ \text{almost surely}. 
\end{align}
But by the martingale convergence theorem
\begin{align}
\norm{ X_{0,0} - X_{0,0}^{(m)} }_2 \to 0 \ \text{as} \to \infty,
\end{align}
\end{proof}

\subsection{Moderate Deviations for Extreme Eigenvalues}

According to \cite{jiang2021moderate} the Marcehnko-Pastur Law implies that the empirical spectral distribution of $W$ converges to a deterministic distribution with support $\left[ \left( 1 - \sqrt{ \beta } \right)_{+}^2 , \left( 1 + \sqrt{ \beta } \right)^2 \right]$, where $x_{+} = \text{max} \left\{ 0, x \right\}$.

\begin{theorem}
Let $k \to \infty$ as $n \to \infty$. Moreover, suppose that each of the entries $c_{ij}$ is symmetric around 0 and $| c_{ij} | < M < \infty$ almost surely or $c_{ij}$ is standard normal. If $\text{Var} \left( c_{11}^2 \right) > 0$, we have

\begin{enumerate}

\item[(a)] (Moderate deviation for $\lambda_{ \text{min} }$) For any $\nu \leq 0$, 
\begin{align}
\underset{ n \to \infty }{ \text{lim} } \frac{1}{ \ell_n^2 } \text{log} \mathbb{P} \left( \ell_n^{-1} n^{1 / 2} \left( \lambda_{ \text{min} } - 1 \right) \leq \nu \right) = - \frac{1}{ 2 \sigma_{ \infty } } \nu^2; 
\end{align}

\item[(a)] (Moderate deviation for $\lambda_{ \text{max} }$) For any $\nu \geq 0$,

\begin{align}
\underset{ n \to \infty }{ \text{lim} } \frac{1}{ \ell_n^2 } \text{log} \mathbb{P} \left( \ell_n^{-1} n^{1 / 2} \left( \lambda_{ \text{max} } - 1 \right) \geq \nu \right) = - \frac{1}{ 2 \sigma_{ \infty } } \nu^2. 
\end{align} 
\end{enumerate} 
In both cases, $\sigma_{\infty}^2 = \text{lim}_{ k \to \infty } \sigma_k^2 = \text{max} \left\{ 2, \text{Var} ( c_{11}^2 ) \right\}$. 
\end{theorem}

As a corollary we can obtain the following moderate deviation for the condition number $\lambda_{\text{max} } / \lambda_{\text{min} }$, which is important quantity that can indicate whether the matrix is ill-conditioned.

\newpage

\subsubsection{Moderate Deviations for $\lambda_{ \text{max} }$, $\lambda_{ \text{min} }$ with fixed $k$} 

Suppose there exist two sequences $b_n \left( b_n \to \infty \right)$, $r_n$ so that a family of random variables $\left\{ Z_n , n \geq 0 \right\}$ with values in topological vector space $\mathcal{X}$ (equipped with $\sigma-$field $\mathcal{B}$) satisfies a fluctuation theorem, say $b_n ( Z_n - r_n )$ converges to some nontrivial distribution. We say $Z_n$ satisfies the moderate deviation principle with speed $\lambda_n \to \infty$ and with good rate function $I(.)$ if the level sets $\left\{  I \leq \ell \right\}$ are compact for all $\ell > 0$ and for any closed set $F$ and open set $G \in \mathcal{B}$,
\begin{align*}
\underset{ n \to \infty }{ \text{lim sup} } \frac{1}{ \lambda_n } \text{log} \mathbb{P} \left( \tilde{b}_n ( Z_n - r_n ) \in F  \right) \leq - \underset{ x \in F }{ \text{inf} } I(x)
\\
\underset{ n \to \infty }{ \text{lim sup} } \frac{1}{ \lambda_n } \text{log} \mathbb{P} \left( \tilde{b}_n ( Z_n - r_n ) \in G  \right) \geq - \underset{ x \in G }{ \text{inf} } I(x).
\end{align*}   

Here, $\tilde{b}_n$ is a sequence satisfying $\tilde{b}_n \to \infty$ and $\tilde{b}_n / b_n \to 0$. The form of rate function $I(.)$ is closely related to the limit distribution of $b_n ( Z_n - r_n )$. Therefore, if the rate function is continuous we have that 
\begin{align*}
\underset{ n \to \infty }{ \text{lim} } \frac{1}{ \lambda_n } \text{log} \mathbb{P} \left( \tilde{b}_n ( Z_n - r_n ) \leq x \right) = - \underset{ u \leq x }{ \text{inf} } I(u) 
\\
\underset{ n \to \infty }{ \text{lim} } \frac{1}{ \lambda_n } \text{log} \mathbb{P} \left( \tilde{b}_n ( Z_n - r_n ) \geq y \right) = - \underset{ u \geq x }{ \text{inf} } I(u)
\end{align*}

\subsubsection{Discussion on High Dimensional Results}

The literature has indeed documented the inconsistency of sample covariance matrices for especially high-dimensional settings. In order to be able to investigate the features of our novel tail risk matrix, it is necessary to consider the distribution theory of the entries of the matrix. Overall to assess the quality of a matrix estimate, we can use the norm operator such as the spectral radius. Below we present some useful results (e.g. see, \cite{chen2013covariance}).  
\begin{lemma}
Let $Z_i$ be $\textit{i.i.d}$  $\mathcal{N} \left( \mathbf{0}, \boldsymbol{\Sigma}_p  \right)$ and $\lambda_{ \text{max} } \left( \boldsymbol{\Sigma}_p \right) \leq \bar{k} < \infty$. Then if $\boldsymbol{\Sigma}_p = [ \sigma_{ab} ]$, 
\begin{align}
\mathbb{P} \left[ \left| \sum_{i=1}^n \left( Z_{ij} Z_{ik} - \sigma_{jk} \right) \right| \geq  n \nu \right] \leq c_1 \text{exp} \left( - c_2 n \nu^2 \right) \ \ \ \text{for} \ \ \ | \nu | \leq \delta
\end{align}
where $c_1, c_2$ and $\delta$ depend on $\bar{k}$ only. 
\end{lemma}
Suppose we have $n$ temporally observed $p-$dimensional vectors $\left( \mathbf{z}_i \right)_{ i = 1}^n$ having mean zero and covariance matrix $\Sigma_i = \mathbb{E} \left( \mathbf{z}_i, \mathbf{z}_i \right)$ whose dimension is $p \times p$. Our goal is to estimate the covariance matrices $\Sigma_i$ and their inverses $\Omega_i = \Sigma_i^{-1}$ based on the data matrix $Z_{ p \times n} = \left( \mathbf{z}_1,..., \mathbf{z}_n \right)$. In the classical situation where $p$ is fixed, $n \to \infty$ and $\mathbf{z}_i$ are mean zero independent and identically distributed $\textit{i.i.d}$ random vectors.

\newpage

In particular, it is well known that the sample covariance matrix $\hat{ \boldsymbol{\Sigma} }_n = \frac{1}{n} \sum_{ i = 1 }^n \boldsymbol{X}_i \boldsymbol{X}_i^{ \top }$, is a consistent and well behaved estimator of $\boldsymbol{\Sigma}$, and $\hat{ \boldsymbol{\Omega} }_n = \hat{\boldsymbol{\Sigma}}_n^{-1}$ is a natural and good estimator of $\boldsymbol{\Omega}$. However, when the dimensionality $p$ grows with $n$, then $\hat{\boldsymbol{\Sigma}}_n$ is no longer a consistent estimate of $\boldsymbol{\Sigma}$ in the sense that its eigenvalues do not converge to those of $\boldsymbol{\Sigma}$. Moreover, it is clear that $\hat{ \boldsymbol{\Omega} }_n$ is not defined when $\hat{\boldsymbol{\Sigma}}_n$ is not invertible in the high-dimensional case with $p >> n$. Define with 
\begin{align*}
T_u \left( \widehat{\boldsymbol{\Sigma}}_u \right) = \mathbf{Q} \hat{ \boldsymbol{\Lambda} } \mathbf{Q}^{\top} = \sum_{ j = 1}^p \hat{\lambda}_j \mathbf{q}_j \mathbf{q}_j^{\top}
\end{align*}
the  eigen-decomposition, where $\mathbf{Q}$ is an orthonormal matrix and $\hat{ \boldsymbol{\Lambda} }$ is a diagonal matrix. For $v > 0$, 
\begin{align}
\tilde{S}_v = \sum_{j=1}^p \left( \hat{\lambda}_j \vee v \right) \mathbf{q}_j \mathbf{q}_j^{\top},  
\end{align}
where $0 < v \leq \sqrt{p} \bar{\omega}$ and $\omega^2$ is the rate of convergence. Let $\mu_1,...,\mu_p$ be the diagonal elements of $\mathbf{Q}^{\top} \boldsymbol{\Sigma} \mathbf{Q}$. Then, we have by Theorem 2.1 that $\sum_{j=1}^p \left( \hat{\lambda}_j - \mu_j \right)^2 \leq p^2 \bar{\omega}^2$, and consequently
\begin{align*}
\left| \tilde{S}_v - \boldsymbol{\Sigma} \right|^2_F 
&\leq 
2 \left| \tilde{S}_v - T_u \left( \hat{\boldsymbol{\Sigma}} \right) \right|^2_F + 2 \left| T_u \left( \hat{\boldsymbol{\Sigma}}_u \right) - \boldsymbol{\Sigma} \right|_F^2 
\\
&\leq 
2 \sum_{j=1}^p \left( \hat{\lambda}_j - \left( \hat{\lambda}_j  \vee v \right) \right)^2 + 2 \bar{ \omega }^2 p^2 
\\
&\leq 
2 \sum_{j=1}^p \left( 2 \hat{\lambda}_j^2 \mathbf{1} \left\{ \hat{\lambda}_j \leq 0 \right\} + 2 v^2 \right) + 2 \bar{ \omega }^2 p^2  .
\end{align*}
If $\hat{\lambda}_j \leq 0$, since $\mu_i \geq 0$, we have that $\left| \hat{\lambda}_j \right| \leq \left| \hat{\lambda}_j  - \mu_i \right|$. Then,  
\begin{align}
\left| \tilde{S}_v - \Sigma \right|_F^2 \leq 4 v^2 p + 6 \bar{\omega} p^2 \leq 10 \bar{\omega}^2 p^2. 
\end{align}
The eigenvalues of $\tilde{S}_v$ are bounded below by $v$, and thus it is positive definite. Suppose that
\begin{align*}
v = \left( p^{-1} \sum_{ j,k = 1}^p u^2 \times \mathbf{1} \left\{ \left| \hat{\sigma}_{jk} \right| \geq u \right\} \right)^{1 / 2}
\end{align*}
Thus, the same positive-definization procedure also applies to the spectral norm. We obtain that the differences between the eigenvalues and eigenvectors of $\Sigma_m$ and $\hat{\Sigma}_m$ can be bounded by the following 
\begin{align}
\big( \hat{ \lambda }_{r,m} - \lambda_r \big) 
= 
\text{trace} \big\{ e_{r,m} e_{r,m}^{\top} \left( \hat{\boldsymbol{\Sigma}}_m - \boldsymbol{\Sigma} \right) \big\} + \tilde{R}_{r,m}, \ \ \ 
\tilde{R}_{r,m} \leq \frac{ \displaystyle 6 \underset{ \norm{a} = 1 }{ \text{sup} } a^{\top} \left( \hat{\Sigma}_m - \Sigma \right) ^2 a   }{ \displaystyle  \underset{ s }{ \text{min} } \left| \lambda_s - \lambda_r   \right| }
\end{align}

\newpage 

Moreover, we have that 
\begin{align}
\big( \hat{ \zeta }_{r,m} - e_{r,m} \big)
=
 - S_{ r,m } \left( \hat{\Sigma}_m - \Sigma \right) e_{r,m} + R_{r,m}^{*},  
\end{align}
with 
\begin{align}
\norm{ R_{r,m}^{*} } \leq \frac{ \displaystyle 6 \underset{ \norm{a} = 1 }{ \text{sup} } a^{\top} \left( \hat{\Sigma}_m - \Sigma \right) ^2 a   }{ \displaystyle  \underset{ s }{ \text{min} } \left| \lambda_s - \lambda_r   \right| },
\end{align}
where we denote with $S_{r,m} = \sum_{ s \neq r } \frac{1}{ \lambda_s - \lambda_r } e_{s,m} e_{s,m}^{\top}$. Assumption 1 implies that $\mathbb{E} \left( \hat{ \beta }_r \right) = 0$ and Var$\left( \hat{ \beta }_r \right) = \frac{ \lambda_r }{ r }$ and with $\delta_{ii} = 1$ as well as $\delta_{ij} = 0$ for $i \neq j$, we obtain that 
\begin{align*}
\mathbb{E} \left\{ \underset{ \norm{ a } = 1 }{ \text{sup} } a^{\top} \left( \hat{\boldsymbol{\Sigma}}_m - \boldsymbol{\Sigma} \right) \right\} 
&\leq 
\mathbb{E}  \bigg\{  \text{trace} \left[ \left( \hat{\boldsymbol{\Sigma}}_m - \boldsymbol{\Sigma} \right)^2 \right]  \bigg\} 
\\
&= 
\mathbb{E} \left\{ \sum_{j,k = 1}^m \left[ \frac{1}{n} \left( \beta_{ji} - \bar{\beta}_j \right) \left( \beta_{ki} - \bar{\beta}_k \right) - \delta_{jk} \lambda_j \right]^2 \right\}
\\
&\leq 
\mathbb{E} \left\{ \sum_{j,k = 1}^{ \infty } \left[ \frac{1}{n} \left( \beta_{ji} - \bar{\beta}_j \right) \left( \beta_{ki} - \bar{\beta}_k \right) - \delta_{jk} \lambda_j \right]^2 \right\}
\\
&= 
\frac{1}{n} \left( \sum_j \sum_k \mathbb{E} \left( \beta_{ji}^2 \beta_{ki}^2 \right) \right) + o( n^{-1} ) = \mathcal{O} \left( n^{-1} \right)
\end{align*}
for all $m$. 

Furthermore, since  
\begin{align}
\mathsf{trace} \big\{ e_{r,m} e_{r,m}^{\top} \left( \hat{\boldsymbol{\Sigma}}_m - \boldsymbol{\Sigma}_m \right) \big\} = \frac{1}{n} \sum_{ i = 1 }^n \left( \beta_{ri} - \bar{ \beta }_r \right)^2 - \lambda_r
\end{align}
Therefore after applying the central limit theorem we obtain that
\begin{align*}
\sqrt{n} \left( \hat{\lambda}_r - \lambda_r \right) 
&= 
\frac{1}{ \sqrt{n} } \sum_{ i = 1 }^n \left( \beta_{ri} - \hat{\beta}_r \right)^2 - \lambda_r + \mathcal{O}_p \left( n^{- 1 / 2} \right)
\\
&= 
\frac{1}{ \sqrt{n} } \sum_{ i = 1 }^n \left\{ \left( \beta_{ri} \right)^2 - \mathbb{E} \left[  \left( \beta_{ri} \right)^2 \right] \right\} + \mathcal{O}_p \left( n^{- 1 / 2} \right)
\to \mathcal{N} \left( 0, \boldsymbol{\Lambda}_r \right).
\end{align*}
Obviously the event $\hat{\lambda}_{r-1} > \hat{\lambda}_{r} > \hat{\lambda}_{r+1}$ occurs with probability 1. 


\color{black}

\newpage 

\section{Nonstationary Time Series Regressions}

\subsection{Testing for unit root in time series regression}

\begin{example}
Consider the zero-mean Gaussian AR(1) model where $\left\{ y_t : 1 \leq t \leq T \right\}$ generated as below
\begin{align}
y_t = \rho y_{t-1} + \epsilon_t,
\end{align}
where $y_0 = 0$ and $\epsilon_t \sim_{\textit{i.i.d} } \mathcal{N} (0,1)$. 
Define with 
\begin{align}
S_T = \frac{1}{T} \sum_{t=1}^T y_{t-1} \Delta y_t \ \ \ \ \text{and} \ \ \ \  H_T = \frac{1}{T} \sum_{t=1}^T y^2_{t-1}
\end{align}
The large sample behaviour of $\big( S_T, H_T \big)$ is well understood. More specifically, under local-to-unity asymptotics with $c = T ( \rho - 1 )$ held fixed as $T \to \infty$, 
\begin{align}
\big( S_T, H_T \big) \overset{ d }{ \to } \left( \int_0^1 W_c(r) dW_c(r), \int_0^1 W_c(r)^2 dr \right), 
\end{align}
where $W_c(r) = \displaystyle \int_0^1 \mathsf{exp} \big\{ c(r-s) \big\} dW(s)$
and $W(.)$ is a standard Wiener process. 
\end{example}

\begin{example}
Consider the framework of \cite{phillips1987time}, such that $\left\{ y_t \right\}$ be a time series generated by
\begin{align}
y_t &= \alpha y_{t-1} + u_t, \ \ \ t = 1,2,... \ \ \ \text{with} \ \alpha = 1
\end{align}
\end{example}
Consider the sequence of partial sums $\left\{ S_t \right\}$, then we obtain that 
\begin{align}
X_T (r) = \frac{1}{ \sqrt{T} \sigma } S_{ [Tr] } = \frac{1}{ \sqrt{T} \sigma } S_{j-1}
\end{align}
for $\frac{( j - 1)}{T} \leq r \leq \frac{j}{T}$ for $j \in \left\{1,..., T \right\}$, where $[Tr]$ denotes the integral part of Tr. Then, $X_T(r)$ lies in $\mathcal{D} = \mathcal{D}[0,1]$, the space of real valued functions on the interval $[0,1]$ that are right continuous and have finite left limits. Under very general conditions the random element $X_T(r)$ obeys a central limit theory on the function space $\mathcal{D}$. In particular, we have that, as $T \to \infty$, 
\begin{align}
X_T(r) \Rightarrow W(r). 
\end{align}  
which implies a weak convergence on the associated probability measure. In particular, in this case, the probability measure of $X_T(r)$ converges weakly to the probability measure of the standard Brownian motion $W(r)$. We focus on the asymptotic behaviour of the sample moments of the process $\left\{ S_t \right\}$ and the innovations $\left\{ u_t \right\}$. We express the limit distributions as functions of standard Brownian motion $W(r)$.

\newpage

All integrals are understood to be taken over the interval $[0,1]$, while integrals such as $\int W$,  $\int W^2$, $\int r W$ are understood to be taken with respect to Lebesgue measure such that we write $W_1 = W(1)$. Then, as $T \to \infty$, we obtain that 
\begin{align}
\frac{1}{ T \sqrt{T} } &\sum_{t=1}^T S_t \Rightarrow \sigma \int_0^1 W dr
\\
\frac{1}{T^2} &\sum_{t=1}^T S_t^2 \Rightarrow \sigma^2 \int_0^1 W^2 dr
\\
\frac{1}{ T \sqrt{T} } &\sum_{t=1}^T t u_t \Rightarrow \sigma^2 \int_0^1 W^2 dr = \sigma \left( W_1 - \int W \right)
\end{align}

\subsubsection{Power Functions for Unit Root Tests}

Consider the power functions for unit root tests by considering the sequence of local alternatives as
\begin{align}
\alpha = e^{c / T} = \left( 1 + \frac{c}{T} \right)
\end{align}
where $c=0$ reduces to the null hypothesis, $c > 0$ gives local explosive alternatives and $c < 0$ corresponds to local stationary alternatives. Moreover, the asymptotic theory for the sample moments of near-integrated time series converge weakly to corresponding functionals of a diffusion process rather than standard Brownian motion. Specifically, we have that, as $T \to \infty$, 
\begin{align}
\frac{1}{ T \sqrt{T} } \sum_{t=1}^T y_t &\Rightarrow \sigma \int_0^1 J_c
\\
\frac{1}{ T^2  } \sum_{t=1}^T y^2_t &\Rightarrow \sigma \int_0^1 J^2_c
\\
\frac{1}{ T^{ 5/2 } } \sum_{t=1}^T t y_t &\Rightarrow \sigma \int_0^1 r J_c
\\
\frac{1}{ T } \sum_{t=1}^T y_{t-1} u_{t} &\Rightarrow \sigma^2 \int_0^1  J_c dW + \lambda,
\end{align}
where
\begin{align}
J_c(r) = \int_0^r e^{(r-s) c} dW(s)
\end{align}
is the OU process generated in continuous time by the stochastic difference equation 
\begin{align}
dJ_c(r) = c J_c(r) dr + dW(r), \ \ \text{with} \ \ J_c(0) = 0. 
\end{align}
A particular feature of the local to unity model is that the localizing parameter is identifiable but it is not consistently estimable. 

\newpage 

Consider the following expression which gives the signal to noise ratio
\begin{align}
\frac{ \mathsf{Var} ( x_t )  }{ \mathsf{Var} ( u_t )  } \sim \frac{ \displaystyle \frac{1}{n} \sum_{t=1}^n \left( \frac{y_{t-1} }{n} \right)^2 }{ \displaystyle \sigma^2 } \overset{ p }{ \to } 0,
\end{align}
and so the signal from $x_t$ is too weak relative to the error variation to produce a consistent estimator of the localizing coefficient $c$. Although methods have been developed to utilize the way in which the limit distribution depends on the localizing coefficient, the failure of a consistent estimation has been a challenging task when conducting inference in these models. The dependence of the limit distribution on $c$ also affects resampling procedures such as the bootstrap, which are known to be inconsistent in models of this type because exactly of this dependence of functionals to the nuisance parameter of persistence.

\subsubsection{Time Series Regression with a Unit Root}

The consistent estimates $s_u^2$ and $s^2_T$ are used to develop new tests for unit roots that apply under very general conditions. We define the statistics
\begin{align}
Z_{ \alpha } &= T \left( \hat{\alpha} - 1 \right) - \frac{1}{2}   \frac{ \displaystyle \left( s^2_T - s_u^2 \right) }{ \displaystyle \left( T^{-2} \sum_{t=1}^T y^2_{t-1} \right) } 
\\
Z_t &= \left( \sum_{t=1}^T y^2_{t-1} \right)^{1 / 2} \left( \hat{\alpha} - 1 \right) / s_T - \frac{1}{2} \left( s^2_T - s_u^2 \right) \left[ s_T  \left( T^{-2} \sum_{t=1}^T y^2_{t-1} \right) ^{1/2} \right]^{-1}. 
\end{align}

where $Z_{ \alpha }$ is a transformation of the standardized estimator $T \left( \hat{\alpha} - 1 \right)$ and $Z_t$ is a transformation of the regression $t$ statistic. Then, the limiting distribution of $Z_{ \alpha }$ and $Z_t$ are given by 
\begin{theorem}
If the conditions of Theorem 4.2 are satisfied, then as $T \to \infty$, 
\begin{align}
Z_{ \alpha } &\Rightarrow \frac{ \left( W(1)^2 - 1 \right) / 2}{ \int_0^1 W(t)^2 dt }
\\
Z_{ t } &\Rightarrow \frac{ \displaystyle \left( W(1)^2 - 1 \right) / 2}{ \displaystyle \left\{ \int_0^1 W(t)^2 dt \right\}^{1 / 2} }
\end{align}
under the null hypothesis that $\alpha = 1$. 
\end{theorem}

\begin{remark}
The Theorem above demonstrates that the limiting distribution of the two statistics $Z_{ t }$ and  $Z_{ \alpha }$ are invariant within a very wide class of weakly dependent and possible heterogeneous distributed innovations $\left\{ u_t \right\}_{ t = 1}^{ \infty }$. Thus, the limiting distribution of $Z_{ \alpha }$ is identical to that of $T \left( \hat{\alpha} - 1 \right)$, when $\sigma_u^2 = \sigma^2$.

\newpage

An important property of $I(1)$ variables is that there can be linear combinations of these variables that are $I(0)$. If this is so then these variables are said to be cointegrated. 
Notice that econometric cointegration analysis can be used to overcome difficulties associated with stochastic trends in time series and applied to test whether there exist combinations of non-stationary series that are themselves stationary.       
\end{remark}

\subsection{Regression Asymptotics using Martingale Convergence Methods}

Considers weak convergence arguments to stochastic integral approximations. In particular for the case of a martingale-difference sequence $\left( \epsilon_t^2 | \mathcal{F}_{t-1} \right) = \sigma^2_{\epsilon}$ for all $t$ and $\mathsf{sup}_{ t \in \mathbb{Z} } \mathbb{E} \left( \epsilon_t^p   \ \mathcal{F}_{t-1} \right) = \sigma^2_{\epsilon} < \infty$ almost surely for some $p > 2$. Then by Donsker's theorem for the partial sum process we have that 
\begin{align}
\frac{1}{ \sqrt{n} } \sum_{t=1}^{ \floor{nr} } \epsilon_t \to \sigma_{\epsilon} W(r), 
\end{align} 
where $\big\{ W(s), s \geq 0 \big\}$ denotes the standard Brownian motion, which implies that the bilinear form: 
\begin{align}
\frac{1}{n} \sum_{t=2}^{ \floor{nr} } \left( \sum_{j=1}^{ t - 1 }  \epsilon_j \right) \epsilon_t \Rightarrow \sigma_{\epsilon}^2 \int_0^r W(v) dW(v)
\end{align}
Notice that the particular approach also has drawbacks. One of the drawbacks is that the approach is problem specific in certain ways. For instance it cannot be directly used in the case of statistics such as $\sum_{t=1}^n y_{t-1} u_t$, where $y_t = \alpha_n y_{t-1} + u_t$, for $t = 1,...,n$ and $\alpha_n \to 1$ as $n \to \infty$, that are central to the study of local deviations from a unit root in time series regression. Strong approximations to partial sums of independent random variables, along with an application of PS device, allow one to obtain invariance principles under independence and stationarity assumptions with explicit rates of convergence. For instance, using the Hungarian construction if $( \epsilon_t )_{t \in \mathbb{Z} }$ satisfy the conditions then
\begin{align}
\left| \frac{1}{ \sqrt{n} } \sum_{t=1}^{ \floor{nr} } \epsilon_t - \sigma_{\epsilon} W(r) \right| = o_{a/s} \left( n^{ 1/p - 1/2}   \right)
\end{align}
Notice that some existing results available in the literature on convergence to stochastic integrals can be applied to obtain convergence results. Denote with $
N_{n,r} = \frac{ \epsilon_0 }{ \sqrt{n} } + \sum_{t=1}^{ \floor{nr} } \frac{ \epsilon_t }{ \sqrt{n} }$. Then, the following stochastic integral representation holds for the statistic on the right-hand size 
\begin{align}
\frac{1}{n} \sum_{t=2}^{ \floor{nr} } \left( \sum_{j=1}^{ t - 1 }  \epsilon_j \right) \epsilon_t \Rightarrow \int_0^r N_{n,s} d N_{n,s}
\end{align}
Since $N_{n,r} \to W(r)$, then the latter result implies that the convergence result to the stochastic integral holds, provided that the sequence of processes $\left\{ N_{n,t} \right\}$ satisfies the uniform tightness condition.

\newpage

\subsubsection{Unification of the Limit theory of Autoregression} 

Consider the following first order autoregression model  (e.g., see, \cite{giraitis2006uniform})
\begin{align}
y_t = \alpha y_{t-1} + \epsilon_t, \ \ \ t = 1,...,n
\end{align}
with martingale-difference errors $\epsilon_t$. We treat the stationary $| \alpha | < 1$, unit root $\alpha = 1$, local to unity, and explosive cases together in what follows and show how the limit theory for all these cases may be formulated in a unified manner within the martingale convergence framework. Consider the stationary and unit root cases. For $r \in (0,1]$, define the recursive least squares estimator $\hat{\alpha}_r = \sum_{t=1}^{ \floor{nr} } y_{t-1} y_t / \sum_{t=1}^{ \floor{nr} } y_{t-1}^2$,  
\begin{align}
\left( \frac{1}{ \sigma_{\epsilon}^2  } \displaystyle \sum_{t=1}^{ \floor{nr} } y^2_{t-1} \right)^{1 / 2} \big( \hat{\alpha}_r - \alpha  \big) 
=
\frac{ \displaystyle \sum_{t=1}^{ \floor{nr} } y_{t-1} \epsilon_t  }{ \displaystyle \left( \sum_{t=1}^{ \floor{nr} } y^2_{t-1} \sigma^2_{\epsilon}  \right)^{1 / 2} }
= 
\frac{ X_n(r) }{ \left( \tilde{C}_n^{\prime} (r) \right)^{1 / 2} },
\end{align}
where $X_n(r)$ is the martingale given by 
\begin{align}
X_n(r) = 
\begin{cases}
\displaystyle \frac{1}{ \sqrt{n} } \sum_{t=1}^{ \floor{nr} } y_{t-1} \epsilon_t, & \text{for} \ |\alpha | < 1
\\
\\
\displaystyle \frac{1}{ n } \sum_{t=1}^{ \floor{nr} } y_{t-1} \epsilon_t, & \text{for} \ \alpha = 1 
\end{cases}
\end{align}
Then, by Theorem 4.1 it holds that 
\begin{align}
X_n(r) \to X(r)
= 
\begin{cases}
\sigma_{\alpha} \sigma_{\epsilon} W(r), \ \ \text{for} \ \ | \alpha | < 1, 
\\
\\
\displaystyle \sigma_{\epsilon}^2 \int_0^r W(r) dW(r), \ \ \text{for} \ \ \alpha = 1,
\end{cases}
\end{align}
Then, it can be proved that 
\begin{align*}
\left( \frac{1}{ \sigma_{\epsilon}^2  } \displaystyle \sum_{t=1}^{ \floor{nr} } y^2_{t-1} \right)^{1 / 2} \big( \hat{\alpha}_r - \alpha  \big) 
&=
\frac{ \displaystyle X_n(r) }{ \displaystyle \left( \tilde{C}_n^{\prime}(r) \right)^{1 / 2} }
\overset{ d }{ \to } 
\frac{ X(r) }{ \left( C(r) \right)^{1 / 2} }
=
\begin{cases}
\displaystyle \frac{1}{ r^{1/2} } W(r), & | \alpha | < 1
\\
\\
\frac{ \displaystyle \int_0^r W(v) dW(v) }{  \displaystyle \left(  \int_0^r W(v)^2 dv \right)^{1 / 2} }, & | \alpha | = 1,
\end{cases}
\end{align*}
which unifies the limit theory for the stationary and unit root autoregression.

\newpage 

Defining the error variance estimator $s_r^2 = \frac{1}{ \floor{nr} } \sum_{ t=1 }^{ \floor{nr} } \big( y_t - \hat{\alpha}_r y_{t-1} \big)^2$ and noting that $s_r^2 \overset{ p }{ \to } \sigma^2_{\epsilon}$ for $r > 0$, we have the corresponding limit theory for the recursive t-statistic such that the following limit result holds
\begin{align}
t_{ \hat{\alpha} }(r) 
=
\left( \frac{1}{ s^2_r } \sum_{ t=1 }^{ \floor{nr} } y_{t-1}^2   \right) \big( \hat{\alpha}_r - \alpha \big) 
=
\frac{ \displaystyle \sum_{ t=1 }^{ \floor{nr} } y_{t-1} \epsilon_t }{ \displaystyle \left(  \sum_{ t=1 }^{ \floor{nr} } y^2_{t-1} \sigma^2_{\epsilon} \right)^{1/2} } \frac{ \sigma_{\epsilon} }{ s_r } 
=
\frac{ X_n(r) }{\left( \tilde{C}^{\prime}_n(r) \right)^{1 / 2} } \frac{ \sigma_{\epsilon} }{ s_r } 
\end{align}
The theory also extends to cases where $\alpha$ lies in the neighborhood of unity such that $\alpha = \left( 1 + \frac{c}{n}\right)$. Then, it follows that 
\begin{align}
X_n(r) \overset{ d }{ \to } X(r) 
= 
\sigma_{\epsilon}^2 \int_0^r J_c(v) dW(v), 
\end{align}
where $J_c (v) = \displaystyle \int_0^v e^{ c(v-s)} dW(s)$ is a linear diffusion. Then, we obtain that 
\begin{align}
\left( \frac{1}{ s^2_r } \sum_{ t=1 }^{ \floor{nr} } y_{t-1}^2   \right)^{1 / 2} \big( \hat{\alpha}_r - \alpha \big) 
=
\frac{ \displaystyle X_n(r) }{ \displaystyle \left( \tilde{C}_n^{\prime}(r) \right)^{1 / 2} }
\overset{ d }{ \to } 
\frac{ X(r) }{ \left( C(r) \right)^{1 / 2} }
\overset{d}{=}
\frac{ \displaystyle \int_0^1 J_c(v) dW(v) }{ \displaystyle \left( \int_0^1 J_c(v)^2 dv \right)^{1 / 2} }.
\end{align}
Furthermore, we can investigate the effect on the asymptotic theory of the proposed test statistics when there are moderate deviations from unity of the form $\alpha = \left( 1 + \displaystyle \frac{c}{ n^{b} } \right)$ for some $b \in (0,1)$ and $c < 0$. In particular, under the assumption that the model covariates are generated as near unit root processes with $b \in (0,1)$ and $c < 0$ then the regressors are considered to be mildly integrated and the functional takes the following form
\begin{align}
X_n(r) = \frac{1}{ n^{(1 + b ) / 2} } \sum_{t=1}^{ \floor{nr} } y_{t-1} \epsilon_t, \ \ \ \ \alpha = \left( 1 + \frac{c}{n^b} \right), c < 0, \  b \in (0,1),
\end{align}

\medskip

\begin{theorem}
Let $X^T = \left\{ X_t, 0 \leq t \leq T \right\}$ be observations of the mean reversion process with mean reversion function. Then, for any fixed $0 < s_1 < s_2 < 1$, under the null hypothesis, 
\begin{align}
\underset{ s \in [ s_1,s_2 ] }{ \mathsf{sup}  } \Lambda_T(s)  \overset{ \mathcal{D} }{ \to } \underset{ s \in [ s_1,s_2 ] }{ \mathsf{sup}  } \frac{ \norm{ W(s) - sW(1) }^2 }{ s(1-s) }
\end{align} 
as $T \to \infty$, where $\norm{.}$ is the Euclidean norm and $W$ is a $(p+1)-$dimensional standard Brownian motion. 
\end{theorem}

\newpage

\subsection{Time Series Regression}

One of the main differences with independent data where the OLS estimator satisfies certain conditions is that, time-series data commonly violate one of the classical assumptions. In particular, the independence assumption holds that the error terms corresponding to different point in time are not correlated. However, when the error terms are serially correlated (autocorrelated), the OLS method produces biased estimates of the standard errors of the regression coefficients. The persistence of a random shock is what distinguishes an AR error process from an MA error process. An AR(p) error process can be written as an infinite sum of past random shocks. Each shock persists indefenitely, although its importance diminishes over time. For example, in an AR(1) process, the effect persists indefenitely, although it decreases over time and is effectively zero after some relatively small number of time periods ( \cite{choudhury1999understanding}).

\subsubsection{Unit Roots versus Deterministic Trends}

Consider the following model 
\begin{align}
y_t = \mu + \rho y_{t-1} + u_t    
\end{align}
Then, assume that the true data generating process is given by  
\begin{align}
H_1: y_t = \mu + \phi t + \rho y_{t-1} + u_t, \ \ \ \text{with} \ \ | \rho | < 1, \phi \neq 0. 
\end{align}

\subsubsection{The Engle-Granger Approach}

Consider a vector time series process $X_t \in \mathbb{R}^k$ such that each component is $X_{it}$ is integrated or order one, such that each of the series $X_{it}$ contains a unit root, but $\Delta X_{it}$ is a zero mean stationary process. Thus, using the Wold decomposition theorem we can write: 
\begin{align}
\Delta X_t = C(L) u_t = \left( \sum_{j=0}^{\infty} C_j L^j \right) u_t = \sum_{j=0}^{\infty} c_j u_{t-j}
\end{align}
where $L$ is the lag operator and $u_t$ is a $k-$variate white-noise process.

Working in the context of a bivariate system with at most one cointegrated vector, \cite{engle1987co} propose estimating the cointegrated vector $\alpha = \left( 1, \alpha_2 \right)$ by regressing the first component $X_{1t}$ of $X_t$ on the second component $X_{2t}$, using OLS, and then testing whether the OLS residuals $\hat{Z}_t$ have a unit root using the augmented Dickey-Fuller test. Therefore, if the test rejects the unit root hypothesis, and thus accepts the cointegration hypothesis, one may substitute for $\alpha^{\prime} X_{t-1}$ the OLS residual $\hat{Z}_{t-1}$ and estimate the parameter matrices in $A^{*} (L)$ by OLS, assuming that $A^{*}(L)$ is a finite-order lag polynomial. Notice that the above approach is only applicable if there is at most one cointegrated vector. Systems with dimension greater than two, however, may have multiple cointegrated vectors.

\newpage

\subsubsection{Testing the hypothesis of no serial correlation}

In some cases we are interested to test the hypothesis that the series $Y_t$ is serially uncorrelated against the alternative that is serially correlated. The most appropriate test of this hypothesis is based on the OLS of an AR$(p)$ model. Consider the following model
\begin{align}
Y_t = \alpha_0 + \alpha_1 Y_{t-1} + \alpha_2 Y_{t-2} + ... + \alpha_p Y_{t-p} + e_t
\end{align}

with $e_t$ a MDS. Thus, in this model the series $Y_t$ is serially uncorrelated if the slope coefficients are all zero. For example, when the cointegrating vector is known and the covariates of the cointegration vector are strictly exogenous, the optimal test for the null of a unit root in the known cointegrating vector collapes to the usual unit root test for a single variable. From a testing perspective, an unknown cointegrating vector is a nuisance parameter that is identified only under the alternative hypothesis, so optimal tests against the alternative of cointegration take the form of a weighted average test where the weights are over the possible cointegrating vectors. 

\subsubsection{A Cointegrated VAR}

A VAR has several equivalent representations that are valuable for understanding the interactions between exogeneity, cointegration and economic policy analysis. To start, the levels form of the $s-$th order Gaussian VAR for $x$ is
\begin{align}
x_t = K q_t + \sum_{j=1}^s A_j x_{t-j} + \varepsilon_t, \ \ \ \varepsilon_t \sim \mathcal{N} \left( 0, \Sigma \right).     
\end{align}
where $K$ is an $N \times N_0$ matrix of coefficients of the $N_0$ deterministic variables $q_t$.

Suppose that we have a vector $Y_t = \left[ y_{1t}, y_{2t},..., y_{nt} \right]^{\prime}$ that does not satisfy the conditions for stationarity. One way to achieve stationarity might be to model $\Delta y_t$, rather than $y_t$ itself. However, differencing can discard important information about the equilibrium relationships between the variables. This is because another way to achieve stationarity can be through linear combinations of the levels of the variables. Thus, if such linear combinations exist then we have cointegration and the variables are said to be \textbf{cointegrated}. 
    
Cointegration has some important implications: 
\begin{itemize}
        
\item[1.] Estimates of the cointegrating relationships are super-consistent, they converge at rate $T$ rather than $\sqrt{T}$. Moreover, the system implies a set of dynamic long-run equilibria between the variables

\item[2.] Modelling cointegrated variables allows for separate short-run and long-run dynamic responses. 
         
\end{itemize}

\medskip

\begin{definition}
Suppose that $y_t$ is $I(1)$. Then $y_t$ is cointegrated if there exists an $N \times r$ matrix $\beta$, of full column rank and where $0 < r < N$, such that the $r$ linear combinations, $\beta^{\prime} y_t = u_t$, are $I(0)$. 
\end{definition}

\newpage

\begin{remark}
The dimension $r$ is the cointegration rank and the columns of $\beta$ are the cointegrating vectors. Testing for cointegrating relations that economic theory predicts should exist, implies that the null hypothesis of noncointegration is not rejected. However, under the presence of breaks there is a need for implementing tests of the null of non-cointegration against alternatives allowing cointegrating relations subject to breaks. Although, the dates of breaks as well as their number are unknown a priori (see, \cite{saikkonen2006break} and \cite{shin1994residual} among others).  
\end{remark}

\textbf{Identification, Estimation and Forecasting with Time Series Regressions:}

\begin{itemize}

\item[1.] \textbf{Existence of stationary solution:} Usually this refers to finding sufficient conditions that ensure the existence of a weakly dependent stationary and ergodic solution $Z_t = \left( Y_t, X_t \right)$.

\item[2.] \textbf{Inference Problem:} The inference problem consists of the estimation procedure, the consistency of the corresponding estimator as well as deriving the asymptotic distribution of this estimator. 

\item[3.] \textbf{Significance Test of Parameter:} Usually we employ a Wald-type test statistic for linear restrictions on the parameters of the cointegrating regression model. 

\item[4.] \textbf{Model selection:} This step can be done either using a direct model selection approach, using an information criterion or using an exogenously generated procedure. The crucial step is to consider conditions that ensure the weak and strong consistency of the proposed procedure. 
     
\end{itemize}

\begin{example}
Consider the following data generating process as below
\begin{align}
\Delta X_t = \alpha \beta^{\prime} X_{t-1} + \sum_{i=1}^{k-1} \Gamma_i \Delta X_{t-i} + \varepsilon_t, \ \ \ \text{for} \ t = 1,...,T, 
\end{align}
where $\left\{ \varepsilon_t \right\}$ is \textit{i.i.d} with mean zero and full-rank covariance matrix $\Omega$, and where the initial values $X_{1-k},..., X_0$ are fixed. We are interested in the null hypothesis $H_0: \beta = \tau$. Thus, when $\tau$ is a known $(p \times r)$ matrix of full column rank $r$, the subspace spanned by $\beta$ and $\tau$ are identical. 

\end{example}

\begin{example}
Consider the time series vector $\boldsymbol{Y}_t$, which is an $N-$dimensional random vector generated by the VAR model as below
\begin{align}
\boldsymbol{Y}_t = A_1 \boldsymbol{Y}_{t-1} + ... + A_p \boldsymbol{Y}_{t-p} + \boldsymbol{u}_t, \ \ \ t \in \left\{ 1,..., T \right\}.    
\end{align}
Define the $N (p+1)$ vector $\boldsymbol{X}_t = \big( \boldsymbol{Y}_{t-p}^{\top},...,  \boldsymbol{Y}_{t-1}^{\top}, \boldsymbol{Y}_t^{\top} \big)^{\top}$ and let $\Sigma_x = Var ( \boldsymbol{X}_t ) = \mathbb{E} \left[ \boldsymbol{X} \boldsymbol{X}^{\top} \right]$ and $\Gamma_i = \mathbb{E} \left[ \boldsymbol{Y}_t \boldsymbol{Y}_{t-i}^{\top} \right]$ the autocovariance matrix and $\Sigma_x = 
\begin{pmatrix}
\Gamma_0 & \Gamma_1^{\top} & ... & \Gamma_p^{\top}    
\end{pmatrix}$.
\end{example}

\medskip

\begin{example}
Let $\boldsymbol{x}_t$ be an $I(1)$ vector of $n$ components, each with possibly deterministic trend in mean. Suppose that the system can be written as a finite-order vector autoregression:
\begin{align}
\boldsymbol{x}_t = \boldsymbol{\mu} + \boldsymbol{\pi}_1 \boldsymbol{x}_{t-1}  + \boldsymbol{\pi}_2 \boldsymbol{x}_{t-2} + ... + \boldsymbol{\pi}_k \boldsymbol{x}_{t-k} + \boldsymbol{\varepsilon}_t, \ \ \ t = 1,..., T    
\end{align}

\newpage 

Then, the model can be rewritten in error-correction form as below
\begin{align*}
\Delta \boldsymbol{x}_t 
= 
\boldsymbol{\mu} + \boldsymbol{\Gamma}_1 \Delta \boldsymbol{x}_{t-1} + \boldsymbol{\Gamma}_2 \Delta \boldsymbol{x}_{t-2} + ... + \boldsymbol{\Gamma}_{k-1} \Delta \boldsymbol{x}_{t-k+1} + \boldsymbol{\pi} \boldsymbol{x}_{t-k} + \boldsymbol{\varepsilon}_t   
\equiv
\boldsymbol{\mu} + \sum_{i=1}^{k-1} \boldsymbol{\Gamma}_i (1 - L) L^i \boldsymbol{x}_i + \boldsymbol{\pi} \boldsymbol{x}_{t-k} + \boldsymbol{\varepsilon}_t.  
\end{align*}
Therefore, we get the following system equation representation for $t \in \left\{ 1,..., T\right\}$
\begin{align}
\boldsymbol{\pi} (L) \boldsymbol{x}_t &= \boldsymbol{\mu} + \boldsymbol{\varepsilon}_t, \ \     
\boldsymbol{\pi} (L) = (1 - L) \boldsymbol{I}_n - \sum_{i=1}^{k-1} \boldsymbol{\Gamma}_i (1-L) L^i - \boldsymbol{\pi} L^k    
\\
\boldsymbol{\Gamma}_i &= - \boldsymbol{I}_n + \boldsymbol{\pi}_1 + \boldsymbol{\pi}_2 + ... + \boldsymbol{\pi}_i, \ \ \ \ i = 1,...,k  
\end{align}
\end{example}

\begin{example}
The asymptotic size and power of the augmented Dickey-Fuller test for a unit root are studied by \cite{paparoditis2018asymptotic}. In particular, the authors show that the limiting distribution of the augmented DF test under the null hypothesis of a unit root is valid under general set of assumptions on the dependence structure of innovations. Consider the following specification
\begin{align}
X_t = \rho X_{t-1} + \sum_{j=1}^p a_{j,p} \Delta X_{t-j} + e_{t,p}
\end{align}
where $X_t$ is a linear, infinite order autoregressive process such that
\begin{align}
X_t = X_{t-1} + U_t, \ \ \ U_t = \sum_{j=1}^{ \infty } a_j U_{t-j} + e_t.
\end{align}
Notice that stationarity and causality of $\left\{ U_t \right\}$ is ensured by assuming that $\sum_{j=1}^{ \infty } | j |^s | a_j | < \infty$ for some $s \geq 1 $ and $\sum_{j=1}^{ \infty } a_j z^j \neq 0$ for all $|z| \leq 1$.
\end{example}

\begin{example}
Consider a general autoregressive distributed lag ARDL $(p,q)$ model where a series, $y_t$, is a function of a constant term, $\boldsymbol{\alpha}_0$, past values of itself stretching back $p-$periods, contemporaneous and lagged values of an independent variable, $x_t$, of lag order $q$, and \textit{i.i.d} error term:
\begin{align}
y_t = \boldsymbol{\alpha}_0 + \sum_{i=1}^p \boldsymbol{\alpha}_i y_{t-i} + \sum_{j=0}^q \boldsymbol{\beta}_j \boldsymbol{x}_{t-j} + \boldsymbol{\epsilon}_t,    
\end{align}
A commonly used model is the ARDL $(1,1)$ model given by $y_t = \boldsymbol{\alpha}_0 + \boldsymbol{\alpha}_1 y_{t-1} + \boldsymbol{\beta}_0 x_t + \boldsymbol{\beta}_1 x_{t-1} + \boldsymbol{\epsilon}_t$. Then, the contemporaneous effect of $x_t$ on $y_t$ is given by $\boldsymbol{\beta}_0$. The magnitude of $\boldsymbol{\alpha}_1$ informs us about the memory property of $y_t$. Assuming that $0 < \alpha_1 < 1$, larger values indicate that movements in $y_t$ take longer to dissipate. The long-run effect (or long-run multiplier) is the total effect that a change in $x_t$ has on $y_t$. The above econometric specification corresponds to a conditional mean function. One can also consider the corresponding conditional quantile functional form which implies a cointegrating relation around a certain quantile level of the distribution function (see, \cite{galvao2013quantile} and \cite{cho2015quantile}). 
\end{example}

\newpage 

\subsection{Cointegrating Regression}

\subsubsection{Distribution of the OLS Estimate for a Special Case}

Following, \cite{hamilton2020time} let $y_{1t}$ be a scalar and $\boldsymbol{y}_{2t}$ be a $( p \times 1 )$ vector satisfying \begin{align}
\label{coint}
y_{1t} &= \boldsymbol{\alpha} + \boldsymbol{\gamma}^{\top} \boldsymbol{y}_{2t} + u_{1t}
\\
\boldsymbol{y}_{2t} &= \boldsymbol{y}_{2,t-1} + \boldsymbol{u}_{2t}
\end{align} 
In particular, if $y_{1t}$ and $\boldsymbol{y}_{2t}$ are both $I(1)$ but $u_{1t}$ and $\boldsymbol{u}_{2t}$ are $I(0)$, then $d \equiv ( p + 1 )$, the $d-$dimensional vector $\left( y_{1t},  \boldsymbol{y}_{2,t-1}  \right)^{\prime}$ is cointegrated with the cointegrated relation given by expression \eqref{coint}. Therefore, by considering the case that the innovation sequence of the system follows a Gaussian distribution, we have that the following holds
\begin{align}
\begin{bmatrix}
u_{1t}
\\
\boldsymbol{u}_{2t}
\end{bmatrix}  
\overset{ \textit{i.i.d}  }{ \sim } 
\mathcal{N} \left( 
\begin{bmatrix}
0
\\
\boldsymbol{0}
\end{bmatrix}, 
\begin{bmatrix}
\sigma_1^2 & \boldsymbol{0}^{\prime}
\\
\boldsymbol{0} & \boldsymbol{\Omega}_{22}
\end{bmatrix}
\right).
\end{align}
Notice that although we consider a Gaussian cointegrated system we assume that the explanatory variables $\boldsymbol{y}_{2,t}$ in \eqref{coint} are independent of the error term $u_{1t}$ for all $t$. Therefore, conditional on $\big(  \boldsymbol{y}_{2,1}, \boldsymbol{y}_{2,2},..., \boldsymbol{y}_{2,T} \big)$, the OLS estimates have a Gaussian distribution given by 
\begin{align}
\begin{bmatrix}
\left(  \hat{\boldsymbol{\alpha} }_T - \boldsymbol{\alpha} \right)
\\
\left(  \hat{\boldsymbol{\gamma}}_T - \boldsymbol{\gamma} \right)
\end{bmatrix}    
&\big(  \boldsymbol{y}_{2,1}, \boldsymbol{y}_{2,2},..., \boldsymbol{y}_{2,T}  \big) 
\nonumber
\\
&= 
\begin{bmatrix}
T \ & \ \displaystyle \sum_{t=1}^T \boldsymbol{y}_{2,T}^{\prime} 
\\
\displaystyle \sum_{t=1}^T \boldsymbol{y}_{2,T} \ & \  \displaystyle \sum_{t=1}^T \boldsymbol{y}_{2,T} \boldsymbol{y}_{2,T}^{\prime} 
\end{bmatrix}^{-1} 
\begin{bmatrix}
\displaystyle  \sum_{t=1}^T u_{1t}
\\
\displaystyle \sum_{t=1}^T \boldsymbol{y}_{2,T}  u_{1t}
\end{bmatrix}
\overset{ \textit{i.i.d}  }{ \sim }  \mathcal{N} 
\left(  \begin{bmatrix}
0
\\
\boldsymbol{0}
\end{bmatrix},  
\begin{bmatrix}
T \ & \ \displaystyle \sum_{t=1}^T \boldsymbol{y}_{2,T}^{\prime} 
\\
\displaystyle \sum_{t=1}^T \boldsymbol{y}_{2,T} \ & \  \displaystyle \sum_{t=1}^T \boldsymbol{y}_{2,T} \boldsymbol{y}_{2,T}^{\prime} 
\end{bmatrix}^{-1}  \right)
\end{align}

\begin{remark}
The above conditional distribution is that is needed to justify the small-sample application of the usual OLS t and F-tests under the null hypothesis. Thus, consider a hypothesis test involving $m$ restrictions on $\boldsymbol{\alpha}$ and $\boldsymbol{\gamma}$ of the following form 
\begin{align}
\boldsymbol{R}_{ \boldsymbol{\alpha} }  \boldsymbol{\alpha} +   \boldsymbol{R}_{ \boldsymbol{\gamma} } \boldsymbol{\gamma} = \boldsymbol{r},
\end{align}
where $\boldsymbol{R}_{ \boldsymbol{\alpha} }$ and $\boldsymbol{r}$ are known $( m \times 1 )$ vectors and $\boldsymbol{R}_{ \boldsymbol{\gamma} }$ is a known $( m \times d )$ matrix which describes these restrictions. 
Under the null hypothesis we have that $\mathbb{H}_0: \boldsymbol{R} \boldsymbol{\beta} = \boldsymbol{r}$ where the Wald statistic is written as 
\begin{align}
W = \left( \boldsymbol{R} \widehat{ \boldsymbol{\beta} } - \boldsymbol{r} \right)^{\top} \left[  \boldsymbol{R} \boldsymbol{J}_{ \boldsymbol{\beta} \boldsymbol{\beta} }^{-1} ( \hat{ \boldsymbol{\theta} } ) \boldsymbol{R}^{\top} \right]^{-1}   \left( \boldsymbol{R} \widehat{ \boldsymbol{\beta} } - \boldsymbol{r} \right)
\end{align}
where $\boldsymbol{\beta} = \big( \boldsymbol{\beta}_1^{\top},..., \boldsymbol{\beta}_p^{\top} \big)^{\top}$.
\end{remark}

\newpage 

Therefore, the Wald form of the OLS $F-$test of the null hypothesis is given by 
\begin{align}
\big(  \boldsymbol{R}_{ \boldsymbol{\alpha} } \hat{ \boldsymbol{\alpha} } + \boldsymbol{R}_{ \boldsymbol{\gamma} } \hat{ \boldsymbol{\gamma} }  - \boldsymbol{r} \big)^{\prime} \left\{ s^2_T \big[ \boldsymbol{R}_{ \boldsymbol{\alpha} } \ \ \boldsymbol{R}_{ \boldsymbol{\gamma} }    \big] \begin{bmatrix}
T \ & \ \displaystyle \sum_{t=1}^T \boldsymbol{y}_{2,T}^{\prime} 
\\
\displaystyle \sum_{t=1}^T \boldsymbol{y}_{2,T} \ & \  \displaystyle \sum_{t=1}^T \boldsymbol{y}_{2,T} \boldsymbol{y}_{2,T}^{\prime} 
\end{bmatrix}^{-1} 
\begin{bmatrix}
\boldsymbol{R}_{ \boldsymbol{\alpha} }^{\prime}
\\
\\
\boldsymbol{R}_{ \boldsymbol{\gamma} }^{\prime}
\end{bmatrix} \right\}^{-1} \big(  \boldsymbol{R}_{ \boldsymbol{\alpha} } \hat{ \boldsymbol{\alpha} } + \boldsymbol{R}_{ \boldsymbol{\gamma} } \hat{ \boldsymbol{\gamma} }  - \boldsymbol{r} \big) / m
\end{align}

where
\begin{align}
s_T^2 = \frac{1}{T-n} \sum_{t=1}^T \big( y_{1t} - \hat{\boldsymbol{\alpha}}_T - \boldsymbol{\gamma}_T^{\prime} \boldsymbol{y}_{2t} \big)^2.    
\end{align}
In other words, conditional on the vector $\big(  \boldsymbol{y}_{2,1}, \boldsymbol{y}_{2,2},..., \boldsymbol{y}_{2,T} \big)$ it follows that the above expression when we replace with the population variance, has a $\chi^2_m$ distribution. 

Thus, conditional on the data $\big(  \boldsymbol{y}_{2,1}, \boldsymbol{y}_{2,2},..., \boldsymbol{y}_{2,T} \big)$ the OLS F-test could be viewed as the ratio of a $\chi^2_m$ variable to the independent $\chi^2_{ T - p }$ variable $( T - p ) \sigma_T^2 / \sigma_1^2$ with numerator and denominator each divided by its degree of freedom. Therefore, the OLS F test has an exact $F ( m, T - p )$ distribution. Hence, despite the $I(1)$ regressors and complications of cointegration, the correct approach is to is to estimate the model by OLS and use standard $t$ or $F$ statistics to test any hypotheses about the cointegrating vector.

\subsubsection{Estimation of cointegrating vectors}     

We discuss estimation of $\beta$ with the following representation 
\begin{align}
y_t = \beta^{\prime} x_t + e_t
\end{align}
from the observable vector $z_t = \left( x_t^{\prime}, y_t \right)^{\prime}$ introduced in the previous section. The setting is or cointegration we assume that holds requires that the unobservable process $e_t$ is $I(d_e)$, $d_e < d_y$, while we assume that $\beta$ is identified for $t=1,...,n$   

For a generic column vector or scalar sequence $a_t$, $t = 1,...,n$ define the discrete Fourier transform below
\begin{align}
w_a ( \lambda ) = \frac{1}{ \left( 2 \pi n \right)^{1 / 2} } \sum_{t=1}^n a_t e^{i t \lambda },
\end{align} 

With also a column vector or scalar sequence $b_t, t=1,...,n$ possibly identical to $a_t$, define the cross periodogram as below
\begin{align}
I_{\alpha \beta} = w_{\alpha} \left( \lambda \right) w_{\beta} \left(  - \lambda \right).
\end{align}

\newpage 

Now denote by $\lambda_j = 2 \pi j / n$, for integer $j$, the Fourier frequencies, and define the averaged cross-periodogram. The case $m = \floor{ n / 2}$, where $\floor{.}$ denotes the integer part, is of particular interest, as we deduce that 
\begin{align}
\widehat{F}_{ab} \left(  \floor{ \frac{n}{2} } \right) = \frac{1}{2} \sum_{t=1}^n ( a_t - \bar{a} ) ( b_t - \bar{b} )^{\prime},  
\end{align}
the mean-corrected sample covariance, with $\bar{a} = n^{-1} \sum_{t=1}^n a_t$. We observe that $\widehat{F}_{ab} (m)$ represent the contributions from frequencies $[1, \lambda_m ]$ to the sample covariances above.  We estimate $\beta$ by the frequency domain least squares (FDLS) statistic given by 
\begin{align}
\hat{\beta}_m = \widehat{F}_{xx} (m)^{-1} \widehat{F}_{xy} (m)
\end{align}

\subsubsection{Cointegrating Vector and Equilibrium Points}

A typical cointegrating system has the following form 
\begin{align}
y_{2t} = A y_{2t-1} + u_{2t}
\end{align}
and require the coefficient matrix $A$ to have stable roots, then the new system is a conventional SEM. Furthermore, it is a triangular system since it is written in a reduced form. In particular, when $\Sigma$ is not block-diagonal, the MLE is obtained by using OLS on the augmented regression equation such that 
\begin{align}
y_{1t} = \beta^{\prime} y_{2t} + \gamma^{\prime} \Delta y_{2t} + u_{1.2t}
\end{align}  

Consider $\widehat{\boldsymbol{\beta}}$ to be the single equation OLS estimator of $\boldsymbol{\beta}$. Then, it holds that 
\begin{align}
T \left( \widehat{\boldsymbol{\beta}} -  \boldsymbol{\beta}  \right)
\Rightarrow
\left( \int_0^1 S_2 S_2^{\prime} \right)^{-1} \left( \int_0^1 S_2^{\prime} d S_1 + \sigma_{21} \right).
\end{align}  
 
Then, consider the fully modified OLS estimator given as below
\begin{align}
\boldsymbol{\beta}^{**} = \big( Y_2^{\prime} Y_2 \big)^{-1} \big(  Y_2^{\prime} y_1 - T \widehat{\sigma}_{21} \big) 
\end{align}
whose asymptotics are as below
\begin{align}
T \left( \boldsymbol{\beta}^{**} - \beta \right) \Rightarrow \left( \int_0^1 S_2 S_2^{\prime} \right)^{-1} \left( \int_0^1 S_2^{\prime} d S_1 \right).
\end{align}
Then, the further modification for the endogeneity of $y_{2t}$ is required to remove the correlation between the Brownian motion $S_1$ and $S_2$.

\newpage

Therefore, this is achieved by constructing 
\begin{align}
y_{1t}^{+} = y_{1t} - \widehat{\sigma}^{\prime}_{21} \widehat{\Sigma}_{22}^{-1} \Delta y_{2t}, 
\end{align} 
Then, the fully modified OLS estimator employs both the serial correlation and endogeneity corrections and is given by the following expression
\begin{align}
\boldsymbol{\beta}^{+} = \big( Y_2^{\prime} Y_2 \big)^{-1} \big(  Y_2^{\prime} y_1^{+} - T \widehat{\delta}^{+} \big) 
\end{align} 
Therefore, with these corrections the new estimator $\boldsymbol{\beta}^{+}$  has the same asymptotic behaviour as the full system MLE. Observed that it is a two-step estimator, and relies on the preliminary construction of $y_{1t}^{+}$ and $\widehat{\delta}^{+}$. Then, fully modified test statistics are based on $\boldsymbol{\beta}^{+}$ which can be constructed in the usual way. Therefore, we define the t-ratios as below 
\begin{align}
t_i^{+} = \left( \beta_i^{+} - \beta_i \right) / s_i^{+}.  
\end{align}

\subsection{Unit Roots, Cointegration and Structural Breaks}

\begin{example}
Consider that $y_t$ is generated by the following model
\begin{align}
y_t = \mu_0 + \mu_1 t + \alpha y_{t-1} + \beta (L) \Delta y_{t-1} + \varepsilon_t, \ \ \ t = 1,...,T     
\end{align}
where $\beta (L)$ is the lag polynomial of known order $p$ with the roots of $1 - \beta(L) L$ outside the unit circle. Under the null hypothesis, it holds that $\alpha = 1$ and $\mu_1 = 0$.  The above example can be estimated using the OLS approach without restrictions on $\mu_0, \mu_1$ or $\alpha$, when $y_t$ is regressed on $1, t, y_{t-1}, \Delta y_{t-1},..., \Delta y_{t-p}$, the test statistic for $\alpha = 1$, is the standard Dickey-Fuller test for a unit root against a trend-stationary alternative. We consider the asymptotic representations for rolling estimators and test statistics, however unlike rolling coefficient estimators (in expectation) is  constant through the sample. 

Thus, the rolling estimator $\bar{\theta}$ is  
\begin{align}
\tilde{\theta} ( \delta; \delta_0 ) = \left(  \sum_{t = \floor{n(\delta - \delta_0) } + 1 }^{\floor{n \delta} }  Z_{t-1} Z_{t-1}^{\prime} \right)^{-1} \left(  \sum_{t = \floor{n(\delta - \delta_0) } + 1 }^{\floor{n \delta} } Z_{t-1} y_t \right)    
\end{align}
The estimators and the test statistics are computed using the full $n$ observations for $k \in \left\{ k_0, k_0 + 1,...,T - k_0 \right\}$, where $k_0 = \floor{ T \delta_0 }$. Then, the stochastic processes constructed from the sequential estimators and Wald test statistic are as below. 

\begin{remark}
An overview of the use of partial sum processes for break detection is studied by \cite{katsouris2022partial} (see, also \cite{katsouris2023predictability}). Moreover, \cite{xiao2001testing} considers a test statistic for the null hypothesis of stationarity against an autoregressive unit root alternative. In particular, the author employs partial-sums of residuals based on the underline econometric specification. Furthermore, a bootstrap-based approach for detecting multiple breakpoints is examined by \cite{kejriwal2020bootstrap}. 
\end{remark}

\newpage

For $\delta_0 \leq \delta \leq 1 - \delta_0$, 
\begin{align}
\tilde{\theta} ( \delta ) &= \left( \sum_{t=1}^n Z_{t-1} \left( \floor{n\delta} \right) Z_{t-1} \left( \floor{n\delta} \right)^{\prime}  \right)^{-1} \left( \sum_{t=1}^n Z_{t-1} ( \floor{n \delta} ) y_t \right)    
\\
\mathcal{Y}_n \left( \tilde{\theta} ( \delta ) - \theta \right) &= \Gamma_n \left( \delta \right)^{-1} \Psi(\delta), 
\end{align}
and 
\begin{align}
\tilde{F}_n (\delta) = \big[ R \tilde{\theta} ( \delta ) - r \big]^{\prime} \left[  R       \left( \sum_{t=1}^n Z_{t-1} \left( \floor{n\delta} \right) Z_{t-1} \left( \floor{n\delta} \right)^{\prime} \right)^{-1} R^{\prime} \right]^{-1}  \big[ R \tilde{\theta} ( \delta ) - r \big] / q \tilde{\sigma}^2 (\delta)
\end{align}
where 
\begin{align}
\Gamma_n (\delta) = \mathcal{Y}_n^{-1} \left( \sum_{t=1}^n Z_{t-1} \left( \floor{n\delta} \right) Z_{t-1} \left( \floor{n\delta} \right)^{\prime} \right) \mathcal{Y}_n^{-1}    
\end{align}
\end{example}

\subsubsection{Testing for Structural Change in Cointegrated Systems}

Following the framework of  \cite{seo1998tests}, and by the invariance principle proposed by  \cite{phillips1986multiple}, it holds that 
\begin{align}
n^{-1/ 2} \sum_{t=1}^{ \floor{ns} } u_t &\Rightarrow W(s) \equiv BM ( \Sigma )
\\ 
n^{-1/ 2} x_{ \floor{ns} } &\Rightarrow C(1) W(s)
\end{align}
Therefore, we need to show that 
\begin{align}
\mathbb{P} \left( \underset{ s \in [0,1] }{ \text{sup} } n^{-1/ 2} \left| x_{ \floor{ns} } - C(1) \sum_{t=1}^{ \floor{ns} } u_t  \right| > \epsilon \right) \leq \mathbb{P} \left(  \underset{ s \in [0,1] }{ \text{sup} } n^{-1/ 2} \left| \Phi (L) u_{ \floor{ns} } \right| > \epsilon \right) \to 0. 
\end{align}

For instance, if $\left\{ \Phi(L) u_t \right\}$ is uniformly square integrable, we can then apply the following result. 
\begin{align}
n^{-1/ 2} \sum_{t=1}^{ \floor{ns} } w_t \Rightarrow \gamma^{\prime} \Phi(1) W(s).
\end{align}
Therefore, we need to show that 
\begin{align}
\mathbb{P} \left( \underset{ s \in [0,1] }{ \text{sup} } n^{-1/ 2} \left| w_{ t } - \gamma^{\prime} \Phi(1) \sum_{t=1}^{ \floor{ns} } u_t  \right| > \epsilon \right) \leq \mathbb{P} \left(  \underset{ s \in [0,1] }{ \text{sup} } n^{-1/ 2} \left| \gamma^{\prime} \Phi_1 (L) u_{ t } \right| > \epsilon \right) \to 0 
\end{align}
where $\Phi_1 (L) = \left( \Phi(L) - \Phi(1) \right) / (1 - L)$.

\newpage

Moreover, a key ingredient for deriving the limiting distribution of the Wald-type and LR-type statistics is consider the weak convergence of partial-sum functionals. We define weak convergence on the space $C\left( [0,1] \right)$ with respect to the uniform metric. Thus, we define weak convergence of the projected sequence 
\begin{align}
R_{12t}(\lambda) 
= 
x_{2t-1} - \left( \sum_{t=1}^{\floor{\lambda n} } x_{2t-1} z_{t-1}^{\prime} \right) \cdot \left( \sum_{t=1}^{\floor{\lambda n} } z_{t-1} z_{t-1}^{\prime} \right)^{-1} \cdot z_{t-1} 
\end{align}
Related asymptotic moments for the particular framework are as below: 
\begin{align}
\frac{1}{n} \sum_{t=1}^{ \floor{\lambda n } } R_{12} ( \lambda ) u_t^{\prime} 
\Rightarrow
\int_0^{\lambda}  W_2(\lambda ) dW^{\prime}(s),
\\
\frac{1}{ n^2 } \sum_{t=1}^{ \floor{\lambda n } } R_{12t} ( \lambda ) R_{12t} ( \lambda )^{\prime} 
\Rightarrow
\int_0^{\lambda} W_2(\lambda ) W_2^{\prime}(s) ds.
\end{align} 
Therefore, under the null hypothesis of no structural change, the model can be estimated by existing methods. Then, it holds that 
\begin{align*}
F(\lambda)^b 
&= F(\lambda) - V(\lambda) V(1) F(1),
\\
F(\lambda)^b 
&=
\int_0^{\lambda} B_2 ( s, \lambda ) B_2 ( s, \lambda )^{\prime} ds,
\\
B_2 ( s, \lambda ) 
&=
B_2(s) - \frac{1}{\lambda} B_2(s) ds.
\end{align*}
Notice that the expression $\boldsymbol{B}(s) = BM ( \boldsymbol{\Omega} )$ represents a Brownian motion with long-run variance $\boldsymbol{\Omega}$. Due to uniformly integrability, the following limit result holds
\begin{align}
\frac{1}{n} \sum_{t=1}^{ \floor{ \lambda n} } x_t w_t^{\prime} 
=
\frac{1}{n} \sum_{t=1}^{ \floor{ \lambda n} } \big( x_{t-1} + \Delta x_t \big) w_t^{\prime} \Rightarrow \int_0^{\lambda} C(1) W dW^{\prime} \Phi(1)^{\prime} \gamma + \lambda \Lambda_1 
\end{align}
By Lemma 9 of \cite{seo1998tests}, it holds that
\begin{align*}
\frac{1}{n} \sum_{t=1}^{ \floor{ \lambda n} } R_{2t} (\lambda)R_{2t} (\lambda)^{\prime} 
&=
\left( \frac{1}{n} \sum_{t=1}^{ \floor{ \lambda n} } w_{t-1} w_{t-1}^{\prime} \right) - \left( \frac{1}{n} \sum_{t=1}^{ \floor{ \lambda n} } w_{t-1} z_{t-1}^{\prime} \right)
=
\left( \frac{1}{n} \sum_{t=1}^{ \floor{ \lambda n} } z_{t-1} z_{t-1}^{\prime} \right)^{-1} \left( \frac{1}{n} \sum_{t=1}^{ \floor{ \lambda n} } z_{t-1} w_{t-1}^{\prime} \right)
\\
&\overset{ p }{ \to } \lambda \mathbb{E} \big( w_0 w_0^{\prime} \big) - \lambda \mathbb{E} \big( w_0 z_0^{\prime} \big) \mathbb{E} \big( z_0 z_0^{\prime} \big)^{-1}  \big( z_0 w_0^{\prime} \big) \equiv \lambda \boldsymbol{Q}
\end{align*}
uniformly in $\uptau \in \mathcal{T}$.

Thus, to derive the limit distribution under the null hypothesis, $
\frac{1}{n} \sum_{t=1}^{ \floor{\lambda n } } z_t z_t^{\prime} \lambda Q$, uniformly in $\lambda \in [0,1]$, for Q some positive definite matrix and for all $t$ and $n^{- 1 / 2} \sum_{t= 1}^{ \floor{ \lambda n } } z_t u_t \Rightarrow s Q^{- 1 / 2} W_q (s)$, where $W_q(s)$ is a $q-$vector of independent Wiener processes. This imposes a form of moment homogeneity, that is, a homogeneous distribution throughout the sample. The particular assumption holds only under the null of no structural break and is a necessary condition for deriving the limiting distributions of the test statistics. 

\newpage

\begin{remark}
According to \cite{seo1998tests}, even though $F( \lambda )$ is distributed as mixed normal with covariance matrix, it follows that $F( \lambda )^b$ is not a Brownian bridge, as defined in the stationary case, although it is still tied down. Hence, the asymptotic distributions here are different from those found by \cite{andrews1993tests}. Specifically, because the distribution of $LM_1^{\beta}$ is a $\chi^2$ for a known $\lambda$, these tests for structural change of the cointegrating vector are standard only if we know the change point. Therefore, the distribution of the tests depends only on the number of parameters and the admissible range of the break point.
\end{remark}

\subsubsection{Trend Stationarity and Structural Break Testing}

Consider the following model for the scalar random variable $x_t$ such that 
\begin{align}
x_t = \beta_1 + \beta_2 t + \beta_3 DT_t \left( \tau^* \right) + e_t, \ \ \ t = 1,..., T    
\end{align}
The shocks, $e_t$, are assumed to follow a zero mean fractionally integrated process of order $d$, denoted by $e_t \in I(d)$. Furthermore, the deterministic trend break term, $DT_t \left( \uptau^* \right)$, is defined for a generic $\uptau$ such that $DT_t \left( \uptau \right) := \left( t - \floor{\uptau T} \right) \boldsymbol{1} \left\{ t \geq  \floor{\uptau T} \right\}$. Therefore, where a true trend break occurs, which implies that we are under the alternative hypothesis and so $\beta_3 \neq 0$, we assume that the true break fraction is such that $\uptau^* \in [ \uptau_U, \uptau_L ] =: \Lambda \subset [0,1]$, where the quantiles $\uptau_L$ and $\uptau_U$ are trimming parameters. The formulation of the model is obtained using the vectors such that $z_t(\tau):= \left( 1, t, DT_t(\uptau) \right)^{\prime}$ and $\beta := \left( \beta_1, \beta_2, \beta_3 \right)^{\prime}$. Therefore, we can obtain the OLS estimate of $\beta$ as below
\begin{align}
\widehat{\beta}(\uptau) := \left( \sum_{t=j}^T z_t(\uptau) z_t(\uptau)^{\prime} \right)^{-1} \left( \sum_{t=j}^T  z_t(\uptau) y_t \right),     
\end{align}
Moreover, we define the corresponding de-trended residuals as below
\begin{align}
\widehat{u}_t (\uptau) := y_t - z_t(\uptau)^{\prime}  \widehat{\beta}(\uptau)     
\end{align}
where the particular sequence is obtained for $t = 1,...,T$. 

Therefore, under the null hypothesis we have that the estimated $\eta_t$ are obtained by taking the corresponding fractional differences of these OLS de-trended residuals. In other words, if the true break fraction $\uptau^*$, was known then one would simply evaluate the Likelihood ratio for break testing, denoted with $LM (\uptau)$ at the location $\uptau = \uptau^*$, and thus the resulting test statistic is denoted by $LM (\uptau^* )$. However, most of the empirical applications in the literature do not have a priori information regarding the exact location of the break-point location within the full sample. Our focus, is on the case where $\uptau^{*}$ is unknown and so, our proposed test will be based on evaluating $LM (\uptau)$ at $\hat{\uptau}$, which is obtained as the minimum RSS estimate 
\begin{align}
\hat{\uptau} := \underset{ \uptau \in \Lambda }{ \mathsf{arg min} } \sum_{t=1}^T \big(  \widehat{u}_t(\uptau) \big)^2,     
\end{align}
whose exact form is determined according the exact value of $d_0$ being tested under the null hypothesis. 

\newpage 

Specifically if $d_0$ lies in the region $( -0.5, 0.5 )$ then we estimate $\uptau^{*}$ using the levels of the data and testing the null hypothesis that the long memory parameter in the levels data is $d_0$, whereas if $d_0$ lies in the range $(0.5, 1.5)$ we instead estimate $\uptau^{*}$ using the first differences of the data and test the null hypothesis that the long memory parameter in the first difference data is $( d_0 - 1 )$ (see, \cite{arai2007testing}).

Therefore, in order to be able to determine the limiting distribution of the test statistic, we will need to determine the large sample behaviour of LM evaluated at the estimated break-point location, denoted by $LM ( \hat{\uptau} )$, which is done by comparing it to the infeasible LM statistic. In other words, the large sample behaviour of $LM ( \hat{\uptau} )$ clearly depends on the large sample properties of the estimates $\hat{\uptau}$ and $\widehat{\beta} ( \widehat{\uptau} )$ evaluated at $\uptau = \widehat{\uptau}$. Now assume that the underline data generating process has a break under the alternative hypothesis, which implies that we can define the Pitman drift such that $H_c: \theta := \theta_T = c / \sqrt{T}$.

Then, for $d_0 \in \left( -0.5, 0.5 \right)$, define a generic element $\alpha$ and the diagonal matrix such that 
\begin{align}
K_T ( \alpha ) := \mathsf{diag} \big\{ T^{ \frac{3}{2} - \alpha}, T^{ \frac{3}{2} - \alpha } \big\}.     
\end{align}

\begin{remark}
Therefore, a consequence of Theorem 1 is that, $LM(\hat{\uptau}) - LM = o_p(1)$ irrespective of whether $\beta_3 \neq 0$ or $\beta_3 = 0$. Conseqenently, regardless of the value of $\beta_3$, the LM converges in distribution to a $\chi^2-$random variate with centrality parameter
\begin{align}
LM ( \widehat{\uptau} ) \overset{d}{\to} \chi_1^2 ( c^2 \omega^2 )     
\end{align}
which implies that the test statistic retains asymptotic optimality. Moreover, since $LM( \widehat{\uptau} ) \overset{d}{\to} \chi^2_1$ under $H_0$, standard critical values can still be easily obtained. 
\end{remark}

\begin{remark}
Theorem 1 is based on establishing that the difference between the LM-type statistics based on $\widehat{\varepsilon}_t$ and $\widehat{\varepsilon}_t ( \widehat{\tau} )$ is asymptotically negligible. A key part of the derivation of the theorem is proving that $\widehat{A} - \widehat{A} (\widehat{\tau}) = o_p( T^{-1/2} )$ and since the difference between the LM-type test statistic based on $\widehat{\varepsilon}_t$ and $\widehat{\varepsilon}_t (\widehat{\uptau})$ crucially depends on the term $\Delta_{+}^{ d_0 } \big( DT_t (\widehat{\uptau}) - DT_t (\uptau^{*}) \big)$, on showing that
\begin{align}
\sum_{t=1}^T \left( \sum_{j=1}^{t-1} j^{-1} \Delta_{+}^{d_0} \big( DT_t (\widehat{\uptau}) - DT_t (\uptau^{*}) \big) \right)  \widehat{\varepsilon}_t = o_p( T^{1/2} )   
\end{align}
\end{remark}
Under the alternative hypothesis $H_c$, by setting $r = 0$ and $\alpha = \delta_0$, the FCLT holds such that
\begin{align}
T^{ - \left( \frac{1}{2} + \delta_0 \right) } \sum_{t=1}^{ \floor{ \tau T } } u_t \Rightarrow \sigma_{\infty} W( \uptau ; \delta_0 ).   
\end{align}

\begin{remark}
More typically in applied work, it is natural that the break-point $\uptau$ is thought to be unknown. Therefore, in this case the testing procedure is nonstandard because a nuisance parameter $\uptau$ appears only under the alternative hypothesis. Therefore, testing problems in which the nuisance parameter appears only under the alternative hypothesis require the development of tests with specific optimality properties. 
\end{remark}

\newpage

\subsubsection{Cointegration Testing under Structural Breaks}

Further to the two previous examples, in this section we consider cointegration testing under structural breaks (see, \cite{campos1996cointegration},  \cite{johansen2000cointegration} and \cite{kasparis2012dynamic}). In order to test for structural breaks in cointegration models a residual-based test can be derived from single equation models. However, the system equation approach appears limited when a structural break in cointegrating vectors is incorporated in the conventional specification form due to the presence of parameter instability (see, \cite{saikkonen2006break}). The literature proposes two different estimation and testing approaches, the first approach consists of the test for the null of cointegration while the second approach considers the test for the null of no cointegration (see, \cite{wagner2022residual}). Specifically, the authors consider the Lagrange Multiplier (LM) test and derive its limiting distribution which is free of nuisance parameter dependencies except for the number of $I(1)$ regressors (degrees of freedom) and the location of the structural break.    

To derive the asymptotic theory of the test statistic we derive an invariance principle and a corresponding functional central limit theorem to the partial-sum process $n^{ -1/2 } \sum_{t=1}^{ \floor{nr} } v_t$, which weakly converges to a two-dimensional standard Brownian motion $B(r) = \left( B_1(r), B_2(r) \right)^{\prime}$. In particular, the test proposed by \cite{shin1994residual} is constructed based on the regression residual of $y_{1t}$ on 1 and $y_{2t}$, denoted by $\hat{v}_{1t}$.

Thus, the test statistic is given by
\begin{align}
V_n = n^{-2} \frac{1}{ \hat{\sigma}^2 } \sum_{t=1}^n S_t^2, \ \ \ \ S_t = \sum_{j=1}^t \hat{v}_{1j} \ \ \ \text{and} \ \ \ \hat{\sigma}^2 = n^{-1} \sum_{t=1}^n \hat{v}^2_{1t}.     
\end{align}
We show that $V_n$ diverges to infinity but $V_n / n$ weakly converges to a distribution that is positive almost surely. Denote with $x_t = \left( 1, y_{2t} \right)^{\prime}$ and $e_t = v_{1t} + \mu_2 d_{2t}$ such that $y_{1t} = \left( \mu_1, \beta \right) x_t + e_t$. The residual $\hat{v}_{1t}$ is 
\begin{align}
\hat{v}_{1t} = e_t - \left( \sum_{t=1}^n e_t x_t^{\prime} \right) \left( \sum_{t=1}^n x_t x_t^{\prime}    \right)^{-1} x_t.     
\end{align}
The following weak convergence results hold
\begin{align}
\boldsymbol{D}_n^{-1} \left( \sum_{t=1}^n \boldsymbol{x}_t \boldsymbol{x}_t^{\prime} \right) &\Rightarrow \int_0^1 \boldsymbol{X}(s) \boldsymbol{X}(s)^{\prime} ds
\\
n^{-1/2} \boldsymbol{D}_n^{-1} \sum_{t=1}^n \boldsymbol{x}_t e_t &\Rightarrow \mu_2 \int_0^1 \boldsymbol{X}(s) d_2(s) ds
\end{align}
Furthermore, it holds that 
\begin{align}
n^{-1} \sum_{t=1}^{\floor{nr}} e_t \Rightarrow \mu_2 \int_0^r d_2(s) ds, \ \ \ n^{-1/2} \boldsymbol{D}_n^{-1} \sum_{t=1}^{\floor{nr} } \boldsymbol{x}_t \Rightarrow \int_0^r \boldsymbol{X}(s) ds. 
\end{align}
where the normalization matrix is given by $\boldsymbol{D}_n = \mathsf{diag} \left( n^{1/2}, n \right)$, $\boldsymbol{X}(s) = \left[ 1, B_2(s) \right]^{\prime}$. 

\newpage 

Moreover, $d_2 (s)$ is a step-function such that,  $d_2 (s) = 1$ if $s > 0.5$ and $0$ otherwise. Then, it holds that 
\begin{align*}
n^{-3} \sum_{t=1}^n S_t^2 \Rightarrow \mu_2^2 \int_0^1 \left\{ \int_0^r \left[ d_2(s) - \left( \int_0^1 \boldsymbol{X}(u) d_2(u) du \right)^{\prime} \times \left( \int_0^1 \boldsymbol{X}(u) \boldsymbol{X}(u)^{\prime} du \right)^{-1} \boldsymbol{X}(s) \right] ds \right\}^2 dr.     
\end{align*}
and the variance under the null hypothesis is estimated to be 
\begin{align}
\hat{\sigma}^2 \Rightarrow 1 + \mu_2^2 \left\{ \int_0^1 d_2(s) ds - \left( \int_0^1 \boldsymbol{X}(s) d_2 ds \right)^{-1} \times \left( \int_0^1 X(s) X(s)^{\prime} ds \right)^{-1} \int_0^1 X(s) d_2(s) ds \right\}.     
\end{align}

\subsubsection{Testing for a Structural Break in Cointegrating Regression}

Consider the following predictive regression model 
\begin{align}
y_t &= \beta_0 + \boldsymbol{\beta}^{\prime}_1 \boldsymbol{x}_t + \epsilon_t, 
\\
\boldsymbol{x}_t &= \boldsymbol{R}_T \boldsymbol{x}_{t-1} + \boldsymbol{\eta}_t 
\end{align}
where $y_t$ is a scalar vector and $\boldsymbol{x}_t \in \mathbb{R}^{p \times 1}$ is $p-$dimensional regressor vector. Furthermore, consider the case in which $\boldsymbol{\Phi}_T= \boldsymbol{I}_p$. Define the vector $\boldsymbol{u}_t = \big( \epsilon_t, \boldsymbol{\eta}_t^{\prime} \big)^{\prime}$ and the following long-run covariance matrices
\begin{align*}
\boldsymbol{\Omega}_T := \underset{ T \to \infty }{ \mathsf{lim} } \frac{1}{T} \sum_{t=1}^T \sum_{j=1}^T \mathbb{E} \big[ \boldsymbol{u}_j \boldsymbol{u}_t^{\prime} \big] 
= 
\begin{pmatrix}
\displaystyle \underset{ T \to \infty }{ \mathsf{lim} } \frac{1}{T} \sum_{t=1}^T \sum_{j=1}^T  \mathbb{E} \big[ \epsilon_j^{\prime} \epsilon_t  \big] \ & \  \displaystyle \underset{ T \to \infty }{ \mathsf{lim} } \frac{1}{T} \sum_{t=1}^T \sum_{j=1}^T  \mathbb{E} \big[ \epsilon_j^{\prime} \boldsymbol{\eta}_t \big]
\\
\\
\displaystyle \underset{ T \to \infty }{ \mathsf{lim} } \frac{1}{T} \sum_{t=1}^T \sum_{j=1}^T  \mathbb{E} \big[ \boldsymbol{\eta}_j^{\prime} \epsilon_t \big] \ & \ \displaystyle \underset{ T \to \infty }{ \mathsf{lim} } \frac{1}{T} \sum_{t=1}^T \sum_{j=1}^T  \mathbb{E} \big[ \boldsymbol{\eta}_j^{\prime} \boldsymbol{\eta}_t \big]
\end{pmatrix}
\end{align*}
such that 
\begin{align}
\boldsymbol{\Omega}_T \overset{ d }{ \to } \boldsymbol{\Omega} := 
\begin{pmatrix}
\boldsymbol{\Omega}_{\epsilon \epsilon} \ & \  \boldsymbol{\Omega}_{\epsilon \eta}
\\
\\
\boldsymbol{\Omega}_{\eta \epsilon} \ & \ \boldsymbol{\Omega}_{\eta \eta} 
\end{pmatrix}
\end{align}
Similarly, we define 
\begin{align*}
\boldsymbol{\Lambda}_T := \underset{ T \to \infty }{ \mathsf{lim} } \frac{1}{T} \sum_{t=1}^T \sum_{j=1}^t \mathbb{E} \big[ \boldsymbol{u}_j \boldsymbol{u}_t^{\prime} \big] 
= 
\begin{pmatrix}
\displaystyle \underset{ T \to \infty }{ \mathsf{lim} } \frac{1}{T} \sum_{t=1}^T \sum_{j=1}^t  \mathbb{E} \big[ \epsilon_j^{\prime} \epsilon_t  \big] \ & \  \displaystyle \underset{ T \to \infty }{ \mathsf{lim} } \frac{1}{T} \sum_{t=1}^T \sum_{j=1}^t  \mathbb{E} \big[ \epsilon_j^{\prime} \boldsymbol{\eta}_t \big]
\\
\\
\displaystyle \underset{ T \to \infty }{ \mathsf{lim} } \frac{1}{T} \sum_{t=1}^T \sum_{j=1}^t  \mathbb{E} \big[ \boldsymbol{\eta}_j^{\prime} \epsilon_t \big] \ & \ \displaystyle \underset{ T \to \infty }{ \mathsf{lim} } \frac{1}{T} \sum_{t=1}^T \sum_{j=1}^t  \mathbb{E} \big[ \boldsymbol{\eta}_j^{\prime} \boldsymbol{\eta}_t \big]
\end{pmatrix}
\end{align*}

\newpage

\begin{align}
\boldsymbol{\Lambda}_T \overset{ d }{ \to } \boldsymbol{\Lambda} := 
\begin{pmatrix}
\boldsymbol{\Lambda}_{\epsilon \epsilon} \ & \  \boldsymbol{\Lambda}_{\epsilon \eta}
\\
\\
\boldsymbol{\Lambda}_{\eta \epsilon} \ & \ \boldsymbol{\Lambda}_{\eta \eta} 
\end{pmatrix}
\end{align}
Furthermore, we define with 
$\boldsymbol{\Omega}_{ \epsilon \eta } = \boldsymbol{\Omega}_{ \epsilon \epsilon } - \boldsymbol{\Omega}_{ \epsilon \eta } \boldsymbol{\Omega}_{ \eta \eta }^{-1} \boldsymbol{\Omega}_{ \eta \epsilon }$ and $\boldsymbol{\Lambda}_{ \eta \epsilon  }^{+} = \boldsymbol{\Lambda}_{ \eta \epsilon } - \boldsymbol{\Omega}_{ \eta \eta }  \boldsymbol{\Lambda}_{ \eta \eta }^{-1} \boldsymbol{\Lambda}_{ \eta \epsilon }$. Then, we estimate the cointegrating regression model via OLS and obtain the residuals $\hat{\epsilon} = y_t - \boldsymbol{\beta}^{\prime} \boldsymbol{z}_t$, where  $\boldsymbol{z}_t = \big( 1, \boldsymbol{x}_t^{\prime} \big)^{\prime}$ and  $\boldsymbol{\beta} = \big( \beta_0, \boldsymbol{\beta}_1^{\prime} \big)$. Moreover, we define the following vector of residuals
\begin{align}
\hat{ \boldsymbol{u} }_t = \big( \hat{\epsilon}_t , \left( \Delta z_t - \Delta \bar{z} \right)^{\prime} \big)^{\prime},
\end{align}
Notice that $\Delta \bar{z}$ represents the sample mean of $\Delta z$. Furthermore, we denote with $\hat{u}_t$ the estimated vector of residuals as well as with $\hat{\boldsymbol{\Omega}}$ and $\hat{\boldsymbol{\Lambda}}$ the corresponding estimators of the long-run covariance matrices. We also define the transformed dependent variable as
\begin{align}
y_t^{+} := y_t - \hat{\boldsymbol{\Omega}}_{\epsilon \eta} \hat{\boldsymbol{\Omega}}_{\eta \eta}^{-1} \big( \Delta z_t - \Delta \bar{z} \big).  
\end{align}
Therefore, we obtain the fully modified (FM) estimator as below
\begin{align}
\hat{ \boldsymbol{\beta} }^{+} &= \left( \sum_{t=1}^T  \boldsymbol{z}_t \boldsymbol{z}_t^{\prime} \right)^{-1} \left( \sum_{t=1}^T y_t^{+} \boldsymbol{z}_t^{\prime} - \big[ 0  \ \ \hat{\boldsymbol{\Lambda}}^{+ \prime}_{\eta \epsilon} \big] \right) 
\\
\hat{ \boldsymbol{\beta} }^{+} &= \left( 
\begin{bmatrix}
\displaystyle \sum_{t=1}^T  1  \ \ & \ \ \displaystyle  \sum_{t=1}^T \boldsymbol{x}_t
\\
\\
\displaystyle \sum_{t=1}^T \boldsymbol{x}_t^{\prime}  \ \ & \ \ \displaystyle \sum_{t=1}^T \boldsymbol{x}_t^{\prime} \boldsymbol{x}_t
\end{bmatrix}^{-1}
\right)
\left( \sum_{t=1}^T y_t^{+} \boldsymbol{z}_t^{\prime} - \big[ 0  \ \ \hat{\boldsymbol{\Lambda}}^{+ \prime}_{\eta \epsilon} \big] \right) 
\end{align}
The associated residual vector is given by $
\hat{ \epsilon }^{+} = y_t^{+} - \boldsymbol{\beta}^{+} \boldsymbol{x}_t$. Furthermore, we define the following 
\begin{align}
\hat{s}_t = \left[ \boldsymbol{x}_t \hat{ \epsilon }^{+} - \begin{bmatrix}
0
\\
\hat{\boldsymbol{\Lambda}}^{+ \prime}_{\eta \epsilon}
\end{bmatrix} \right].
\end{align}
Define the following test statistic
\begin{align}
\boldsymbol{F}_k = \boldsymbol{S}_{Tk}^{\prime} \left[ \hat{\Omega}_{\epsilon . \eta } \boldsymbol{V}_{Tk} \right]^{-1} \boldsymbol{S}_{Tk}
\end{align}
\begin{align}
\boldsymbol{S}_{Tk} = \sum_{t=1}^k \hat{s}_t, \ \ \  \boldsymbol{V}_{Tk} = \big( \boldsymbol{M}_{Tk} - \boldsymbol{M}_{Tk} \boldsymbol{M}_{TT}^{-1} \boldsymbol{M}_{Tk} \big) , \ \ \ \boldsymbol{M}_{Tk} = \sum_{t=1}^k \boldsymbol{x}_t \boldsymbol{x}_t^{\prime}
\end{align}

\newpage 

The F-test is a test for a structural break at a known point which is asymptotically distributed under the null hypothesis as a $\chi^2-$distributed random variable with $p$ degree of freedoms. However, when the break-point is unknown, then the test is $F_{ \mathsf{sup} } = \mathsf{sup}_{ k \in \mathcal{B} } F_k$, where  $\mathcal{B} \in (0,1)$.  Moreover, a framework for structural break detection in predictive regression models is proposed by \cite{katsouris2023predictability, katsouris2023structural}.

\subsubsection{Unit root and Cointegrating limit theory }

\begin{example}
Consider the following data generating process
\begin{align}
y_t &= \mu + x_t^{\prime} \beta + u_t
\\
x_t &= x_{t-1} + v_t.
\end{align}
Stacking the error process defines $\eta_t = \left[ u_t , \ v_t^{\prime} \right]^{\prime}$. Furthermore, it is assumed that $\eta_t$ is a vector of $I(0)$ processes in which case $x_t$ is a non-cointegrating vector of $I(1)$ processes and there exists a cointegrating relationship among $[ y_t, x_t^{\prime} ]^{\prime}$ with cointegrating vector $[ 1, - \beta^{\prime}  ]^{\prime}$. To review existing theory and to obtain the key theoretical results in the paper, assumptions about $\eta_t$ are required. It is sufficient to assume that $\eta_t$ satisfies a functional central limit theorem (FCLT) of the form given below
\begin{align}
\frac{1}{ \sqrt{n} } \sum_{t=1}^{ \floor{nr} } \eta_t \Rightarrow B(r) \Rightarrow \Omega^{1/2} W(r), \ \ \ r \in [0,1],    
\end{align}
Define the partial sum process such that 
\begin{align}
\widehat{S}_t = \sum_{j=1}^t \widehat{\eta}_t     
\end{align}
We start by establishing functional central theorems for 
\begin{align*}
\frac{1}{T} \sum_{t=1}^{ \floor{nr}  } \widehat{u}_t 
&=
\frac{1}{\sqrt{T} } \sum_{t=1}^{ \floor{nr} } u_t - \frac{ \floor{nr} }{n} n^{1/2} \left(  \widehat{\mu} - \mu  \right) - \frac{1}{n \sqrt{n} } \sum_{t=1}^{ \floor{nr} } x_t^{\prime} n \left(  \widehat{\beta} - \beta \right)
\\
&\Rightarrow 
\int_0^r dB_u(s) - r \int_0^r B_v^{*} (s)^{\prime} ds \Theta.
\end{align*}
Using the definition of $\widehat{\eta} = \left[ \widehat{u}_t, v_t^{\prime}   \right]^{\prime}$ and stacking now leads to the following asymptotic theory result
\begin{align}
\frac{1}{ \sqrt{n} } \widehat{S}_{ \floor{nr} } = \frac{1}{ \sqrt{n} } \sum_{t=1}^{ \floor{nr} } \widehat{\eta}_t \Rightarrow 
\begin{bmatrix}
\displaystyle \int_0^r dB_u(s) - \int_0^r B_v^{*} (s)^{\prime} ds \Theta
\\
\displaystyle B_v(r)
\end{bmatrix}
\end{align}
Under the stated assumption it holds that
\begin{align}
\frac{1}{T} \sum_{t=2}^T S_{t-1}^{\eta} \eta_t^{\prime} \Rightarrow \int B(r) dB(r) + \Lambda.    
\end{align}
\end{example}

\newpage 

\begin{proposition}
As $b \to 0$, the fixed-b limiting distribution of $\widehat{\theta}_b^{+}$ converges in probability to the traditional limit distribution.     
\end{proposition}

\begin{remark}
Notice that these results show that the performance of the FM-OLS estimator relies critically on the consistency approximation of the long-run variance estimators being accurate and that moving around the bandwidth and kernel impacts the sampling behaviour of the FM-OLS estimator. However, it is well-known that non-parametric kernel long run variance estimators suffer from bias and and sampling variability which as a result can affect the the accuracy of the traditional approximation. Further details on fixed-b asymptotics are discussed in       \cite{vogelsang2013fixed} and \cite{vogelsang2014integrated}.
\end{remark}

\begin{example}
Consider the cointegrating system as below
\begin{align}
\boldsymbol{y}_t &= \boldsymbol{A} x_t + \boldsymbol{u}_t,
\\
\boldsymbol{x}_t &= \boldsymbol{x}_{t-1} + \boldsymbol{v}_t
\end{align}
where $\boldsymbol{y}_t$ is an $m-$dimensional vector and $\boldsymbol{x}_t$ is $p-$dimensional such that $\boldsymbol{\eta}_t = \left( \boldsymbol{u}_t^{\prime}, \boldsymbol{v}_t^{\prime}  \right)^{\prime}$ is an $(m + K)-$dimensional vector of innovations and $\boldsymbol{A}$ is an $(m \times K)$ matrix of cointegrating coefficients.  

The mixture process of the limit theory is given by 
\begin{align}
\mathsf{vec} \left\{ n \left( \hat{ \boldsymbol{A} }^{+} - \boldsymbol{A}  \right) \right\} \Rightarrow \mathcal{MN} \left( \boldsymbol{0},  \left( \int_0^1 \boldsymbol{B}_{\uptau}^{+} \boldsymbol{B}_{\uptau}^{+ \prime} \right)^{-1} \otimes \boldsymbol{\Omega}_{yy.x}   \right),
\end{align}
where $\hat{A}^{+}$ is the FM regression estimator with a conditional long-run covariance matrix of $u_t$ given $v_t$ expressed as below
\begin{align}
\boldsymbol{\Omega}_{yy.x} := \boldsymbol{\Omega}_{yy} - \boldsymbol{\Omega}_{yx} \boldsymbol{\Omega}_{xx}^{-1} \boldsymbol{\Omega}_{xy}
\end{align}
Then, the FM regression estimator has the explicit form given by
\begin{align}
\hat{ \boldsymbol{A} }^{+} = \left( \hat{ \boldsymbol{Y} }^{+ \prime} \boldsymbol{X} - n \hat{ \boldsymbol{\Delta} }^{+}_{yx} \right) \times \left( \boldsymbol{X}^{\prime} \boldsymbol{X} \right)^{-1}
\end{align}
where 
\begin{align}
\boldsymbol{X} = \big[ x_1^{\prime},..., x_n^{\prime} \big], \ \ \    \hat{ \boldsymbol{Y} }^{+} = \big[ \hat{y}_1^{+ \prime},..., \hat{y}_n^{+ \prime} \big]^{\prime} \in \mathbb{R}^{ n \times m} 
\end{align}
such that 
\begin{align}
\hat{\boldsymbol{y}}_t^{+} = \boldsymbol{y}_t - \hat{ \boldsymbol{\Omega} }_{yx}  \hat{ \boldsymbol{\Omega} }_{xx}^{-1} \Delta \boldsymbol{x}_t,  \ \ \ \text{and} \ \ \ \boldsymbol{\Delta}_{yx}^{+} = \boldsymbol{\Delta}_{yx} - \boldsymbol{\Omega}_{yx} \boldsymbol{\Omega}_{xx}^{-1} \boldsymbol{\Delta}_{xx}      
\end{align}
where $\hat{ \boldsymbol{\Omega} }_{yx}  \hat{ \boldsymbol{\Omega} }_{xx}^{-1}$ and  $\hat{\boldsymbol{\Delta}}_{yx}^{+}$ are consistent estimates of $\boldsymbol{\Omega}_{yx}  \hat{ \boldsymbol{\Omega} }_{xx}^{-1}$ and  $\boldsymbol{\Delta}_{yx}^{+}$ respectively.

\newpage

These matrices can be constructed in the familiar fashion using semiparametric lag kernel methods with residuals from a preliminary cointegrating least squares regression. Notice that setting with $u_{y.xt} = u_{yt} - \Omega_{yx} \Omega_{xx}^{-1} \Delta x_t$ and $U_{y.x} = \big[ u_{y.x1}^{\prime},..., u_{y.xn}^{\prime} \big]^{\prime}$ as the corresponding data matrix, we have that 
\begin{align}
\hat{\boldsymbol{y}}_t^{+} 
= 
\boldsymbol{y}_t - \hat{ \boldsymbol{\Omega} }_{yx}  \hat{ \boldsymbol{\Omega} }_{xx}^{-1} \Delta \boldsymbol{x}_t 
=
\boldsymbol{A} \boldsymbol{x}_t + \boldsymbol{u}_{0.xt} 
\end{align}
In terms of asymptotic theory we have that 
\begin{align}
\xi_{ y . xn}^{+} (s) := \xi_{ yn } (s) - \boldsymbol{\Omega}_{yx} \boldsymbol{\Omega}_{xx}^{-1} \xi_{ xn}(s) \Rightarrow \boldsymbol{B}_{y.x} (s) \equiv BM \big( \boldsymbol{\Omega}_{y.xx}  \big)
\end{align}
Notice that the unit root limit theory given by expressions (17) and (18) of \cite{magdalinos2009limit} involves the demeaned process $B^{\mu}(s)$ although there is no intercept in the regression. In particular, the demeaning effects arises because as shown in expression (16), in the direction of the initial condition, the time series is dominated by a component that behaves like a constant. Therefore, the sample moment matrix is no longer asymptotically singular such that 
\begin{align}
\frac{1}{n^2} \sum_{t=1}^n \boldsymbol{x}_{t-1} \boldsymbol{x}_{t-1}^{\prime} \Rightarrow \int_0^1 J_c^{*} (r) J_c^{* \prime} (r)
\end{align}
Then, from the first-order vector autoregression of $y_t$ on $y_{t-1}$ we obtain the regression coefficient matrix 
\begin{align}
\hat{A} = \left( \sum_{t=1}^T y_t y_{t-1}^{\prime} \right) \left( \sum_{t=1}^T y_{t-1} y_{t-1}^{\prime} \right)^{-1}. 
\end{align} 
The asymptotic behaviour of $\hat{A}$ is described by a corresponding functional Brownian motion. To be more precise consider standardized deviations of $\hat{A}$ about $I_n$ such that
\begin{align}
T \left( \hat{A} - I_p \right) = \left( \frac{1}{n} \sum_{t=1}^T u_t y_{t-1}^{\prime} \right) \left( \frac{1}{n^2} \sum_{t=1}^T y_{t-1} y_{t-1}^{\prime} \right)^{-1}. 
\end{align}   
Therefore, to obtain the asymptotic behaviour of the above statistic we write the sample second moment $\left( \frac{1}{T^2} \sum_{t=1}^n y_{t-1} y_{t-1}^{\prime} \right)$ as a quadratic functional of the random element $X_n(r)$, at least up to a term of $o_p(1)$. 
\begin{align}
\frac{1}{n^2} \sum_{t=1}^n y_{t-1} y_{t-1}^{\prime} = \int_0^1 X_n(r) X_n(r)^{\prime} dr + o_p(1)  
\end{align}
Then, by an application of the continuous mapping theorem we can establish that 
\begin{align}
\frac{1}{n^2} \sum_{t=1}^n y_{t-1} y_{t-1}^{\prime} 
&\Rightarrow \int_0^1 B(r) B(r)^{\prime} dr, \ \ \ \text{as} \ \ T \to \infty. 
\\
\frac{1}{n} \sum_{t=1}^n u_t y_{t-1}^{\prime} &\Rightarrow \int_0^1 dB(r) B(r)^{\prime} 
\end{align}

\end{example}

\newpage

\subsubsection{FM transformation in Cointegrating Regression}

A challenging issue in the cointegrating regression literature is that when $u_t$ is uncorrelated with $v_t$ and hence uncorrelated with $x_t$, it follows that \textit{(i)} $\lambda_{uv} = \boldsymbol{0}$, $\Delta_{uv} = \boldsymbol{0}$, and \textit{(ii)} $B_u(r)$ is independent of $B_v(r)$. Furthermore, because of the independence between the Brownian motions $B_u(r)$ and $B_v(r)$ in this case, one can condition on $B_v(r)$ to show that the limiting distribution of $T \left( \widehat{\beta} - \beta \right)$ is a zero mean Gaussian mixture. Therefore, one can also show that the $t$ and Wald statistics for testing the hypotheses about $\beta$ have the usual $\mathcal{N}(0,1)$ and chi-square limits assuming serial correlation in $u_t$ is handled using consistent robust standard errors. On the other hand, when regressors are endogenous, the limiting distribution of $T \left( \widehat{\beta} - \beta \right)$ is obviously more complicated because of the correlation between $B_u(r)$ and $B_v(r)$ and the presence of nuisance parameters in the vector $\Delta_{vu}$. 

Thus, we can no longer condition on $B_v(r)$ to obtain an asymptotic normal result and $\Delta_{vu}$ introduces an asymptotic bias. In particular, inference is difficult in this situation because nuisance parameters cannot be removed by simple scaling methods. Therefore, the FM-OLS estimator of \cite{phillips1990statistical} (see, also \cite{phillips1990asymptotic}) is employed to asymptotically remove $\Delta_{vu}$ and to deal with the correlation between  $B_u(r)$ and $B_v(r)$. In practise FM-OLS estimation requires the choice of bandwidth and kernel. While bandwidth and kernel play no role asymptotically when considering the consistency results for $\widehat{\Omega}$ and $\widehat{\Delta}$, in finite samples they affect the sampling distributions of $\widehat{\theta}^{+}$ and thus of $t$ and Wald statistics based on the FM-OLS estimator of $\widehat{\theta}^{+}$. In other words, to obtain an approximation for $\widehat{\theta}^{+}$ that reflects the choice of bandwidth and kernel, the natural asymptotic theory to use is the fixed-b theory. However, fixed-b theory has primarily been developed for models with stationary regressors, which means that some additional work is required to obtain analogous results for cointegrating regressions. A major difference is that the first component of $\widehat{\eta}_t$, that is, $\widehat{u}_t$, is the residual from a cointegrating regression, which leads to dependence of the corresponding limit partial sum process on the number of integrated regressors and the deterministic components.

Specifically, the FM-OLS estimator for the cointegrating predictive regression model is defined as below
\begin{align}
A_{FM} = \left( \sum_{t=1}^T \boldsymbol{y}_t^{+} \boldsymbol{x}_t^{\prime} - T \hat{\kappa} \hat{\boldsymbol{\Gamma}} \right) \left( \sum_{t=1}^T    \boldsymbol{x}_t \boldsymbol{x}_t^{\prime} \right)^{-1}
\end{align} 
where $y_t^{+} = y_t - \hat{\boldsymbol{\Omega}}_{12} \hat{\boldsymbol{\Omega}}_{22}^{-1} D \boldsymbol{x}_t$. Then, the asymptotic distribution of the FM-OLS estimator
\begin{align}
T \big( \dot{A}_{FM} - A \big) \Rightarrow \left( \int_0^1 d B_{1.2} (r) B_2(r)^{\prime} \right) \left( \int_0^1 B_2(r) B_2(r)^{\prime} dr \right)^{-1}, 
\end{align}  
which is the same as the asymptotic distribution of the CCR estimator. Now, we may use the residuals
\begin{align}
\dot{u}_t^{+} 
= 
\boldsymbol{y}_t^{+} - \dot{A}_{FM} \boldsymbol{x}_t 
= 
u_t - \hat{\boldsymbol{\Omega}}_{12} \hat{\boldsymbol{\Omega}}_{22}^{-1} \Delta \boldsymbol{x}_t  - \big( \dot{A}_{FM} - A \big) \boldsymbol{x}_t
\end{align}

\newpage

The FM-OLS is optimal in the sense that is has the same limiting distribution as Gaussian maximum likelihood. Then estimators can estimated with long-run covariance matrices. 

The OLS residuals can be written as below
\begin{align}
\hat{u}_t = y_t - \hat{A} \boldsymbol{x}_t = u_t - \left( \hat{A} - A \right) \boldsymbol{x}_t 
\end{align}
Consider the partial sum processes $S^{u}_t = \sum_{i=1}^t u_i$ and $S^{x}_t = \sum_{i=1}^t x_i$, which implies that 
\begin{align}
\frac{1}{T^2} \sum_{t=1}^T \hat{S}_t \hat{S}_t^{\prime} \Rightarrow \int_0^1 \bigg[ \boldsymbol{B}_1(r) - \alpha \boldsymbol{B}_2(r) \bigg] \bigg[ \boldsymbol{B}_1(r) - \alpha \boldsymbol{B}_2(r) \bigg]^{\prime} dr,
\end{align}
where
\begin{align}
\alpha = \bigg\{ \int_0^1 \boldsymbol{B}_2(r) d \boldsymbol{B}_1(r)^{\prime} + \boldsymbol{\Gamma}_{21} \bigg\} \left\{ \int_0^1 \boldsymbol{B}_2(r) \boldsymbol{B}_2(r)^{\prime}   \right\}
\end{align}
In particular, these results show that eliminating nuisance parameters is not easy unless $\boldsymbol{x}_t$ is strictly exogenous. Thus, the endogeneity aspect of the model implies that there is a dependence between the nonstationary regressor $X_t$ and the stationary error $u_t$.

\subsubsection{FM estimation with nearly integrated regressors}

In particular, the autoregressive root of the regressor which is parametrized with the local-to-unity specification allows to capture the near unit root behaviour of many predictor variables and is less restrictive than the pure unit root specification. Moreover, the OLS estimator of $\beta$ does not have an asymptotically mixed normal distribution due to the correlation between $u_t$ and $v_t$. From the paper of \cite{hjalmarsson2007fully} the fully modified OLS estimator is given by 
\begin{align}
\hat{\beta}^{+} 
= 
\left( \sum_{t=1}^n \underline{\boldsymbol{y}}^{+}_t \underline{\boldsymbol{x}}_{t-1} n - \widehat{\boldsymbol{\Lambda}}_{12} \right) \left( \sum_{t=1}^n    \underline{\boldsymbol{x}}_{t-1} \underline{\boldsymbol{x}}_{t-1}^{\prime} \right)^{-1}.
\end{align}
Define with $\omega_{11.2} = \omega_{11} - \omega_{21}^2 \omega_{22}^{-1}$ and $B_{1.2} = B_1 - \omega_{21} \omega_{22}^{-1} B_2 \equiv BM ( \omega_{11.2} )$. Then, as $T \to \infty$, the following limiting distribution holds
\begin{align}
T \left( \hat{\beta}^{+} - \beta \right) \Rightarrow \left( \int_0^1 d B_{1.2} \underline{J}_c \right) \left( \int_0^1     \underline{J}_c^2 \right)^{-1} \equiv \mathcal{MN} \left( 0, \omega_{11.2} \left( \int_0^1 \underline{J}_c^2 \right)^{-1}    \right).
\end{align} 
Furthermore, under the assumption that both $u_t$ and $v_t$ are martingale difference sequences it can be shown that the OLS estimation of the augmented regression 
\begin{align}
y_t = \alpha + \beta x_{t-1} + \gamma \Delta_c \boldsymbol{x}_t + u_{t.v},
\end{align}
yields an estimator of $\beta$ with an asymptotic distribution identical to that of $\hat{\beta}^{+}$.

\newpage

\subsection{Dynamic Seemingly Unrelated Cointegrating Regression}

We follow the framework proposed by \cite{mark2005dynamic}. More specifically, we consider a fixed number of $N$ cointegrating regressions each with $T$ observations.

\begin{assumption}
Each equation of the SUR system $i \in \left\{ 1,..., N \right\}$ has the triangular representation: 
\begin{align}
y_{it} &= \boldsymbol{x}_{it}^{\prime} \boldsymbol{\beta}_i + u_{it}
\\
\Delta \boldsymbol{x}_{it} &= \boldsymbol{e}_{it}
\end{align}
where $\boldsymbol{x}_{it}$ and $\boldsymbol{e}_{it}$ are $(k \times 1)$ dimensional vectors. 
\end{assumption} 

Moreover, we denote with $\boldsymbol{u}_{t} = \left( u_{1t}, ..., u_{NT} \right)^{\prime}$ and $\boldsymbol{e}_{t} = \left( e^{\prime}_{1t} ,... , e^{\prime}_{Nt} \right)$ and with $\boldsymbol{w}_t = \left( \boldsymbol{u}_{t}^{\prime},     \boldsymbol{e}_{t}^{\prime} \right)^{\prime}$ an $N(k+1)$dimensional vector with the orthonormal Wold moving average representation such that $\boldsymbol{w}_t = \boldsymbol{\Psi}(\boldsymbol{L}) \boldsymbol{\epsilon}_t$. Notice that $\boldsymbol{\epsilon}_t$ is a martingale difference sequence, such that $\mathbb{E} \left[ \boldsymbol{\epsilon}_t \right] = \boldsymbol{0}$ and $\mathbb{E} \big[ \boldsymbol{\epsilon}_t \boldsymbol{\epsilon}_t^{\prime} \big] = \boldsymbol{I}_k$ and finite fourth moments.  Therefore, the endogeneity problem shows up as correlation between the $i-$th equilibrium error $u_{it}$ and potentially an infinite number of leads and lags of the first differenced regressors from all of the equations of the system such as $\Delta \boldsymbol{x}_{jt} = \boldsymbol{e}_{jt}$ for $(i, j) \in \left\{ 1,..., N \right\}$. 

Consider the following formulation of the model 
\begin{align}
y_{it} = \boldsymbol{x}_{it}^{\prime} \boldsymbol{\beta}_i + \boldsymbol{z}_{pt}^{\prime} \boldsymbol{\delta}_{pi} + u_{it}
\end{align}

Let $\boldsymbol{y}_t = \left( y_{1t},..., y_{Nt} \right)^{\prime}$, $\boldsymbol{u}_{t} = \left( u_{1t},..., u_{Nt} \right)^{\prime}$ and $\boldsymbol{\beta} = \left( \boldsymbol{\beta}_1^{\prime} ,...,     \boldsymbol{\beta}_N^{\prime} \right)^{\prime}$, $\boldsymbol{\delta} = \left( \boldsymbol{\delta}_1^{\prime} ,...,     \boldsymbol{\delta}_N^{\prime} \right)^{\prime}$. Moreover we have the following matrices
\begin{align}
\boldsymbol{Z}_{pt} = \left( \mathbf{I}_N \otimes \boldsymbol{z}_{pt}   \right), \ \ \boldsymbol{X}_{t} = \text{diag} \left( \boldsymbol{x}_{1t} ,... ,  \boldsymbol{x}_{Nt} \right) \ \ \text{and} \ \ \boldsymbol{W}_t = \left( \boldsymbol{X}_{t}^{\prime} , \boldsymbol{Z}_{pt}^{\prime} \right)^{\prime}.
\end{align}
Then, the equations can be stacked together in a system as below
\begin{align}
\boldsymbol{y}_t = \left( \boldsymbol{\beta}^{\prime} , \boldsymbol{\delta}_p^{\prime} \right) \boldsymbol{W}_t + \boldsymbol{u}_t. 
\end{align}
Moreover, denote the long-run covariance matrix of $\boldsymbol{u}_t$ by $\boldsymbol{\Omega}_{uu}$. Then, an expression for the DSUR estimator with known $\boldsymbol{\Omega}_{uu}$ is given by 
\begin{align}
\begin{bmatrix}
\hat{ \boldsymbol{\beta} }_{DSUR}
\\
\\
\hat{ \boldsymbol{\delta} }_{DSUR}
\end{bmatrix}
=
\left( \sum_{t = p+1}^{T-p} \boldsymbol{W}_{t} \boldsymbol{\Omega}_{uu}^{-1} \boldsymbol{W}_{t}^{\prime} \right)^{-1} \left( \sum_{t = p+1}^{T-p} \boldsymbol{W}_{t} \boldsymbol{\Omega}_{uu}^{-1} \boldsymbol{y}_{t} \right).
\end{align}
This, it can be shown that $\hat{ \boldsymbol{\beta} }_{DSUR}$ is asymptotically mixed normal.

\newpage

Hypothesis testing for the linear restrictions $\boldsymbol{R} \boldsymbol{\beta} = \boldsymbol{r}$, where $\boldsymbol{R}$ is a $q \times Nk$ matrix of constants and $\boldsymbol{r}$ is a $q-$dimensional vector of constants. In practise, the testing hypothesis can be also written as below
\begin{align}
\mathbb{H}_0: \beta_1 = ... = \beta_N
\end{align}
Notice that the above formulation of the null hypothesis provides a  way for conveniently formulating a test of homogeneity restrictions on the cointegrating vectors. 

Then, the DSUR Wald test statistic is expressed as below
\begin{align}
\mathcal{W}_{dsur} = \left( \boldsymbol{R} \hat{ \boldsymbol{\beta} }_{dsur} -  \boldsymbol{r} \right)^{\prime}  \big[ \boldsymbol{R} \hat{ \boldsymbol{V} }_{dsur} \boldsymbol{R}^{\prime} \big]^{-1}  \left( \boldsymbol{R} \hat{ \boldsymbol{\beta} }_{dsur} -  \boldsymbol{r} \right)
\end{align}
where
\begin{align}
\hat{ \boldsymbol{V} }_{dsur} = \sum_{t = p+1}^{T-p} \boldsymbol{X}_{t} \boldsymbol{\Omega}_{uu}^{-1} \boldsymbol{X}_{t}^{\prime} 
\end{align}
The Wald statistic $\mathcal{W}_{DSUR}$ is asymptotically distributed as a chi-square variate with $q$ degrees of freedom under the null hypothesis. Moreover, in applications we replace $\boldsymbol{\Omega}_{uu}$ with a consistent estimator, $\hat{ \boldsymbol{\Omega} }_{uu} \overset{ p }{ \to } \boldsymbol{\Omega}_{uu}$. Estimation of the long-run covariance matrix is discussed below. Such an estimator might be called a "feasible" DSUR estimator.

\paragraph{Two Step DSUR}

The first step purges endogeneity by least squares and the second step estimates $\boldsymbol{\beta}$ by running SUR on the least squares residuals obtained from the first-step regression. Under standard regularity conditions, the two-step DSUR estimator is asymptotically equivalent to the DSUR estimator $\hat{ \boldsymbol{\beta} }_{DSUR}$ discussed above. Stacking the equations together as $\hat{ \boldsymbol{y} }_t = \hat{ \boldsymbol{X} }_t^{\prime} \boldsymbol{\beta} + \hat{ \boldsymbol{u} }_t$ and running SUR gives the two-step DSUR estimator gives, 
\begin{align}
\hat{ \boldsymbol{\beta} }_{2sdsur} = \left[ \sum_{t = p+1}^{T-p}   \hat{ \boldsymbol{X} }_t \boldsymbol{\Omega}_{uu}^{-1} \hat{ \boldsymbol{X} }_t^{\prime} \right]^{-1} \left[ \sum_{t = p+1}^{T-p}   \hat{ \boldsymbol{X} }_t \boldsymbol{\Omega}_{uu}^{-1} \hat{ \boldsymbol{y} }_t \right]
\end{align}

\paragraph{Restricted DSUR}

Next, we consider the estimation of the cointegration vector under the homogeneity restrictions $\boldsymbol{\beta}_1 = ... = \boldsymbol{\beta}_N = \boldsymbol{\beta}$. As in two-step DSUR, endogeneity can first be purged by regressing $y_{it}$ and each element of $\boldsymbol{x}_{it}$ on $\boldsymbol{z}_{pt}$. Let $\hat{y}_{it}$ and $\hat{x}_{it}$ denote the resulting regression errors. Therefore, the estimation problem becomes one of estimating $\boldsymbol{\beta}$, in the system of equations $\hat{y}_{it} = \hat{\boldsymbol{x} }_{it}^{ \prime } \boldsymbol{\beta} + \hat{u}_{it}$. Stacking the system equations together:  
\begin{align}
\hat{ \boldsymbol{y} }_t = \hat{ \boldsymbol{x} }_t^{\prime} \boldsymbol{\beta} + \hat{ \boldsymbol{u} }_t
\end{align} 

\newpage

\paragraph{Asymptotic Properties}
\

Let $\boldsymbol{W}(r)$ be a vector standard Brownian motion for $0 \leq r \leq 1$. 

\begin{proposition}
Let $T^{*} = (T - 2p)$. Under the conditions of Assumptions above, 
\begin{enumerate}

\item[(i)] $T^{*} \left( \hat{\boldsymbol{\beta}}_{dsur} - \boldsymbol{\beta}_{dsur} \right)$ and $\sqrt{ T^{*} } \left( \hat{\boldsymbol{\delta}}_{p, dsur} - \boldsymbol{\delta}_{p, dsur} \right)$ are asymptotically independent. 

\item[(ii)] If $\boldsymbol{B}_{e} = \text{diag} \left( \boldsymbol{B}_{ e_1} , ..., \boldsymbol{B}_{ e_N} \right)$, $\hat{ \boldsymbol{V} }_{dsur} = \sum_{t = p+1}^{T-p} \boldsymbol{X}_{t} \boldsymbol{\Omega}_{uu}^{-1} \boldsymbol{X}_{t}^{\prime}$ and $\boldsymbol{R}$ is a $q \times Nk$ matrix of constants such that $\boldsymbol{R} \boldsymbol{\beta} = \boldsymbol{r}$, then as $T^{*} \to \infty$, 
\begin{align}
T^{*} \left( \hat{ \boldsymbol{\beta} }_{dsur} - \boldsymbol{\beta} \right) \Rightarrow \left( \int \boldsymbol{B}_{e} \boldsymbol{\Omega}_{uu}^{-1} \boldsymbol{B}_{e}^{\prime} \right)^{-1} \left( \int \boldsymbol{B}_{e} \boldsymbol{\Omega}_{uu}^{-1} d \boldsymbol{B}_{u} \right)
\end{align}
and 
\begin{align}
 \left( \boldsymbol{R} \hat{ \boldsymbol{\beta} }_{dsur} -  \boldsymbol{r} \right)^{\prime}  \big[ \boldsymbol{R} \hat{ \boldsymbol{V} }_{dsur} \boldsymbol{R}^{\prime} \big]^{-1}  \left( \boldsymbol{R} \hat{ \boldsymbol{\beta} }_{dsur} -  \boldsymbol{r} \right) \overset{ D }{ \to } \chi^2_q.
\end{align}
\end{enumerate}
\end{proposition}

Notice that a functional central limit theory applies such that
\begin{align}
\displaystyle  \frac{1}{ \sqrt{T^{*}} } \sum_{ t = p+1}^{ \floor{ (T-p) r} } \boldsymbol{w}_t \overset{ D }{ \to } \left( \boldsymbol{B}_u^{\prime} , \boldsymbol{B}_v^{\prime} \right)^{\prime}
\end{align}
with the long-run covariance matrix $\boldsymbol{\Omega} = \mathsf{diag} \big( \boldsymbol{\Omega}_{uu} , \boldsymbol{\Omega}_{ee} \big)$. 

Due to the block diagonality of $\boldsymbol{\Omega}$, it can be easily verified that $\boldsymbol{B}_u$ and $\boldsymbol{B}_e$ are independent as discussed in the frameworks of \cite{mark2005dynamic} and   \cite{saikkonen1993estimation}. 
Similarly \cite{phillips1999linear} develop the asymptotic theory for regression models with nonstationary panel data for which the time series component is an integrated process and where both $T$ and $n$ are large. Recently, \cite{chen2023seemingly} consider a framework for robust estimation of a Seemingly Unrelated Regression System for VAR models with explosive roots based on the approach of \cite{magdalinos2009limit} and \cite{magdalinos2020least}.

\medskip

\begin{remark}
Notice that the possibility of interpreting cointegration vectors as economic long-run relations is the main reason why the vector autoregressive model has become widely used in the empirical analysis of economic data. Specifically, the statistical concept of cointegration in the $I(1)$ model, involving linear combination of levels of the variables, corresponds to the economic concept of a long-run static steady-state relation. Similarly, the statistical concept of multicointegration or polynomial cointegration, involving linear combinations of both levels and differences, corresponds to the economic concept of a long-run dynamic steady-state relation. 
\end{remark}

\newpage 

\subsection{Predictive Regression Models}

\subsubsection{Regression with Integrated Regressors}

Various examples related to integrated processes and the limit theory in regression with integrated predictors can be found in the book of \cite{banerjee1993co}. Among others, related literature includes the papers of \cite{cavanagh1995inference}, \cite{stock1991confidence}, \cite{jansson2006optimal}. Below we present an example (see, corresponding chapter from \cite{banerjee1993co}), to shed light on the related asympotics. 

\begin{example} (Cointegrating Regression)

Consider the following bivariate system of co-integrated variables $\{ y_t \}_{t=1}^{\infty}$ and $\{ x_t \}_{t=1}^{\infty}$.  
\begin{align}
y_t &= \beta x_t + u_t \\
\Delta x_t &= \epsilon_t
\end{align}
with $u_t \sim N(0, \sigma^2_u )$, \ $\epsilon_t \sim N(0, \sigma^2_{\epsilon} )$ and $\mathbb{E}( u_t \epsilon_s ) = \sigma_{u \epsilon } \ \forall \ t \neq s$.
\end{example}
The OLS estimator of $\beta$ is given by 
\begin{align}
\hat{ \beta } = \left( \sum_{t=1}^T x_t^2 \right)^{-1} \left( \sum_{t=1}^T y_t x_t \right)
\end{align}
Thus, 
\begin{align}
T \left( \hat{ \beta } - \beta \right)   = \left( T^{-2} \sum_{t=1}^T x_t^2 \right)^{-1} \left( T^{-1} \sum_{t=1}^T x_t u_t \right)
\end{align}
Note that we have a regression with an integrated regressor, since
\begin{align}
x_t = x_{t-1} +  \epsilon_t
\end{align}
Therefore, it follows that 
\begin{align}
\left( \frac{1}{T^2} \sum_{t=1}^T x_t^2 \right) \Rightarrow \sigma^2_{ \epsilon } \int_0^1 W_{ \epsilon } (r)^2 dr.
\end{align}
In order to derive the limiting distribution of the model parameter in the case of integrated regressors, we shall first derive the limiting distribution of  $\left( T^{-1} \sum_{t=1}^T x_t u_t \right)$.

To do this, we condition $u_t$ on $\epsilon_t$ as given below
\begin{align}
u_t = \phi \epsilon_t + v_t, \ \ \phi = \frac{ \sigma_{u \epsilon} }{ \sigma^2_{\epsilon} } \ \ \text{and} \ \ \sigma^2_v = \sigma^2_u -   \frac{ \sigma^2_{u \epsilon} }{ \sigma^2_{\epsilon} } 
\end{align}

\newpage

Define $W_{\epsilon} (r)$ and $W_{v} (r)$ to be two independent Wiener processes on $\mathcal{C}[0,1]$. Therefore, 
\begin{align}
\left( T^{-1} \sum_{t=1}^T x_t u_t \right)
&= T^{-1}  \sum_{t=1}^T x_t \left(  \phi \epsilon_t + v_t \right)
\\
&= \phi \left( T^{-1} \sum_{t=1}^T x_t \epsilon_t \right) + \left( T^{-1} \sum_{t=1}^T x_t v_t \right)
\end{align}
Substituting $x_t = x_{t-1} +  \epsilon_t$ into the above gives
\begin{align}
\left( T^{-1} \sum_{t=1}^T x_t u_t \right)
&= \phi \left( T^{-1} \sum_{t=1}^T (  x_{t-1} +  \epsilon_t ) \epsilon_t \right) + \left( T^{-1} \sum_{t=1}^T (  x_{t-1} +  \epsilon_t ) v_t \right)
\\
&= \phi \left( T^{-1} \sum_{t=1}^T x_{t-1} \epsilon_t \right) + \phi \left( T^{-1} \sum_{t=1}^T  \epsilon^2_t \right) 
\\
&\ \ \  + \left( T^{-1} \sum_{t=1}^T x_{t-1} v_t \right) + \left( T^{-1} \sum_{t=1}^T \epsilon_t v_t \right)
\end{align}
We have the following asymptotic results
\begin{align}
T^{-1} \sum_{t=1}^T  \epsilon^2_t  & \overset{ p }{ \to } \sigma^2_{\epsilon} 
\\
T^{-1} \sum_{t=1}^T \epsilon_t v_t & \overset{ p }{ \to }  0
\\
 T^{-1} \sum_{t=1}^T x_{t-1} v_t & \Rightarrow \sigma_{ \epsilon } \sigma_{ v } \int_0^1 W_{\epsilon} (r) dW_v (r) \equiv \int_0^1 B_{\epsilon} (r) dB_v (r)
\\
T^{-1} \sum_{t=1}^T x_{t-1} \epsilon_t & \Rightarrow  \frac{ \sigma^2_{ \epsilon } }{2}  \bigg[  W^2_{\epsilon} (1) - 1 \bigg] \ \ \text{(see, \cite{phillips1987time})} 
\end{align}
Therefore, we have that 
\begin{align}
T^{-1} \sum_{t=1}^T x_{t} u_t \Rightarrow \left\{ \phi \left(  \frac{ \sigma^2_{ \epsilon } }{2}  \bigg[  W^2_{\epsilon} (1) - 1 \bigg]  \right)   + \phi \sigma^2_{\epsilon} +   \sigma_{ \epsilon } \sigma_{ v } \int_0^1 W_{\epsilon} (r) dW_v (r)  \right\}
\end{align}
Furthermore, \cite{phillips1988asymptotic} proved the following result
\begin{align}
\int_0^1 W_{\epsilon} (r) dW_v (r) \Rightarrow N \left( 0, \int_0^1 W_{\epsilon} (r) dr   \right)
\end{align}
Therefore, under the null hypothesis $H_0: \beta = 0$, 
\begin{align}
T \hat{\beta} \Rightarrow \left\{ \phi \frac{ \sigma^2_{ \epsilon } }{2}   \bigg[  W^2_{\epsilon} (1) + 1 \bigg]  +   \sigma_{ \epsilon } \sigma_{ v } \bigg[ \int_0^1 W_{\epsilon} (r) dW_v (r) \bigg] \right\} \left( \sigma^2_{ \epsilon } \int_0^1 W_{ \epsilon } (r)^2 dr \right)^{-1}.
\end{align}

\newpage

Thus, the $t-$statistic, denoted as $\mathcal{T}_{\beta = 0}$ for testing the null hypothesis, $H_0: \beta = 0$, 
\begin{align}
\mathcal{T}_{\beta = 0} 
= \frac{ \hat{\beta} }{  \hat{\sigma}^2_u \left( \displaystyle \sum_{t=1}^T x_t^2 \right)^{- \frac{1}{2} }} 
= T \frac{ \hat{\beta} }{  \hat{\sigma}^2_u \left( \displaystyle T^{-2} \sum_{t=1}^T x_t^2 \right)^{- \frac{1}{2} }} 
\end{align}
has the following limiting distribution
\begin{align}
\mathcal{T}_{\beta = 0}  
&\Rightarrow
\left\{ \phi \frac{ \sigma^2_{ \epsilon } }{2}   \bigg[  W^2_{\epsilon} (1) + 1 \bigg]  +   \sigma_{ \epsilon } \sigma_{ v } \bigg[ \int_0^1 W_{\epsilon} (r) dW_v (r) \bigg] \right\} \left( \sigma^2_{ \epsilon } \int_0^1 W_{ \epsilon } (r)^2 dr \right)^{-\frac{1}{2}} \times \frac{1}{ \sigma_u^2   }
\\
&\equiv 
\frac{ \phi }{ 2 } \frac{  \sigma_{ \epsilon } }{ \sigma_u  } \bigg[  W^2_{\epsilon} (1) + 1 \bigg] \left(  \int_0^1 W_{ \epsilon } (r)^2 dr \right)^{-\frac{1}{2}} + \frac{  \sigma_{ v } }{ \sigma_u  } N(0,1)
\end{align}

The above limiting distribution indicates that the $t-$ratio of $\hat{\beta}$ does not follow a standard normal distribution unless $\phi = 0$, which in that cases implies that $x_t$ is exogenous for the estimation of $\beta$. In particular, when  $\phi \neq 0$ then the first term of the above limiting distribution gives rise to second-order or endogeneity bias, which although asymptotically negligible in estimating $\beta$ due to super consistency, can be important in finite samples.

\begin{example}
We follow \cite{cavanagh1995inference}. Consider the following recursive system
\begin{align}
y_t &= \mu_y + \gamma x_{t-1} + \epsilon_{2t} \\
\label{auto}
x_t &= \mu_x + v_t 
\end{align}
where the sequence $v_t$ is generated as
\begin{align}
\label{error}
\big(1 - \alpha L \big) b(L) v_t = \epsilon_{1t}
\end{align} 
where $b(L) = \sum_{j=0}^k b_j L^j$, $b_0 = 1$, and $\epsilon_t = ( \epsilon_{1t}, \epsilon_{2t})^{\prime}$.
\begin{assumption}
The vector $\epsilon_t = ( \epsilon_{1t}, \epsilon_{2t})^{\prime}$ is a \textit{martingale difference sequence} such that 
\begin{enumerate}
\item[(i)] $\mathbb{E} \left(  \epsilon_t \epsilon_t^{\prime} | \epsilon_{t-1}, \epsilon_{t-2},... \right) = \Sigma > 0.$
\item[(ii)] $\mathbb{E} \left( \epsilon_{it}^4 < \infty \right)$, $i =1,2$.
\end{enumerate}
\end{assumption}
Let $\delta = \text{corr}( \epsilon_{1t}, \epsilon_{2t} )$ and assume that $E v_0^2 < \infty$. The roots of $b(L)$ are assumed to be fixed and less than 1 in absolute value. 
\begin{itemize}
\item If $|\alpha| < 1$ and $\alpha$ is fixed, then $x_t$ is integrated of order 0, that is, $I(0)$.
\item If $\alpha = 1$ and $\rho$ is fixed, then $x_t$ is integrated of order 1, that is, $I(1)$.
\end{itemize}

\newpage

Therefore, we consider $\alpha$ to be the largest autoregressive root of the univariate representation of $x_t$. Thus, we can jointly write expressions $x_t = \mu_x + v_t$ and $\big(1 - \alpha L \big) b(L) v_t = \epsilon_{1t}$ in augmented Dickey-Fuller, (ADF) form as below
\begin{align}
\Delta x_t = \tilde{\mu}_x + \beta x_{t-1} + \alpha(L) \Delta x_{t-1} + \epsilon_{1t} 
\end{align}
where 
\begin{align}
\tilde{\mu}_x = (1 - \alpha ) b(1) \mu_x, \ \ \beta = ( \alpha -1 ) b(1), \ \ \ a_i = - \sum_{j = i + 1}^k \tilde{\alpha}_j
\end{align}
with $\tilde{\alpha}( L ) = L^{-1} \big[ 1 - \big(1 - \alpha L \big) b(L)  \big]$. We consider the problem of testing the null hypothesis that $\gamma = \gamma_0$ or, equivalently, constructing confidence intervals for $\gamma$. Thus, for this problem, the root $\alpha$ is a nuisance parameter. Let $B = (B_1, B_2)$ be a two-dimensional Brownian motion with covariance matrix $\tilde{\Sigma}$ such that
\begin{align}
\tilde{\Sigma} = 
\begin{bmatrix}
1       &  \delta \\
\delta  &  1 \\
\end{bmatrix}
\end{align} 
and let $J_c$ be the diffusion process\footnote{Notice that the Ornstein–Uhlenbeck stochastic process is a time-homogeneous Ito diffusion process.} defined by
\begin{align}
J_c(r) = c J_c(r) ds + dB_1(s), 
\end{align}
where $J_c(0) = 0$. Let $J_c^{\mu} (s) = J_c(s) - \int_0^1 J_c(r) dr$. Under the local-to-unity model $\alpha = (1 + c / T)$ then,
\begin{align}
\left\{ \sigma_{11}^{-1/2} T^{-1/2} \sum_{t=1}^{ \floor{T \pi }} \epsilon_{1t} , \  \sigma_{22}^{-1/2} T^{-1/2} \sum_{t=1}^{ \floor{T \pi }} \epsilon_{2t}, \ \omega^{-1} T^{-1/2} x^{\mu}_{ \floor{T \pi } }      \right\} 
\Rightarrow 
\bigg\{ B_1( \pi ), B_2( \pi ), J_c^{\mu}( \pi )  \bigg\}
\end{align}
jointly, where $\omega^2 = \sigma_{11} / b(1)^2$, \ $\displaystyle x_t^{\mu} = x_t - \frac{1}{T-1} \sum_{t=2}^T x_{t-1}$. Therefore, it follows that $t_{\beta}$ and $t_{\gamma}$ have the following joint limiting representation:
\begin{align}
\label{joint}
\big( t_{\beta}, t_{\gamma} \big)
\Rightarrow 
\bigg\{ \tau_{1c} + c \Theta_c, \tau_{2c} \bigg\} \equiv \bigg\{  \tau_{1c} + c \theta_c , \delta \tau_{1c} + (1 - \delta^2)^{1/2} \mathcal{Z}  \bigg\}
\end{align}
where 
\begin{align}
\tau_{1c} &= \left[ \int \bigg( J_c^{\mu}( r )  \bigg)^2  \right]^{-1/2} \times \left\{ \int J_c^{\mu}(r) dB_1  \right\}
\\
\tau_{2c} &= \left[ \int \bigg( J_c^{\mu}( r )  \bigg)^2  \right]^{-1/2} \times \left\{ \int J_c^{\mu}(r) dB_2  \right\}
\end{align}

\newpage

Furthermore, 
\begin{align}
\Theta_c = \left[ \int \bigg( J_c^{\mu}( r )  \bigg)^2  \right]^{1/2}
\end{align}
and $\mathcal{Z}$ is a standard normal random variable distributed independently of $\big( B_1, J_c \big)$. The final expression in \eqref{joint} is obtained by writing 
\begin{align}
B_2 = \delta B_1 + (1 - \delta^2)^{1/2} \widetilde{B}_2
\end{align}
where $\widetilde{B}_2$ is a standard Brownian motion distributed independently of $B_1$. The limiting distribution of $t_{\gamma}$ depends on both $c$ and $\delta$. However, $\delta$ is consistently estimated by the sample correlation between $\hat{\epsilon}_{1t}$ and $\hat{\epsilon}_{2t}$, therefore we treat $\delta$ as known for the purposes of the asymptotic theory. Moreover, a joint test of $c$ and $\gamma$ can be performed using an appropriate Wald statistic. 

We define the following vector 
\begin{align}
\phi_T ( \gamma_0, c_0 ) = \bigg[ T \hat{\beta} - c_0 \hat{b}(1), T( \hat{\gamma} - \gamma_0 ) \bigg]^{\prime}
\end{align}
where $\hat{b}(1) = 1 - \sum_{j=1}^k \hat{a}_{j-1}$ with $\left\{ \hat{a}_{j} \right\}$ are the estimators of $\left\{ a_{j} \right\}$ from the OLS estimation of the ADF parametrization. Let $\hat{\Sigma}$ be the $2 \times 2$ matrix with typical element
\begin{align}
\hat{\sigma}_{ij} = \displaystyle \frac{1}{T-1} \sum_{t=1}^T e_{it} e_{jt}
\end{align}
where $e_{1t}$ and $e_{2t}$ are the residuals of the corresponding equations. 

Consider the test statistic 
\begin{align}
W \left( \gamma_0, c_0 \right) = \frac{1}{2} \phi_T ( \gamma_0, c_0 )^{\prime} \left( \hat{\Sigma} \ T^{-2} \sum_{t=2}^T \left( x^{\mu}_{t-1} \right)^2 \right) \phi_T ( \gamma_0, c_0 ).
\end{align}

Furthermore, extensions of the calculations in \cite{stock1991confidence} show that, under the null hypothesis $\mathbb{H}_0: \gamma = \gamma_0$ and $\beta = 0$ (jointly), that is, $( \gamma, c ) = ( \gamma_0, c_0 )$ 
\begin{align}
\label{joint2}
W \left( \gamma_0, c_0 \right) \Rightarrow \frac{1}{2} \left( \tau^2_{1c_0}  + \mathcal{Z}^2 \right).
\end{align}

\begin{remark}
The key difficulty for tests of the hypothesis $\gamma = \gamma_0$ using either $t_{\gamma}$ or $W \left( \gamma_0, c_0 \right)$ is that the limiting distributions of these statistics\footnote{The exception is if $\delta = 0$, in which case $t_{\gamma}$ has a standard normal distribution for all values of $c$, as well as for $\alpha$ fixed, that is, $| \alpha | < 1$.} depend on the local-to-unity parameter $c$. Although $\alpha$ is constently estimable, $c$ is not, so the asymptotic inference in this case cannot in general rely on simply substituting a suitable estimator $\hat{c}$ for $c$ when selecting critical values for tests of $\gamma$.
\end{remark}

\end{example}

\newpage

\subsubsection{Optimal Inference in Predictive Regressions with Persistent Regressors}

Following \cite{jansson2006optimal} consider the case where $\left\{ ( y_t, x_t^{\prime} \right\}$ is generated by the predictive regression model with the local to unity specification for the autocorrelation matrix. 

Let $\hat{\Omega}$ be a consistent estimator of 
\begin{align}
\Omega = 
\begin{bmatrix}
\omega_{yy} & \omega_{yx} \\
\omega_{xy} & \omega_{xx}
\end{bmatrix}
= 
\underset{ T \to \infty }{ \text{lim} }
T^{-1} 
\sum_{t=1}^T \sum_{s=1}^T \mathbb{E} \left[ \begin{pmatrix}  \epsilon_{yt} \\ \Psi (L) \epsilon_{xt} \end{pmatrix} \begin{pmatrix}  \epsilon_{ys} \\ \Psi (L) \epsilon_{xs} \end{pmatrix}^{\prime}    \right],
\end{align}
which is the long-run variance of $\left( \epsilon_{yt}, \Psi(L) \epsilon_{xt} \right)^{\prime}$. 
\begin{example}
Consider the general I(1) vector process examined by \cite{phillips1986multiple} 
\begin{align}
x_t = x_{t-1} + v_t , \ \ \text{where} \ \ x_0 = 0, 
\end{align}
and $v_t$ is a weakly stationary stochastic process with unconditional variance $\mathbb{E} \left( v_t v_t^{\prime} \right) = G$ and long run covariance matrix $\Omega = G + \Lambda + \Lambda^{\prime}$. Notice that the matrix $\Lambda$ has the following representation
\begin{align}
\Lambda = \sum_{j=1}^{\infty} \mathbb{E} \left( v_0 v_j^{\prime} \right)
\end{align}
By considering the case of non-IID errors $x_T / \sqrt{T}$ converges to the vector Brownian motion BM$( \Omega )$:  
\begin{align}
T^{-1 / 2} \sum_{ t = 1 }^{[Tr]} v_t \Rightarrow B(r). 
\end{align}
Then, we obtain that 
\begin{align}
T^{-1} \sum_{t=1}^T x_{t-1} v_t^{\prime} &\Rightarrow \int_0^1 B(r) dB(r)^{\prime} + \Lambda, 
\\
T^{-1} \sum_{t=1}^T x_{t} v_t^{\prime} &\Rightarrow \int_0^1 B(r) dB(r)^{\prime} + G + \Lambda.
\end{align}
\end{example}

\begin{example}
(Cointegration in Systems of Equations) Consider the following DGP 
\begin{align}
y_t &= \beta z_t + v_t \\
v_t &= \rho v_{t-1} + \epsilon_{1t} \\
z_t &= z_{t-1} + \epsilon_{2t}
\end{align}
and 
\begin{align}
\begin{pmatrix}
\epsilon_{1t} \\
\epsilon_{2t}
\end{pmatrix}
\sim \mathcal{N} 
\left(  
\begin{pmatrix}
0 \\
0
\end{pmatrix}, 
\begin{pmatrix}
\sigma_1^2 & \theta \sigma_1 \sigma_2  \\
\theta \sigma_1 \sigma_2 & \sigma_2^2
\end{pmatrix}
\right)
\end{align}

\newpage

Moreover, the system can be represented using the usual ECM form
\begin{align}
\begin{bmatrix}
\Delta y_t \\
\Delta z_t
\end{bmatrix}
=
\begin{bmatrix}
\rho - 1 & \beta (1 - \rho) \\
0        & 0
\end{bmatrix}  
\begin{bmatrix}
y_{t-1} \\
z_{t-1}
\end{bmatrix}
+
\begin{bmatrix}
u_{1t} \\
u_{2t}
\end{bmatrix}
\end{align}
where $u_{1t} = \beta \epsilon_{2t} + \epsilon_{1t}, u_{2t} = \epsilon_{2t}$, and $\mathbb{E} \left( u u^{\top} \right) = \Lambda$, with 
\begin{align}
\Lambda = 
\begin{bmatrix}
\beta^2 \sigma_2^2 + \sigma_1^2 + 2 \beta \theta \sigma_1 \sigma_2 & \beta \sigma_2^2 + \theta  \sigma_1 \sigma_2 \\
\beta^2 \sigma_2^2 + \theta \sigma_1 \sigma_2       & \sigma_2^2
\end{bmatrix}  
\end{align}

\end{example}

Therefore, within our setting we aim to compare the efficiency of the likelihood ratio test and the Wald type statistics as structural break detectors for the predictive regression model with persistent regressors at an unknown break-point. In order to study the asymptotic behaviour of the individual components we consider limit theorems involving invariance principles for partial sums within the moderate deviations framework developed by \cite{phillips2007limit}. 

Notice that the stochastic process $x_{ \floor{ns}, n} := x_{ \floor{ns} } / d_n$ on the Skorohod space $\mathcal{D}[0,1]$ converges weakly to a Gaussian process $G(s)$.  Furthermore, on a suitably expanded probability space there exists a process $\left( x_{ t, n}^{0}, 1 \leq t \leq n \right) =_d \left( x_{ t, n}, 1 \leq t \leq n \right)$ such that 
\begin{align*}
\underset{ 0 \leq s \leq 1 }{ \text{sup} } \left| x_{ \floor{ns}, n}^0 - G(s) \right| = o_p(1)
\end{align*}

\subsubsection{Moderate Deviations from a Unit Root in Autoregressive Processes}

\begin{example}
Consider the autoregressive time series model (see, \cite{phillips2005limit})
\begin{align}
x_t = \rho_n x_{t-1} + \epsilon_t , \ \ t \in \left\{ 1,..., n \right\}
\end{align}
such that the autocorrelation coefficient is defined as
\begin{align}
\rho_n = \left( 1 + \frac{c}{n^{ \gamma } } \right), \ \ \gamma \in (0,1)
\end{align}
Consider the component random elements $x_{[ n^{\gamma} \bullet ]}$ of the Skorohod space. Furthermore, denote with 
\begin{align}
W_{ n^{\gamma} } (\bullet) := \frac{1}{ n^{\gamma / 2} } \sum_{ j = 1}^{[n^{\gamma} \bullet ]} \epsilon_j
\end{align}
\end{example}
It is possible to approximate the partial sum process on $\mathcal{D} [ 0, + \infty )$ for $x_{[ n^{\gamma} \bullet ]}$ by the Stieltjes integral
\begin{align}
U_{ n^{\gamma} } (\bullet) := \int_0^{ \bullet} e^{ c( \bullet - r ) } dW_{ n^{ \gamma } } (r) = \frac{1}{ n^{ \gamma / 2} } \sum_{ i = 1 }^{[n^{ \gamma } \bullet ]} e^{ \frac{c}{ n^{ \gamma } } \left( n^{ \gamma } \bullet - j \right)} \epsilon_j.
\end{align}

\newpage

Therefore for each $\gamma \in (0,1)$ and $c < 0$, we have that 
\begin{align}
\underset{ t \in [0, n^{1 - \gamma} ] }{ \text{sup}  } \left| \frac{1}{n^{ \gamma / 2}} x_{[ n^{\gamma} t]} - U_{ n^{\gamma} } (t)  \right| = o_p(1) \ \ \ \text{as} \ \ \ n \to \infty.
\end{align}
\begin{remark}
Therefore, by considering the above invariance law we can operate in the familiar framework of \cite{phillips1987towards} where $U_{ n^{\gamma} } (t)$, and hence the time series $x_n$ with appropriate normalization, converges to the linear diffusion $\int_0^t e^{ c(t-s)} dW(s)$, where $W$ is Brownian motion with variance $\sigma^2$. Theorem 4.6 in the paper of PM provides a bridge between stationary and local to unity autoregressions with weakly dependent innovation errors. Furthermore, when the innovation error sequence is a linear process, then the least squares estimator has been found to satisfy a Gaussian limit theory with an asymptotic bias. 
\end{remark}
\color{black}

Notice that with the seminal paper of \cite{phillips2007limit}, the authors introduce limit theorems and invariance principles for moderate deviations from unit root. The particular framework is employed in time series models such as autoregressive and predictive regression models with predictors assumed to be generated with the LUR specification. A strong approximation over the interval $[ 0, n^{1 - \alpha} ]$ for the partial sum process of \textit{i.i.d} errors can be constructed via an expanded probability space with a Brownian motion $W(.)$ with variance $\sigma^2$ for which
\begin{align}
\underset{ t \in  [ 0, n^{1 - \alpha} ] }{ \text{sup} } \left| W_{n^{\alpha} } - W(t) \right| = o_{a.s} \left( \frac{1}{ n^{ \frac{\alpha}{2} - \frac{1}{\nu} } } \right)
\end{align}

Therefore, we have that one of the main results of the paper which provide a uniform approximation of $n^{- \alpha / 2} y_{ \floor{ n^{\alpha} . } }$ by $V_{ n^{\alpha } }$ on $[ 0, n^{1 - \alpha} ]$. For each $\alpha \in (0,1)$ and $c < 0$ we have that 
\begin{align}
\underset{ t \in [0, n^{1 - \alpha }] }{ \text{sup} } \ \left| \frac{1}{n^{\alpha / 2}}  y_{ \floor{ n^{\alpha} t } }  - V_{ n^{\alpha } } (t) \right| = o_p(1), \ \ \text{as} \ n \to \infty
\end{align} 
The importance of the above result lies in the fact that an embedding of the random element $V_{ n^{\alpha } } (t)$ to the linear diffusion $J_c(t) :=  \displaystyle \int_0^t e^{ c(t-r) } dB(r)$ is possible. Using integration by parts it can be shown that,
\begin{align}
\underset{ t \in [0, n^{1 - \alpha }] }{ \text{sup} } \ \left|   V_{ n^{\alpha } } (t) - J_c(t) \right| \leq 2 \underset{ t \in [0, n^{1 - \alpha }] }{ \text{sup} } \ \left| B_{ n^{\alpha } } (t) - B(t) \right|. 
\end{align}
Moreover, by Lemma 3.1 we obtain the following expression 
\begin{align}
\underset{ t \in [0, n^{1 - \alpha }] }{ \text{sup} } \  \left| V_{ n^{\alpha } } (t) - J_c(t) \right| = \mathcal{O} \left( \frac{1}{ n^{ \frac{\alpha}{2} - \frac{1}{\nu} } } \right), \ \ \text{as} \ \ n \to \infty
\end{align} 
Therefore, we obtain that 
\begin{align}
\underset{ t \in [0, n^{1 - \alpha }] }{ \text{sup} } \ \left|   \frac{1}{ n^{\alpha / 2} } y_{ \floor{ n^{\alpha} t } }  - J_c(t)  \right| = \mathcal{O} \left( \frac{1}{ n^{ \frac{\alpha}{2} - \frac{1}{\nu} } } \right), \ \ \text{as} \ \ n \to \infty.
\end{align}

\newpage

We note that the limit theory is established through a combination of a functional law to a diffusion and a central limit law to a Gaussian random variable. A more immediate consequence is the limit law of the random element $y_{ \floor{ n^{\alpha} . } }$ on the original probability space. For all $j = 0,..., \floor{ n^{1 - \alpha} } - 1$ and $p \in [0,1]$ we obtain $\frac{1}{ n^{ \alpha / 2} } y_{ \floor{ n^{\alpha } j} + \floor{ n^{\alpha } p} } $.

\begin{lemma}
(Lemma 5.1 in \cite{phillips2005limit}) For each $\alpha \in (0,1)$ and $c > 0$, we have that 
\begin{enumerate}
\item[(a)]
\begin{align}
\underset{ t \in [ 0, n^{1- \alpha} ]  }{ \text{sup} } \left| \int_0^t \rho_n^{ - n^{\alpha} s } dB_{ n^{\alpha} } (s) - \int_0^t e^{ -cs } dB(s)  \right| = o_p \left( \frac{1}{ n^{ \frac{\alpha}{2} - \frac{1}{\nu} } } \right)
\end{align}

\item[(b)]
\begin{align}
\underset{ t \in [ 0, n^{1- \alpha} ]  }{ \text{sup} } \left| \int_0^t \rho_n^{ - \left( \floor{ n^{\alpha} t } - \floor{ n^{\alpha} s }  \right)  } dB_{ n^{\alpha} } (s) - J_{-c} (t) \right| = o_p \left( \frac{1}{ n^{ \frac{\alpha}{2} - \frac{1}{\nu} } } \right)
\end{align}
as $n \to \infty$, on the same probability space. 
\end{enumerate}
\end{lemma}

We analyze each of the two terms separately.  The term containing the block components can be written:
\begin{align*}
U_{1n} 
&= 
\rho_n^{ - 2 \kappa_n } \sum_{ j = 0 }^{ \floor{n^{1 - \alpha} } - 1 } \frac{1}{ n^{ 2 \alpha } } \sum_{ k = 1 }^{ \floor{ n^{\alpha} } } y^2_{ \floor{ n^{\alpha} j } + k }
\\
&=
\rho_n^{ - 2 \kappa_n }  \int_0^{ \floor{ n^{1 - \alpha} } } \left( \int_0^r \rho_n^{ \floor{ n^{\alpha} r } - n^{ \alpha } s  } dB_{ n^{ \alpha } } (s) \right)^2 dr + o_p(1). 
\end{align*}
Taking the inner integral along $[0, r] = [ 0, \floor{ n^{1 - \alpha} } ]$ we have, up to $o_p(1)$ that
\begin{align}
U_{1n} 
&= 
\left( \int_0^{ \floor{ n^{1 - \alpha} } }  \rho_n^{ - n^{\alpha} s} dB_{ n^{\alpha} } (s) \right)^2 \rho_n^{ - 2 \kappa_n } \int_0^{ \floor{ n^{1 - \alpha} } } \rho_n^{ 2 \floor{ n^{\alpha} r } } dr + R_n, 
\end{align}
where the remainder term $R_n$ is shown in the Appendix to be $o_p(1)$.

The second integral on the right side can be evaluated directly to obtain
\begin{align}
\int_0^{ \floor{ n^{1 - \alpha} } }  \rho_n^{ 2 \floor{ n^{\alpha} r } } dr
=
\frac{ \rho_n^{ 2 \kappa_n } }{ 2 c } \left[ 1 + O(1) \right], \ \ \text{as} \ \ n \to \infty.
\end{align}
Furthermore, we obtain that 
\begin{align*}
U_{1n} 
= 
\frac{1}{2c} \left( \int_0^{ \floor{ n^{1 - \alpha} } }  e^{  - cs } dB(s) \right)^2  + o_p \left( \frac{1}{ n^{ \frac{\alpha}{2} - \frac{1}{\nu} } } \right) 
\\
= 
\frac{1}{2c} \left( \int_0^{ \infty }  e^{  - cs } dB(s) \right)^2  + o_p \left( \frac{1}{ n^{ \frac{\alpha}{2} - \frac{1}{\nu} } } \right) 
\end{align*}

\newpage

Then, $\Psi_n^{\prime} \to_{\text{a.s}} \Psi = \sum_{ j=1 }^{ \infty } \rho^{-j} u_j^{ \prime }$, and it follows by the Shorokhod representation theorem that $\Psi_n \to_d \Psi$. Then, by joint weak convergence of $\Psi_n$ and $Z_n$ it follows that $\left( Z_n, \Psi_n \right) \Rightarrow \left( Z, \Psi \right)$ as $n \to \infty$, with $Z =_d \Psi$. The limiting random variables $\Psi$ and $Z$ can be shown to be independent by modifying Anderson's argument adjusted for weakly dependent errors. The idea is that, as $n \to \infty$, $Z_n$ can be approximated by the first $\floor{ L_n }$ elements of the sum $\sum_{ j = 1}^n \rho^{ -j } u_j$ whereas $\Psi_n$ can be approximated by the last $\floor{ L_n }$ elements of the sum $\sum_{ j = 1}^n \rho^{ - (n-j) -1} u_j$, where $\left( L_n \right)_{ n \in \mathbb{N} }$ is a sequence increasing to $\infty$ with $L_n \leq n / 3$ for each $n$. 

Accordingly, we define 
\begin{align}
Z_n^{*} := \sum_{ j = 1}^{ \floor{L_n } } \rho^{-j} u_j \ \ \text{and} \ \ \Psi_n^{*} := \sum_{ j = n- \floor{L_n } + 1 }^{ n }   \rho^{-( n - j )} u_j = \sum_{ k = 1}^{ \floor{ L_n } - 1 } \rho^{-k} u_{ n - k + 1}. 
\end{align}
We may further approximate 
\begin{align*}
\Psi_n^{*} 
= 
\sum_{ k = 1 }^{ \floor{L_n} - 1 } \rho^{-k} \sum_{ s = 0 }^{ \infty } c_s \epsilon_{ n - k + 1 - s }
&=
\sum_{ k = 1 }^{ \floor{L_n} - 1 } \rho^{-k} \sum_{ s = 0 }^{ \floor{L_n} } c_s \epsilon_{ n - k + 1 - s } + \sum_{ k = 1 }^{ \floor{L_n} - 1 } \rho^{-k} \sum_{ s = \floor{L_n} + 1 }^{ \infty } c_s \epsilon_{ n - k + 1 - s }
\\
&=
\Psi_n^{**} + \sum_{ k = 1 }^{ \floor{L_n} - 1 } \rho^{-k} \sum_{ s = \floor{L_n} + 1 }^{ \infty } c_s \epsilon_{ n - k + 1 - s },  
\end{align*}
We have that, $\Psi_n - \Psi_n^{*} = \sum_{ k = \floor{ L_n } }^n \rho^{ -k } u_{ n - k + 1}$, and so
\begin{align}
\mathbb{E} \left| \Psi_n - \Psi_n^{*} \right| \leq \mathbb{E} \left| u_1 \right| \sum_{ k = \floor{ L_n } + 1 }^n | \rho |^{-k} = \mathcal{O} \left( | \rho | \right)^{ - L_n }
\end{align}   
Similar to the sample variance, the asymptotic behaviour of the sample covariance is partly determined by elements of the time series $y_{t-1} u_t$ that do not belong to the block components 
\begin{align}
\bigg\{ y_{ \floor{ n^{\alpha} j } + k - 1 } u_{ \floor{ n^{\alpha} j } + k  }: j = 0,..., \floor{ n^{1 - \alpha} } - 1, k = 1,..., \floor{ n^{\alpha } } \bigg\}
\end{align}
Obtaining limits for the block components and the remaining time series separately in a method similar to that used for the sample variance will work. For the limit as $\alpha \to 1$, we have that $n^{1 - \alpha} \to 1$, and so $\floor{ n^{1 - \alpha }} = 1$ for $\alpha = 1$, in which case $j = 0$ in the aforementioned blocking scheme. Therefore, the invariance principle proposed by \cite{phillips1987time} $n^{ - 1 / 2} y_{ \floor{np} } \Rightarrow J_c(p)$ on $\mathcal{D} \left( [0,1] \right)$ which yield the usual local to unity limit result
\begin{align}
n \left( \hat{\rho} - \rho \right) \Rightarrow \frac{ \displaystyle \int_0^1 J_c(r) dB(r) }{ \displaystyle \int_0^1 J_c(r)^2 dr }. 
\end{align} 
In other words continuity in the limit theory cannot be achieved at the boundary with the conventional local to unity asympotics, at least without using the blocking construction.

\newpage

\paragraph{Supplementary Results: Main Proofs}

\begin{proposition}
For each $x \in [0,M], M > 0$, possibly depending on $n$, and real- valued, measurable function $f$ on $[0, \infty )$
\begin{align}
\frac{1}{ n^{\alpha / 2} } \sum_{ i = 1 }^{ \floor{ n^{\alpha} x }  } f \left( \frac{i}{ n^{\alpha} } \right) u_i = \int_0^x f(r) dB_{ n^{\alpha} } (r). 
\end{align} 
\end{proposition}

An immediate consequence of the above Proposition is the following identity. For each $x \in [ 0, n^{1 - \alpha} ]$ and $m \in \mathbb{N}$
\begin{align}
\frac{1}{ n^{\alpha / 2} } \sum_{ i = 1 }^{ \floor{ n^{\alpha} x } }  f \left( \frac{i}{ n^{\alpha} } \right) u_{ i + m } 
=
\int_0^x f(r)  dB_{ n^{\alpha} } \left( r + \frac{m}{ n^{\alpha} }  \right). 
\end{align}

\begin{proposition}
For each $\alpha \in (0,1)$ we have that 
\begin{align}
\underset{ 0 \leq t \leq n }{ \text{max} } \left|  \frac{ \tilde{ \epsilon }_t }{ n^{\alpha / 2} }   \right| = o_p(1), \ \ \text{as} \ \ n \to \infty.
\end{align}
\end{proposition}

\begin{proof}
The arguments follows Phillips (1999). Summability conditions of $\sum_{j=1}^{ \infty } j | c_j |$ ensures that $\tilde{\epsilon}_t = \sum_{j=1}^{ \infty } \tilde{c}_j \epsilon_{t-j}$ converges absolutely almost surely. Therefore, Fatou's lemma and the Minkowski inequality give the following 
\begin{align*}
\mathbb{E} | \tilde{\epsilon}_t |^{\nu} 
&\leq 
\underset{ N \to \infty }{ \text{ lim inf} } \ \mathbb{E} \left| \sum_{j=0}^N \tilde{c}_j \epsilon_{t-j} \right| \leq \underset{ N \to \infty }{ \text{ lim inf} } \left[  \sum_{j=0}^N \big(  \mathbb{E} \left| \tilde{c}_j \epsilon_{t-j} \right|^{\nu} \big)^{ \frac{1}{ \nu} } \right]^{\nu}
\\
&= 
\mathbb{E} | \epsilon_0 |^{\nu} \underset{ N \to \infty }{ \text{ lim inf} } \left( \sum_{j=0}^N \left| \tilde{c}_j \right|     \right)^{ \nu } = \mathbb{E} | \epsilon_0 |^{\nu} C_2^{\nu}  
\end{align*}
where $C_2 = \sum_{ j = 0 }^{ \infty } \left| \tilde{c}_j \right| < \infty$.

Thus, for any $\delta > 0$ the Markov inequality gives
\begin{align*}
\mathbb{P} \left( \underset{ 0 \leq t \leq n }{ \text{max} } | \tilde{\epsilon}_t | > \delta n^{\alpha / 2} \right) 
\leq 
\sum_{ t = 0}^n \mathbb{P} \left( | \tilde{ \epsilon }_t | > \delta n^{ \alpha / 2 } \right)  
\leq
\sum_{ t = 0}^n \frac{ \mathbb{E} | \tilde{ \epsilon }_t |^{ \nu } }{ \delta^{\nu} n^{\nu \alpha / 2 } }
\\
\leq \frac{ \mathbb{E} | \tilde{ \epsilon }_0 |^{ \nu } C_2^{\nu} }{  \delta^{\nu} } \frac{  n + 1 }{ n^{ \nu \alpha / 2} } = o(1). 
\end{align*} 
if and only if $\frac{ \nu \alpha }{ 2 } > 1$. 
\end{proof}

\newpage

\begin{proposition}
The following two results hold
\begin{enumerate}
\item[(a)] Let $y_{nt}^{*} := \sum_{ i = 1}^n$. Then, for each $\alpha \in ( 0, \frac{1}{2} ]$
\begin{align}
\frac{ 1 }{ n^{ \frac{1 + 3 \alpha}{ 2 } } } \sum_{ t = 1 }^n y_{t-1} \widetilde{ \epsilon }_t
=
\frac{ 1 }{ n^{ \frac{1 + 3 \alpha}{ 2 } } } \sum_{ t = 1 }^n y^{*}_{nt-1} \widetilde{ \epsilon }_t + o_p(1), \ \ \text{as} \ n \to \infty.
\end{align}

\item[(b)] Let $\gamma_m (h) = \mathbb{E} \left[ \tilde{\epsilon}_t u_{t-h} \right] = \sigma^2 \sum_{ j = 0 }^{\infty} c_j \tilde{c}_{j + h}$ for $h \geq 0$ and $m_n = \sum_{ i = 1}^{ \infty } \rho_n^{i-1} \gamma_m (i)$.  
Then it holds that, 
\begin{align}
\underset{ n \to \infty }{ \text{lim} } m_n = \sum_{i=1}^{\infty} \gamma_m(i).
\end{align}
\end{enumerate}
\end{proposition}

\begin{proposition}
For part (a), we can write the following expression 
\begin{align*}
\frac{ 1 }{ n^{ \frac{1 + 3 \alpha}{ 2 } } } \sum_{ t = 1 }^n  \left( y^{*}_{nt-1} - y_{t-1} \right) \tilde{\epsilon}_t 
&= 
\frac{ 1 }{ n^{ \frac{1 + 3 \alpha}{ 2 } } } \sum_{ t = 1 }^n \left[ \left( \sum_{ t = 1 }^n  \rho_n^i u_{t-i-1} - y_0 \rho_n^t \right) \right] \tilde{\epsilon}_t 
\\
&=
\frac{ 1 }{ n^{ \frac{1 + 3 \alpha}{ 2 } } } \sum_{ t = 1 }^n \sum_{ t = 1 }^n  \rho_n^i u_{t-i-1} \tilde{\epsilon}_t + o_p \left( \frac{ 1 }{ n^{ \frac{1 + 3 \alpha}{ 2 } } } \right),
\end{align*}
\end{proposition}

Since, by Proposition A3 and the fact that $\sum_{t=1}^n \left| \rho_n \right|^t = \mathcal{O} \left( n^{\alpha} \right)$, 
\begin{align}
\left| \frac{ 1 }{ n^{ \frac{1 + 3 \alpha}{ 2 } } } \sum_{ t = 1 }^n  y_0 \rho_n^t \tilde{\epsilon}_t  \right| \leq \left| \frac{y_0}{ n^{\alpha / 2} } \right| \underset{ 1 \leq t \leq n }{ \text{max} } \left| \frac{ \tilde{\epsilon}_t }{ n^{\alpha / 2} } \right| \frac{ 1 }{ n^{ \frac{1 +  \alpha}{ 2 } } }  \sum_{ t = 1 }^n  \left| \rho_n \right|^t = o_p \left( \frac{ 1 }{ n^{ \frac{1 + 3 \alpha}{ 2 } } } \right). 
\end{align}

\paragraph{Proof of Lemma 4.1}

Consider the stochastic process
\begin{align}
\xi_t := \sum_{ j = 0}^{ \infty } \sum_{ j = 0}^{ \infty } c_j \tilde{c}_j \epsilon_{t-j} \epsilon_{t-i}  
\end{align}
is a stationary process with autocovariance function given by 
\begin{align}
\gamma_{ \xi } (h) = \sigma^4 \sum_{ j = h }^{ \infty } \sum_{ i = j + 1 }^{ \infty } c_{j-h} \tilde{c}_{ i - h } \tilde{c}_i, \ \ h \in \mathbb{Z}. 
\end{align}

Then by Theorem of Inrangimov and Linnik (1971)
\begin{align}
\mathbb{E} \left[ \left( \sum_{t=1}^n \xi_t \right)^2 \right] \leq n \sum_{ h = - \infty }^{ \infty } \left| \gamma_{\xi} (h) \right|, 
\end{align}

\newpage

Therefore, 
\begin{align}
\mathbb{E} \left[ \left( \frac{1}{ n^{ \frac{1 + \alpha}{2} } }   \sum_{t=1}^n \xi_t \right)^2 \right] \leq \frac{1}{ n^{ \alpha } } \sum_{ h = - \infty }^{ \infty } \left| \gamma_{\xi} (h) \right| < \frac{2}{ n^{\alpha} } \sum_{ h = 0 }^{ \infty } \left| \gamma_{\xi} (h) \right| = \mathcal{O} \left( \frac{1}{ n^{ \alpha } } \right) 
\end{align} 
provided that $\sum_{ h = 0}^{ \infty } \left| \gamma_{\xi} (h) \right| < \infty$. 

To show summability of the covariance function $\gamma_{ \xi }$, write
\begin{align*}
\sum_{ h = 0 }^{ \infty } \left| \gamma_{\xi} (h) \right| 
&\leq 
\sigma^4 \sum_{ h = 0 }^{ \infty } \sum_{ j = h }^{ \infty } \left| c_{j-h} c_j \right| \left| \sum_{ i = j + 1 }^{ \infty } \tilde{c}_{i-h} \tilde{c}_i \right|
\\
&\leq  
\sigma^4 \sum_{ h = 0 }^{ \infty } \sum_{ j = h }^{ \infty } \left| c_{j-h} c_j \right| \left( \sum_{ i = j + 1 }^{ \infty } \tilde{c}_{i-h}^2 \right)^{1 / 2}  \left( \sum_{ i = j + 1 }^{ \infty } \tilde{c}_{i}^2 \right)^{1 / 2}
\\
&\leq
\sigma^4 \left( \sum_{ i = 0 }^{ \infty } \tilde{c}_{i}^2 \right) \sum_{ h = 0 }^{ \infty }  \sum_{ j = h }^{ \infty } \left| c_{j-h} \right| \left| c_{j} \right|  
\\
&= C_4 \sum_{ i = 0 }^{ \infty } \left| c_{j} \right| \sum_{ j = h }^{ \infty } \left| c_{h} \right| \leq \sigma^4 C_4 C_1^2 < \infty. 
\end{align*}

Once the fitted stationary process is obtained, the variance $\Sigma$ of $\left\{ w_t \right\}$ can be consistently estimated as
\begin{align}
\hat{ \Sigma } = \frac{1}{n} \sum_{ t = 1}^n \hat{w}_t \hat{w}_t. 
\end{align}

\paragraph{Summary} Although it is not shown explicitly, the idea of stationary transformations can be used also for the essentially singular model to construct the regression which yields asympotically invariant, with respect to nuisance parameters, tests.

\begin{lemma}
Suppose that $X_{ni} = \left[ 1 + \frac{z_n}{n} \right] X_{n,i-1} + v_{ni}$, for $i = 1,...,n, X_{n0} = 0$, where $z_n \to z$ and $v_{ni} = \sum_{j = 0}^{ \infty } c_j (a_n) e_{i-j}$ with $\alpha_n \to \alpha$ and with the $c_j$'s satisfying $\sum_{ j = 0 }^{\infty} \left| c_j ( \alpha ) \right| < \infty$. Then, the stochastic process $n^{- 1/ 2} X_{ \floor{nt} }$ converges in law in the Skorokhod space to $\int_0^t \text{exp} \left[ z(t-s) \right] dW(t)$, where $W(t)$ is the two-dimensional Brownian motion with variance given by 
\begin{align}
\left( \sum_{ r = 0 }^{\infty} c_r( \alpha ) \right) \left( \sum_{ r = 0 }^{\infty} c_r( \alpha ) \right)^{\prime} 
\end{align}
is finite and non-singular. 
\end{lemma}

The use of moderate deviation principles are useful in understanding the various asymptotic theory results that apply across the parameter space of the autoregressive coefficient of the nonstationary autoregressive model. In particular, for the case of quantile autoregressions and quantile predictive regression models a relevant framework is proposed by \cite{katsouris2022asymptotic}.

\newpage

\subsubsection{Multivariate Predictive Regression Model}

Consider the multivariate predictive regression model with lag-augmentation 
\begin{align}
\label{model1}
\boldsymbol{y}_{t}   &=  \boldsymbol{A} \boldsymbol{y}_{t-1} + \boldsymbol{u}_{t}, \ \ \ \ \ \ \   
\\
\label{model2}
\boldsymbol{x}_t     &= \mathbf{R}_n \boldsymbol{x}_{t-1} + \boldsymbol{v}_{t}, \ \ \ \ \ \ \ \text{with} \ \ \ \mathbf{R}_n = \left( \boldsymbol{I}_p - \frac{ \boldsymbol{C}_p }{ n^{\alpha} }  \right)
\end{align}
where $\left( \boldsymbol{y}_{t}, \boldsymbol{x}_t \right) \in \mathbb{R}^{p \times 1 }$ are $p-$dimensional vectors for $t \in \left\{ 1,..., n \right\}$ and $\left\{ \boldsymbol{A} \right\} \in \mathbb{R}^{p \times p }$ are $(p \times p)$ matrices of coefficients and $\boldsymbol{C} = \mathsf{diag} \left\{ c_1,..., c_p \right\}$ is a $(p \times p)$ diagonal matrix with the coefficients of persistence. The degree of persistence of the regressors is determined by the value of the unknown persistence coefficients $c_i$'s which are assumed to be positive constants. Moreover, the common exponent rate for the persistence coefficients take values,  $\alpha = 1$, which covers the local-unit-root (LUR) regressors  and $\alpha \in (0,1)$ the case of mildly integrated (MI) regressors as defined by  \cite{magdalinos2009limit} and \cite{kostakis2015Robust}. When $\alpha = 0$ we assume the presence of stationary regressors, for a suitably restricted coefficient matrix $\left( \boldsymbol{I}_p - \boldsymbol{C}_p \right)$. A more recent application is given by \cite{yang2020testing}. 
\begin{assumption}
\label{Assumption1}
Let $\boldsymbol{e}_t = \left( \boldsymbol{u}^{\prime}_{t}, \boldsymbol{\epsilon}_{t}^{\prime} \right)^{\prime}$ be a $2p-$dimensional vector. The innovation sequence $\boldsymbol{e}_t$ is a conditionally homoscedastic \textit{martingale difference sequence} (m.d.s)  such that the following conditions hold:   
\begin{enumerate}
\item[A1.] $\mathbb{E} \left( \boldsymbol{e}_{t} | \mathcal{F}_{t-1} \right) = \boldsymbol{0}$, where $\mathcal{F}_t = \sigma \left( \boldsymbol{e}_t, \boldsymbol{e}_{t-1},... \right)$ is an increasing sequence of $\sigma-$fields.
\item[A2.] $\mathbb{E} \left( \boldsymbol{e}_{t} \boldsymbol{e}_{t}^{\prime} | \mathcal{F}_{t-1} \right) = \boldsymbol{\Sigma}_{ee}$, where $\boldsymbol{\Sigma}_{ee} \in \mathbb{R}^{p \times p}$ is a positive-definite covariance matrix, which has the following form: 
\begin{align*}
\boldsymbol{\Sigma}_{ee} =  
\begin{bmatrix}
\boldsymbol{\Sigma}_{uu} &  \boldsymbol{\Sigma}_{uv} \\
\boldsymbol{\Sigma}_{vu} & \boldsymbol{\Sigma}_{vv}
\end{bmatrix} > 0.
\end{align*}
such that $\boldsymbol{\Sigma}_{uu}, \boldsymbol{\Sigma}_{uv},  \boldsymbol{\Sigma}_{vv} \in \mathbb{R}^{p \times p}$ where $p$ is the number of regressors.
\item[A3.] The innovation to $x_t$ is a linear process with the representation below 
\begin{align*}
\boldsymbol{v}_t 
:= 
\boldsymbol{\Psi}( \boldsymbol{L} ) \boldsymbol{\epsilon}_t 
\equiv
\displaystyle \sum_{j=0}^{\infty} \boldsymbol{\Psi}_j \boldsymbol{\epsilon}_{t-j}, \ \ \ \boldsymbol{\epsilon}_t \sim^{i.i.d} \left(  \boldsymbol{0} , \boldsymbol{\Sigma}_{ee} \right)
\end{align*}
where $\left\{ \boldsymbol{\Psi} \right\}_{j=0}^{\infty}$ is a sequence of absolute summable constant matrices such that $\sum_{j=0}^{ \infty} \boldsymbol{\Psi}_j$ has full rank and $\boldsymbol{\Psi}_0 = \boldsymbol{I}_p$ with  $\boldsymbol{\Psi} (1) \neq 0$, allowing for the presence of serial correlation in the innovations of the predictive regression model. Suppose that(A1)-(A3) hold, the following \textit{FCLT} applies 
\end{enumerate}
\begin{align}
\label{fclt}
\frac{1}{ \sqrt{n} } \sum_{t=1}^{ \floor{ n r} } \boldsymbol{w}_t
=
\frac{1}{ \sqrt{n} } \sum_{t=1}^{ \floor{ n r} } 
\begin{bmatrix}
\boldsymbol{u}_{t} \\
\boldsymbol{v}_{t}
\end{bmatrix}
\Rightarrow \ \text{BM} \left(  \boldsymbol{\Sigma}_{w} \right) =
\begin{bmatrix}
\boldsymbol{B}_{u}(r) \\
\boldsymbol{B}_{v}(r)
\end{bmatrix}
\equiv \boldsymbol{\Sigma}_{w}^{1/2} \boldsymbol{W}(r),
\end{align} 
where $\boldsymbol{w}_t = \left( \boldsymbol{u}_t, \boldsymbol{v}_t \right)^{\prime}$
and $\left\{ \boldsymbol{W}(r) : 0 \leq r \leq 1 \right\}$ is the standard (vector) BM, in $\mathcal{D}\left( [0,1] \right)$.
\end{assumption}

\newpage

Moreover, it holds that $B(r) = \omega W(r)$ for some $r \in [0,1]$ and $\omega := \sum_{j = - \infty}^{\infty} \gamma_u(j) = C(1)^2 \sigma^2$, the one-sided covariance of $v_t$. Moreover, by definition it holds that $W_c(r) = \int_0^r e^{ c(r-s)} dW(s)$. Notice that, expression \eqref{fclt} provides an invariance principle for the partial sum process of $\boldsymbol{w}_t$, which implies that the partial sum of $\boldsymbol{w}_t$, weakly converges to the stochastic quantity $\boldsymbol{\Sigma}_{w}^{1/2} \boldsymbol{W}(r)$, that is, a $p-$dimensional Brownian process with covariance matrix $\boldsymbol{\Sigma}_{w}$. Within our context we employ the local-unit-root limit theory as proposed by the seminal paper of \cite{phillips1987time}, such that $\boldsymbol{x}_{[nr]} / \sqrt{n} \Rightarrow \boldsymbol{J}_c(r)$, where $\boldsymbol{J}_c(r)$ is a $p-$dimensional Gaussian process defined as below 
\begin{align}
\label{OUprocess}
\displaystyle \boldsymbol{J}_c(r) = \int_0^r e^{(r-s)\boldsymbol{C}_p } d \boldsymbol{W}(s), \ \ \ \ r \in (0,1).
\end{align} 
that satisfies the Black-Scholes differential equation
$d \boldsymbol{J}_c(r) \equiv \boldsymbol{C}_p \boldsymbol{K}_c(r) + d \boldsymbol{B}_v(r)$, with $\boldsymbol{K}_c(r) =0$, implying also that $\boldsymbol{K}_c(r) \equiv \sigma_v \boldsymbol{J}_c(r)$, where $\displaystyle \boldsymbol{J}_c(r) =  \int_0^r e^{(r-s)\boldsymbol{C}_p} d \boldsymbol{W}(s)$. Moreover, since in this paper we consider time series nonstationarity in the form of regressors exhibiting mildly integradeness, then we also employ the asymptotic results proposed by \cite{phillips2007limit} since the invariance principle $x_{[nr]}$ requires appropriate normalization. Notice that $\boldsymbol{K}_c(r)$ represents the Ornstein-Uhlenbeck, (OU) process, which encompasses the unit root case such that $\boldsymbol{J}_c(r) \equiv \boldsymbol{B}_v(r)$, for $\boldsymbol{C}_p = 0$. Moreover, the assumption of a local-unit-root specification for the predictors of the model such that $\boldsymbol{x}_t = \left( \boldsymbol{I}_p - \frac{ \boldsymbol{C}_p }{n^{ \alpha } } \right)$ $\boldsymbol{x}_{t-1} + \boldsymbol{v}_t$ and more specifically by allowing the autocorrelation coefficient to be of the form $\boldsymbol{\Phi}_n  = \left( \boldsymbol{I}_p - \frac{ \boldsymbol{C}_p }{ n^{ \alpha } } \right)$, permits to consider the $\boldsymbol{K}_c(r)$ Gaussian process given by \eqref{OUprocess},  as a stochastic approximation. 

\textit{Proof.} Expanding the expression $x_t z_t^{\prime}$, vectorizing and summing over $t \in (1,...,n)$ 
\begin{align*}
x_t z_t^{\prime} = R_n x_{t-1}  z_{t-1}^{\prime} R_n^{\prime}  +   R_n x_{t-1}u_{xt}^{\prime} +  u_{xt} z_{t-1}^{\prime} R_n^{\prime} + u_{xt}u_{xt}^{\prime}
\end{align*}
We have that 
\begin{align*}
\mathcal{L} 
= 
\left[  I_{K^2} - R_{zn} \otimes R_n \right] \frac{1}{n} \sum_{t=1}^n \text{vec} \left( x_{t-1} z_{t-1}^{\prime}  \right) 
&=  
\left[ I_{K} + o_p(1) \right] \text{vec}  \left[  \frac{1}{n} \sum_{t=1}^n x_{t-1} u_{xt}^{\prime} + \frac{1}{n} \sum_{t=1}^n u_{xt} z_{t-1}^{\prime} +  \frac{1}{n} \sum_{t=1}^n u_{xt} u_{xt}^{\prime} \right]
\\
&= 
\text{vec}  \left[  \frac{1}{n} \sum_{t=1}^n x_{t-1} u_{xt}^{\prime}  + \Lambda^{\prime}_{xx} + \mathbb{E} \left( u_{x1} u_{x1}^{\prime}  \right) \right] + o_p(1)
\end{align*}
Therefore, for $\delta < \gamma$, which holds since $\gamma = 1$ and $\delta \in (0,1)$, then we obtain
\begin{align}
\left[  I_{K^2} - R_{zn} \otimes R_n \right] = -\frac{1}{ n^{\delta} } \left( C_z \otimes I_K \right) \left[ I_K   +  O_p \left( \frac{1}{ n^{\gamma -\delta } } \right) \right] 
\end{align}
Note for example, that 
\begin{align*}
R_{zn} \otimes R_n = \left( I_K + \frac{C_z}{n^{\delta}} \right) \otimes \left( I_K + \frac{C}{n^{\gamma}} \right) = I_{K^2} + \frac{ C \otimes I_K }{ n^{\gamma}} + \frac{ C_z \otimes I_K }{ n^{\delta}} + \frac{ C_z \otimes C }{ n^{\gamma + \delta}} 
\end{align*}

\newpage

Thus, 
\begin{align}
- \frac{1}{ n^{ 1 + \delta} } \left( C_z \otimes I_K \right)  \sum_{t=1}^n \text{vec} \left( x_{t-1} z_{t-1} \right) &\Rightarrow \left( \int_0^1  J_C dB^{\prime}_x   + \Omega_{xx}  \right)  \\
\frac{1}{ n^{ 1 + \delta} }  \sum_{t=1}^n x_{t-1} z_{t-1}^{\prime} &\Rightarrow - \left( \int_0^1  J_C dB^{\prime}_x   + \Omega_{xx}  \right) C_z^{-1}
\end{align}

\textbf{Part (L2)}. Since, $J_c(r) = \int_0^r e^{-c(r-s)} dB_v(s)$ and using expression (14) in PM for the univariate predictive regression with single regressor we have that 
\begin{align}
\frac{ \sum_{t=1}^n z_{t-1}^2 }{ n^{ 1 + \delta} } \overset{ p }{ \to } V_{c_z} := \int_0^{\infty} \sigma_v^2 e^{-2r c_z}   dr = \frac{ \sigma^2_v }{ 2 c_z }
\end{align}
The corresponding result for the the univariate predictive regression with multiple regressors is as below
\begin{align}
\frac{1}{ T^{1 + \delta } }  \sum_{t=1}^T z_{t-1} z_{t-1}^{\prime} \overset{ p }{ \to } V_{zz} :=  \int_0^{\infty} e^{rC_z} \Omega_{xx} e^{rC_z} dr  
\end{align}
Recall that  by definition the functional $\underline{J}_c(r) = \int_0^r e^{(r-s)C} d \underline{B}(s)$, is considered to be a vector diffusion process and satisfies the stochastic differential equation system 
\begin{align}
d \underline{J}_c(r) = C \underline{J}_c(r) dr + d \underline{B}(r), \ \ \underline{J}_c(0) = 0
\end{align}
Therefore, using the initial condition as well we can also write
\begin{align}
\underline{J}_c(r) = \underline{B}(r) + C \int_0^r e^{(r-s)C} \underline{B}(s) ds
\end{align}
Since $\underline{J}_c(r)$ is a Gaussian process and for fixed $r$ then it can be easily proved that the finite dimensional distribution (see e.g., \cite{phillips1988regression})
\begin{align}
\underline{J}_c(r) \Rightarrow N(0, Q) \ \ \text{where} \ Q = \int_0^r e^{(r-s)C} \Omega e^{(r-s)C^{\prime}} ds
\end{align}
where $\Omega$ the $p \times p$ covariance matrix of $\underline{B}(r)$ the $p-$dimensional Brownian motion. 

\textbf{Part (L3)}. The weakly convergence result to a mixed Gaussian distribution below
\begin{align}
\frac{1}{ T^{ \frac{ 1 + \delta }{ 2 } } } \sum_{t=1}^T \left( z_{t-1} \otimes u_t  \right)  
\Rightarrow N \left( 0 , V_{zz} \otimes \Sigma_{ uu } \right)   
\end{align}
shows that the limit distribution of $T^{ - (1 + \delta )/2} \sum_{t=1}^T \left( z_{t-1} \otimes u_t  \right)$ is Gaussian with mean zero and covariance matrix equal to the probability limit of $T^{ - (1 + \delta )/2} \sum_{t=1}^T \left( z_{t-1} \otimes u_t  \right)$, which is equal to $V_{zz} \otimes \Sigma_{ uu }$, where $V_{zz} := \int_0^{\infty} e^{r C_z} \Omega_{vv} e^{r C_z} dr$. Similarly the limit distribution of $T^{ - (1 + \delta )/2} \sum_{t=1}^T \left( x_{t-1} \otimes u_t  \right)  
\Rightarrow N \left( 0 , V_{xx} \otimes \Sigma_{ uu } \right)$, where $V_{xx} := \int_0^{\infty} e^{r C} \Omega_{vv} e^{r C} dr$, is proved in 
Lemma 3.3 of \cite{magdalinos2009limit}.

\newpage

We prove the weakly convergence result given above for the univariate predictive regression (applied to both the case of single versus multiple regressors). See also case 2 of Lemma B3 in Appendix of KMS, where the result for the multivariate predictive regression is proved. We consider the partial sum  process $S_t = \sum_{j = 1}^t z_j$ where $z_j$ is the IVX instrument. Therefore, we obtain $\Delta S_{t-1} \equiv S_{t-1} - S_{t-2} = z_{t-1}$, i.e.,
\begin{align*}
\frac{1}{ T^{ 1 + \delta} } \sum_{t=1}^T x_{t-1} z_{t-1} &= \frac{1}{ T^{ 1 + \delta} } \sum_{t=1}^T x_{t-1} \Delta S_{t-1} = \frac{1}{ T^{ 1 + \delta} } \sum_{t=1}^T x_{t-1} \left(  S_{t-1} - S_{t-2}   \right) \\
&= \frac{1}{ T^{ 1 + \delta} } \sum_{t=1}^T S_{t-1}  x_{t-1} - \frac{1}{ T^{ 1 + \delta} } \sum_{t=1}^T S_{t-2} x_{t-1} \\
&= \frac{1}{ T^{ 1 + \delta} } \left\{ \left[ \sum_{t=1}^T S_{t-1}  x_{t-1} - \sum_{t=1}^T S_{t-2}  x_{t-2} \right] -  \sum_{t=1}^T S_{t-2} \Delta x_{t-1}    \right\} \\
&= \frac{1}{ T^{ 1 + \delta} } \left(  S_{T-1} x_{T-1} - S_0 x_0  \right) - \frac{1}{ T^{ 1 + \delta} } \sum_{t=1}^T S_{t-2} \Delta x_{t-1} 
\end{align*}
The first term of the above expression is considered to be asymptotically negligible, i.e., $o_p(1)$. Thus, we consider the asymptotic behaviour of the second term by expanding the expression further. We have that $\displaystyle \Delta x_t =  \frac{C}{T} x_{t-1} + v_t$, implying  
\begin{align}
\frac{1}{ T^{ 1 + \delta} } \sum_{t=1}^T S_{t-2} \Delta x_{t-1} = \frac{1}{ T^{ 1 + \delta} } \sum_{t=1}^T  S_{t-2} v_{t-1}  + \frac{C}{ T^{ 2 + \delta} } \sum_{t=1}^T  S_{t-2} x_{t-2}
\end{align}
Therefore, the above gives the following
\begin{align*}
\frac{1}{ T^{ 1 + \delta} } \sum_{t=1}^T x_{t-1} z_{t-1} 
&= 
- \frac{1}{ T^{ 1 + \delta} } \left\{  \sum_{t=1}^T  S_{t-2} v_{t-1} + \frac{C}{T} \sum_{t=1}^T  S_{t-2} x_{t-2}   \right\}
=  
- \frac{1}{ T^{ 1 + \delta} } \left\{ \sum_{t=1}^T  S_{t-2} \left( v_{t-1}  +  \frac{C}{T} x_{t-2} \right)  \right\} \\
&=  
- \frac{1}{ T^{ 1 + \delta} } \left\{ \sum_{t=1}^T  \sum_{j = 1}^{t-2} z_j \left( v_{t-1}  +  \frac{C}{T}  x_{t-2} \right)  \right\}
\end{align*}
We prove the limiting behaviour of the above terms using the convergence of the martingale expression (16) in PM and the fact that $x_t$ is a local to unity process.

Then, the following weakly convergence of  \eqref{term1} and \eqref{term2} follows
\begin{align}
\label{term1}
\frac{1}{ T^{ 1 + \delta } }  \sum_{t=1}^T  \sum_{j = 1}^{t-2} z_j v_{t-1} 
& \Rightarrow \int_0^1  \underline{J}_c(r) d B_v C_z^{-1} 
\\
\label{term2}
\frac{C}{ n^{ 2 + \delta} }  \sum_{t=1}^n  \sum_{j = 1}^{t-2} z_j x_{t-2} 
& \Rightarrow   \Omega_{vv} C_z^{-1} 
\\
\frac{1}{ T^{ 1 + \delta} } \sum_{t=1}^T x_{t-1} z_{t-1} &\Rightarrow - \left( \int_0^1  \underline{J}_c(r) d B_v  + \Omega_{vv}  \right)   C^{-1}_z
\end{align}

\newpage

\subsubsection{Predictive Regression with Multiple Regressors}

Consider the predictive regression model such as 
\begin{align}
y_t &= \beta_0 + \boldsymbol{\beta}^{\prime}_1 \boldsymbol{x}_{t-1} + u_t, 
\\
\boldsymbol{x}_t &= \boldsymbol{R}_T \boldsymbol{x}_{t-1} + \boldsymbol{v}_t 
\end{align}
where $y_t$ is a scalar vector and $\boldsymbol{x}_{t-1}$ is a $k-$dimensional vector of regressors. Denote with $\boldsymbol{\eta}_t = \big[ u_t,  \boldsymbol{v}_t^{\prime} \big]^{\prime}$. 

Then, the predictive regression above includes a model intercept and furthermore we do not demean the random variables in order to remove the model intercept. In addition, we assume that $\boldsymbol{R}_T \neq \boldsymbol{I}_T$ which would correspond to the integrated regressor case.

\textbf{Fully Modified OLS} Under the aforementioned conditions the functional law 
\begin{align*}
\displaystyle T^{- 1 / 2} \sum_{t=1}^{ \floor{Tr} } \overset{ d }{ \to } \boldsymbol{B}(r) \equiv BM ( \Omega )
\end{align*}
holds for partial sums of $\boldsymbol{\eta}_t$. Define the partition $\boldsymbol{B} = \big( B_u, \boldsymbol{B}_v^{\prime} \big)^{\prime}$.
\begin{proposition} 
Under Assumptions it holds that 
\begin{align}
T  \left( \widehat{A}^{+} - A \right) \overset{ d }{ \to } \left( \int_0^1 d \boldsymbol{B}_{u.v} \boldsymbol{B}^{\prime}_v \right) \left( \int_0^1 \boldsymbol{B}_v \boldsymbol{B}^{\prime}_v \right)     
\end{align} 
\end{proposition}
When $\Omega_{uu.v}$ has full rank, the rate of convergence of the FM-OLS estimator is determined by the rates of weak convergence of the sample covariances and the rate of nonparametric estimation of $\Omega$ and $\Gamma^{+}$ does not apply any role. 
In particular, the OLS estimator $\widehat{A}^{+} = Y^{\prime} X \left( X^{\prime} X \right)^{-1}$ and employs corrections for endogeneity in the regressor $x_t$, leading to the transformed dependent variable 
\begin{align}
\widehat{y}_t^{+} = y_t - \widehat{\Omega}_{0x} \widehat{\Omega}_{xx}^{-1} \big( x_t - x_{t-1} \big) 
\end{align}
and a bias correction term involving 
\begin{align}
\widehat{\Delta}_{0x}^{+} = \widehat{\Delta}_{0x} - \widehat{\Omega}_{0x} \widehat{\Omega}_{xx}^{-1} \widehat{\Delta}_{xx}
\end{align}
which is constructed in the usual way using consistent nonparametric estimators of submatrices of the long run and one sided long run quantities $\Omega$ and $\Gamma^{+}$. In particular, the FM-OLS estimators removes asymptotic bias and increases efficiency by correcting both the long run serial correlation in $u_t$ and endogeneity in $x_t$ causes by the long run correlation between $u_{0t}$ and $u_{xt}$.

\newpage

\begin{example}[Mildly Explosive Autoregression]
\

Consider the martingale difference property of $\left\{ \hat{\epsilon}_{n,t} - \mathbb{E}_{ \mathcal{F}_{t-1} } \left[ \epsilon_{1t} \right] : t \geq 1 \right\}$ which yields that 
\begin{align}
\norm{ \sum_{t=1}^{ k_n } a_{n,t} \big( \epsilon_{1t} - \mathbb{E}_{ \mathcal{F}_{t-1} } \left[ \epsilon_{1t} \right] \big)^2 } \leq \Delta_n^2 \sum_{t=1}^{ k_n } a_{n,t}^2 \to 0. 
\end{align}
by some choice of $( \Delta_n )_{ n \in \mathbb{N} }$. Furthermore, the fact that $( y_t )$ is $\mathcal{F}_t-$adapted implies that,  for any fixed integer $M > 0$, 
\begin{align}
\sum_{t=1}^{ k_n } a_{n,t} y_t = \sum_{m=0}^{ M - 1} \sum_{t=1}^{k_n} a_{n,t} \big( \mathbb{E}_{ \mathcal{F}_{t-m} } \left[ y_{t} \right] - \mathbb{E}_{ \mathcal{F}_{t-m-1} } \left[ y_{t} \right] \big) + \sum_{t=1}^{k_n} a_{n,t} \mathbb{E}_{ \mathcal{F}_{t-M} } [ y_t ]      
\end{align}
Thus, for each $m$, it holds that 
\begin{align}
\epsilon_t^{ (m) } := \mathbb{E}_{ \mathcal{F}_{t-m} } [ y_t ] - \mathbb{E}_{ \mathcal{F}_{t-m-1} } [ y_t ]     
\end{align}
is a $\mathcal{F}_{t-m}-$martingale difference process that inherits the uniform integrability property from $y_t$, and so, $\epsilon_t^{(m)}$ satisfies the required conditions. Then, applying the triangle inequality we obtain the following
\begin{align}
\norm{ \sum_{t=1}^n a_{n,t} y_t }_1 &\leq \sum_{m=0}^{ M - 1 } \norm{ \sum_{t=1}^{k_n} a_{n,t} \epsilon_t^{(m)} }_1 + \sum_{ t=1 }^{ k_n } \left| a_{n,t} \right| \norm{ \mathbb{E}_{ \mathcal{F}_{t-M} } [ y_t ] }_1
\\
&\leq M \times \underset{ 0 \leq m < M }{ \mathsf{max} } \norm{ \sum_{t=1}^{k_n} a_{n,t} \epsilon_t^{(m)} }_1 + \zeta \underset{ n \in \mathbb{N} }{ \mathsf{sup} } \sum_{ t = 1}^{ k_n } \left| a_{n,t} \right| \psi_M.
\end{align}
Let $\delta > 0$ be arbitray. Since $C := \zeta \mathsf{sup}_{ n \in \mathbb{N} } \sum_{ t = 1}^{ k_n } \left| a_{n,t} \right| < \infty$ and $\psi_M \to \infty$, there exists $M_0 ( \delta ) \in \mathbb{N}$ such that $\psi_M \leq \delta / (2C)$.
\end{example}

\begin{example}
Consider the case in which all regressors, $X_t$, is a $p \times 1$ vector of $I(1)$ variables such that 
\begin{align}
X_t = X_{t-1} + \eta_t, \ \ 1 \leq t \leq n
\end{align} 
where $X_0 = \mathcal{O}_p (1)$ and $\underset{ 1 \leq t \leq n }{ \mathsf{max} } \ \mathbb{E} \left(  \norm{ \eta_t }^q \right) \leq C < \infty$ for some $q > 8$.

Furthermore, $\left\{ u_t \right\}_{ t=1}^n$ is independent of $\left\{ ( X_t, Z_t ) \right\}_{t=1}^n$ and $\left\{ \big( \left( u_t, \eta_t \right), \mathcal{F}_{n, t-1}  \big) \right\}_{ 2 \leq t \leq n}$ and forms a martingale difference sequence with $\sigma_v^2 = \mathbb{E} \left( v_t^2 \right) < \infty$, such that 
\begin{align}
\underset{ 2 \leq t \leq 2 }{ \mathsf{sup} } \left| \mathbb{E} \left(    v_t^2 | \mathcal{F}_{n,t-1} \right) - \sigma_v^2 \right| \to 0
\end{align}
Denote with 
\begin{align}
\label{convergence.result}
B_{n,\eta} (r) \equiv \frac{1}{ \sqrt{n} } \sum_{t=1}^{ \floor{nr} } \eta_t.
\end{align}

\newpage

There exists a vector Brownian motion $B_{\eta}$ such that $B_{n,\eta} (r) \Rightarrow B_{\eta}(r)$, on $D[0,1]^d$ as $n \to \infty$, where $D[0,1]^d$ is the space of cadlag functions on $[0,1]^d$ equipped with Skorohod topology, and $B_{\eta}$ is a $d-$dimensional multivariate Brownian motion with a finite positive definite covariance matrix $\Sigma_{\eta}$ such that 
\begin{align}
\Sigma_{\eta} = \underset{ n \to \infty }{ \mathsf{lim} } \text{Var} \left( n^{- 1 / 2} \sum_{t=1}^n \eta_t \right)
\end{align}

The martingale difference condition is required to employ the generalized martingale central limit theorem when we derive the limiting distribution of our proposed test statistic under the null hypothesis. Furthermore, taking the weak convergence result in  \eqref{convergence.result} as an assumption is commonly done in the econometrics literature. The conditions for the multivariate functional central limit theorem for partial sums of weakly dependent random vectors can be found \cite{wooldridge1988some}. Thus, by applying the continuous mapping theorem it follows that 
\begin{align}
\underset{ 0 \leq r \leq 1 }{ \mathsf{sup} } \left| B_{n, \eta}   \right| \overset{ d }{ \to } \underset{ 0 \leq r \leq 1 }{ \mathsf{sup} } \left| B_{\eta} \right|
\end{align}
Since it is bounded in probability, that is, $\underset{ 0 \leq r \leq 1 }{ \mathsf{sup} } \left| B_{\eta}(r) \right| = \mathcal{O}_p(1)$, hence we have that 
\begin{align}
\underset{ 1 \leq t \leq n }{ \mathsf{max} } \left| X_t \right| = \mathcal{O}_p \left( \sqrt{n} \right), 
\end{align}
\end{example}

\subsubsection{Testing for Structural Break in Predictive Regression}

In this section, we present the framework of \cite{georgiev2018testing}. The approach of the particular paper is to consider between two important test statistics which are commonly used in detecting structural change, that is, the sup Wald test of \cite{andrews1993tests} and the Cramer-von-Mises type statistics of \cite{nyblom1989testing}. On the one hand the sup Wald test proposed by \cite{andrews1993tests} is motivated by alternatives where the parameters display a small number of breaks at deterministic points in the sample, while under the alternative hypothesis of the Cramer-von-Mises test statistic the coefficients of the model are random and slowly evolve through time. Using a fixed regressor wild bootstrap procedure their framework allows for both conditional and unconditional heteroscedasticity in the data.     

Consider the time-varying predictive regression model 
\begin{align}
\label{specification1}
y_t = \alpha_t + \beta_t x_{t-1} + \epsilon_{yt}, \ \ t = 1,...,T.
\end{align}
where $\epsilon_{yt}$ is a mean zero innovation process and $x_t$ is an observed process, generated by  
\begin{align}
x_t &= \mu + s_{xt},\ \  s_{xt} = \rho_x s_{xt-1} + \epsilon_{xt}
\end{align}
where $\rho_x := \left( 1 - \frac{c_x}{T} \right)$, with $c_x \leq 0$ which allows to consider two types of persistence properties for the regressors of the model, that is, (i) strongly persistence unit root, and (ii) local-to-unit root.

\newpage

\begin{remark}
Note that \eqref{specification1} provides a time-varying representation of the predictive regression model, which allows both the model intercept and the slope coefficients to vary over time. In particular, the parameters are expressed as 
\begin{align}
\alpha_t := \alpha + a s_{\alpha t} , \ \ \ \beta_t := \beta + b s_{\beta t}
\end{align}
For instance, in the context of a time-invariant model, that is, $\alpha_t = \alpha$ and $\beta_t = \beta$. Therefore, with near-unit root predictors then the framework of Cavanagh et al (1995) applies. Moreover, when $\beta$ is fixed, then the predictive regression model can be interpreted as a co-integrating regression     because $y_t$ is a near-unit root process. The main focus of the proposed framework is on testing the null hypotheses that the intercept and the slope parameters are constant over time against the alternative that they vary over time through the sequences of associated time-varying coefficients, $s_{\alpha t}$ and $s_{\beta t}$. In particular, this can be done by testing the restrictions that $a = 0$ and $b = 0$. 
\end{remark}

\paragraph{Stochastic Coefficient Variation}

The assumption of time-varying parameters in the predictive regression model implies that the time-varying components of the model parameters, that is, $s_{\alpha t}$ and $s_{\beta t}$ follow (near-) unit root processes expressed as below
\begin{align}
\begin{bmatrix}
s_{\alpha t}
\\
s_{\beta t}
\end{bmatrix}
=
\begin{bmatrix}
\rho_{\alpha} & 0 \\
0 & \rho_{\beta} \\
\end{bmatrix}
\begin{bmatrix}
s_{\alpha t-1}
\\
s_{\beta t-1}
\end{bmatrix}
+
\begin{bmatrix}
\epsilon_{\alpha t}
\\
\epsilon_{\beta t}
\end{bmatrix}
\end{align}
where the autocorrelation coefficient is expressed as below
\begin{align}
\rho_{\alpha} := \left( 1 - \frac{ c_{\alpha } }{T} \right), \ \ \rho_{\beta} := \left( 1 - \frac{ c_{\beta } }{T} \right), \ \ \text{with} \ \ c_{\alpha} \geq 0 \ \ \text{and} \ \ c_{\beta} \geq 0.
\end{align}
which are unit-root or local-to-unit root autoregressive processes. 

The non-stochastic coefficient variation mechanism includes the testing framework  proposed by \cite{andrews1993tests}. Under the assumption of one-time breakpoint, then $s_{ \alpha t}$ and $s_{ \beta t}$ are modelled as below
\begin{align}
s_{ \alpha t} = s_{ \beta t} = D_t \left( \floor{ \pi_0 T } \right)
\end{align}
where we define $D_t \left( \floor{ \pi_0 T } \right) := \mathbf{1} \left( t \leq \floor{ \pi T } \right)$ with $\floor{ \pi T }$ denoting a generic shift point with associated break fraction $\pi$. We take the true shift fraction $\pi_0 \in \Pi$, where $\Pi = [ \pi_1, \pi_2]$ with $0 < \pi_1 < \pi_2 < 1$.

Consider the testing hypothesis given by 
\begin{align}
H_{0}: a \neq 0 \ \ \text{versus} \ \ \ b \neq 0.
\end{align}

Moreover, consider the fitted regression model 
\begin{align}
\label{reg}
y_t = \hat{\alpha} + \hat{\beta} x_{t-1} + \hat{\beta}_0 \Delta x_t + \hat{e}_t, \ \ \ t= 1,...,T
\end{align}
Let $\hat{\sigma}^2 := T^{-1} \sum_{t=1}^T \hat{e}_t^2$, with $\hat{e}_t^2$ the OLS residual from the fitted regression given by \eqref{reg}. Moreover, we denote with $\underline{X} = [ 1 \ x_{t-1} ]^{\prime}$.

Then, the LM statistic is constructed as 
\begin{align}
LM := \frac{1}{T \hat{\sigma}^2 } \sum_{j=1}^T \left( \sum_{t=1}^j \underline{X}_t \hat{e}_t \right) \left( \sum_{t=1}^T \underline{X}_t \underline{X}_t^{\prime}    \right)^{-1} \left( \sum_{t=1}^j \underline{X}_t \hat{e}_t \right)
\end{align}
Note that the corresponding single parameter LM statistics are given by 
\begin{align}
LM_1 &:= \frac{1}{T^2 \hat{ \sigma}^2 } \sum_{j=1}^T \left( \sum_{t=1}^j \hat{e}_t \right)^2 \\
LM_2 &:= \frac{1}{ \left( T^2 \hat{ \sigma}^2 \sum_{t=1}^T x_{t-1}^2 \right) } \sum_{j=1}^T \left( \sum_{t=1}^j x_{t-1} \hat{e}_t \right)^2
\end{align}
where LM$_1$ corresponds to the test statistic relating the the intercept alone and LM$_2$ to the slope coefficient alone. 
\begin{theorem}
Consider the model in (1)-(3) and let Assumption 2 hold. Under the  null hypothesis and under the same local alternatives , the following converge jointly as $T \to \infty$, in the sense of weak convergence of random measures on $\mathbb{R}:$
\begin{align}
LM_x \big| x \to_w \int_0^1 \mathbf{J}^{\prime}(r) \left[ \mathbf{V}(1) \right]^{-1} \mathbf{J}^{\prime}(r) dr \bigg| B_1 
\end{align}
and
\begin{align}
LM^{*}_{ x } \big| x, y \to_w \int_0^1 \mathbf{J}_0^{\prime}(r) \left[ \mathbf{V}(1) \right]^{-1} \mathbf{J}_0^{\prime}(r) dr \bigg| B_1 
\end{align}  
where $\mathbf{J}_0(r) := \int_0^r \mathbf{A}(s) dY(s), r \in [0,1]$.

Similarly, in the sense of weak convergence of random measures on $\mathbb{R}$, the following convergence jointly as $T \to \infty:$
\begin{align}
\text{supF}_x \big| x \to_w \text{sup}_{r \in \Lambda } \bigg\{ \mathbf{J}^{\prime}(r) \bigg[ \mathbf{V}(r) - \mathbf{V}(r) \mathbf{V}(1)^{-1}\mathbf{V}(r) \bigg]^{-1} \mathbf{J}(r)  \bigg\} \bigg| B_1
\end{align}
and
\begin{align}
\text{supF}^{*}_x \big| x,y \to_w \text{sup}_{r \in \Lambda } \bigg\{ \mathbf{J}_0^{\prime}(r) \bigg[ \mathbf{V}(r) - \mathbf{V}(r) \mathbf{V}(1)^{-1}\mathbf{V}(r) \bigg]^{-1} \mathbf{J}_0(r)  \bigg\} \bigg| B_1 := \mathcal{J}_0 \big| B_1. 
\end{align}
\end{theorem}

\newpage 

\paragraph{Supplementary Examples}

In this section we provide additional examples related to testing for parameter instability in predictive regression models. Further derivations can be found in the studies of \cite{katsouris2023predictability, katsouris2023structural, katsouris2023testing}. First, consider the case in which the predictive regression model includes a model intercept but is assumed to be stable before conducting any statistical inference on the existence of a structural break on the slopes of the model.  

\begin{example}
Consider the following predictive regression model
\begin{align}
\label{model1}
y_{t} &= \alpha + \boldsymbol{ \beta }_{1}^{\prime} x_{t-1} I \{ t \leq k  \} + \boldsymbol{ \beta }_{2}^{\prime}  x_{t-1} I \{ t > k \} + u_{t}  
\end{align}
where $y_t$ is a scalar (i.e., univariate predictant) and $x_t$ is $p-$dimensional vector of predictors (i.e., multiple predictors). Moreover, $k = \floor{ T \pi}$ denotes the unknown break-point location in the sample such that $\pi \in (0,1)$. The predictors of the predictive regression model are generated as a LUR process given by
\begin{align}
x_t = \left( \mathbf{I}_p - \frac{ \mathbf{C} }{T^{ \gamma_x} } \right) x_{t-1} + v_t
\end{align}
where $C >0$, i.e., $c_i > 0$ for all $i \in \left\{ 1,..., p \right\}$ and $\gamma_x \in (0,1)$ (MI). We are interested to test:
\begin{align}
\mathbb{H}_0: \boldsymbol{ \beta }_1 = \boldsymbol{ \beta }_2 \ \ \ \text{against} \ \ \ \mathbb{H}_1: \boldsymbol{ \beta }_1 \neq \boldsymbol{ \beta }_2,
\end{align}
\begin{align}
\label{Wald.test.ena}
\mathcal{W}_T^{OLS}( \pi ) 
&= \frac{1}{ \hat{\sigma}_u^2 } \left( \hat{\boldsymbol{ \beta }}_1 - \hat{\boldsymbol{ \beta }}_2  \right)^{\prime}  \left[ \boldsymbol{ \mathcal{R} } \left( \boldsymbol{X}^{*\prime} \boldsymbol{X}^{*} \right)^{-1} \boldsymbol{ \mathcal{R} }^{\prime}\right]^{-1} \left( \hat{\boldsymbol{ \beta }}_1  - \hat{\boldsymbol{ \beta }}_2 \right) 
\end{align} 
Therefore, when a model intercept is included in the model and is assumed to remain stable, then we shall consider the demeaned versions of the variables and reformulated the regression specification. The model can be written as:
\begin{align}
y_t^{*} = \beta_1 x_{t - 1}^{*} I \{ t \leq k  \} + \beta_2 x_{t - 1}^{*} I \{ t > k  \} + u_t^{*} \ \ \ \ \ \ \ 
\end{align}
such that the above predictive regression model considers the demeaned variables as below:
\begin{align*}
y_t^{*} = y_t - \frac{1}{T} \sum_{ t = 1}^T y_t, \ x_{1t - 1}^{*} = x_{1t - 1} - \frac{1}{T} \sum_{ t = 1}^T x_{1t - 1}, \ x_{2t - 1}^{*} = x_{2t - 1} - \frac{1}{T} \sum_{ t = 1}^T x_{2t - 1}, \ u_t^{*} = u_t - \frac{1}{T} \sum_{ t = 1}^T u_t. 
\end{align*} 
where $x_{1t-1} := x_{t-1} I \{ t \leq k  \}$ and $x_{2t-1} := x_{t-1} I \{ t > k  \}$. In matrix notation we use $X_1^{*} = \left( X_1 - \bar{X}_1 \right)$ and $X_2^{*} = \left( X_2 - \bar{X}_2 \right)$ where $\bar{X}_1 = \frac{1}{T} \sum_{t=1}^T x_{1t-1}$ and $\bar{X}_2 = \frac{1}{T} \sum_{t=1}^T x_{2t-1}$.   
Also, $u_t^{*}$ denotes the average-corrected error sequence of the predictive regression model given by \eqref{model1}.

\begin{remark}
Notice that while the regressors $\boldsymbol{X}_1$ and $\boldsymbol{X}_2$ are orthogonal regressors, the corresponding regressors $\boldsymbol{X}_1^{*}$ and $\boldsymbol{X}_2^{*}$ which occur after considering the demeaned versions of their counterparts are not orthogonal. For this reason one has to be careful when constructing the estimators for the parameters $\hat{\boldsymbol{ \beta }}_1$ and $\hat{\boldsymbol{ \beta }}_2$ with respect to the regressors $\boldsymbol{X}_1^{*}$ and $\boldsymbol{X}_2^{*}$, as well as the corresponding covariance matrix. 
\end{remark}

\newpage

More specifically, while asymptotically the conventional OLS estimators can be used without further modifications in finite-samples the standard formula will not hold. Similarly the covariance matrix can be simplified only asymptotically, which implies having off-diagonal terms converging in probability to zero, however this result does not hold in finite samples. In practise, we have the following specification i.e., a predictive regression model with multiple regressors written as below
\begin{align*}
y_t &= \alpha + \beta_1^{\prime} \boldsymbol{x}_{1t-1} + \boldsymbol{x}_{2t-1} \beta_2^{\prime} + u_t  
\\
\bar{y}_T &= \alpha + \beta_1^{\prime} \bar{ \boldsymbol{x} }_{1T} + \beta_2^{\prime} \bar{ \boldsymbol{x} }_{2T}  + \bar{u}_T  
\\
\left( y_t - \bar{y}_T\right)
&= 
\beta_1^{\prime} \left(  \boldsymbol{x}_{1t-1} -  \bar{\boldsymbol{x} }_{1T} \right) + \beta_2^{\prime} \left(  \boldsymbol{x}_{2t-1} -  \bar{\boldsymbol{x} }_{2T} \right)  + \left( u_t  - \bar{u}_T  \right)
\end{align*}
And by simplifying the notation we write 
\begin{align*}
y_t^{*} = \beta_1 x_{1t - 1}^{*} + \beta_2 x_{2t - 1}^{*} + u_t^{*} \ \ \ \ \ \ \ 
\end{align*}
where $x_{1t-1} := x_{t-1} I \{ t \leq k  \}$ and $x_{2t-1} := x_{t-1} I \{ t > k  \}$. 
\begin{remark}
Specifically, continuing the discussion from Remark 1, the expressions $\hat{\boldsymbol{ \beta }}_1 = \left( \boldsymbol{X}_1^{* \prime}  \boldsymbol{X}_1^{\prime} \right) \boldsymbol{X}_1^{* \prime} u^{*}$ and $\hat{\boldsymbol{ \beta }}_2 = \left( \boldsymbol{X}_2^{* \prime}  \boldsymbol{X}_2^{\prime} \right) \boldsymbol{X}_2^{* \prime} u^{*}$ will not hold in finite samples due to the fact that $\boldsymbol{X}_1^{* \prime}  \boldsymbol{X}_1^{\prime} \neq 0$ in finite samples. Moreover since we operate within the predictive regression model where regressors are generated as LUR processes, we have to note that these results hold for the mildly integrated case. For instance, in the case of near stationary regressors these effects are not contributing even after demeaning, which just implies that are still valid asymptotically for the near stationary case. On the other hand, for the case of LUR regressors, i.e., $\gamma_x = 1$, then the demeaning effect is contributing to finite samples and therefore we need to check the terms for the parameter estimators and the corresponding covariance matrix.
\end{remark}
Therefore, we have the following expression
\begin{align}
\label{eq1}
T^{ \frac{1 + \gamma_x}{2}} \left( \hat{\boldsymbol{ \beta }}_1 - \boldsymbol{ \beta }^0 \right) 
&=   \left( \frac{1}{ T^{ 1 + \gamma_x } } \sum_{t=1}^{ T } x^{*}_{1t-1} x_{1t-1}^{* \prime} \right)^{-1} \left( \frac{1}{ T^{ \frac{1 + \gamma_x}{2}} } \sum_{t=1}^{ T } x^{*}_{1t-1} u^{*}_{t} \right) + o_p(1).
\end{align}
\color{black}
Consider the first term of \eqref{eq1}
\begin{align*}
\sum_{t=1}^{ T } x^{*}_{1t-1} x_{1t-1}^{* \prime} = 
\sum_{t=1}^{ T }  \left( x_{1t - 1} - \frac{1}{T} \sum_{ t = 1}^T x_{1t - 1} \right) \left( x_{1t - 1} - \frac{1}{T} \sum_{ t = 1}^T x_{1t - 1} \right)^{\prime} 
= 
\sum_{t=1}^{ T } x_{1t - 1} x_{1t - 1}^{\prime}  - T \bar{x}_{1t - 1} \bar{x}_{1t - 1}^{\prime}  
\end{align*}
Therefore, we have that 
\begin{align*}
\frac{1}{ T^{ 1 + \gamma_x } } \sum_{t=1}^{ T } x^{*}_{1t-1} x_{1t-1}^{* \prime} 
&= 
\frac{1}{ T^{ 1 + \gamma_x } } \sum_{t=1}^{ T } x_{1t - 1} x_{1t - 1}^{\prime}  - \frac{T}{ T^{ 1 + \gamma_x } } \left( \frac{1}{T} \sum_{ t = 1}^T x_{1t - 1} \right) \left( \frac{1}{T} \sum_{ t = 1}^T x_{1t - 1} \right)^{\prime}  
\\
&=
\frac{1}{ T^{ 1 + \gamma_x } } \sum_{t=1}^{ T } x_{1t - 1} x_{1t - 1}^{\prime} - \left( \frac{\sum_{ t = 1}^T x_{1t - 1}}{ T^{ \frac{\gamma_x}{2} + 1 }}  \right) \left( \frac{\sum_{ t = 1}^T x_{1t - 1}}{ T^{ \frac{\gamma_x}{2} + 1 }}  \right)^{\prime}
\end{align*}

\newpage

Since as $T \to \infty$, $\frac{\sum_{ t = 1}^T x_{1t - 1}}{ T^{ \frac{\gamma_x}{2} + 1 }} \to_p 0$ for $\gamma_x \in (0,1)$, then we obtain that
\begin{align}
\underset{ T \to \infty }{ \text{plim} } \left\{ \frac{1}{ T^{ 1 + \gamma_x } } \sum_{t=1}^{ T } x^{*}_{1t-1} x_{1t-1}^{* \prime} \right\} 
&\equiv 
\underset{ T \to \infty }{ \text{plim} } \left\{
\frac{1}{T^{ 1 + \gamma_x } } \sum_{t=1}^{ T } x_{1t-1} x_{1t-1}^{\prime} \right\} 
\nonumber
= 
\underset{ T \to \infty }{ \text{plim} } \left\{
\frac{1}{T^{ 1 + \gamma_x } } \sum_{t=1}^{ T } x_{t-1} x_{t-1}^{\prime} I \left\{ t \leq k \right\} \right\} 
\nonumber
\\
&=
\underset{ T \to \infty }{ \text{plim} } \left\{
\frac{1}{T^{ 1 + \gamma_x } } \sum_{t=1}^{ \floor{T \pi} } x_{t-1} x_{t-1}^{\prime} \right\}  
\nonumber
\Rightarrow \pi \mathbf{V}_c
\end{align}
Consider the second term of \eqref{eq1}
\begin{align*}
\sum_{t=1}^{ T } x^{*}_{1t-1} u^{*}_{t}
=
\sum_{t=1}^{ T } \left( x_{1t - 1} - \frac{1}{T} \sum_{ t = 1}^T x_{1t - 1} \right) \left( u_t - \frac{1}{T} \sum_{ t = 1}^T u_t \right)
=
\sum_{t=1}^{ T } x_{1t - 1} u_t  - \bar{u}_{t} \sum_{t=1}^{ T }  x_{1t - 1} - \bar{x}_{t - 1}  \sum_{t=1}^{ T } \left( u_t - \bar{u}_{t} \right)
\end{align*}
Since the last term above is zero, we have that 
\begin{align}
\frac{1}{ T^{ \frac{1 + \gamma_x}{2}} } \sum_{t=1}^{ T } x^{*}_{1t-1} u^{*}_{t} = \frac{1}{ T^{ \frac{1 + \gamma_x}{2}} } \sum_{t=1}^{ T } x_{1t - 1} u_t - \bar{u}_{T} \frac{ \sum_{t=1}^{ T }  x_{1t - 1}}{ T^{ \frac{1 + \gamma_x}{2}} }
\end{align}
In other words, 
\begin{align*}
\frac{1}{ T^{ \frac{1 + \gamma_x}{2}} } \sum_{t=1}^{ T } x^{*}_{1t-1} u^{*}_{t} = \frac{1}{ T^{ \frac{1 + \gamma_x}{2}} } \sum_{t=1}^{ T } x_{1t - 1} u_t + o_p(1).
\end{align*}
\begin{proof}
Notice that $\bar{u}_{T} = \frac{1}{T} \sum_{t=1}^{ T } u_t$, where $u_t \sim_{ \textit{i.i.d} } (0, \sigma^2)$ which implies that $\bar{u}_{T} = \mathcal{O}_p \left( \frac{1}{ \sqrt{T} } \right)$. Another way to think about it is that since $T^{-1} \sum_{t=1}^{ T } u_t \overset{ p }{ \to } \mathbb{E} \left( u_t \right)$ then $T^{- 3 / 2} \sum_{t=1}^{ T } u_t = o_p(1)$. Notice that for the LUR process, $x_t = \rho x_{t-1} + u_t$, we have that $u_t$ is a martingale difference sequence with respect to $\mathcal{F}_t = \sigma \left( u_t, u_{t-1}, ... \right)$ satisfying $\mathbb{E} \left( u_t^2 | \mathcal{F}_t \right) = \sigma^2 > 0$ for all $t$, and $\left( u_t^2 \right)_{ t \in \mathbb{Z} }$ is a uniformly integrable sequence. Therefore, by expanding the autoregressive model, where $\rho_T = \left( 1 - \frac{c}{T}^{\gamma_x} \right)$, for some $c > 0 $ and $\gamma_x \in (0,1)$, we obtain the following expression: 
\begin{align*}
\left( 1 - \rho_T \right) \sum_{t=1}^T x_{t-1} &= \sum_{t=1}^T u_t  - \left( x_T - x_0 \right)
\end{align*}
Notice that $x_T$ dominates (faster asymptotic convergence) the term $\sum_{t=1}^T x_{t-1}$ in the mildly integrated case, i.e., $\gamma_x \in (0,1)$ since the quantity $x_T$ is $\mathcal{O}_p \left( T^{ - \frac{\gamma_x }{2} } \right)$ while the quantity $\sum_{t=1}^T x_{t-1}$ is $o_p(1)$. We can see the first result by considering the order of convergence of the term $\mathbb{E} \left[ x_T \right] = \sum_{j=1}^t \rho_T^{t-j} x_{t-j}$ which is $\mathcal{O}_p \left( T^{ \frac{\gamma_x }{2} } \right)$. Moreover, it also holds that $\sum_{t=1}^T u_t = o_p(1)$. Thus, the order of convergence is 
\begin{align}
\bar{u}_{T} \frac{ \sum_{t=1}^{ T }  x_{1t - 1}}{ T^{ \frac{1 + \gamma_x}{2}} } \overset{ p }{ \to }  \mathcal{O}_p (1) \frac{1}{ \sqrt{T} } \frac{ \sum_{t=1}^{ T }  x_{1t - 1}}{ T^{ \frac{\gamma_x}{2}} } = \mathcal{O}_p (1) o_p(1) = o_p(1)
\end{align}
\end{proof}

\newpage 

Recall that if for some sequence $M_n$ it holds that $m_n = \mathcal{O} ( \frac{1}{n} )$ then the sequence $n M_n = \mathcal{O} ( 1 )$ is a bounded sequence. Since as $T \to \infty$, $\bar{u}_{t} \to 0$ while $T^{ - \frac{1 + \gamma_x}{2}} \sum_{t=1}^{ T }  x_{1t - 1}$ is bounded, then we have that
\begin{align}
\underset{ T \to \infty }{ \text{plim} } \left\{ \frac{1}{ T^{ \frac{1 + \gamma_x}{2}} } \sum_{t=1}^{ T } x^{*}_{1t-1} u^{*}_{t} \right\} 
&\equiv  
\underset{ T \to \infty }{ \text{plim} } \left\{  \frac{1}{ T^{ \frac{1 + \gamma_x}{2}} } \sum_{t=1}^{ T } x_{1t - 1} u_t \right\} 
=
\underset{ T \to \infty }{ \text{plim} } \left\{  \frac{1}{ T^{ \frac{1 + \gamma_x}{2}} } \sum_{t=1}^{ T } x_{t - 1} u_t I \left\{ t \leq k \right\} \right\}
\nonumber
\\
&=
\underset{ T \to \infty }{ \text{plim} } \left\{  \frac{1}{ T^{ \frac{1 + \gamma_x}{2}} } \sum_{t=1}^{ \floor{T \pi} } x_{t - 1} u_t \right\}
\nonumber
\Rightarrow 
\mathcal{N} \bigg( 0,  \pi \sigma_{u}^2  \otimes  \mathbf{V}_c \bigg) 
\end{align}
Therefore, we obtain the convergence result below:
\begin{align}
\label{result1}
T^{ \frac{1 + \gamma_x}{2}} \left( \hat{\boldsymbol{ \beta }}_1 - \boldsymbol{ \beta }^0 \right) 
\Rightarrow  \bigg( \pi \mathbf{V}_c \bigg)^{-1} \mathcal{N} \bigg( 0,  \pi \sigma_{u}^2  \otimes  \mathbf{V}_c \bigg) 
= \frac{1}{\pi} \mathbf{V}_c^{-1}  \mathcal{N} \bigg( 0,  \pi \sigma_{u}^2  \otimes  \mathbf{V}_c \bigg) 
\end{align}
We also have the following expression
\begin{align}
\label{eq2}
T^{ \frac{1 + \gamma_x}{2}} \left( \hat{\boldsymbol{ \beta }}_2 - \boldsymbol{ \beta }^0 \right) 
&=   \left( \frac{1}{ T^{ 1 + \gamma_x } } \sum_{t=1}^{ T } x^{*}_{2t-1} x_{2t-1}^{* \prime} \right)^{-1} \left( \frac{1}{ T^{ \frac{1 + \gamma_x}{2}} } \sum_{t=1}^{ T } x^{*}_{2t-1} u^{*}_{t} \right) + o_p(1).
\end{align}
Consider the second term of \eqref{eq2}
\begin{align*}
\sum_{t=1}^{ T } x^{*}_{2t-1} u^{*}_{t}
=
\sum_{t=1}^{ T } \left( x_{2t - 1} - \frac{1}{T} \sum_{ t = 1}^T x_{2t - 1} \right) \left( u_t - \frac{1}{T} \sum_{ t = 1}^T u_t \right)
=
\sum_{t=1}^{ T } x_{2t - 1} u_t  - \bar{u}_{t} \sum_{t=1}^{ T }  x_{2t - 1} - \bar{x}_{t - 1}  \sum_{t=1}^{ T } \left( u_t - \bar{u}_{t} \right)
\end{align*}
Since the last term above is zero, we have that 
\begin{align}
\frac{1}{ T^{ \frac{1 + \gamma_x}{2}} } \sum_{t=1}^{ T } x^{*}_{2t-1} u^{*}_{t} = \frac{1}{ T^{ \frac{1 + \gamma_x}{2}} } \sum_{t=1}^{ T } x_{2t - 1} u_t - \bar{u}_{T} \frac{ \sum_{t=1}^{ T }  x_{2t - 1}}{ T^{ \frac{1 + \gamma_x}{2}} }
\end{align}
In other words, 
\begin{align*}
\frac{1}{ T^{ \frac{1 + \gamma_x}{2}} } \sum_{t=1}^{ T } x^{*}_{2t-1} u^{*}_{t} = \frac{1}{ T^{ \frac{1 + \gamma_x}{2}} } \sum_{t=1}^{ T } x_{2t - 1} u_t + o_p(1).
\end{align*}
and the proof of the above holds due to similar arguments as in the case of $\frac{1}{ T^{ \frac{1 + \gamma_x}{2}} } \sum_{t=1}^{ T } x^{*}_{1t-1} u^{*}_{t}$. Using similar arguments we can show that
\begin{align}
\underset{ T \to \infty }{ \text{plim} } \left\{ \frac{1}{ T^{ 1 + \gamma_x } } \sum_{t=1}^{ T } x^{*}_{2t-1} x_{2t-1}^{* \prime} \right\} 
&\equiv 
\underset{ T \to \infty }{ \text{plim} } \left\{
\frac{1}{T^{ 1 + \gamma_x } } \sum_{t=1}^{ T } x_{2t-1} x_{2t-1}^{\prime} \right\} 
\nonumber
= 
\underset{ T \to \infty }{ \text{plim} } \left\{
\frac{1}{T^{ 1 + \gamma_x } } \sum_{t=1}^{ T } x_{t-1} x_{t-1}^{\prime} I \left\{ t > k \right\} \right\} 
\nonumber
\\
&=
\underset{ T \to \infty }{ \text{plim} } \left\{
\frac{1}{T^{ 1 + \gamma_x } } \sum_{t=\floor{T \pi} + 1}^{ T } x_{t-1} x_{t-1}^{\prime} \right\}  
\Rightarrow ( 1 - \pi ) \mathbf{V}_c
\end{align}

\newpage

and
\begin{align}
\underset{ T \to \infty }{ \text{plim} } \left\{ \frac{1}{ T^{ \frac{1 + \gamma_x}{2}} } \sum_{t=1}^{ T } x^{*}_{2t-1} u^{*}_{t} \right\} 
&\equiv  
\underset{ T \to \infty }{ \text{plim} } \left\{  \frac{1}{ T^{ \frac{1 + \gamma_x}{2}} } \sum_{t=1}^{ T } x_{2t - 1} u_t \right\} 
=
\underset{ T \to \infty }{ \text{plim} } \left\{  \frac{1}{ T^{ \frac{1 + \gamma_x}{2}} } \sum_{t=1}^{ T } x_{t - 1} u_t I \left\{ t > k \right\} \right\}
\nonumber
\\
&=
\underset{ T \to \infty }{ \text{plim} } \left\{  \frac{1}{ T^{ \frac{1 + \gamma_x}{2}} } \sum_{t=\floor{T \pi} + 1}^{ T } x_{t - 1} u_t \right\}
\nonumber
\Rightarrow 
\mathcal{N} \bigg( 0,  ( 1 - \pi ) \sigma_{u}^2  \otimes  \mathbf{V}_c \bigg) 
\end{align}
Therefore, we obtain the convergence result below:
\begin{align}
\label{result2}
T^{ \frac{1 + \gamma_x}{2}} \left( \hat{\boldsymbol{ \beta }}_2 - \boldsymbol{ \beta }^0 \right) 
\Rightarrow  
\bigg( ( 1 - \pi ) \mathbf{V}_c \bigg)^{-1} \mathcal{N} \bigg( 0, ( 1 - \pi ) \sigma_{u}^2  \otimes  \mathbf{V}_c \bigg) 
\equiv
\frac{1}{ 1 - \pi } \mathbf{V}_c^{-1} \mathcal{N} \bigg( 0, ( 1 - \pi ) \sigma_{u}^2  \otimes  \mathbf{V}_c \bigg) 
\end{align}
Now, using expressions \eqref{result1} and \eqref{result2} we obtain the following simplified expression
\begin{align}
\label{betas}
T^{ \frac{1 + \gamma_x }{2}} \left( \hat{\boldsymbol{ \beta }}_1 - \hat{\boldsymbol{ \beta }}_2 \right) 
&\Rightarrow 
\left\{ \frac{1}{\pi} \mathbf{V}_c^{-1}  \mathcal{N} \bigg( 0,  \pi \sigma_{u}^2  \otimes  \mathbf{V}_c  \bigg)  - \frac{1}{1 - \pi} \mathbf{V}_c^{-1} \mathcal{N} \bigg( 0,  ( 1 - \pi ) \sigma_{u}^2  \otimes  \mathbf{V}_c  \bigg)  \right\} 
\nonumber
\\
&= \left\{  \frac{ \mathbf{V}_c^{-1} \mathcal{N} \bigg( 0,  \pi \sigma_{u}^2  \otimes  \mathbf{V}_c \bigg) -  \mathbf{V}_c^{-1}  \mathcal{N} \bigg( 0, \sigma_{u}^2  \otimes  \mathbf{V}_c  \bigg)  }{ \pi ( 1 - \pi )  }  \right\}
\nonumber
\\
&= 
\left\{  \frac{ \mathbf{V}_c^{-1 / 2} \mathbf{V}_c^{-1 / 2} \mathcal{N} \bigg( 0,  \pi \sigma_{u}^2  \otimes  \mathbf{V}_c \bigg) - \mathbf{V}_c^{-1 / 2} \mathbf{V}_c^{-1 / 2}  \mathcal{N} \bigg( 0, \sigma_{u}^2  \otimes  \mathbf{V}_c  \bigg)  }{ \pi ( 1 - \pi )  }  \right\}
\nonumber
\\
&= 
\left\{  \frac{ \mathbf{V}_c^{-1 / 2} }{  \pi ( 1 - \pi ) } \left[ \mathbf{V}_c^{-1 / 2} \mathcal{N} \bigg( 0,  \pi \sigma_{u}^2  \otimes  \mathbf{V}_c \bigg) - \mathbf{V}_c^{-1 / 2}  \mathcal{N} \bigg( 0, \sigma_{u}^2  \otimes  \mathbf{V}_c  \bigg) \right] \right\}
\nonumber
\\
&= 
\left\{  \frac{ \mathbf{V}_c^{-1 / 2} }{  \pi ( 1 - \pi ) } \left[  \mathcal{N} \bigg( 0,  \pi \sigma_{u}^2 \otimes  \mathbf{I}_p  \bigg) - \mathcal{N} \bigg( 0, \sigma_{u}^2  \otimes  \mathbf{I}_p  \bigg) \right] \right\}
\end{align}
where $\gamma_x \in (0,1)$ and $\pi$ the unknown break fraction. We construct the OLS Wald statistic for testing the null hypothesis $H_0: \boldsymbol{ \mathcal{R} } \boldsymbol{\beta} = 0$, which is equivalent to the null of no parameter instability, 
such that $\boldsymbol{ \mathcal{R} } = \left[ \mathbf{I}_p \ - \mathbf{I}_p \right]$ and $\boldsymbol{\beta} = \left( \boldsymbol{\beta}_1, \boldsymbol{\beta}_2 \right)$, where $\boldsymbol{\beta}$ is the model parameter and $\hat{ \boldsymbol{\beta} } = \left( \boldsymbol{\beta}_1, \boldsymbol{\beta}_2 \right)$ the OLS estimator.
\begin{align}
\label{wald}
\mathcal{W}_T( \pi ) 
&= \frac{1}{ \hat{\sigma}_u^2 } \left[ T^{ \frac{1 + \gamma_x }{2}} \left( \hat{\boldsymbol{\beta}}_1 - \hat{\boldsymbol{\beta}}_2  \right) \right]^{\prime}  \left[ \boldsymbol{ \mathcal{R} } \left( \frac{ \boldsymbol{X}^{* \prime} \boldsymbol{X}^{*} }{ T^{1 + \gamma_x } } \right)^{-1} \boldsymbol{ \mathcal{R} }^{\prime}\right]^{-1} \left[ T^{ \frac{1 + \gamma_x }{2}} \left( \hat{\boldsymbol{\beta}}_1 - \hat{\boldsymbol{\beta}}_2 \right) \right]
\end{align}
where
\begin{align}
\label{cov}
\left[ \boldsymbol{ \mathcal{R} } \left( \frac{ \boldsymbol{X}^{* \prime} \boldsymbol{X}^{*} }{ T^{1 + \gamma_x } } \right)^{-1} \boldsymbol{ \mathcal{R} }^{\prime} \right] 
\Rightarrow 
\left[ \bigg( \pi \mathbf{V}_c \bigg)^{-1}  + \bigg( (1 - \pi) \mathbf{V}_c \bigg)^{-1}  \right] 
= \left[ \frac{1}{\pi} \mathbf{V}_c  + \frac{1}{1 - \pi }  \mathbf{V}_c \right]   
\nonumber
= \frac{ \mathbf{V}_c^{-1} }{ \pi (1 - \pi) }
\end{align}

\newpage

Therefore, based on expressions \eqref{betas}, \eqref{wald} and \eqref{cov} the simplified asymptotic result for the OLS-Wald statistic is given by 
\begin{align}
\mathcal{W}_T( \pi ) 
&\Rightarrow 
\frac{1}{ \sigma_{u}^2 }  \left\{  \frac{ \mathbf{V}_c^{-1 / 2} }{  \pi ( 1 - \pi ) } \left[  \mathcal{N} \bigg( 0,  \pi \sigma_{u}^2  \otimes  \mathbf{I}_p \bigg) - \mathcal{N} \bigg( 0, \sigma_{u}^2  \otimes  \mathbf{I}_p  \bigg) \right] \right\}^{\prime} \left[ \frac{ \mathbf{V}_c^{-1} }{ \pi (1 - \pi) } \right]^{-1} 
\nonumber 
\\
&\ \ \ \ \ \times
\left\{  \frac{ \mathbf{V}_c^{-1 / 2} }{  \pi ( 1 - \pi ) } \left[  \mathcal{N} \bigg( 0,  \pi \sigma_{u}^2 \otimes \mathbf{I}_p  \bigg) - \mathcal{N} \bigg( 0, \sigma_{u}^2  \otimes  \mathbf{I}_p  \bigg) \right] \right\}
\nonumber
\\
&= 
\frac{1}{ \sigma_{u}^2 } \frac{1}{ \pi (1 - \pi) } \left[  \mathcal{N} \bigg( 0,  \pi \sigma_{u}^2  \otimes  \mathbf{I}_p \bigg) - \mathcal{N} \bigg( 0, \sigma_{u}^2  \otimes  \mathbf{I}_p  \bigg) \right]^{\prime} \left( \mathbf{V}_c^{-1 / 2} \right)^{\prime} \mathbf{V}_c \mathbf{V}_c^{-1 / 2} 
\nonumber 
\\
&\ \ \ \ \ \ \ \ \ \ \ \ \ \ \ \ \ \times
\left[  \mathcal{N} \bigg( 0,  \pi \sigma_{u}^2 \otimes \mathbf{I}_p  \bigg) - \mathcal{N} \bigg( 0, \sigma_{u}^2  \otimes  \mathbf{I}_p  \bigg) \right] 
\end{align}
Notice that since
\begin{align}
\left( \mathbf{V}_c^{-1 / 2} \right)^{\prime} \mathbf{V}_c \mathbf{V}_c^{-1 / 2}  = \left( \mathbf{V}_c^{-1 / 2} \right)^{\prime} \mathbf{V}_c^{1/2} \left( \mathbf{V}_c^{1/2} \right)^{\prime} \mathbf{V}_c^{-1 / 2} = \mathbf{I}_p
\end{align}
Then the limit result for the Wald OLS statistic simplifies to the following expression
\begin{align*}
\mathcal{W}_T( \pi ) 
&\Rightarrow 
\frac{1}{ \sigma_{u}^2 } \frac{1}{ \pi (1 - \pi) }
\left[  \mathcal{N} \bigg( 0,  \pi \sigma_{u}^2  \otimes  \mathbf{I}_p \bigg) - \mathcal{N} \bigg( 0, \sigma_{u}^2  \otimes  \mathbf{I}_p  \bigg) \right]^{\prime}  
\left[  \mathcal{N} \bigg( 0,  \pi \sigma_{u}^2  \otimes  \mathbf{I}_p \bigg) - \mathcal{N} \bigg( 0, \sigma_{u}^2  \otimes  \mathbf{I}_p  \bigg) \right]
\\
\\
&\Rightarrow
\frac{ \bigg[ \mathbf{W}_p (\pi) - \pi \mathbf{W}_p (1) \bigg]^{\prime} \bigg[ \mathbf{W}_p (\pi) - \pi \mathbf{W}_p (1) \bigg] }{ \pi ( 1 - \pi ) } :=  \frac{  \boldsymbol{ \mathcal{BB} }_p ( \pi )^{\prime}  \boldsymbol{ \mathcal{BB} }_p ( \pi )  }{ \pi (1 -  \pi) }
\end{align*}
where $\mathbf{W}_p (\pi)$ is a Brownian motion with covariance matrix $\sigma_{u}^2$. In summary, we have proved that the limiting distribution of the OLS-Wald statistic for testing the null hypothesis of parameter constancy for the parameters of the predictive regression model when: (i) a model intercept is included but is assumed to remain stable through out the sample and (ii) the regressors are modelled as mildly integrated regressors,it weakly convergence to a normalized Brownian Bridge, i.e., the standard NBB limit result of Andrews is preserved. Therefore, similarly the corresponding limiting distribution of the sup OLS-Wald statistic is the sup of a NBB process. In other words, it follows that
\begin{align}
\widetilde{\mathcal{W}}^{OLS}(\pi) \Rightarrow \underset{ \pi \in [ \pi_1, \pi_2 ]}{\text{sup}} \frac{  \boldsymbol{ \mathcal{BB} }_p ( \pi )^{\prime}  \boldsymbol{ \mathcal{BB} }_p ( \pi )  }{ \pi (1 -  \pi) }, 
\end{align}
where $ \boldsymbol{ \mathcal{BB} }_p ( . )$ is a $p-$dimensional standard Brownian bridge, such that $ \boldsymbol{ \mathcal{BB} }_p ( . ) := \left[ \mathbf{W}_p (.) - \pi \mathbf{W}_p (1) \right]$. 

Notice that in the case of a fixed break-point, say $\pi \equiv \pi_0$, then since $\mathbf{W}_p( \pi_0 ) - \pi_0 \mathbf{W}_p(1) \equiv \mathcal{N} \big( 0, \pi_0 (1 - \pi_0) \big)$, then it holds that $\displaystyle \frac{  \boldsymbol{ \mathcal{BB} }_p ( \pi_0 )^{\prime}  \boldsymbol{ \mathcal{BB} }_p ( \pi_0 )  }{ \pi_0 (1 -  \pi_0 ) } \equiv \sum_{j=1}^p \mathcal{N}(0,1)^2 = \chi^2 (p)$ which shows that the OLS-Wald is free of any nuisance parameters.

\end{example}

\newpage 

\subsubsection{Main Aspects on Predictive Regressions}

A research question of interest is whether future values of one series $\left\{ Y_t \right\}$ can be predicted from lagged values of another series $\left\{ X_t \right\}$, where the two series obey the following model: 
\begin{align}
Y_t &=  \alpha + \beta X_{t-1} + u_t, 
\\
X_t &= \mu + \rho X_{t-1} + v_t, 
\end{align}
It is well known that the OLS estimate of $\beta$ is biased when the errors $\left\{ u_t, v_t \right\}$ are correlated, with the amount of bias increasing as $\rho$ get closer to the unit boundary. There is currently no known theoretical justification that inference (e.g., a t-test) based on the bias-corrected OLS estimate will have improved size properties relative to the test based on the uncorrected OLS estimate. In practise, there are several studies in the literature that provide examples where even the use of the exact bias correction does not result in accurate finite-sample inference because the resulting t-statistic can still be very far from normal.  Therefore, from an inferential point of view, issues such as the bias of point estimates may be irrelevant in small samples, and thus it is also crucial to examine the likelihood. At first glance, it may seem surprising that a likelihood ratio test (LRT) may provide well-behaved hypothesis tests in situations when the t-statistic does not because the to test statistics are closely related (see, \cite{phillips2014restricted}). 

In this section, we consider the estimation of the predictive regression model from the statistical point of view, motivating the use of the restricted likelihood, which is free of the nuisance intercept parameter and hence able to imitate the likelihood of the no-intercept univariate model with its attendant small curvature. A curvature-related approach to tackle the predictive regression problem for the model was also taken by \cite{jansson2006optimal}. Due to the presence of model intercepts $( \mu, \eta )$, we are motivated to seek a likelihood that does not involve the location parameters and yet possesses small curvature properties similar to those of the model with known location parameters. Thus, the restricted likelihood has exactly these properties. In practise, the idea of restricted likelihood was originally proposed by KS (1970) precisely as a means of eliminating the effect of nuisance location parameters when estimating the parameters of the model covariance structure in a linear model. 

We focus our attention on the use of the restricted likelihood for carrying inference on $\beta$ in the bivariate predictive regression model. Notice that in the particular framework, the restricted likelihood is the exact likelihood of the vector $\big\{ ( Y_t - Y_{t-1} )_{t=2}^n, ( X_t - X_{t-1} )_{t=2}^n \big\}$. In practise, the restricted maximum likelihood (REML) estimates are found to be asymptotically efficient. A standard method for testing the two-sided hypothesis, $H_0: \beta = 0$ versus $H_A: \beta \neq 0$ is the LRT, which compares the log-likelihood evaluated at the unrestricted estimates of the parameters to the log-likelihood evaluated at the parameter estimates obtained under the restriction that the null hypothesis $H_0: \beta = 0$ is true. Then, it is easy to generalize the REML likelihood in two directions that are both of practical interest. One generalization is to the case where the predictor series is a multivariate AR(1) process. The other generalization is to the case where the univariate predictor follows a higher order AR process.

\newpage

In other words, the inclusion of an intercept in the model causes the likelihood ratio to lose this property, thus pointing to the intercept as the source of the problem. Therefore, this motivates the use of the restricted likelihood, which is free of the nuisance intercept parameter and hence able to imitate the likelihood of the no-intercept univariate model with its attendant small curvature. Thus, we are indeed able to obtain theoretical results that demonstrate the the LRT based on the restricted likelihood (RLRT) has good finite-sample performance for both estimation and inference in this context.

\paragraph{Multivariate Regression}

Assume that the data $\big( Y_1,..., Y_n, \boldsymbol{X}_0,...,  \boldsymbol{X}_n^{\prime} \big)$ follow
\begin{align}
Y_t &= \mu + \boldsymbol{\beta}^{\prime} \boldsymbol{X}_{t-1} + u_t,
\\
\boldsymbol{X}_{t} &= \boldsymbol{\eta} + \boldsymbol{R} \boldsymbol{X}_{t-1} + \boldsymbol{v}_t,
\end{align}
where $u_t = \boldsymbol{\phi}^{\prime} \boldsymbol{v}_t + e_t$, such that $\big( e_t,  \boldsymbol{v}_t^{\prime} \big)^{\prime} \sim \mathcal{N} \big( \boldsymbol{0}, \mathsf{diag} \left( \sigma_e^2, \boldsymbol{\Sigma}_v   \right) \big)$ is an \textit{i.i.d} series and $\boldsymbol{R}$ is a $(k \times k)$ matrix with all eigenvalues less than unity in absolute value. Define with $\boldsymbol{\Sigma}_v \equiv \mathsf{Var} ( \boldsymbol{v}_t )$ and $\boldsymbol{\Sigma}_{ \boldsymbol{X} } \equiv \mathsf{Var} \left( \boldsymbol{X}_{t} \right)$ as 
\begin{align}
\mathsf{vec} \left( \boldsymbol{\Sigma}_{ \boldsymbol{X} } \right) = \big( I_{K^2} - \boldsymbol{R} \otimes \boldsymbol{R} \big)^{-1} \mathsf{vec} \left( \boldsymbol{\Sigma}_{v} \right)
\end{align}
and define with 
\begin{align}
\hat{\boldsymbol{K}} 
= 
\left[ \boldsymbol{\Sigma}_{ \boldsymbol{X} }^{-1} + n ( \boldsymbol{I} - \boldsymbol{R} )^{\prime} \boldsymbol{\Sigma}_{v}^{-1} ( \boldsymbol{I} - \boldsymbol{R} ) \right]^{-1} \left[ \boldsymbol{\Sigma}_{ \boldsymbol{X} }^{-1} \boldsymbol{X}_0 +  ( \boldsymbol{I} - \boldsymbol{R} )^{\prime} \boldsymbol{\Sigma}_{v}^{-1} ( \boldsymbol{I} - \boldsymbol{R} ) \sum_{t=1}^n \boldsymbol{X}_t  \right]
\end{align}

Then, the REML log-likelihood up to an additive constant for the model is given by 
\begin{align*}
L_M 
&= 
- \left( \frac{n-1}{2} \right) \mathsf{log} \sigma_e^2 - \frac{1}{ 2 \sigma_e^2 } S \left( \boldsymbol{\phi}, \boldsymbol{\beta}, \boldsymbol{R}  \right) - \frac{1}{2} \mathsf{log} | \boldsymbol{\Sigma }_{ \boldsymbol{X} } | - \frac{n}{2} \mathsf{log} | \boldsymbol{\Sigma }_{ v } |
\nonumber
\\
&\ \ \ -
\frac{1}{2} \mathsf{log} \left| \boldsymbol{\Sigma }_{ \boldsymbol{X} } + n ( \boldsymbol{I} - \boldsymbol{R} )^{\prime} \boldsymbol{\Sigma}_{v}^{-1} ( \boldsymbol{I} - \boldsymbol{R} )  \right|
\nonumber
\\
&\ \ \ -
\frac{1}{2} \left\{ \left( \boldsymbol{X}_0 - \hat{\boldsymbol{K}} \right)^{\prime} \boldsymbol{\Sigma }_{ \boldsymbol{X} }^{-1} \left( \boldsymbol{X}_0 - \hat{\boldsymbol{K}} \right) + \sum_{t=1}^n \left[  \boldsymbol{X}_t - \hat{\boldsymbol{K}} - \boldsymbol{R} \left( \boldsymbol{X}_{t-1} - \hat{\boldsymbol{K}} \right) \right]^{\prime} \boldsymbol{\Sigma}_v^{-1} \left[  \boldsymbol{X}_t - \hat{\boldsymbol{K}} - \boldsymbol{R} \left( \boldsymbol{X}_{t-1} - \hat{\boldsymbol{K}} \right) \right] \right\}
\end{align*}
where
\begin{align}
S \left( \boldsymbol{\phi}, \boldsymbol{\beta}, \boldsymbol{R}  \right)  
= 
\sum_{t=1}^n \bigg[ Y_t^{\mu} - \boldsymbol{\phi}^{\prime} \boldsymbol{X}_t^{\mu} - \big( \boldsymbol{\beta}^{\prime} - \boldsymbol{\phi}^{\prime} \boldsymbol{R} \big) \boldsymbol{X}_{t-1}^{\mu} \bigg]^2
\end{align}
and
\begin{align}
\boldsymbol{X}_t^{\mu} = \boldsymbol{X}_t - \frac{1}{n} \sum_{t=1}^n \boldsymbol{X}_t \ \ \ \ \text{and} \ \ \ \ \boldsymbol{X}_{t-1}^{\mu} = \boldsymbol{X}_{t-1} - \frac{1}{n} \sum_{t=1}^n \boldsymbol{X}_{t-1} 
\end{align}

\newpage 

\begin{remark}
Notice that in the case where $\boldsymbol{R}$ is assumed to be diagonal matrix, the predictive regression model is no longer a seemingly unrelated regression (SUR) system and hence OLS will no longer be efficient. However, REML will clearly retain efficiency, no matter what the form of $\boldsymbol{R}$ is, thus giving it an advantage in terms of both asymptotic efficiency and power over any OLS-based procedure. 
\end{remark}

Furthermore, since the dimension of the parameter space is very large in the vector case,  it is not feasible to obtain a result such as Theorem 3 in the most general case. However, in the case where $\boldsymbol{R}$ is a diagonal matrix and where $\big( \sigma_e^2, \boldsymbol{\phi}, \boldsymbol{\Sigma}_v \big)$ are assumed known with $\boldsymbol{\Sigma}_v$ diagonal, we are able to obtain the following result on the finite-sample behaviour of the RLRT for testing $\mathbb{H}_0: \boldsymbol{\beta} = \boldsymbol{0}$.

\paragraph{Higher Order Autoregressive Regressors}

Let the observed data $\big( Y_1,..., Y_n, X_{-p+2},..., X_n \big)$ follow 
\begin{align}
Y_t &= \mu + \beta X_{t-1} + u_t, 
\\
X_t &= \eta + \alpha_1 X_{t-1} + .... + \alpha_p X_{t-p}  + v_t
\end{align}

where $u_t = \phi v_t + e_t$ and $\big( e_t , v_t \big) \sim \mathcal{N} \big( \boldsymbol{0}, \mathsf{diag} \left( \sigma_e^2, \sigma_v^2 \right) \big)$ are an $\textit{i.i.d}$ series. Moreover, assume that all the roots of the polynomial $z^p - \sum_{s=1}^p z^{p-s} \alpha_s$ lie within the unit circle. Define with 
\begin{align}
\boldsymbol{Y}^{\mu} = \boldsymbol{Y} - \boldsymbol{1} \bar{Y}, \ \ \ \ \boldsymbol{X}^{\mu} = \bigg[ \boldsymbol{X}_1 - \boldsymbol{1} \bar{X}_1,..., \boldsymbol{X}_{-p + 1} - \boldsymbol{1} \bar{X}_{-p + 1} \bigg], 
\end{align}
where $X_i = \frac{1}{n} \boldsymbol{1} \boldsymbol{X}_i$. 

\medskip

\begin{exercise}

Consider the predictive regression model formulated as below
\begin{align}
y_t 
&= 
\alpha \left( \frac{t}{n} \right) +  \boldsymbol{x}_{t-1}^{\prime} \boldsymbol{\beta} \left( \frac{t}{n} \right) + u_t
\\
\boldsymbol{x}_t &= \left( \boldsymbol{I}_p - \frac{ \boldsymbol{C}_p }{ n^{\gamma} } \right) \boldsymbol{x}_{t-1} + \boldsymbol{v}_t
\end{align}
\begin{itemize}
\item Under the null hypothesis we have a predictive regression model with fixed parameter vector which is expresses as $y_t = \alpha + \boldsymbol{x}_{t-1}^{\prime} \boldsymbol{\beta} + u_t$.

\item Under the alternative hypothesis, $\mathcal{H}_1$, $\alpha_t = \alpha \left( \frac{t}{n} \right)$ and $\beta_t = \boldsymbol{\beta} \left( \frac{t}{n} \right)$ are changing over time. 

\end{itemize}
Therefore, the OLS estimator might not be suitable since there exists no parameter vector $\boldsymbol{\theta}$ such that $\mathbb{E} \left( y_t | \boldsymbol{X}_{t-1}  \right) = \boldsymbol{X}_{t-1}^{\prime} \boldsymbol{\theta}$ almost surely under the alternative hypothesis, where $\boldsymbol{X}_{t-1} = ( 1, \boldsymbol{x}_{t-1}^{\prime} )^{\prime}$ such that $\boldsymbol{\theta} = \left( \alpha, \boldsymbol{\beta}^{\prime} \right)^{\prime}$. What about a corresponding IVX estimator? Can it consistently estimate the time-varying parameter $\theta_t$ given the assumptions of the predictive regression modelling framework? 

Provide a suitable asymptotic theory analysis, clearly indicating assumptions and theorems employed for the literature that justify the use of a parametric or nonparametric estimation approach with desirable statistical properties. You can also provide a small Monte Carlo simulation study. 
\end{exercise}

\newpage 

\begin{exercise}
Consider the following predictive regression model 
\begin{align}
y_t &= \beta_1 x_{1,t-1} + \beta_2 x_{2,t-1} + \beta_3 x_{3,t-1} + \epsilon_t \\
x_t &= 
\begin{pmatrix}
1 - \frac{c_1}{n} & 0    & 0     \\
0    & 1 - \frac{c_2}{n} & 0     \\
0    & 0    & 1 - \frac{c_3}{n}
\end{pmatrix}
x_{t-1} + u_t \\
u_t &= 
\begin{pmatrix}
0.28 & 0    & 0     \\
0    & 0.32 & 0     \\
0    & 0    & -0.14 
\end{pmatrix}
u_{t-1} + v_t  
\end{align}

\begin{itemize}
\item Write an R Script that simulates the vector $\mathsf{data} := [ y_t, \mathbf{x}_t ]$ based on the above DGP.

\item Hence, using suitable values for the sample size $n$ and for the nuisance parameters of persistence $(c_1, c_2, c_3)$ such that regressors are either stationary or mildly explosive, demonstrates whether or not the predictive accuracy of the model is improved in the case when all regressors are stationary in comparison to the case of mixed integration order. 
\end{itemize}
\end{exercise}

\begin{exercise}
Consider the following predictive regression models
\begin{align}
\label{model1}
y_{t} &= \boldsymbol{x}_{1t-1}^{\prime} \boldsymbol{ \delta }_1 + u_{t}
\\
\label{model2}
y_{t} &= \boldsymbol{x}_{1t-1}^{\prime} \boldsymbol{ \beta }_1 + \boldsymbol{x}_{2t-1}^{\prime} \boldsymbol{ \beta }_2 + v_{t} 
\end{align}
where the resulting out of sample forecast errors are obtained as 
\begin{align}
\hat{e}_{1,t} = y_{t} - \boldsymbol{x}_{1t-1}^{\prime} \hat{ \boldsymbol{ \delta } }_{1t}  \ \ \ \ \text{and} \ \ \ \ \hat{e}_{2,t} = y_{t} - \boldsymbol{x}_{t-1}^{\prime} \hat{ \boldsymbol{ \beta } }_{t}
\end{align} 
Denote with $\boldsymbol{x}_{t-1} = \left( \boldsymbol{x}_{1t-1}^{\prime}, \boldsymbol{x}_{2t-1}^{\prime} \right)$ and $\boldsymbol{\beta} = \left( \boldsymbol{\beta}_{1}^{\prime}, \boldsymbol{\beta}_{2}^{\prime} \right)$ and set $p = p_1 + p_2$. The nested environment is commonly employed for predictive accuracy testing (see \cite{hansen2015equivalence}). Therefore, one step ahead forecasts of $y_{t+1}$ based on the econometric specifications given by expressions \eqref{model1} and \eqref{model2} are generated recursively as $\hat{y}_{1,t|t-1} = \boldsymbol{x}_{1t-1}^{\prime} \hat{ \boldsymbol{ \delta } }_{1t}$ and $\hat{y}_{2,t|t-1} = \boldsymbol{x}_{t-1}^{\prime} \hat{ \boldsymbol{ \beta} }_{t}$, for $t \in \left\{ k_0, ..., T - 1 \right\}$, 
\begin{align}
\hat{ \boldsymbol{ \delta } }_{1t} &= \left( \sum_{j=1}^t \boldsymbol{x}_{1j-1} \boldsymbol{x}_{j-1}^{\prime} \right)^{-1} \left( \sum_{j=1}^t \boldsymbol{x}_{1j-1}^{\prime} y_j \right)
\\
\hat{ \boldsymbol{ \beta} }_{t} &= \left( \sum_{j=1}^t \boldsymbol{x}_{j-1} \boldsymbol{x}_{j-1}^{\prime} \right)^{-1} \left( \sum_{j=1}^t \boldsymbol{x}_{j-1}^{\prime} y_j \right)
\end{align}
Moreover, based on the following test statistic provide suitable asymptotic theory analysis. 
\begin{align}
T_n = \frac{1}{ \hat{\sigma}^2_{\epsilon} } \left\{ \sum_{t = n_{\kappa} + 1 }^n \big( y_t - \tilde{y}_{t|t-1} \big)^2 - \big( y_t - \hat{y}_{t|t-1} \big)^2 \right\}  
\end{align} 
\end{exercise}

\newpage 

\begin{exercise}
Consider the structural change monitoring environment of \cite{chu1996monitoring} 
\begin{align}
y_t = x_t^{\prime} \beta + \epsilon_t, \ t =1,...,n  
\end{align}
Under the null the parameter vector is constant $(\beta_t = \beta \ \forall t)$ against the alternative of temporary shifts. The alternative of a single structural change occurs at an unknown change point $k > n$. 
\begin{align}
y_t =
\begin{cases}
   x_t^{\prime}  \beta_1 + \epsilon_t, & \text{if } \ t=1,...,k \\
   x_t^{\prime}  \beta_2 + \epsilon_t, & \text{if}  \  t=k + 1,...,n
\end{cases}
\end{align}
Treating $k \geq n$ then the two model estimates ex-ante and ex-post the structural change are given 
\begin{align}
\hat{\beta}_{1,T}(k) &=  \left(  \sum_{t=1}^k x_t x_t^{\prime}   \right)^{-1}    \left(  \sum_{t=1}^k x_t y_t  \right)   \\
\hat{\beta}_{2,T}(k) &=  \left(  \sum_{t=k+1}^T x_t x_t^{\prime}   \right)^{-1}  \left(  \sum_{t=k+1}^T x_t y_t  \right) 
\end{align}
Moreover, consider the construction of a test statistic based on the difference between the model estimates in-sample and the the full-sample given by
\begin{align}
\underset{ n \leq k \leq T }{ \text{max}  }   \  \frac{k}{\hat{\sigma}_T \sqrt{T}} \norm{  Q^{1/2}_t \left( \hat{\beta}_{1,T}( k )  - \hat{\beta}_T   \right)    }
\end{align}
where $\hat{\sigma}_T$ is an unbiased estimator of $\sigma_2$ and $Q_T = \frac{1}{T} \sum_{i=1}^T x_i x_i^{\prime}$. Then, the model parameter that depends on the window length $h$ is given by
\begin{align}
\tilde{\beta} \big( k, [Th] \big) =  \left(  \sum_{t=k+1}^{k + [Th] } x_t x_t^{\prime}   \right)^{-1}  \left(  \sum_{t=k+1}^{k + [Th] } x_t y_t  \right)   
\end{align}
Therefore, assuming that the structural change occurs in the monitoring period with length exactly the size of the moving window, $\floor{ Th }$, then the proposed class of ME tests against the alternative of structural change has the following form
\begin{align}
\text{ME}_{ T, [Th] } =  \underset{ n \leq k \leq T }{ \text{max}  } \  \frac{ [Th]  }{\hat{\sigma}_T \sqrt{T}} \norm{  Q^{1/2}_t \left( \hat{\beta}_{1,T}( k )  - \hat{\beta}_T   \right)    }
\end{align}

\begin{itemize}
\item[(a)] Consider suitable modifications of the asymptotic theory that corresponds to the above test in order to accommodate the IVX estimator of the predictive regression model below:
\begin{align}
y_t  &= \beta x_{t-1} + u_{0t} \\
x_t  &= \rho x_{t-1} + u_{xt}, \ \ \ \rho = \left( 1 + \frac{c}{n^{\alpha}}  \right)
\end{align}

\newpage

where the IVX-instrumentation is obtained with
\begin{align}
\tilde{z_t} = \sum_{j=1}^t \rho_{nz}^{t-j} \Delta x_j, \ \ \ \rho_{nz} = 1 + \frac{c_z}{n^{\beta}} , \ \text{where } \ \beta \in (0,1), c_z < 0
\end{align}
and the IVX estimator has the following form:
\begin{align}
\hat{\beta}_{ivx} = \frac{ \displaystyle \sum_{t=1}^n \tilde{z}_{t-1} y_t }{ \displaystyle \sum_{t=1}^n \tilde{z}_{t-1} x_{t-1} } = \beta + \frac{ \displaystyle \sum_{t=1}^n \tilde{z}_{t-1} u_{0t} }{ \displaystyle \sum_{t=1}^n \tilde{z}_{t-1} x_{t-1} }
\end{align}
Does the limiting distribution of $n^{ \frac{1 + \beta }{2}} \left( \hat{\beta}_{ME-ivx} - \beta  \right) \implies \tilde{\Psi}$ is mixed Gaussian? Clearly indicate all necessary assumptions and related limit results.

\item[(b)] Consider the following alternative hypothesis which implies the presence of heterogeneous persistence and temporary parameter instability in the predictive regression system
\begin{align}
y_t =
\begin{cases}
   \beta_1  x_{1,t-1}^{\prime}    + u_{0t}, \ x_{1,t}  = \rho_1 x_{1,t-1} + u_{xt} &   \  t=1,...,k      \\
   \beta_2  x_{2,t-1}^{\prime}    + u_{0t}, \ x_{2,t}  = \rho_2 x_{2,t-1} + u_{xt} &   \  t=k + 1,..., k+l  \\
   \beta_1  x_{1,t-1}^{\prime}    + u_{0t}, \ x_{1,t}  = \rho_1 x_{1,t-1} + u_{xt} &   \  t=k + l + 1,...,T
\end{cases}
\end{align}
where $\rho_1 = 1 + \frac{c_1}{n^{\alpha}}$,  $\rho_2 = 1 + \frac{c_2}{n^\alpha}$ with $c_1 < 0$ and $c_2 > 0$ and $\alpha \in (0,1)$ in both cases. Since, we have a moving window in practise the model estimation and IVX instrumentation occurs in each rolling window. Moreover, you may assume that a different instrumentation is constructed for each regime such that $\tilde{z}_{1,t} = z_{1,t} + \frac{c_1}{n^{\beta_1}} \psi^{(1)}_{nt}$ and $\tilde{z}_{2,t} = z_{2,t} + \frac{c_2}{n^{\beta_2}} \psi^{(2)}_{nt}$.

Using the ME-IVX test function which corresponds to a self-normalized test statistic representing a difference of the IVX estimators between the in-sample against the estimator within the monitoring period with a corresponding fixed window of size $h$:   
\begin{align}
M_n (r|h)= \frac{[nh]}{\hat{\sigma} \sqrt{n}} . \Omega^{1/2}_{(n)} . \bigg[ \hat{\beta}_{t,{([n r]-[nh],[nh] )}}^{\text{IVX}} - \hat{\beta}_{t,(n)}^{\text{IVX}}    \bigg] 
\end{align}
for some $h$, $0<h \leq 1$ the moving data window as a percentage of the historical period, obtain suitable asymptotic theory analysis that demonstrates the statistical properties of the test. Relevant studies to check include: \cite{busetti2004tests} and \cite{horvath2020sequential}.
\end{itemize}

\end{exercise}

\newpage 

\section{Resampling for Time Series Regressions}

Seminal studies for aspects related to implementation of resampling methods in time series analysis include \cite{politis1994stationary}, \cite{paparoditis2001tapered, paparoditis2003residual} among others.

\subsection{Preliminary Theory on Bootstrap Methodologies}

Let $\left\{ X_n \right\}_{ n \in \mathbb{Z}}$ be a stationary sequence of random variables with common continuous distribution function $F(t) = \mathbb{P} \left( X_0 \leq t \right)$ on a probability space. Assume $0 \leq X_0 \leq 1$. The empirical process is defined as 
\begin{align}
F_n(t) &= \frac{1}{n} \sum_{i=1}^n \mathbf{1} \left\{  X_i \leq t \right\}
\\
B_n(t) &= \sqrt{n} \big( F_n(t) - F(t)  \big), 
\end{align}
where $F_n(t)$ is the empirical distribution of $\left\{ X_n \right\}$. We have that $B_n(t) \to_d B(t)$ in the spaces $D[0,1]$ endowed with Skorohod topology where $B$ is a Gaussian process specified by $\mathbb{E} \left( B(t) \right) = 0$ for every $t$, and for every $t$ and $s$, 
\begin{align}
cov \big( B(s), B(t) \big) = \sum_{h = - \infty}^{h = + \infty} cov \big( \mathbf{1} \left\{ X_0 \leq s \right\}, \mathbf{1} \left\{ X_k \leq t \right\} \big).     
\end{align}
In addition, 
\begin{align}
\underset{ n \to \infty }{ \text{lim} } \ cov \big( B_n(s), B_n(s) \big) 
\end{align}
The block-based bootstrap estimators of the mean and empirical process are defined as below. Let $k$ and $l$ be two integers such that $n = kl$. Let $T_{n1}, T_{n2},..., T_{nk}$ be i.i.d random variables each having uniform distribution on $\left\{ 1,...,n \right\}$. Define the triangular array $\left\{ X_{ni}, 1 \leq i \leq n + l  \right\}$by $X_{ni} = X_i$ for $1 \leq i \leq n$ and $X_{ni} = X_{i - n }$ for $1 \leq i \leq n + l$. In other words, we extend our sample of size $n$ by another $l$ observations, namely, $X_1,...,X_l$. Then, the bootstrapped estimator of the mean is defined as below
\begin{align}
\bar{X}_n^{*} = \frac{1}{k} \sum_{i=1}^k \frac{1}{l} \sum_{j = T_{ni}}^{T_{ni} + l - 1} X_j,
\end{align} 
and the bootstrapped estimator of the empirical process is given by 
\begin{align}
F_n(t)^{*} =  \frac{1}{k} \sum_{i=1}^k \frac{1}{l} \sum_{j = T_{ni}}^{T_{ni} + l - 1} \mathbf{1} \left\{ X_j \leq t \right\}.
\end{align}

\newpage 

Then, the boostrapped empirical process is defined as 
\begin{align}
B_n(t)^{*} = n^{1/2} \big[ F_n(t)^{*} - F_n(t) \big].
\end{align}

We denote with $\mathbb{E}^{*}$, $var^{*}$ the moments under the conditional probability measure $\mathcal{P}^{*}$ induced by the resampling mechanism, that is, $\mathcal{P}^{*}$, is the conditional probability given $\left( X_1,...,X_n \right)$. In this paper, we focus on the weak convergence of $B_n^{*}(t)$ to $B(t)$ in the Skorohod topology on $D[0,1]$, almost surely.

\paragraph{Simulating the bootstrap distribution}

We follow the steps below to obtain the approximation to the bootstrap distribution $G_n \left( x, \widehat{F}_n \right)$. We consider that $b = 1,...,B$ are the bootstrap replications. 

\begin{enumerate}
\item[Step 1.] Draw a random sample from the cdf $\widehat{F}_n$. The realized sample $x_1^{*b},...., x_n^{*b}$ is your bootstrap sample. 

\item[Step 2.]  Calculate the bootstrap version of $Q_n( x, F)$, which is $Q_n( x^{*b}, F)$. We denote this as short-hand by $Q_n^{*b}$. 

\item[Step 3.]  After repeating steps 1 and 2 B times, then we collect all the estimated bootstrapped quantities $Q_n^{*1},..., Q_n^{*B}$ to calculate a related test statistic as function of $G_n \left( x, \widehat{F}_n \right)$.     
\end{enumerate}

Now, how the sample in Step 1 of Algorithm 1 is drawn, depends on the form of the estimator $\widehat{F}_n$. Two main bootstrap methodologies are commonly used: 

\begin{itemize}
\item[1.] \textbf{Parametric bootstrap:} $F$ is estimated by $F \left( x | \widehat{\theta}_n \right)$. In this case, the bootstrap sample, $X_1^{*}, ...,X_n^{*}$ is drawn from the cdf $F \left( x | \widehat{\theta}_n \right)$, where $\widehat{\theta}_n$ is an estimate of $\theta$ based on the original sample $x_1,...,x_n$.

\item[2.] \textbf{Nonparametric (iid) bootstrap:} $F$ is estimated by EDF $\widehat{F}_n^{E}$. In this case, the bootstrap sample, $X_1^{*}, ...,X_n^{*}$ is drawn from the cdf $\widehat{F}_n^{E}$, which implies that $X_1^{*}, ...,X_n^{*}$ is drawn with replacement from the original sample $x_1,...,x_n$.  
\end{itemize}

\subsubsection{Bootstrap in regression models}

Consider the following regression model for $i = 1,...,n$
\begin{align}
Y_i = \alpha + \beta X_i + \epsilon_i
\end{align}

\paragraph{Pair bootstrap} The pairs bootstrap is a direct extension of the iid bootstrap and is based on the bivariate EDF $\widehat{F}_n (x,y)$ for the sample $(X_1, Y_1),..., (X_n, Y_n)$ are a random sample with bivariate cdf F. The pairs bootstrap builds the bootstrap sample $( X_1^{*}, Y_1^{*}  ),...,( X_n^{*}, Y_n^{*} )$ by drawing pairs with replacement from the sample $(X_1, Y_1),..., (X_n, Y_n)$. It is crucial that $(X_i, Y_i)$ are kept together as a pair, otherwise if $\textbf{Y}^{*}$ and  $\textbf{X}^{*}$ were drawn separately the bootstrap would assume there is no relation between the two.

\newpage

\subsection{On Bootstrap Asymptotics}

The bootstrap is a method for estimating the distribution of an estimator or test statistic by resampling one's data or a model estimated from the data. The methods that are available for implementing the bootstrap and the improvements in accuracy that it achieves relative to first-order asymptotic approximations depend on whether the data are an i.i.d random sample or a time series. If the data are iid, the bootstrap can be implemented by sampling the data randomly with replacement or by sampling a parametric model of the distribution of the data. The distribution of a statistic is estimated by its empirical distribution under sampling from the data or parametric model (see, \cite{hardle2003bootstrap}).  Therefore, the situation is more complicated when the data are a time series because bootstrap sampling must be carried out in a way that suitably captures the dependence structure of the DGP. 

\subsubsection{Why the bootstrap provides Asymptotic Refinements}

The term asymptotic refinements refers to approximations to distribution functions, coverage probabilities that are more accurate than those of first-order asymptotic distribution theory (see, \cite{beran1988prepivoting}). We focus on the distribution function of the asymptotically $N(0,1)$ statistic $T_n = n^{1 / 2} \left( \theta_n - \theta \right) / s_n$.   

Let $\widehat{ \mathcal{P} }$ denote the probability measure induced by bootstrap sampling, and let $\widehat{T}_n$ denote a bootstrap analogue of $T_n$. If the data are i.i.d, then it suffices to let $\widehat{ \mathcal{P} }$ be the empirical distribution of the data. Bootstrap samples are drawn by sampling the data $\left\{ X_i : i = 1,...,n  \right\}$ randomly with replacement. If $\left\{ \widehat{X}_i: i = 1,...,n   \right\}$ is such a sample, then 
\begin{align}
\widehat{T}_n = n^{1 / 2} \left( \widehat{\theta}_n - \theta_n \right) / \widehat{s}_n, 
\end{align}

where $\hat{\theta}_n = \theta( \hat{m}_n )$ and $\hat{m}_n = n^{-1} \sum_{i=1}^n \widehat{X}_i$, and $\widehat{s}_n$ is obtained by replacing the $\left\{ X_i \right\}$ with $\left\{ \hat{X}_i \right\}$ in the formula for $s_n^2$.

\subsubsection{The Block Bootstrap}

The block bootstrap is the best-known method for implementing the bootstrap with time-series data. It consists of dividing the data into blocks of observations and sampling the blocks randomly with replacement. The blocks may be non-overlapping. Define $Y_i = \left\{ X_i,....,X_{i-q} \right\}$. With non-overlapping blocks of length $\ell$, block 1 is observations $\left\{ Y_j : j = 1,..., \ell \right\}$, block 2 is observations $\left\{ Y_{j+1} : j = 1,..., \ell \right\}$, and so forth. The bootstrap sample is obtained by sampling blocks randomly with replacement and laying them end-to-end in the order sampled. The procedure of sampling blocks of $Y_i$'s instead of  $X_i$'s is called the blocks-of-blocks bootstrap (see, \cite{andrews2004block}). The bootstrap estimators of distribution functions are smaller with overlapping blocks that with non-overlapping ones. The rates of convergence of the error made with overlapping and non-overlapping blocks are the same, and theoretical arguments will be based on non-overlapping blocks with the understanding that the results also apply to overlapping blocks (see, \cite{politis1994stationary}).

\newpage

Moreover, the errors made by the stationary bootstrap are larger than those of the bootstrap with non-stochastic  block lengths and either overlapping or non-overlapping blocks. Therefore, the stationary bootstrap is unattractive relative to the bootstrap with non-stochastic block lengths. Regardless of whether the blocks are overlapping or non-overlapping, the block length must increase with increasing sample size $n$ to make bootstrap estimators of moments and distribution functions consistent. The block length must also increase with increasing $n$ to enable the block bootstrap to achieve asymptotically correct coverage probabilities for confidence intervals and rejection probabilities for tests. Thus, when the objective is to estimate a moment of distribution function, the asymptotically optimal block length may be defined as the one that minimizes the asymptotic mean-square error of the block bootstrap estimator.  

\subsubsection{Studentization}

We address the problem of Studentizing $n^{1 / 2} \hat{\Delta}_n = n^{1 / 2} \left[ \theta \left( \hat{m}_n \right) - \theta \left( \hat{E} \hat{m}_n \right) \right]$.  The source of blocking distorts the dependence structure of the DGP. To illustrate the essential issues with a minimum of complexity, assume that the blocks are non-overlapping, $\theta$ is the identity function, and $\left\{ X_i \right\}$ is a sequence of uncorrelated (though not necessarity independent) scalar random variables. For example, many economic time series are martingale difference sequences, and therefore, serially uncorrelated. 

Let $V$ denote the variance operator relative to the process that generates $\left\{ X_i \right\}$. Then, 
\begin{align}
n^{1 / 2} \Delta_n = n^{1 / 2} \left( m_n - \mu \right), \ \ \text{and} \ \  n^{1 / 2} \widehat{ \Delta }_n = n^{1 / 2}  \left( \hat{m}_n - \mu \right)
\end{align}
with $V \left( n^{1 / 2} \Delta_n \right) = \boldsymbol{E} \left( X_1 - \mu  \right)^2$. The natural choice of $s_n^2$ is the sample variance, $s_n^2 = n^{-1} \sum_{i=1}^n \left( X_i - \mu \right)^2$, in which case $
s_n^2 - \text{var} \left(  n^{1 / 2} \Delta_n \right) = \mathcal{O}_p \left( \Delta_n \right) = \mathcal{O}_p \left( n^{ - 1 / 2} \right)$. Let $\ell$ and $B$ denote the block length and number of blocks respectively, and assume that $B \ell = n$. Let $\hat{V}$ denote the variance operator relative to the block bootstrap DGP. An obvious bootstrap analog of $s_n^2$ is $\hat{s}_n^2 = n^{- 1} \sum_{i = 1}^n \left( \hat{X}_i - \hat{m}_n \right)^2$, which leads to the Studentized statistic, $\widetilde{T}_n \equiv n^{1 / 2} \hat{\Delta}_n / \hat{s}_n$. However, we have that $\widehat{V} \left( n^{1 / 2}  \hat{\Delta}_n \right) = \tilde{s}_n^2$, where 
\begin{align}
\tilde{s}_n^2 = n^{- 1} \sum_{b = 0}^B \sum_{i = 1}^{\ell} \sum_{j=1}^{\ell} \left( X_{b \ell} - m_n \right) \left( X_{b \ell + j} - m_n  \right)
\end{align} 
Moreover, $\hat{s}_n^2 - \tilde{s}_n^2 = \mathcal{O}_p \left( \left( \ell / n \right)^{1 / 2} \right)$ almost surely. The consequences of this relatively large error in the estimator of the variance of $n^{1 / 2} \hat{\Delta}_n$ can be seen by carrying out Edgeworth expansions of 
\begin{align}
\mathbb{P} \left( n^{1 / 2} \Delta_n / s_n \leq z \right) \ \ \ \hat{ \mathbb{P} } \left( n^{1 / 2} \hat{\Delta}_n / \hat{s}_n \leq z \right)
\end{align}
This problem can be mitigated by Studentizing $n^{1 / 2} \hat{\Delta}_n$ with $\tilde{s}_n$ or the estimator 
\begin{align}
n^{- 1} \sum_{b = 0}^{B - 1} \sum_{i = 1}^{\ell} \sum_{j = 1}^{\ell} \left( X_{b \ell + i } - \hat{m}_n \right) \left( X_{b \ell + j} - \hat{m}_n  \right) 
\end{align}

\newpage

\subsubsection{The Sieve Bootstrap for Linear Processes}

A substantial improvement over the performance of the block bootstrap is possible if the DGP is known to be a linear process. That is, the DGP has the form
\begin{align}
\label{sieve.spec}
X_i - \mu = \sum_{j=1}^{\infty} \alpha_j \left( X_{i - j} - \mu \right) + U_i, 
\end{align} 
where $\mu = \mathbb{E} \left( X_i \right)$ for all $i$, $\left\{ U_i \right\}$ of \textit{i.i.d} random variables, and $\left\{ X_i \right\}$ may be a scalar or a vector process. Assume that $\sum_{j=1}^{\infty} \alpha_j^2 < \infty$ and all of the roots of the power series $1 - \sum_{j=1}^{\infty} \alpha_j z^j$ are outside of the unit circle. For instance, \cite{paparoditis1996bootstrapping} proposed approximating \eqref{sieve.spec} by an AR$(p)$ model in which $p = p(n)$ increases with increasing sample size.      

Let $\left\{ a_{nj}: j = 1,...,p \right\}$ denote least squares or Yule-Walker estimates of the coefficients of the approximating process, and let $\left\{ U_{nj} \right\}$ denote the centered residuals. The sieve bootstrap consists of generating bootstrap samples according to the process
\begin{align}
\hat{X}_i - m = \sum_{j=1}^p a_{nj} \left( \hat{X}_{i-j} - m \right) + \hat{U}_j, 
\end{align}
where $m = n^{- 1} \sum_{i = 1}^n X_i$ and the $\hat{U}_j$ are sampled randomly with replacement from the $U_{nj}$. 

Furthermore, \cite{paparoditis2018sieve} considers the Sieve bootstrap for functional time series. In particular, bootstrap procedures for Hilbert space-valued time series proposed so far in the literature, are mainly attempts to adapt, to the infinite dimensional functional framework, of bootstrap methods that have been developed for the finite dimensional time series case.

\subsubsection{Discussion on Bootstrap Asymptotics}

A different not so popular approach when deriving asymptotics for bootstrapped-based estimators and test statistics is the use of the Stein's approximation. In terms of the implementation of Stein’s method via econometric models, \cite{bandi2007simple} consider a parametric estimation method that matches parametric estimates of the drift and diffusion functions to their functional counterparts. More precisely, the authors consider a corresponding estimate of the distance between the process and the limiting Wiener process, with respect to a certain metric. On the other hand, it is argued that the Stein’s method proves weak convergence type results through approximations to expectations without direct use of characteristic functions, allowing it to be used in complex problems with dependence. Consequently, in the particular paper the authors show that consistency can be proved for any type of bootstrap with exchangeable weights using Stein’s method. Furthermore, \cite{barbour1990stein} employs Stein’s method to derive rates of convergence to a Wiener process limits. These limit results demonstrate that the Stein operator can be utilized to find the solution  of an equilibrium distribution based on a collection of OU processes.

\newpage

\subsection{Bootstrapping in Time Series Regressions}

The bootstrap proposed by Efron (1979) has proven to be a powerful nonparametric tool for approximating the sampling distribution and variance of a statistic of interest. In particular, the bootstrap estimates the asymptotic distribution of $\tau_n$, that is, 
\begin{align}
G_n^{*} (x) := \mathbb{P}^{*} ( \tau_n \leq x ) \to_p G_{\infty} ( x ) 
\end{align}

\begin{example}
Consider the following triangular array (see, \cite{parker2006unit}). Previous results on the residual-based block bootstrap are established by  \cite{paparoditis2003residual})
\begin{align}
x_{t,n} = \rho_n x_{t-1,n} + u_t, \ \ \ \rho_n = \left( 1 + \frac{c}{n} \right) , c < 0 
\end{align}
Denote with $y_t = \displaystyle \sum_{j=1}^t \hat{u}_j$. Furthermore, recall that 
\begin{align}
S_n(r) = S_n \left( \frac{t-1}{n} \right), \ \ \ \text{where} \ \ \frac{t-1}{n} \leq r \leq \frac{t}{n}. 
\end{align}
Then, it follows from Theorem 1 and the continuous mapping theorem:
\begin{align*}
\frac{1}{n^2} \sum_{t=2}^n y^2_{t-1} 
&= 
\frac{1}{n^2} \sum_{t=2}^n \left( \sum_{j=1}^t \hat{u}_j \right)^2 
= 
\frac{1}{n} \sigma_n^2 \sum_{t=2}^n S^2_n \left( \frac{t-1}{n} \right) dr
\\
&=
\sigma_n^2 \sum_{t=2}^n  \int_{(t-1)/n}^{t / n} S^2_n (r) dr = \sigma_n^2 \int_0^1 S^2_n (r) dr 
\overset{ d }{ \to } \int_0^1 W^2(r) dr. 
\end{align*}
Notice that to obtain the limit results as shown in the example above, we consider the behaviour of the continuous standardized partial sum process $\big\{ S_n(r), 0 \leq r \leq 1 \big\}$ given by 
\begin{align}
S_n(r) &= \frac{1}{ \hat{ \sigma}_n^2 } \frac{1}{ \sqrt{n}} \sum_{j=1}^{\floor{nr} } \hat{u}_j,
\\
\hat{ \sigma}_n^2 &= \text{Var} \left( \frac{1}{ \sqrt{n}} \sum_{j=1}^{\floor{nr} } \hat{u}_j  \right).  
\end{align}
To establish asymptotic theory results one considers the limit behaviour of the following expression
\begin{align}
\begin{bmatrix}
n^{- 1/ 2} \hat{\beta}
\\
\\
n \left( \hat{\rho} - 1 \right)
\end{bmatrix}
= 
\begin{bmatrix}
1   &   \displaystyle \frac{1}{ n^{3 / 2} } \sum_{t=2}^n x_{t-1} 
\\
\\
\displaystyle \frac{1}{ n^{3 / 2} } \sum_{t=2}^n x_{t-1} & \displaystyle \frac{1}{ n^2 } \sum_{t=2}^n x^2_{t-1}
\end{bmatrix}
\end{align}
\end{example}

\newpage

\begin{example}[$I(1)$ Bootstrap Samples]
Estimate the regression 
\begin{align}
\Delta y_t = \sum_{ j = 1}^{ \ell_{\tau} } \pi_j \Delta y_{t - j } + e_t^{*} , \ \ t = \ell_T + 2,...., T
\end{align}

Consider the following wild bootstrap (see, \cite{cavaliere2015bootstrap} and \cite{cavaliere2020inference})  
\begin{align}
\tilde{\tau}_n := a_n^{-1} \sum_{t=1}^n \epsilon_t^{*} 
\end{align}

where $\epsilon_t^{*} = \hat{\epsilon}_t w_t^{*}$, $w_t^{*} \sim (0,1)$. It can be shown that 
\begin{align}
P^{*} \left( \tilde{ \tau }_n^{*} \leq x \right) = \mathbb{P} \left( \frac{1}{a_n} \sum_{t=1}^n \hat{\epsilon}_t w_t^{*} \leq x \bigg| \left\{ \epsilon_t \right\} \right) \overset{ w }{ \to } \mathbb{P} \left( \sum_{t=1}^{\infty} \delta_t Z_t \leq x \bigg| Z \right)
\end{align}
\end{example}
The bootstrap mimics a particular conditional distribution of the original statistic. In classic bootstrap inference, asymptotic bootstrap validity is understood and established as convergence in probability (or almost surely) of the CDF of $\tau$, say $F$. This convergence, along with continuity of $F$ implies that sup$_{ x \in \mathbb{R} } \left| F_n^{*}(x) - F(x) \right| \to 0$, in probability (or almost surely). In other words, it is often the case that bootstrap validity can be addressed through the lens of a conditioning argument. In this regard, we consider related asymptotic results which concern the probability that, for a sequence of random elements $X_n$, it holds that the bootstrap $p-$value is uniformly distributed in large sample conditionally on $X_n$
\begin{align}
\mathbb{P} \left( p_n^{*} \leq q | X_n \right) \overset{ p }{ \to } q, \ \ \ \forall \ q \in (0,1).
\end{align} 
This property that we call "bootstrap validity conditional on $X_n$", implies unconditional validity. Thus, conditional bootstrap validity given $X_n$ implies that the bootstrap replicates asymptotically the property of conditional tests and confidence intervals to have, conditionally on $X_n$, constant null rejection probability and coverage probability respectively. When dealing with random limit distributions, the usual convergence concept employed to establish bootstrap validity, that is, weak convergence in probability, can only be used in some special cases. Instead, we use the probabilistic concept of weak convergence of random measures. As $F$ is sample path-continuous then, $\left( \tau_n, F_n^{*} \right) \overset{ p }{ \to } \left( \tau, F \right)$ on the Skorokhod-representation space, whereas $\left( \tau_n, F_n^{*} \right) \overset{ w }{ \to } \left( \tau, F \right)$ on a general probabilistic space. In other applications, the idea of similar proof would be to choose $X_n$ through $D_n-$measurable random elements such that $\tau_n^{*}$ depends on the data essentially through $X_n$ (see, \cite{cavaliere2020inference}). Further details on the bootstrap for nonstationary autoregressions can be found in \cite{katsouris2023bootstrapping} (see, also \cite{reichold2022bootstrap}).

Overall, the following bootstrap consistency result holds
\begin{align}
\underset{ x \in \mathbb{R} }{ \text{sup} } \big| \mathbb{P}^{*} \left( \tau_n^{*} \leq x \right) - \mathbb{P} \left( \tau_n^{*} \leq x \right) \big| \to_p 0
\end{align}
Therefore, the bootstrap p-value satisfies $p_n^{*} := \mathbb{P}^{*} \left( \tau_n^{*} \leq x \right) |_{ x = \tau_n } = G_n^{*} ( \tau_n ) \to_w U[0,1]$.

\newpage

\begin{example}
Consider the first order autoregression with a unit root given by 
\begin{align}
y_t = \rho y_{t-1} + \epsilon_t, \ \ \rho = 1 \ \ \epsilon_t \sim \ i.i.d \left( 0, \sigma^2 \right)
\end{align} 
Let $J_c(r)$ represent the OU process with mean reversion parameter $c$. We define with 
\begin{align}
H_c(x) := \mathbb{P} \left( \frac{ \displaystyle \int_0^1 J_c d J_c }{ \displaystyle \int_0^1 J_c^2 } \leq x \right)
\end{align} 
Let $\tau_n := n \left( \hat{\rho} - 1 \right) \to \tau_{\infty} :=  \displaystyle \int_0^1 J_0 d J_0 \bigg/ \displaystyle \int_0^1 J_0^2 du$, that is, $\mathbb{P} \left( \tau_n \leq x \right) \to H_0 (x)$.
\end{example}

\begin{example}
Consider the location model with non-stationary stochastic volatility 
\begin{align}
y_t = \mu + \sigma_t \epsilon_t 
\end{align}
Then, the assumption of non-stationary stochastic volatility implies that, $\sigma_n(r)$ on $\mathcal{D}[0,1]$ is the approximate of $\left\{ \sigma_t \right\}$, that is, $\sigma_n(.) := \sigma_{ [ nr ] } \ \text{with} \ r \in [0,1]$. It can be proved that the standardized estimator of the mean is mixed normal, as $n \to \infty$, that is, 
\begin{align}
\tau_n := \sqrt{n} \left( \hat{\mu}_n - \mu \right) = n^{-1 / 2} \sum_{ t= 1}^n \sigma_t \epsilon_t \to \tau_{ \infty } = \int_0^1 \sigma( u ) dB(u) \equiv \mathcal{MN} \left( 0, V \right)
\end{align}
with $V = \displaystyle \int_0^1 \sigma^2 ( u ) du$.
\end{example}

\begin{example}
Consider the Gaussian Wild Bootstrap, given by the following form
\begin{align}
y_t^{*} := \epsilon_t \eta_t^{*}, \ \ \eta^{*}_t \sim^{\text{i.i.d}} N(0,1) 
\end{align}
and the bootstrap statistic given by 
\begin{align}
\tau_n^{*} := n^{-1 / 2} \sum_{t=1}^n y_t^{*} = n^{-1 / 2} \sum_{t=1}^n  \epsilon_t \eta_t^{*}
\end{align}
Conditional on the original data, with $\hat{V}_n := n^{-1} \sum_{t=1}^n  \epsilon_t$, $\tau_n^{*} | \left\{ y_t \right\} \sim N \left( 0, \hat{V}_n \right)$, 
or equivalently, denoting with $\Phi$ the CDF for $N(0,1)$ we have that
\begin{align}
\mathbb{P}^{*} \left( \tau_n^{*} \leq x \right) = \mathbb{P}^{*} \left( N(0, \hat{V}_n ) \leq x \right) = \Phi \left( \hat{V}_n^{-1 / 2}  x \right), x \in \mathbb{R}.
\end{align} 
Therefore, since $\hat{V}_n \to V:= \displaystyle \int \sigma^2 (u) du$, we have a random limit given by
\begin{align}
\mathbb{P}^{*} \big( \tau_n^{*} \leq x \big) \to \Phi \left( V_n^{-1 / 2} x \right) = \mathbb{P} \big( N(0,V) \leq x | V \big), \ x \in \mathbb{R} \ \text{which implies that} \ \tau_n^{*} \to N(0,V) | V.
\end{align}
 
\end{example}

\newpage 

\section{Open Problems}

Many open problems remain in the econometrics and statistics literature. We give emphasis to the aspects related to the development of an econometric environmnet that incorporates network dependence and time series nonstationarity. Below we present related literature and ask some relevant research questions. 

\subsection{Literature on Measures of Connectedness}

Financial connectedness is a crucial concept in understanding the structure of financial systems as well as to monitor financial contagion, spillover effects and the amplification of systemic risk. In particular, the rise of systemic risk is considered to be due to shocks transmitted via system-wide connectedness. Moreover, such type of connectedness can be traced in the long memory properties of time series within a network and therefore appropriate modelling methodologies require to capture the evolution of spillover effects and the network's spectral response within a infinite time-horizon (e.g., \cite{schennach2013long}). On the other hand, if we are only interested on capturing financial contagion then for example, \cite{glasserman2015likely} provide a framework which depends less on the network topology and link formation. However, our approach aims to include such information and additionally provide a more robust representative of the interactions within financial networks by incorporating features of time series.  

Three important lines of literature provide an in depth examination of this concept. Firstly, \cite{diebold2014network}, introduce a robust framework to study network connectedness via the use of forecast error variance decomposition (FEVD) of an estimated VAR process, as a transformation method that summarizes connectedness between a set of institutions. The significance of their methodology is the provision of directional connectedness measures which are analogous to bilateral imports and exports for each of a set of $N$ countries. Consider, an $N$ dimensional covariance-stationary data generating process (DGP) which describes the evolution of the economic network under examination within a VAR(p) topology, i.e., $\boldsymbol{Y_t}=(y_{1,t},...,y_{N,t})$ for $t=1,...,T$. Then, by the Wold decomposition theorem gives $\boldsymbol{\Phi(L)} \boldsymbol{Y}_t = \boldsymbol{\epsilon_t}$, where $\boldsymbol{\Phi(L)}=\sum_h \boldsymbol{\Phi}_h L^h$ and assuming that $|\Phi(z)|$ lie outside the unit-circle, the VAR process has the MA($\infty$) representation, $\boldsymbol{Y_t}= \boldsymbol{\Theta(L)}\boldsymbol{\epsilon_t}$, where $\boldsymbol{\Theta(L)}$ is an $(N x N)$ matrix of infinite lag polynomials. Note that many of the aspects of connectedness are described with this topology. Also, $\boldsymbol{\epsilon}_t$ is considered to be a white-noise generated by the covariance matrix $\boldsymbol{\Sigma}$. Therefore, the H-step generalized variance decomposition matrix D$^{gH}$ = $[d_{ij}^{gH}]$ is given by \footnote{See also, \cite{lanne2016generalized} }
\begin{align}
d_{ij} = \frac{ \sigma_{ij}^{-1} \sum_{h=0}^{H-1} (e'_i\boldsymbol{\Theta}_h \boldsymbol{\Sigma} e_j   )^2   }{ \sum_{h=0}^{H-1} (e'_i\boldsymbol{\Theta}_h \boldsymbol{\Sigma} \boldsymbol{\Theta}^{'}_h e_i   )  }
\end{align}
where $e_j$ are selection matrices e.g., $e_j=(0,...1,...0)^{'}$, $\boldsymbol{\Theta}_h$ is the coefficient matrix evaluated at the $h$-lagged shock vector and $\sigma_{ij}$ is the $j$th diagonal element of $\boldsymbol{\Sigma}$.

\newpage

The particular methodology is build on the previous work of \cite{diebold2012better}, who introduce generalized spillover measurements, separating directional spillovers (i.e. how much of the shocks to the volatility spillovers across major asset classes) and net pairwise spillovers (i.e., how much each market contributes to the volatility in other markets). The H-step variance decomposition gives the fraction of variable $i$'s H-step forecast error variance due to shocks in variable $j$ and practically is associated with connectedness measures since it provides information regarding the connectedness among the components of the series (e.g., as a network adjacency matrix).  

Secondly, \cite{billio2012econometric}, examine the applications of two econometric methodologies as measures of connectedness  which are designed to capture changes in correlation and causality among financial institutions, that is, (i) principal components (PCA) and (ii) Linear Granger causality. Therefore, using the PCA methodology, the authors observe that during periods when a subset of principal components has an explanatory power greater than a threshold $H$ of the total volatility in the system this is indicative of increased interconnectedness between financial institutions. Furthermore, if we let $\lambda_j$ to represent the $j$th eigenvalue of the diagonal matrix of loadings $\Lambda$, then the critical threshold of increased interconnectedness is defined to be
\begin{align}
 \frac{\sum_{j=1}^n \lambda_j}{\sum_{j=1}^N \lambda_j} \equiv h_n \geq H
\end{align}  
Moreover, there is a mapping between the contribution and exposure of institution $j$ to the risk of the system, which is given by
\begin{align}
 \text{PCAS}_{j,n} = \frac{1}{2} \frac{\sigma^2_j }{\sigma^2_S} \frac{ \partial  \sigma^2_S }{\partial \sigma^2_S } \biggr \rvert_{h_n \geq H} = \sum_{k=1}^n \frac{\sigma^2_j }{\sigma^2_S} L^2_{jk} \lambda_k  \biggr \rvert_{h_n \geq H}
\end{align}    
The above measure gives a degree of connectedness, thus in order to also measure the directionality of connectedness the authors use the linear Granger causality. Linear Granger causality as first introduced by \cite{granger1969investigating, granger1980testing} has been further tested in various frameworks (e.g., VAR(p) models) and for different moments (e.g., first or higher order). For example, let two time series be $\{Y_{i,t}\}_{t=1}^T$ and $\{Y_{j,t}\}_{t=1}^T$ with the corresponding information set available at time $t-1$ be $I_{t-1} \equiv ( I_{i, t-1}, I_{j, t-1} )$. Then we say that if there are enough statistical evidence to support the hypothesis 
\begin{align}
H_1: F( Y_{i,t} | I_{i, t-1} ) \neq  F( Y_{i,t} | I_{t-1} ) ,  \text{for} \ \forall \ F( .| I_{t-1}) \in \mathbb{R} 
\end{align}
then there is additional information in the set of $Y_{j,t-1}$  which helps predict the future values of $Y_{i,t-1}$, that is, time series $\{Y_{j,t}\}_{t=1}^T$ is said to "Granger-cause" $\{Y_{i,t}\}_{t=1}^T$. Based on this idea, the authors introduce various measures of directionality connectedness such as degree of Granger Causality, number of connections, sector-conditional connections, closeness and centrality. The authors also present an extension of this framework to the case of non-linear Granger Causality which captures volatility-based interconnectedness.

\newpage

\subsection{Further Research}

The network and multivariate time series literature have seen plethora of independent applications in recent years, but a unified framework bridging the gap remains an open research question. In this course we have studies various estimation and inference methods from the time series econometrics literature as well as useful limit theorems relevant to network econometrics. Therefore, the study statistical estimation methodologies under network dependence and nonstationarity can indeed be fruitful. A crucial component in this stream of literature is the identification of network structures via shock propagation. Specifically, one can consider a single network with directed links. Although most of the current literature considers some dependence structure for time series data when considering nonstationary processes, some form of weak dependence provides enough information for identification purposes. This motivates the study of cointegration dynamics under network dependence. 

More specifically, \cite{schennach2018long} consider a diffusion process to identify shocks across the nodes of a network. Furthermore, recently \cite{bykhovskaya2022time} consider an estimation and prediction framework where the model specifies the temporal evolution of a weighted network that combines classical autoregression with non-negativity, a positive probability of vanishing and peer effect interactions between weights assigned to edges in the process. Although nonstationarity in time series can be interpreted as time-varying model parameters, we shall focus on cointegrating and unit root dynamics when the underline data structure is defined across nodes of a single network (graph). In other words, we are interested for conditions of network stationarity against explosiveness of the underline network evolution process.  

Consider distance as a measure of dependence, so-called $m-$dependence case. A sequence of random variables $\left\{ X_i, i \geq 1 \right\}$ is a sequence of $m-$dependent random variables if the two sets of random variables $\left\{ X_1, X_2,..., X_i \right\}$ is a sequence of $m-$dependent random variables if the two sets of random variables $\left\{ X_1,..., X_i \right\}$ and $\left\{ X_j, X_{j+1},..., X_n \right\}$ are independent whenever $j - i > m$. Therefore, interesting research questions we are interested to examine is how $m-$dependence expressed in terms of the distance of nodes in a graph can affect the stability of the system as approximated by the persistence properties of regressors as well as the degree of network dependence.

\newpage

\appendix
\numberwithin{equation}{section}
\makeatletter

\section{Elements of Stochastic Processes}

\subsection{Probability Theory}

\begin{proposition}[Martingale Difference CLT] In the notation of the previous theorem, let $\left\{ X_{in}: i \leq n \right\}$
\begin{align}
X_{in} = \frac{ X_i - \mathbb{E} ( X_i ) }{ \sqrt{n} } 
\end{align}
and suppose that for each $n$, $\left\{ \left( X_{in}, \mathcal{A}_{in} \right) : i \leq n \right\}$ is a martingale difference sequence satisfying the Linderberg condition, such that
\begin{align}
W_n = \sum_{i=1}^n \mathbb{E} \left[ X_{in}^2 \boldsymbol{1} \left\{  | X_{in} | \geq \frac{1}{r} \right\} \big| \mathcal{F}_{i-1,n} \right]    
\end{align}
The following statements hold true:
\begin{itemize}
    \item[(i)] If $\displaystyle  \sum_{i=1}^n \mathbb{E} \left[ X_{in}^2 | \mathcal{F}_{i-1,n} \right] \to \sigma^2$, then $\displaystyle z_n = \sum_{i=1}^n X_{in}^2 \overset{d}{\to} \zeta \sim \mathcal{N} \left( 0, \sigma^2 \right)$

    \item[(ii)] If $\sum_{i=1}^n X_{in}^2 \overset{ p }{\to} \sigma^2$, then $z_n \to \zeta \sim \mathcal{N} ( 0, \sigma^2 )$. 
    
\end{itemize}
\end{proposition}
The stochastic process $X: \mathbb{R}_{+} \times \Omega \mapsto \mathcal{B} \left( \mathbb{R}_{+} \right) \otimes \mathcal{E}$ is measurable when this mapping is measurable, that is, for all $A \in \mathcal{B} \left( \mathbb{R} \right)$, 
\begin{align}
\left\{ (t, \omega) : X(t, \omega ) \in A \right\} \in \mathcal{B} \left( \mathbb{R} \right),
\end{align}
Therefore, as a consequence of this measurability and Fubini's theorem, $X(., \omega): \mathbb{R}_{+} \mapsto \mathbb{R}$ is almost surely measurable, while for measurable functions $h: \mathbb{R} \mapsto  \mathbb{R}$,
\begin{align}
Y( \omega ) \equiv \int_{ \mathbb{R}_{+} } h \left(  X( t , \omega)   \right) dt
\end{align}  
is a random variable provided this integral exists. A stochastic process on $\mathbb{R}_{+}$, is necessarily measurable, if for example, the trajectories are either almost surely continuous or almost surely monotonic and right-continuous. We are interested to examine the evolution of a stochastic process, that is, we observe $\left\{ X(s,\omega): 0 < s \leq t \right\}$ for some (unknown) $\omega$ and finite time interval $(0,t]$. It is then natural to consider the $\sigma-$algebra, $
\mathcal{F}_t^{(X)} \equiv \sigma \left\{ X( s, \omega ): 0 < s \leq t    \right\}$, generated by all possible such evolutions. 

Clearly, $
\mathcal{F}_s^{(X)} \subset \mathcal{F}_t^{(X)}$,
for $0 < s < t < \infty$. Generally speaking, an expanding family $\boldsymbol{ \mathcal{F} } = \left\{  \mathcal{F}_t:  0 \leq t < \infty   \right\}$ of $\sigma-$algebras of $\mathcal{E}$ is called a filtration or history. Suppose $X$ is measurable and $\mathcal{F}_t-$adapted. A stronger condition to impose on $X$ is that of progressive measurability with respect to $\mathcal{F}$, meaning that for every $t \in \mathbb{R}_{+}$ and any $A \in \mathcal{B} \left( \mathbb{R} \right)$, 
\begin{align}
\left\{ (s, \omega) : ) < s \leq t, X(s, \omega ) \in A \right\} \in \mathcal{B} \left( (0,t] \right) \times \mathcal{F}_t.
\end{align}

\newpage

Therefore, given any measurable $\mathcal{F}-$adapted $\mathbb{R}-$valued process $X$, we can find an $\mathcal{F}-$progressively measurable process $Y$, that is modification of X, in the sense of being defined on $( \Omega, \mathcal{E}, \mathbb{P} )$ and satisfying 
\begin{align}
\mathbb{P} \left\{ \omega: X(t, \omega) =  Y(t, \omega) \right\} = 1
\end{align}
Therefore, the sets of the form $[s,t] \times U, 0 \leq s < t, U \in \mathcal{F}_t, t \geq 0$, generate a $\sigma-$algebra on $\mathbb{R}_{+} \times \Omega$, which may be called the $\boldsymbol{ \mathcal{F}}-$progressively measurable may be rephrased as the requirement that $X(t, \omega)$ be measurable with respect to the $\boldsymbol{ \mathcal{F}}-$progressive $\sigma-$algebra. A more restrictive condition to impose on $X$ is that it be $\boldsymbol{ \mathcal{F}}-$predictable. Then, call the sub-$\sigma-$algebra of $\mathcal{B} \left( \mathbb{R}_{+} \right) \otimes \mathcal{E}$ generated by the product sets of the form $(s,t] \times U$, where $U \in \mathcal{F}_s, t \geq s$, and $0 \leq s < \infty$, the predictable $\sigma-$algebra, denoted by $\Psi^{\mathcal{F}}$. This terminology is chosen because it reflects what can be predicted at some future time $t$ given the evolution of the process- as revealed by sets $U \in \mathcal{F}_s$-up to the present time $s$. Then $X$ is $\mathcal{F}-$predictable when it is $\Psi^{\mathcal{F}}-$measurable, that is, for any $A \in \mathcal{B} \left( \mathbb{R} \right)$,    
\begin{align}
\left\{ ( t, \omega) : X( t, \omega ) \in A \right\} \in \Psi^{\mathcal{F}}. 
\end{align}
Moreover, the $\mathcal{F}-$predictable process is left-continuous and the left-continuous history $\mathcal{F}_{(-)} \equiv \left\{ \mathcal{F}_{-}  \right\}$ associated with $\mathcal{F}$ appears, here $\mathcal{F}_{0-} = \mathcal{F}_0$ and $\mathcal{F}_{t-} = \text{lim sup}_{s < t} \mathcal{F}_s \vee_{s < t } \mathcal{F}_s$. Notice that if $X( t, \omega )$ is $\mathcal{F}_{t-}-$measurable, its value at $t$ is in fact determined by information at times prior to $t$.   
\begin{definition}
Given a history $\mathcal{F}$, a nonnegative r.v $T: \Omega \mapsto [ 0, \infty ]$ is an $\mathcal{F}-$stopping time if, 
\begin{align}
\left\{ \omega: T ( \omega ) \leq t \right\} \in \mathcal{F}_t , \ \ \ 0 \leq t < \infty. 
\end{align} 
\end{definition}
If $S, T$ are stopping times, then so are $S \vee T$ and $S \wedge T$. Indeed, given a family $\left\{ T_n : n = 1,2,... \right\}$ of stopping times, $\text{sup}_{ n \leq 1 } T_n$ is an $\mathcal{F}-$stopping time, while $\text{inf}_{ n \geq 1} T_n$ is an $\mathcal{F}_{(+)}-$stopping time.   
\begin{lemma}
Let $X$ be an $\mathcal{F}-$adapted monotonically increasing right-continuous process, and let $Y$ be an $\mathcal{F}_0-$measurable r.v. Then $T( \omega ) \equiv \text{inf} \left\{ t : X(t, \omega ) \geq Y( \omega) \right\}$ is an $\mathcal{F}-$stopping time, possibly extended, while if $X$ is $\mathcal{F}-$predictable, then $T$ is an (extended) $\mathcal{F}_{(-)}-$stopping time. 
\end{lemma}
As an important corollary to this result, observe that if $X$ is $\mathcal{F}-$progressive and almost surely integrable on finite intervals, then  
\begin{align}
Y( t, \omega ) = \int_0^t X( s, \omega ) ds
\end{align}
is $\mathcal{F}-$progressive, $Y(T)$ is a r.v if $T < \infty$ and $Y ( t \wedge T )$ is again $\mathcal{F}-$progressive. 

\newpage

\paragraph{Martingales}

\begin{definition}
Let $( \Omega, \mathcal{E}, \mathbb{P} )$ be a probability space, $\mathcal{F}$ a history on $( \Omega, \mathcal{E} )$ and $X(.) \equiv \left\{ X(t) : 0 \leq t < \infty \right\}$ a real-valued process adapted to $\mathcal{F}$ and such that $\mathbb{E} \left( | X(t) | \right) < \infty$ for $0 \leq t < \infty$. Then, $X$ is an $\mathcal{F}-$martingale if $0 \leq s < t < \infty$, 
\begin{align}
\mathbb{E} \left[ X(t) | \mathcal{F}_s \right] = X(s) \ \text{almost surely},
\end{align}
and an $\mathcal{F}-$submartingale if 
\begin{align}
\mathbb{E} \left[ X(t) | \mathcal{F}_s \right] \geq X(s) \ \text{almost surely},
\end{align}
and an $\mathcal{F}-$supmartingale if the reverse inequality holds. 
\end{definition}

\begin{theorem}
(Doob-Meyer). Let $\mathcal{F}-$ be a history and $X(.)$ a bounded $\mathcal{F}-$submartingale with right-continuous trajectories. Then, there exists a unique (up to equivalence) uniformly integrable $\mathcal{F}-$martingale $Y(.)$ and a unique $\mathcal{F}-$predictable cumulative process $A(.)$ such that 
\begin{align}
X(t) = Y(t) + A(t). 
\end{align}  
\end{theorem}
Therefore, for nondecreasing process $A(.)$ with right-continuous trajectories, it can be shown that $\mathcal{F}-$predictability is equivalent to the property that for every bounded $\mathcal{F}-$martingale $Z(.)$ and positive $u$, 
\begin{align}
\mathbb{E} \left[ \int_0^u Z(t) A(dt) \right] = \mathbb{E} \left[ \int_0^u Z(t-) A(dt) \right]. 
\end{align}
Since for any $\mathcal{F}-$adapted cumulative process $\xi$ and any $\mathcal{F}-$martingale $Z$, $\mathbb{E} \left[ Z(u) \int_0^u \xi ( dt )   \right] = \mathbb{E} \left[ Z(t) \int_0^u \xi ( dt ) \right]$, the property above is then equivalent to 
\begin{align}
\mathbb{E} \left[ Z(u) A(u) \right] = \mathbb{E} \left[ \int_0^u Z(t-) A(dt) \right]. 
\end{align} 
A cumulative process with this property is referred to in many texts as a natural increasing process. The theorem can then be reshaped as: every bounded submartingale has a unique decomposition into the sum of a uniformly integrable martingale and a natural increasing function. Therefore, there is a relation between natural increasing and predictable processes.  
\paragraph{Monotone Convergence Theorem}

\begin{theorem}
Let $\left( X_n \right)_{ n \geq 1}$ be a monotonically sequence, that is, $X_n \leq X_{n+1}$ almost surely. Assume that $\mathbb{E} \left[ | X_n | \right] < \infty$ for every $n \geq 1$. Then $X_n ( \omega ) \to X ( \omega )$ for all $\omega$ and some limiting random variable $X$ (that is possibly degenerate), and
\begin{align}
\mathbb{E} \left[ X_n \right] \to \mathbb{E} \left[ X \right].  
\end{align} 
\end{theorem}

\paragraph{The Supercritical Regime}

The main result in this section, is a law of large numbers for the size of the maximal connected component. Below, we write $\zeta_{\lambda} = 1 - \eta_{\lambda}$, so the following theorem shows that there exists a giant component: 

\begin{theorem}
(LLN for giant components) Fix $\lambda > 1$. Then, for every $\nu \in ( \frac{1}{2}, 1 )$, there exists $\delta = \delta ( \nu, \lambda ) > 0$ such that
\begin{align}
\mathbb{P}_{ \lambda } \bigg( | \mathcal{E} | - \zeta_{\lambda} \eta \geq \eta^{\nu} \bigg) = \mathcal{O} \left( \eta^{-\delta} \right). 
\end{align} 
\end{theorem} 
The Theorem says that: a vertex has a large connected component with probability $\zeta_{ \lambda }$. Therefore, at least in expectation, there are  roughly $\zeta_{ \lambda } \eta$ vertices with large connected components. 
\begin{proof}
(page 130) Here, we give an overview of the proof. We rely on an analysis of the number of vertices in connected components of size at least $k$, 
\begin{align}
Z_{ \geq k} = \sum_{ v \in [n] } \boldsymbol{1}\left\{ | \mathcal{E} (v) | \geq k  \right\}.  
\end{align}
The proof contains 4 main steps. In the first step, for $k_n = K \text{log} n$ and $K$ sufficiently large, we compute
\begin{align}
\mathbb{E}_{ \lambda } \left[ Z_{ \geq k} \right] = n \mathbb{P}_{ \lambda } \big( | \mathcal{E} (v) | \geq k \big)
\end{align}
We evaluate $\mathbb{P}_{ \lambda } \big( | \mathcal{E} (v) | \geq k \big)$ using the bound of the theorem above. Moreover, the Proposition below  
\begin{align}
\mathbb{P}_{\lambda} \big( | \mathcal{E} (v) | \geq k_n \big) = \zeta_{\lambda} ( 1 + o(1)). 
\end{align}
In the second step, we use a variance estimate on $Z_{ \geq k }$ in Proposition 4.10, implying that $\forall \ \nu \in ( \frac{1}{2}, 1  )$, 
\begin{align}
| Z_{ \geq k} - \mathbb{E} \left[ Z_{ \geq k} \right] | \leq n^{\nu} 
\end{align}
In the third step, we show that for $k = k_n = K \text{log} n$ for some $K > 0$ sufficiently large and assuming that there is no connected component of size in between $k_n$ and $\alpha n$, for any $\alpha < \zeta_{\lambda}$. This is done by a first moment argument: the expected number of vertices in such connected components is equal to $\mathbb{E}_{\lambda} \left[ Z_{ \geq k} - Z_{ \geq \alpha n}  \right]$, and we use the bound in Proposition and another Proposition, which states that, for any $\alpha < \zeta_{ \lambda }$, there exists $J = J( \alpha ) > 0$ such that, for all $n$ sufficiently large, 
\begin{align}
\mathbb{P}_{ \lambda } \big( k_n \leq | \mathcal{E} (v) | < \alpha n   \big) \leq e^{- k_n J}.  
\end{align}
In the forth step, we prove that for $2 \alpha > \zeta_{\lambda}$, and when there are no clusters of size in between $k_n$ and $\alpha_n$, then it holds that $
Z_{ \geq k_n } = | \mathcal{E}_{ \text{max} } |$.
\end{proof} 

\newpage

\subsection{Elements of Large Deviations Principles}

Following the framework of \cite{gao2011delta}, let $C \left( \bar{B} ( \theta_0, \eta ) \right)$ denote the space of continuous $\mathbb{R}^d-$valued functions on $\bar{B} ( \theta_0, \eta )$ and define $\norm{ f } = \text{sup}_{ \theta \in  \bar{B} ( \theta_0, \eta ) } | f( \theta ) |$ for $f \in C \left( \bar{B} ( \theta_0, \eta ) \right)$. Let $\Psi_0 ( \theta)$ and $\Psi_{ 0n }$ be the restrictions of $\Psi$ and $\Psi_n$ on $\bar{B} ( \theta_0, \eta )$. Moreover, we have that $\left\{ a(n), n \geq 1 \right\}$ satisfies $a(n) \to \infty$ and $\frac{ a(n) }{ \sqrt{n }  } \to 0$ and $\left\{ \psi( X_i , \theta ), i \geq 1 \right\}$ satisfies  
\begin{align}
\frac{ \sqrt{n} }{ a(n) } \underset{ \theta \in \bar{B} ( \theta_0, \eta ) }{ \text{sup} } \left| \Psi_n(\theta) - \Psi (\theta)  \right| \to 0
\end{align}
and
\begin{align}
\underset{ n \to \infty }{ \text{lim sup} } \frac{1}{ a^2 (n) } \ \text{log} \left( n \mathbb{P} \left(  \underset{ \theta \in \bar{B} ( \theta_0, \eta ) }{ \text{sup} }  \left| \psi(X, \theta) \right| \geq \sqrt{n} a(n) \right) \right) = - \infty. 
\end{align}
Let $Y$ be a random variable taking its values in a Banach space and $\mathbb{E} (Y) = 0$. If there exists a sequence of increasing nonnegative functions $\left\{ H_k , k \geq 1 \right\}$ on $( 0, + \infty )$ satisfying 
\begin{align}
\underset{ u \to \infty }{ \text{lim} } u^{-2} H_k (u) = + \infty, \ \ \underset{ k \to \infty }{ \text{lim} } \underset{ n \to \infty }{ \text{lim} } \frac{1}{a^2 (n)} \text{log} \frac{H_k( \sqrt{n} a(n))}{n} = + \infty, 
\end{align}
and
\begin{align}
\mathbb{E} \left( H_k ( \norm{Y} ) \right) < \infty, \ \ \text{for any} \ \ k \geq 1,
\end{align}
then
\begin{align}
\underset{ n \to \infty }{ \text{lim sup} } \frac{1}{ a^2 (n) } \text{log} \bigg( n \mathbb{P} \big( \norm{Y} \geq \sqrt{n} a(n) \big) \bigg) = - \infty. 
\end{align}

\medskip

Let $\left( \xi_k, \mathcal{F}_k \right)$ be a martingale difference sequence. The paper concerns the tail behvaiour of the quadratic form $S_n = \sum_{ k = 1}^n \sum_{ j=1}^{ k - 1} \beta_n^{ k - j} \xi_k \xi_j$, where $\beta_n \to 1$ as $n \to \infty$. Then, the main conclusions about $\mathbb{P} \left( \frac{1}{n} S_n > x_n\right)$, where $x_n \to \infty$ as $n \to \infty$, are obtained using the tail behaviour of a martingale with values in a certain Hilbert space. 

We set with   
\begin{align}
\label{def}
S_n = \sum_{ k = 1}^n \sum_{ j=1}^{ k - 1} \beta_n^{ k - j} \xi_k \xi_j,
\end{align}
where $\left\{ \beta_n \right\}$ is a sequence of real numbers such that $\beta_n \to 1$ as $n \to \infty$.

\newpage

In particular, this paper considers the tail behaviour of the quadratic form $S_n$, which strongly depends on a near the unity boundary behaviour of $\beta_n$. We assume that the sequence $\left\{ \xi_k \right\}$ satisfies
\begin{align}
\label{condition}
n^{-1} \sum_{ i = 1}^{ [nt] } \mathbb{E} \left[ \xi^2 \big| \mathcal{F}_{i-1} \right] \to t, \ \ \text{as} \ \ n \to \infty. 
\end{align}
for each $t \in [0,1]$ and 
\begin{align}
n^{-1} \sum_{ i = 1 }^n \mathbb{E} \left[ \xi^2_i \chi \left\{ | \xi_i | > \epsilon \sqrt{n} \right\} \big| \mathcal{F}_{i-1} \right]
\to 0, \ \ \text{as} \ n \to \infty
\end{align}
for each $\epsilon > 0$. Then, we have that 
\begin{align}
n^{-1} S_n \to_D \int_0^1 Y(t) dW(t), \ \ \text{as} \ \ n \to \infty. 
\end{align}
if $n \left( n - \beta_n^2 \right) \to \gamma \geq 0$, as $n \to \infty$. We have that $W(t), t \in [0,1]$ is a standard Wiener process where $Y(t)$ is an OU process defined by the Ito stochastic differential equation
\begin{align}
dY(t) = - \frac{\gamma}{2} Y(t) dt + dW(t). 
\end{align} 

Moreover, if condition \eqref{condition} is satisfied and $\mathbb{E} \left[ \xi_k^2 | \mathcal{F}_{k-1} \right] = 1$ almost surely for every $k = 1,2,...$, then from the martingale central limit theorem it follows that 
\begin{align}
\left( 1 - \beta_n^2 \right)^{1 / 2}  n^{- 1 / 2 } S_n \to \mathcal{N} (0,1), \ \text{as} \ n \to \infty.
\end{align}
We shall investigate the tail behaviour of the quadratic form $S_n$ assuming the limit theorems above hold. Particularly, it follows that 
\begin{align}
\frac{ \displaystyle \mathbb{P} \left( n^{-1} S_n > x  \right) }{ \displaystyle \mathbb{P} \left( \int_0^1 Y(t) dW(t) > x  \right) } \to 1, \ \text{as} \ n \to \infty, 
\end{align}
uniformly with respect to $x$ in the interval $0 \leq x \leq O(1)$. 

The aim of this section is the extension of the above convergence result to the uniform convergence over a wider range $0 \leq x \leq x_n$, where $x_n \to \infty$. Usually, these results refer to as limit theorems for large deviation probabilities. In particular, this topic is well investigated for sums of \textit{i.i.d} random variables, for sums of \textit{i.i.d} random elements taking values in certain Banach spaces, for martingales and semimartingales. Limit theorems for large deviation probabilities of quadratic form, where $\beta_n \to 1$, are motivated by investigations of nearly nonstationary autoregressive processes.

\newpage

\begin{theorem}
Let $S_n$ be defined as $\eqref{def}$, where $1 - \beta^2 \geq 0$, $n ( 1 - \beta_n^2) \to \gamma \geq 0$, as $n \to \infty$ and let $\left\{ \mathcal{F}_k \right\}-$mds and $\left\{ \xi_k \right\}$ satisfy the following conditions:

\begin{enumerate}

\item[(i)] there exists a finite constant $M > 0$ such that $| \xi_k | \leq M$ almost surely for every $k \geq 1$, 

\item[(ii)] $\mathbb{E} \left[ \xi_k^2 | \mathcal{F}_{k-1} \right] = 1$ almost surely for every $k \geq 1$.
\end{enumerate}
Then the following holds
\begin{align}
\underset{ 0 \leq x \leq x_n }{ \text{sup} } \left| \frac{ \displaystyle \mathbb{P} \left( n^{-1} S_n > x  \right) }{ \displaystyle \mathbb{P} \left( \int_0^1 Y(t) dW(t) > x  \right) } - 1 \right| &\to 0, \text{as} \ \ n \to \infty
\\
x_n . \text{max} \bigg\{ n^{- 1/ 3} : \left| \gamma - n \left( 1 - \beta_n^2 \right) \right|^{ 1 / 3} \bigg\} &\to \ \text{as} \ n \to \infty.
\end{align}
where for each $n$, $\left\{ \hat{\xi}_{nk}, k = 1,2,... \right\}$ is a martingale difference sequence in the space $L_2 \left([0,1], \mu \right)$ where $\mu( dt ) = \delta_1 (dt) + \gamma dt$, and $\delta_{\alpha}$ is the probability measure concentrated at the point $\alpha$. The second fact is a limit theorem for probabilities of large deviations for martingales taking values in Hilbert spaces. 
\end{theorem}

\section{Elements of Real Analysis}

\subsection{Convergence and Cauchy Sequences}

\begin{example}
Consider the space $C^0 \left( [0,1] \right)$ with the usual metric
\begin{align}
d (f,g) = \underset{ 0 \leq t \leq 1 }{ \text{max} } \left| f(t) - g(t)  \right|. 
\end{align}
A sequence $( f_n )_{ n \in \mathbb{N} }$ converges to $f_0$ in $C^0 \left( [0,1] \right)$ if, and only if, 
\begin{align}
\forall \ \epsilon > 0 \exists N \in \mathbb{N} \forall n \leq N \ \forall t \in [0,1]: \left| f_n(t) - f_0(t) \right| < \epsilon.  
\end{align}
\end{example}

\subsection{Polynomial and Power Series Operators}

We have seen that if the output function $y_t$ and the input function $x_t$ are connected by a linear time invariant filter, then we have the relation below
\begin{align}
y_t = \sum_{k=0}^{\infty } \phi_k x_{t-k}, \ \ \     \phi = \left\{ \phi_k : k = 0,1,2,...  \right\}
\end{align}

\newpage 

Lets suppose that we introduce the \textit{lag operator} $L$, defined by $
L^k x_t = x_{t-k}, \ \ \ k = 0,1,2,...$. Therefore, the model can be rewritten in the following form 
\begin{align}
y_t = \left(  \sum_{k=0}^{\infty } \phi_k x_{t-k}  \right) x_t = \Phi (L) x_t    
\end{align}
\begin{example}
Consider the polynomial operator given by 
\begin{align}
a (L) = \sum_{j=0}^k a_j L^j    
\end{align}
which implies that if $y_t = a(L) x_t$ this gives us 
\begin{align}
y_t =  \sum_{j=0}^k a_j x_{t-j} \equiv a_0 x_t + a_1 x_{t-1} + \alpha_2 x_{t-2} + ... + a_k x_{t-k},    
\end{align}
\end{example}

\begin{example}
The limiting form becomes as below
\begin{align}
x_t = \sum_{j=0}^{\infty} \lambda^j \varepsilon_{t-j} = \frac{1}{ 1 - \lambda L } \varepsilon_t    
\end{align}

\end{example}

\begin{example}
Consider the following time series representation 
\begin{align}
x_t = \sum_{j=0}^{k-1} \psi_j u_{t-j}, \ \ \ t = 1,2,...    
\end{align}
where the error terms $u_j$, have the following properties
\begin{align}
\mathbb{E} \left( u_t \right) = 0 \ \ \ \ \text{and} \ \ \ \mathbb{E} \left( u_t u_t^{\prime} \right) = \delta_{ t t^{\prime}} \sigma^2 \ \ \ \forall \ \ k 
\end{align}
Moreover, we can assume that $u_t \sim \mathcal{N} \left( 0, \sigma^2 \right)$ for all $t$, then indeed the process represents a strictly stationary process. Then, its covariance kernel is given by the following expression 
\begin{align*}
K(s,t) := \mathsf{Cov} (X_s, X_t ) \equiv \mathbb{E} \left[ X_s, X_t \right] 
&= \mathbb{E}  \left[ \sum_{i=0}^{k-1} \sum_{j=0}^{k-1} \psi_i \psi_j u_{t-i} u_{t-j} \right]  
\\
&= 
\sum_{i=0}^{k-1} \sum_{j=0}^{k-1} \psi_i \psi_j \delta_{i-i, t-j} \sigma^2
=
\sigma^2 \sum_{i=0}^{ k - |s-t| - 1} \psi_i \psi_j.
\end{align*}
\end{example}

\newpage

\subsection{Dynamic Multipliers Models}

\begin{example}
Consider the general structural model given by 
\begin{align}
y_{1t} &= \beta_{21} y_{2t} + \gamma_{11} y_{1t-1} + \gamma_{011}^{\star} x_{1t} + \gamma_{121}^{\star} x_{2t-1} + u_{1t}
\\
y_{2t} &= \beta_{12} y_{1t} + \gamma_{12} y_{1t-1} + \gamma_{22} y_{2t-1} + \gamma^{\star}_{022} x_{2t} + \gamma_{112}^{\star} x_{1t-1} + u_{2t}
\end{align}
Therefore, the system can be written in the following form 
\begin{align}
\begin{bmatrix}
 I - \gamma_{11} L & - \beta_{21} I 
 \\
 - \beta_{21} I - \gamma_{12} L & I - \gamma_{22} L
\end{bmatrix}  
\begin{pmatrix}
y_{1t}
\\
y_{2t}
\end{pmatrix}
= 
\begin{bmatrix}
\gamma_{011}^{\star} I  &  \gamma^{\star}_{121} L
\\
\gamma^{\star}_{122} L  & \gamma^{\star}_{022} I
\end{bmatrix}  
\begin{pmatrix}
x_{1t}
\\
x_{2t}
\end{pmatrix}
+ 
\begin{pmatrix}
u_{1t}
\\
u_{2t}
\end{pmatrix}
\end{align}    
\end{example}

\begin{remark}
If we change the value of the $j-$th exogenous variable from $\bar{w}_j$ to $\bar{w}_{j} + 1$ and maintain the new level forever, all other exogenous variables remaining fixed, then the ultimate impact of this on the expectation of the $i-$th jointly dependent variable is given by the $(i,j)-$th element of $[ C(1) ]^{-1} C^{\star} (1)$. Therefore, we may call the $(i,j)$ element of the latter the \textit{dynamic multiplier} of the $j-$th exogenous variable relative to the $i-$th endogenous variable. Hence, $[ C(1) ]^{-1} C^{\star} (1)$, is said to be the matrix of \textbf{dynamic multipliers}.        
\end{remark}

\begin{example}
Further examples can be found in \cite{hatanaka1996time}. Suppose that $\{\Delta x_t \}$ is a linear process with zero mean, that is, 
\begin{align}
\label{initial}
\Delta x_t = b_0 \epsilon_t + b_1 \epsilon_{t-1} + ... \ \ \ \ \ \ b_0 = 1,
\end{align}
where $\{ \epsilon_t \}$ is i.i.d., with $E( \epsilon_t ) = 0$ and $E( \epsilon^2_t ) = \sigma^2_{ \epsilon }$.
Suppose that $\{ x_t \}$ is a stationary linear process such that $x_t = c_0 \epsilon_t + c_1 \epsilon_{t-1} + ...$.
Then, 
\begin{align}
\Delta x_t = c_0 \epsilon_t + ( c_1 - c_0 ) \epsilon_{t-1} + ( c_2 - c_1 ) \epsilon_{t-2} + ...
\end{align}
which implies that $b_0 = c_0$, $b_1= c_1 - c_0$, which it follows that $\sum_{j=0}^{\infty} b_j = 0$. Thus, the long-run component of $\Delta x_t$ vanishes. In fact, being a stationary linear process implies that no part of $x_t$ has a permanent impact upon its future. We can use a lag operator such that $\Delta x_t  = b(L) \epsilon_t$, where $b(L) \equiv b_0 + b_1 L + b_2 L^2 + ...$. Thus, it also follows that (i.e., by substitution of L with 1),  $b(1) = \sum_{j=1}^{\infty} b_j$. Therefore, the long-run component of $ \Delta x_t $ is written as $b(1) \epsilon_t$. Furthermore, it follows from
\begin{align}
b(L) - b(1) = -(1 - L ) \left[ b_1 + b_2 (1 + L) + b_3 (1 + L + L^2) + ...  \right]
\end{align}
Then, we get a useful identity
\begin{align}
\label{identity}
b(L) - b(1) = (1 - L) b^{*}(L), 
\end{align}
where $b^{*}(L) = b_0^{*} + b_1^{*} L + b_2^{*} L + ...$, $b_0^{*} = - \sum_{j=1}^{\infty} b_j$, $b_1^{*} = - \sum_{j=2}^{\infty} b_j$ ,....

\newpage 

Thus, using the identity \eqref{identity}, then \eqref{initial} becomes
\begin{align}
\Delta x_t \equiv b(L) \epsilon_t = b(1) + (1 - L) b^{*}(L) \epsilon_t.
\end{align}
which gives that
\begin{align}
x_t = x_0 + \sum_{s=1}^t \Delta x_s = b(1) \sum_{s=1}^t \epsilon_s + b^{*}(L) \epsilon_t + x_0 - b^{*}(L) \epsilon_0.
\end{align}
\end{example}

In terms of the econometric interpretation of the above asymptotic results the mixed Gaussian assumption of the sample moments indicates the stochastic nature of the covariance which depends on the persistence properties of the regressors included in the model. However, the covariance matrix is free of nuisance parameters which makes inference based on the cointegrated predictive regression system, for example in the form of linear restriction testing feasible. Below, we present an example to demonstrate the definition of the mixed normal distribution (Example 13.1.2 derived from \cite{hatanaka1996time}).                                                                                            

\begin{example}
Consider the cointegrated regression with an uncorrelated error 
\begin{align}
y_t = \beta x_t + \epsilon_{2t}
\end{align} 
where $x_t \equiv v_{1t} = \sum_{j=1}^t \epsilon_{1t}$ and the cointegrating vector $(-\beta, 1)$ to correspond to the bivariate process $(x_t, y_t)$. Set $\left( W_1(r), W_2(r) \right)^{\prime}$ with the covariance matrix $\text{diag} \left\{ \sigma_{11},  \sigma_{22} \right\}$ and setting $\sigma_{12} = 0$ then we have the following asymptotic result
\begin{align}
n \left( \hat{\beta} - \beta \right) \overset{ d }{  \to  } \left( \int W_1(r) d W_2(r) \right) \left( \int W_1(r)^2 dr  \right)^{-1} \equiv \xi
\end{align} 
Since $W_1(r)$ and $W_2(r)$ are independent standard Wiener processes then further simplification follow directly from the property of independent Gaussian processes. Define 
\begin{align}
g \equiv \left( \int W_1(r)^2 dr \right)^{-1}
\end{align} 
Then, it can be proved that $\xi | g \sim N(0 , \sigma_{22} g)$ with a conditional p.d.f given by 
\begin{align}
\left( 2 \pi \sigma_{22} g \right)^{-1/2} \text{exp} \left\{ - \frac{1}{2 \sigma_{22} g} \xi^2    \right\} f(g) dg 
\end{align}
where $f(g)$ is the marginal p.d.f. of g and the marginal p.d.f of $\xi$ is given by
\begin{align}
\int_0^{\infty} \left( 2 \pi \sigma_{22} g \right)^{-1/2} \text{exp} \left\{ - \frac{1}{2 \sigma_{22} g} \xi^2    \right\} f(g) dg 
\end{align}
The above distribution is exactly the \textit{mixed Gaussian} and represents a mixture of zero mean normal distributions with different variances.  
\end{example}

\newpage

\section{Elements of Weak Convergence}

\subsection{Joint Weak Convergence in $J_1$ topology}

In this section we explain in more details the implications of having a random limit distribution in the related theory for bootstrapping. In other words, the theoretical result which we aim to prove involves inference based on stochastic limit bootstrap measures\footnote{Notice that our approach in this paper is different from examining the conditional versus unconditional validity of the bootstrap for predictive regression models. }. To begin with, we assume that the standard conditional weak convergence result applies, $\underset{ x \in \mathbb{R} }{ \text{sup} } \big| F_n^{*} (x) - F(x) \big| \to_p 0$. Furthermore, denote with $\mathbb{E}^{*}$ the expectation under the probability measure induced by the standard bootstrap.
\begin{theorem}
Suppose that $\left\{ X_{n,j}, \mathcal{F}_{n,j} \right\}$ is a martingale difference array. Let $\left\{ \mathcal{J}_n(r),  r \in [0,1] \right\}$ be a sequence of adapted time scales and $\left\{ \mathcal{J}(r), r \in [0,1] \right\}$ a continuous, nonrandom function. If it holds, 
\begin{align}
\forall \epsilon > 0, \ \sum_{j=1}^{ \mathcal{J}(1) } \mathbb{E} \left( X_{n,j}^2 \boldsymbol{1} \left\{ | X_{n,j} | > \epsilon  \right\} \big| \mathcal{F}_{n, j-1} \right) &\to_p 0, \ \ \text{as} \ \ n \to \infty
\\
\sum_{j = 1}^{ \mathcal{J}_n(r) } \mathbb{E} \left( X^2_{n,j} \big| \mathcal{F}_{n,j-1} \right) &\to_p  \mathcal{J}(r), \ \ \text{as} \ \ n \to \infty, r \in [0,1],
\end{align}
Then
\begin{align}
\sum_{j = 1}^{ \mathcal{J}_n(r) } X_{n,j} \to_d W \left(  \mathcal{J}(r)   \right), \ \ \text{as} \ \ n \to \infty, \ \ \text{in} \ \ \mathcal{D}[0,1].
\end{align}
\end{theorem}
A sequence $\left\{ P, P_n \right\}$ of probability measures on the metric space $( S, d )$ converges weakly, when 
\begin{align}
\int \phi(x) dP_n(x) \to \int \phi(x) dP(x), \ \ n \to \infty,  
\end{align}
holds true for all $\phi \in C_b ( S, \mathbb{R} )$. 
\begin{theorem}
(Continuous Mapping Theorem) Let $\left\{ X , X_n \right\}$ be a sequence of random elements taking values in some metric space $( S, d )$ equipped with the associated Borel $\sigma-$field. Assume that 
\begin{align}
X_n \Rightarrow X, \ \ \text{as} \ \ n \to \infty, 
\end{align}  
\end{theorem}
If $\phi: S \to S^{\prime}$, is a mapping into another metric space $S^{\prime}$ with metric $d^{\prime}$ that is \textit{almost surely}, continuous on $X ( \Omega ) \subset S$, then 
\begin{align}
\phi ( X_n ) \overset{ d }{ \to } \phi(X), \ \ \text{as} \ \ n \to \infty,
\end{align} 

\newpage

\begin{theorem}
(Joint Weak Convergence) Let $\left\{ X, X_n \right\}$ and $\left\{ Y, Y_n \right\}$ be two sequences taking values in $( S_1, d_1 )$, respectively, $( S_2, d_2 )$, such that some conditions hold. Then,$
( X_n, Y_n ) \Rightarrow ( X, Y)$, as $n \to \infty$, provided at least one of the following conditions is satisfied
\begin{enumerate}
\item[(i)] $Y = c \in S_2$ is a constant, that is, non-random. 

\item[(ii)] $X_n$ and $Y_n$ are independent for all $n$ as well as $X$ and $Y$ are independent. 
\end{enumerate}
\end{theorem}
Recall that we are given a metric space $(S, d)$ such that as $\mathcal{D} \left( [0,1], \mathbb{R} \right)$ equipped with the Borel $\sigma-$field and the associated set $\mathcal{P}(S)$ of probability measures. Then, the space $\mathcal{P}(S)$ can be metrized by the Prohorov metric $\pi$ and the convergence with respect to the Prohorov metric is the weak convergence,  
\begin{align}
P_n \Rightarrow P \ \ \text{if and only if} \ \phi \left( P_n, P \right) \to 0. 
\end{align}
In a metric space convergence can be characterized by a sequence criterion: A sequence converges if and only if any subsequence contains a further convergent subsequence. Therefore, the weak convergence of a sequence $\left\{ P_n, P \right\}$ of probability measures can be characterized in the following way: 
\begin{align}
P_n \Rightarrow P \iff \pi ( P_n, P ) \to 0, \ \ \text{as} \ n \to \infty 
\end{align}  
if and only if any subsequence $\left\{ P_{n_k} : k \geq 1 \right\}$
contains a further subsequence $\left\{ P_{n_{k^{\prime}} } : k \geq 1 \right\}$ such that $P_{n_{k^{\prime}} } \Rightarrow P \iff \pi (  P_{n_{k^{\prime}} }, P ) \to 0$, as $k \to \infty$, which is equivalent to $P_{n_{k^{\prime}} } \Rightarrow P$, as $k \to \infty$. 

Further, in any metric space a subset $A$ is relatively compact, that is, has a compact closure, if every subsequence $\left\{ P_n \right\} \subset A$ has a subsequence $\left\{ P_{n_{k^{\prime}} } : k \geq 1    \right\}$ with $P_{n_{k^{\prime}} } \to P^{\prime}$ as $k \to \infty$, where the limit $P^{\prime}$ is in the closure of $A$ is relatively compact. Applied to out setting this means: A subset $A \subset \mathcal{P} ( \mathcal{S} )$ of probability measures has compact closure $\bar{A}$ if and only if every sequence $\left\{ P_n \right\} \subset A$ has a subsequence $\left\{ P_{ n_k } \right\}$ with converges weakly to some $P^{\prime} \in \bar{A}$, such that, $P_{ n_k } \Rightarrow P^{\prime}$ as $k \to \infty$. 

Here the limit $P^{\prime}$ may depend on the subsequence. Therefore, the weak convergence $P_n \Rightarrow P$ as $n \to \infty$. First, one shows that $\left\{ P_n \right\}$ is relatively compact and then one verifies that all possible limits are equal to $P$. A theorem due to Prohorov allows us to relate the compact sets of $\mathcal{P} (S)$ to the compact sets of $S$. This is achieved by the concept of tightness. A subset $A \subset \mathcal{P} (S)$ of probability measures in $\mathcal{P} (S)$ is called tight if for all $\epsilon > 0$ there exists a compact subset $K_{\epsilon} \subset S$ such that 
\begin{align}
P \left( K_{ \epsilon } \right) > 1 - \epsilon, \ \ \text{for all} \ \ P \in A.  
\end{align}  

\subsubsection{Functional Central Limit Theorems}

\begin{theorem}
(Donsker \textit{i.i.d} case) Let $\xi_1, \xi_2...$ be a sequence of $\textit{i.i.d}$ random variables with $\mathbb{E} \left( \xi_1 \right) = 0$ and $\sigma^2 = \mathbb{E} \left( \xi_1^2 \right) < \infty$. 
Then, 
\begin{align}
\frac{1}{ \sqrt{T} } \sum_{ t= 1}^{ \floor{ T s } } \xi_t \Rightarrow \sigma B (s), \ \ \text{as} \ T \to \infty  
\end{align}

\newpage

where $B$ denotes the standard Brownian motion and $\Rightarrow$ signifies weak convergence in the Skorohod space $\mathcal{D} \left( [0,1], \mathbb{R} \right)$. 
\end{theorem}

\begin{theorem}
Suppose $\xi_1, \xi_2,...$ satisfies a weak invariance principle. Then, 
\begin{align}
S_T(u) = \frac{1}{ \sqrt{T} } S \left( \floor{Tu} \right) \Rightarrow B(u), \ \ \text{as} \ T \to \infty.
\end{align}
\end{theorem}

\begin{proof}
We have that $\left\{ B(u): u \geq 0 \right\}$ is equal in distribution to $\left\{ \frac{1}{ \sqrt{T} } B(Tu) : u \geq 0    \right\}$ for each $T$. Therefore, 
\begin{align}
\underset{ u \in [0,1] }{ \mathsf{sup} } \left| S_T(u) - B(u)   \right| \overset{ d }{ = } \underset{ u \in [0,1] }{ \mathsf{sup} } \left| \frac{1}{ \sqrt{T} } \sum_{t=1}^{ \floor{Tu} } \xi_t - \frac{1}{ \sqrt{T} } B(Tu) \right|. 
\end{align}
Therefore, we can conclude that on a new probability space, 
\begin{align*}
\underset{ u \in [0,1] }{ \mathsf{sup} } \left| S_T(u) - B(u)   \right| 
&\overset{ d }{ = } 
\underset{ u \in [0,1] }{ \mathsf{sup} } \frac{1}{ \sqrt{T} } \left| \sum_{t=1}^{ \floor{Tu} } \xi_t -  B(Tu) \right| 
\\
&=
\frac{1}{ \sqrt{T} } \ \underset{ n \leq T }{ \mathsf{max} } \ \left|  S(n) - B(n) \right| \overset{ p }{ \to } 0
\end{align*}
as $T \to \infty$, which implies that $S_T \Rightarrow B$ as $T \to \infty$, for the original processes. 
\end{proof}

\subsection{Uniform Inference in the $J_1$ topology}

The FCLT allows us to obtain weak convergence of functionals to the $J_1$ topology which is used to derive the asymptotic distribution of estimators, test statistics, which are considered to be point-estimates. One might also be interested to obtain interval estimates such as the case when constructing confidence intervals. More specifically, uniform inference provides the theoretical framework to construct confidence intervals for the sum of coefficients in autoregressive models. In particular the paper of \cite{mikusheva2007uniform} provides a good overview on the subject. Within the proposed framework, the author clarify the difference between uniform and pointwise asymptotic approximations, and show that a pointwise convergence of coverage probabilities for all values of the parameter does not guarantee the validity of the confidence set (see, also \cite{phillips2014confidence}). 

\begin{example}
Assume that we have an AR(1) process with an intercept
\begin{align}
y_t = c + x_t, \ \ x_t = \rho x_{t-1} + \epsilon_t, \ \ t = \{ 1,...,T \}, \ \ \ x_0 = 0.
\end{align}  
Consider that the autoregressive coefficient $\rho$ takes only values on the parameter space $\Theta \in (-1,1)$. The uniform inference framework provides a way for constructing a confidence set for the parameter $\rho$.

\newpage

\begin{definition}
A subset $C(Y)$ of the parameter space $\Theta$ is said to be a \textit{confidence set} at a confidence level 
\begin{align}
(1 - \alpha) \ \  \text{if} \ \ \ \ \underset{ \rho \in \Theta }{  \text{inf} } \ \mathbb{P} \big[ \rho \in C(Y) \big] \geq 1 - \alpha.
\end{align}
\end{definition}
\begin{definition}
A subset $C(Y)$ of the parameter space $\Theta$ is said to be an \textit{asymptotic confidence set} at a confidence level $1 - \alpha$ (or is said to have a uniform asymptotic convergence probability $1 - \alpha)$ if  
\begin{align}
\underset{ T \to \infty }{  \text{lim inf} } \ \underset{ \rho \in \Theta }{  \text{inf} } \ \mathbb{P} \big[ \rho \in C(Y) \big] \geq 1 - \alpha
\end{align}
\end{definition}
Therefore, the requirement of uniform convergence is much stronger than a requirement of pointwise convergence of coverage probabilities
\begin{align}
\label{uniform.conv}
\underset{ T \to \infty }{  \text{lim} } \ \mathbb{P} \big[ \rho \in C(Y) \big] \geq 1 - \alpha \ \ \ \text{for every} \ \rho \in \Theta.
\end{align}
Specifically, the uniform convergence implies that for every value of the parameter space and for any given accuracy, we can find a large enough sample size to provide the required accuracy at this value. More importantly, the condition given by \eqref{uniform.conv} does not necessarily ensures that there is a sample size that provides the required accuracy for all values of the parameter space. Therefore, a solution to the problem of poor coverage probabilities in certain regions of the parameter space, is to use the methodology of inverting tests.
\end{example}

\subsubsection{Testing Methodology}

Let $A( \rho_o )$ be an acceptance region of an asymptotic level-$\alpha$ test for testing the null hypothesis, $\mathbb{H}_0: \rho = \rho_0$. A set is constructed as a set of parameter values for which the corresponding simple hypothesis is accepted, such that $C(Y) = \{ \rho: Y \in A( \rho ) \}$. Let the testing procedure for a test of the hypothesis, $\mathbb{H}_0: \rho = \rho_0$ be based on a test statistic $\phi ( Y, T, \rho_0  )$, and critical values $c_1( T, \rho_0 )$ and $c_2( T, \rho_0 )$. Then, a set $C(Y)$ is defined as below
\begin{align}
\label{CS}
C(Y) = \bigg\{  \rho \in \Theta: c_1( T, \rho_0 ) \leq  \phi ( Y, T, \rho_0  )  \leq  c_2( T, \rho_0 ) \bigg\}.
\end{align}

\subsubsection{Class of Test Statistics}

Let $y_t^{\mu}$ be the demeaned process that corresponds to $y_t$, that is, $y_t^{\mu} = y_t - \frac{1}{T} \sum_{t=1}^T y_{t-1}$. We consider a class of test statistics based on a pair of statistics
\begin{align}
\big( S(T,\rho), R(T,\rho) \big) = \left( \frac{1}{\sqrt{g(T,\rho)}} \sum_{t=1}^T y_{t-1}^{\mu} \big( y_t - \rho y_{t-1} \big), \ \frac{1}{ g(T,\rho)} \sum_{t=1}^T \big( y_{t-1}^{\mu} \big)^2 \right) 
\end{align} 
where $g(T,\rho)$ is a normalization function.

\newpage

We define $g(T,\rho) = \mathbb{E} \left( \sum_{t=1}^T \big( y_{t-1}^{\mu} \big)^2  \right)$. Also note that the family of test statistics $(S,R)$ are invariant with respect to the values of $c$. More specifically, we consider a sequence of sets of possible values of the AR coefficient such that 
\begin{align}
\Theta_T = \left[ - \left( 1 + \frac{\theta}{T}  \right) , \left( 1 + \frac{\theta}{T}  \right)  \right]
\end{align}

We focus on the asymptotic approach of local-to-unity asymptotics which was developed by \cite{chan1987asymptotic} and \cite{phillips1987towards} among others. Specifically, under the assumption of local to unity autoregressive coefficient, that is, defined by a sequence $\rho_T = \left( 1 + c / T  \right)$, for some fixed $c \leq 0$, then as the sample size increases, we have the joint weakly convergence result below
\begin{align}
\label{weakly}
\left( \frac{1}{T} \sum_{t=1}^T x_{t-1} \epsilon_t, \ \frac{1}{T^2}  \sum_{t=1}^T x^2_{t-1} \right) \Rightarrow \left( \int_0^1 J_c(r) dW  (r), \ \int_0^1 J_c^2(r) dr \right)
\end{align}
where the process $J_c$ represents the OU process defined by $J_c(r) = \int_0^r e^{(r-s)c} dW(s)$.

\subsubsection{Validity of Stock's Method}

\cite{stock1991confidence} proposed to construct a confidence set for the largest autoregressive root by using local-to-unity asymptotic approximation. In particular, if the autoregressive coefficient is local to unity, that is, is defined by a sequence $\rho_T = \text{exp} \left\{ \frac{c}{T} \right\}$ for some fixed $c < 0$, then for $T \to \infty$, the weakly convergence result given by expression \eqref{weakly} holds.

Specifically, \cite{phillips1987towards} proved that
\begin{align}
\left( \sqrt{-2c} \int_0^1 J_c(r) dW(r), (-2c) \int_0^1 J^2_c(r) dr  \right) \Rightarrow \big( N(0,1), 1 \big) \ \ \text{as} c \to - \infty
\end{align}
Therefore, consider a pair of statistics
\begin{align}
\big( S^c , R^c \big) = \left( \frac{1}{\sqrt{g(c)}} \int_0^1 J^{\mu}_c(r) dW(r) , \ \frac{1}{ g(c)}  \int_0^1 \big( J^{\mu}_c(r) \big)^2 dr  \right) 
\end{align} 
where $J^{\mu}_c(r) = J_c(r) - \int_0^1 J_c(r) dr$ and $g(c) = \mathbb{E}\left[ \int_0^1 \big( J^{\mu}_c(r)  \big)^2 dr \right]$. Stock's methodology suggests constructing an asymptotic confidence set as defined by \eqref{CS} with $c_1 ( T, \rho)$ and  $c_2 ( T, \rho)$ being $\alpha /2$ and $1 - \alpha /2$ quantiles of the distribution of the statistic $\phi_1 \left( S^{c(T,\rho)}, R^{c(T,\rho)}, T, \rho \right)$. An advantage of Stock's method of constructing confidence intervals for the largest root of the autoregressive coefficient, is that the critical values depend on the one dimensional local parameter $c$ and can be tabulated for commonly used levels of confidence and commonly used statistics.

\newpage

\subsection{Sub-Gaussian processes}

\begin{example}
(M-dependent sequences) Let $X_n, n \in \mathbb{Z}$ be a stationary sequence with $\mathbb{E} X_n = 0$, $\mathbb{E} X^2_n < \infty$. Assume that $\sigma \big( \left\{ X_j, j \leq 0 \right\} \big)$ and  $\sigma \big( \left\{ X_j, j \geq M \right\} \big)$ are independent. 
\end{example}

Notice that Sub-Gaussian processes satisfy the increment bound $\norm{ X_s - X_t }_{ \psi_2 } \leq \sqrt{6} d(s,t)$. Therefore, the   general maximal inequality leads for sub-Gaussian processes to a bound in terms of an entropy integral. 

\begin{lemma}
Let $\left\{ X_t : t \in T \right\}$ be a separable sub-Gaussian process. Then, for every $\delta > 0$, 
\begin{align}
\mathbb{E} \ \underset{ d(s,t) \leq \delta }{ \text{sup} } | X_s - X_t | \leq K \int_0^{\delta} \sqrt{ \text{log} D (\epsilon, d)} d \epsilon,
\end{align}
for a constant $K$. In particular, for any $t_0$, 
\begin{align}
\mathbb{E} \ \underset{ t }{ \text{sup} } | X_t | \leq \mathbb{E} | X_{t_0} | + K \int_0^{\infty} \sqrt{ \text{log} D (\epsilon, d)} d \epsilon.
\end{align}
\end{lemma}
\begin{proof}
Apply the general maximal inequality with $\psi_2 (x) = e^{x^2} - 1$  and $\eta = \delta$. Since $\psi_2^{-1}(m) = \sqrt{ \log(1 + m)}$, we have that $\psi_2^{-1} \big( D^2 ( \delta, d ) \big) \leq \sqrt{2} \psi_2^{-1} \big( D ( \delta, d ) \big)$. Thus, the second term in the maximal inequality can first be replaced by $\sqrt{2} \delta \psi^{-1} \big( D ( \delta, d ) \big)$. Next by incorporated in the first at the cost of increasing the constant. We obtain, 
\begin{align}
\norm{ \underset{ d(s,t) \leq \delta }{ \text{sup} } | X_s - X_t | }_{\psi_2} \leq K \int_0^{\delta} \sqrt{ \text{log} \left( 1 + D(\epsilon, d) \right)} d \epsilon. 
\end{align} 
\end{proof}

\paragraph{Symmetrization}

Let $\epsilon_1,...,\epsilon_n$ be \textit{i.i.d} Rademacher random variables. Furthermore, instead of the empirical process
\begin{align}
f \mapsto ( P_n - P ) f = \frac{1}{n} \sum_{i=1}^n \left( f(X_i) - Pf \right), 
\end{align}
consider the symmetrized process 
\begin{align}
f \mapsto P_n^{o} f = \frac{1}{n} \sum_{i=1}^n \epsilon_i f(X_i),
\end{align}
where $\epsilon_1,...,\epsilon_n$ are independent of $( X_1,..., X_n )$. Both processes have mean function zero.

\newpage

\begin{lemma}
For every nondecreasing, convex $\Phi: \mathbb{R} \to \mathbb{R}$ and class of measurable function $\mathcal{F}$
\begin{align}
\mathbb{E}^{*} \Phi \big( \norm{ P_n - P }_{ \mathcal{F}} \big) \leq \mathbb{E}^{*} \Phi \big( 2 \norm{ P_n^{0} }_{ \mathcal{F}}   \big)
\end{align}
\end{lemma}

\subsection{Clivenko-Cantelli theorems}

In this section, we prove two types of Clivenko-Cantelli theorems. The first theorem is the simplest and is based on entropy with bracketing. The second theorem, uses random $L_1-$entropy numbers and is proved through symmetrization followed by a maximal inequality.  

\begin{definition}
(Covering numbers) The covering numbers $N \left( \epsilon, \mathcal{F}, \norm{.} \right)$ is the minimal number of balls $\left\{ g : \norm{ g - f } < \epsilon \right\}$ of radius $\epsilon$ needed to cover the set $\mathcal{F}$. The entropy (without bracketing) is the logarithm of the covering numbers. 
\end{definition}

\begin{definition}
(bracketing numbers) Given two functions $l$ and $u$, the bracket $[l,u]$ is the set of all functions $f$ with $l \leq f \leq u$. An $\epsilon-$bracket is a bracket $[l,u]$ with $\norm{ u - l } < \epsilon$. Then, the bracketing number $N_{[ \ ]} \left( \epsilon, \mathcal{F}, \norm{.} \right)$ is the minimum number of $\epsilon-$brackets needed to cover $\mathcal{F}$. The entropy with bracketing is the logarithm of the bracketing number. 
\end{definition}

\begin{theorem}
Let $\mathcal{F}$ be a class of measurable functions such that $N_{[ \ ]} \left( \epsilon, \mathcal{F}, L_1 (P) \right) < \infty$ for every $\epsilon > 0$. Then, $\mathcal{F}$ is Glivenko-Cantelli. 
\end{theorem}

\begin{proof}
Fix $\epsilon > 0$. Choose finitely many $\epsilon-$brackets $[ \ell_i , u_i ]$ whose union contains $\mathcal{F}$ and such that $P( u_i - \ell_i ) < \epsilon$ for every $i$. Then, for every $f \in \mathcal{F}$, there is a bracket such that 
\begin{align}
( P_n - P ) f \leq ( P_n - P ) u_i + P( u_i - f  ) \leq ( P_n - P ) u_i + \epsilon
\end{align}
Consequently, 
\begin{align}
\underset{ f \in \mathcal{F} }{ \text{sup} } \left( P_n - P \right) \leq \underset{ i }{ \text{max} } (P_n - P ) u_i + \epsilon. 
\end{align}
The right side converges almost surely to $\epsilon$ by the strong law of large numbers for real variables. 
\end{proof}

\begin{theorem}
Let $\mathcal{F}$ be a $P-$measurable class of measurable functions with envelope F such that $P^{*} F < \infty$. Let $\mathcal{F}_M$ be the class of functions $f \mathbf{1} \left\{ F \leq M \right\}$ when $f$ ranges over $\mathcal{F}$. If log$N_{[ \ ]} \left( \epsilon, \mathcal{F}, L_1 (P_n) \right) = o_p^{*}(n)$ for every $\epsilon$ and $M > 0$, then $\norm{ P_n - P }_{ \mathcal{F} }^{*} \to 0$ both almost surely and in mean. In particular, $\mathcal{F}$ is GC.  
\end{theorem}

\begin{proof}
By the symmetrization lemma, measurability of the class $\mathcal{F}$, and Fubini's theorem 
\begin{align*}
\mathbb{E}^{*} \norm{ P_n - P }_{\mathcal{F}} 
\leq 
2 \mathbb{E}_{X} \mathbb{E}_{\epsilon} \norm{ \frac{1}{n} \sum_{i=1}^n \epsilon_i f(X_i) }_{ \mathcal{F} }
\leq 
2 \mathbb{E}_{X} \mathbb{E}_{\epsilon} \norm{ \frac{1}{n} \sum_{i=1}^n \epsilon_i f(X_i) }_{ \mathcal{F_M} } + 2 \mathbb{P}^{*} F \left\{ F > M  \right\}
\end{align*}

\newpage

by the triangle inequality, for every $M > 0$. For sufficiently large $M$, the last term is arbitrarily small. To prove convergence in mean, it suffices to show that the first term converges to zero for fixed $M$. Fix $X_1,...,X_n$. If $\mathcal{G}$ is an $\epsilon-$net in $L_1( P_n )$ over $\mathcal{F}_M$, then
\begin{align}
\mathbb{E}_{\epsilon} \norm{ \frac{1}{n} \sum_{i=1}^n \epsilon_i f(X_i)   }_{ \mathcal{F}_M } \leq  \mathbb{E}_{\epsilon} \norm{ \frac{1}{n} \sum_{i=1}^n \epsilon_i f(X_i)   }_{ \mathcal{G} } + \epsilon. 
\end{align}
The cardinality of $\mathcal{G}$ can be chosen equal to $N( \epsilon, \mathcal{F}_M, L_1(P_n) )$. Bound the $L_1-$norm on the right using the   Orlicz-norm for $\psi_2 (x) = \text{exp}( x^2 ) - 1$, and using the maximal inequality to find that the last expression does not exceed a multiple of 
\begin{align}
\sqrt{1 + \text{log} N( \epsilon, \mathcal{F}_M, L_1(P_n) ) } \ \underset{ f \in \mathcal{G} }{ \text{sup} } \norm{ \frac{1}{n} \sum_{i=1}^n \epsilon_i f(X_i) }_{ \psi_2 | X } + \epsilon,
\end{align}
where the Orlicz-norm $\norm{ . }_{ \psi_2 | X }$ are taken over $\epsilon_1,..., \epsilon_n$ with $X_1,...,X_n$ fixed.  By Hoeddding's inequality, then can be bounded by $\sqrt{6 / n} \left( P_n f^2 \right)^{1 / 2}$, which is less than $\sqrt{6 / n} M$. 
\end{proof}

\subsection{Donsker Theorems}

\paragraph{Uniform Entropy:} We establish the weak convergence of the empirical process under the condition that the envelope function $F$ be square integrable, combined with the uniform entropy bound
\begin{align}
\int_{0}^{\infty} \sqrt{ \text{log}  N \left( \epsilon, \mathcal{F}_{Q,2}, L_2(Q) \right) } d \epsilon < \infty.  
\end{align}
\begin{theorem}
Let $\mathcal{F}$ be a class of measurable functions that satisfies the uniform entropy bound. Let the class $\mathcal{F}_{\delta} = \left\{ f - g: f,g, \in \mathcal{F}, \norm{ f - g }_{P,2} < \delta \right\}$ and $\mathcal{F}^2_{\infty}$ be $P-$measurable for every $\delta > 0$. If $P^{*} F^2 < \infty$, then $\mathcal{F}$ is $P-$Donsker.  
\end{theorem}

\begin{proof}
Let $\delta_n \to 0$ be a fixed constant. Using Markov's inequality and the symmetrization lemma: 
\begin{align}
\mathbb{P}^{*} \left( \norm{ G_n }_{\mathcal{F}_{ \delta_n }} > x \right) \leq \frac{2}{x} \mathbb{E}^{*} \norm{ \frac{1}{ \sqrt{n} } \sum_{i=1}^n \epsilon_i f( X_i )  }_{ \mathcal{F}_{\delta_n} }. 
\end{align}
Therefore, we can see that the inner expectation is bounded as below
\begin{align}
\mathbb{E}_{\epsilon} \norm{ \frac{1}{ \sqrt{n} } \sum_{i=1}^n \epsilon_i f( X_i )  }_{ \mathcal{F}_{\delta_n} } 
\leq
\int_0^{\infty} \sqrt{ \text{log}  N \left( \epsilon, \mathcal{F}_{\delta_n }, L_2( P_n ) \right) } d \epsilon. 
\end{align}
Notice that for large values of $\epsilon$, the set $\mathcal{F}_{\delta_n }$ fits in a single ball of radius $\epsilon$ around the origin, in which case the integrand is zero.Furthermore, we have that the covering numbers of the class $\mathcal{F}_{\delta_n }$ are bounded by covering numbers of $\mathcal{F}_{\infty} = \left\{ f - g: f,g \in \mathcal{F} \right\}$.  
\end{proof}

\newpage

\section{Elements of Graph Limits Theory}

Following \cite{Lovasz2012Large} we present the following theory. 

\paragraph{The distance of two graphs}

There are many ways of defining the distance of two graphs $G$ and $G^{\prime}$. Suppose that the two graphs have a common node set $[n]$. 

\paragraph{Kernel and Graphons}

This correspondence with simple graphs suggests how to extend some basic quantities associated with graphs to kernels (or at least to graphons). Most important of these is the normalized degree function given by 
\begin{align}
d_w (x) = \int_0^1 W(x,y) dy.
\end{align}

If the graphon is associated with a simple graph $G$, this corresponds to the scaled degree $d_g(x) / v(x)$. Instead of the interval $[0,1]$, we can consider any probability space $\left( \Omega, \mathcal{A}, \pi  \right)$ with a symmetric measurable function $W: \Omega \times \Omega \to [0,1]$.

\paragraph{Generalizing Homomorphisms}

Homomorphism densities in graphs extend to homomorphism densities in graphos and, more generally in kernels. For every $W \in \mathcal{W}$ and multigraph $F = (V, E)$, we define as below
\begin{align}
t( F, W ) = \int_{[0,1]^V} \prod_{ij \in E} W \left( x_i , x_j \right) \prod_{ i \in V } dx_i 
\end{align}

We can think of the interval $[0,1]$ as the set of nodes, and of the value $W(x,y)$ as the weight of the edge $xy$. Then, the formula above is an infinite analogue of weighted homomorphism numbers. We obtain weighted graph homomorphisms as a special case when $W$ is a stepfunction: For every unweighted multigraph $F$ and weighted graph $G$, $t ( F, G ) = t( F, W_G )$.

In particular, we have that $t_x \left( K_2, W \right) = d_W (x)$. We can use the notation $t_{ \boldsymbol{x} }$, where $\mathbf{x} = \left( x_1,..., x_k \right)$. The product of the $k-$ labelled graphs $F_1$ and $F_2$ satisfies 
\begin{align}
t_{ \mathbf{x} } = \left( F_1, F_2, W \right) = t_{ \mathbf{x} } \left( F_1 , W \right) t_{ \mathbf{x} } \left( F_2 , W \right)
\end{align}

If $F^{\prime}$ arises from $F$ by unlabeling node $k$, then
\begin{align}
t_{ x_1,..., x_{k-1} } \left(  F^{\prime}, W \right) = \int_{[0,1]} t_{ x_1,..., x_k } \left( F, W \right) dx_k. 
\end{align}

\newpage

\paragraph{Kernel Operators}

Every function $W \in \mathcal{W}$ defines an operator $T_W : L_1 [0,1] \to L_{ \infty } [0,1]$, by 
\begin{align}
\left( T_W f \right) (x) = \int_0^1 W( x, y ) f(y) dy. 
\end{align}

For example, if we consider $T_W: L_2 [0,1] \to L_2 [0,1]$, this it is a Hilbert-Schmidt operator, and the rich theory of such operators can be applied. It is a compact operator, which has a discrete spectrum, that is, a countable multiset Spec$(W)$ of nonzero (real) eigenvalues $\left\{  \lambda_1, \lambda_2, ...  \right\}$ such that $\lambda_n \to 0$. In particular, every nonzero eigenvalue has finite multiplicity. Furthermore, it has a spectral decomposition
\begin{align}
W(x,y) \sim \sum_k \lambda_k f_k (x) f_k(y), 
\end{align}

where $f_k$ is the eigenfunction belonging to the eigenvalue $\lambda_k$ with $\norm{ f_k }_2 = 1$. The series on the right may not be almost everywhere convergent (only in $L_2$), but one has 
\begin{align}
\sum_{k=1}^{ \infty } \lambda_k^2 = \int_{ [0,1]^2 } W(x,y)^2 dx dy = \norm{ W }_2^2 \leq \norm{ W }_{ \infty }^2 
\end{align} 
A useful consequence of this bound is that if we order the $\lambda_i$ by decreasing absolute value: $| \lambda_1 | \geq  | \lambda_2 | \geq ...$, then
\begin{align}
| \lambda_k | \leq \frac{ \norm{ W }_2  }{ \sqrt{k} }
\end{align} 

It also follows that for every other kernel $U$ on the same probability space, the inner product can be computed from the spectral decomposition as below: 
\begin{align*}
\langle U, W \rangle &= \int_{ [0,1]^2 } U(x,y) W(x,y) dx dy = \sum_k \lambda_k \int_{ [0,1]^2 } U(x,y) W(x,y) dx dy
\\
&= \sum_k \lambda_k \langle f_k, U f_k \rangle
\end{align*}
The spectral decomposition is particularly useful if we need to express operator powers: The spectral decomposition of the $n-$th power is given by 
\begin{align}
W^{o n} (x, y) = \sum_k \lambda_k^n f_k(x) f_k(Y), 
\end{align}
and the series on the right hand side converges to the left hand side almost everywhere $n \geq 2$.

\newpage

\bibliographystyle{apalike}
\bibliography{myreferences1}

\addcontentsline{toc}{section}{References}

\newpage

\end{document}